\documentclass[11pt]{article}

\usepackage[T1]{fontenc}
\usepackage[margin=1in]{geometry}
\usepackage{amsmath,amssymb,amsthm,mathtools,bm,mathrsfs}
\usepackage{algorithm,algpseudocode}
\usepackage{booktabs,tabularx,array,multirow}
\usepackage{enumitem}
\usepackage{microtype}
\usepackage{graphicx}
\usepackage{xcolor}

\usepackage[colorlinks=true,
            linkcolor=blue!55!black,
            citecolor=blue!55!black,
            urlcolor=blue!60!black]{hyperref}
\usepackage[nameinlink,capitalise,noabbrev]{cleveref}
\usepackage{tikz}
\usetikzlibrary{arrows.meta}
\makeatletter
\renewcommand{\theHALG@line}{\thealgorithm.\arabic{ALG@line}}
\makeatother

\newtheorem{theorem}{Theorem}[section]

\newtheorem{proposition}[theorem]{Proposition}
\newtheorem{lemma}[theorem]{Lemma}
\newtheorem{fact}[theorem]{Fact}
\newtheorem{corollary}[theorem]{Corollary}

\theoremstyle{definition}

\theoremstyle{remark}
\newtheorem{remark}[theorem]{Remark}

\crefname{fact}{Fact}{Facts}
\crefname{construction}{Construction}{Constructions}
\Crefname{construction}{Construction}{Constructions}

\AddToHook{env/proposition/begin}{\crefalias{theorem}{proposition}}
\AddToHook{env/lemma/begin}{\crefalias{theorem}{lemma}}
\AddToHook{env/fact/begin}{\crefalias{theorem}{fact}}
\AddToHook{env/corollary/begin}{\crefalias{theorem}{corollary}}
\AddToHook{env/construction/begin}{\crefalias{theorem}{construction}}
\AddToHook{env/definition/begin}{\crefalias{theorem}{definition}}
\AddToHook{env/remark/begin}{\crefalias{theorem}{remark}}
\AddToHook{env/example/begin}{\crefalias{theorem}{example}}

\newcommand{\R}{\mathbb R}
\newcommand{\Z}{\mathbb Z}
\newcommand{\N}{\mathbb N}
\newcommand{\E}{\mathbb E}
\newcommand{\Id}{\mathsf I}
\newcommand{\ctrl}{\operatorname{ctrl}}
\newcommand{\one}{\mathbf 1}
\newcommand{\sP}{\mathcal P}
\newcommand{\cG}{\mathcal G}
\newcommand{\cB}{\mathcal B}
\newcommand{\vol}{\operatorname{vol}}
\newcommand{\Cov}{\operatorname{Cov}}
\newcommand{\Var}{\operatorname{Var}}
\newcommand{\Span}{\operatorname{span}}
\newcommand{\poly}{\operatorname{poly}}
\newcommand{\dtv}{d_{\mathrm{TV}}}
\newcommand{\dtr}{d_{\mathrm{tr}}}
\newcommand{\mix}{\mathrm{mix}}
\newcommand{\op}{\mathrm{op}}
\newcommand{\ket}[1]{\lvert #1\rangle}
\newcommand{\bra}[1]{\langle #1\rvert}

\newcommand{\norm}[1]{\left\lVert #1\right\rVert}

\title{Improved Upper and Lower Bounds for Quantum Convex-Body Volume Estimation}
\author{Ruizhe Zhang\thanks{\texttt{rzzhang@purdue.edu}}\\Purdue University}
\date{}

\begin{document}

  \maketitle

\maketitle

\begin{abstract}
Estimating the volume of a high-dimensional convex body is a fundamental problem in theoretical computer science. Given a quantum membership oracle for a convex body $K\subset\R^d$, we show that estimating $\vol(K)$ to relative error $\varepsilon$ takes $\widetilde O(d^{5/2}+d^{3/2}/\varepsilon)$ queries. This improves the previous quantum upper bound $\widetilde O(d^3+d^{9/4}/\varepsilon)$ of Chakrabarti et al. (ACM TQC 2023) and Cornelissen--Hamoudi (SODA 2023), and achieves a larger quantum speedup over the currently best classical upper bound $\widetilde O(d^{7/2}+d^3/\varepsilon^2)$ of Jia et al. (JACM 2026). Our quantum algorithm combines two new ingredients: an input-dependent effective spectral-gap analysis of quantum walks based on uniform classical warm-start mixing bounds, and a single quantum estimator for products of normalizer ratios in simulated annealing based on classical bridge sampling. We also prove an $\Omega(d)$ quantum query lower bound for constant relative error, improving the previous $\Omega(\sqrt d)$ lower bound.
\end{abstract}

{\hypersetup{linkcolor=black}
\tableofcontents
}

\section{Introduction}\label{sec:introduction}

Estimating the volume of a convex body $K\subset\R^d$ is one of the most important problems in randomized computation and convex geometry. This problem has broad applications in high-dimensional integration~\cite{ApplegateKannan91,Simonovits03,LovaszVempala06Fast}, statistics~\cite{Genz92,MengWong96,GelmanMeng98}, convex optimization~\cite{BertsimasVempala04,KalaiVempala06}, combinatorial counting~\cite{Stanley86,BrightwellWinkler91,DyerKannanMount97}, and quantum state tomography~\cite{BlumeKohout10,ShangEtAl13,ShangEtAl15}. We work in the membership-oracle model: the algorithm can access $K$ only through an oracle that tells whether a queried point lies in $K$, and it is given concentric balls $B(x_0,r)\subseteq K\subseteq B(x_0,R)$ whose radius ratio $R/r$ is at most polynomial in $d$. Given $0<\varepsilon<1/2$, the task is to output an estimate $\widehat V$ such that with probability at least $2/3$,
\begin{align*}
    (1-\varepsilon)\vol(K)\le\widehat V\le(1+\varepsilon)\vol(K),
\end{align*}
where $\vol(K)$ denotes the volume of the convex body $K$, and the cost is the number of membership queries.

Computing the volume exactly is $\#\mathrm{P}$-hard even for explicitly given polytopes~\cite{DyerFrieze88}, and in the membership-oracle model, every deterministic algorithm that makes polynomially many queries can only approximate the volume by a factor that is exponential in $d$ in the worst case~\cite{BaranyFuredi87}. A major breakthrough came with the work of Dyer, Frieze, and Kannan~\cite{DFK91}, who gave the first fully polynomial-time randomized approximation scheme (FPRAS) for this problem. Their algorithm computes the volume as a telescoping product of volume ratios along a sequence of nested bodies, the innermost of which has known volume, and estimates each ratio from the samples of a random walk. All subsequent algorithms follow this simulated annealing framework and improve its three components, namely the random walk sampler, the annealing schedule, and the affine change of coordinates, called \emph{rounding}, which keeps the intermediate distributions well-conditioned~\cite {LovaszSimonovits93,KLS97,LovaszVempala06Annealing,CousinsVempala18,JiaEtAl26}. For a well-rounded body, the Gaussian cooling algorithm of Cousins and Vempala~\cite{CousinsVempala18} uses an annealing sequence consisting of Gaussian distributions restricted to the convex body with increasing variance, which can estimate the volume to relative error $\varepsilon$ with $\widetilde O(d^3/\varepsilon^2)$ membership queries.\footnote{Throughout this paper, $\widetilde O$ suppresses polylogarithmic factors in $d$, $R/r$, and $1/\varepsilon$.} Jia, Laddha, Lee, and Vempala~\cite{JiaEtAl26} showed how to round a general body with $\widetilde O(d^{7/2}\psi_d^{-2})$ queries, where $\psi_d^{-1}$ is the Kannan--Lov\'asz--Simonovits (KLS) constant~\cite{KLS95} defined below. The KLS conjecture, which remains open, asserts that $\psi_d^{-1}$ is bounded above by a universal constant. The current best upper bound is $\psi_d^{-1}=O(\log^{1/4}d)$ by Letwin~\cite{Letwin26}.\footnote{Two days before completing this paper, we learned that Song and Zhang~\cite{SongZhang26KLS} had further improved the bound to $\psi_d^{-1}=O(4^{\log^*(d+2)})$.} Therefore, the rounding algorithm of~\cite{JiaEtAl26} uses $\widetilde O(d^{7/2})$ queries, and combining with the Gaussian cooling of~\cite{CousinsVempala18} gives the state-of-the-art classical query complexity for estimating the volume of a general convex body
\begin{equation}\label{eq:intro-classical}
    \widetilde O\!\left(d^{7/2}+\frac{d^3}{\varepsilon^2}\right).
\end{equation}
Kook and Vempala~\cite{KookVempala25} recently obtained the same bound for the more general problem of integrating a log-concave function using a sampler based on algorithmic diffusion.

Given quantum access to the membership oracle, that is, an oracle that can be queried in superposition, quantum algorithms can speed up both the sampling and the estimation in the annealing framework. Starting from an accessible coherent encoding of a stationary distribution, called a qsample~\cite{aharonov2003adiabatic}, quantum simulated annealing along a slowly varying sequence of Markov chains prepares subsequent qsamples with inverse-square-root dependence on the spectral gap, a quadratic improvement over the classical relaxation time~\cite{Szegedy04,Richter07,WocjanAbeyesinghe08}. For nonnegative random variables with bounded relative variance, quantum mean estimation based on amplitude estimation achieves relative error $\varepsilon$ with $\widetilde O(1/\varepsilon)$ uses of coherent sampling and evaluation circuits and their inverses, compared with $O(1/\varepsilon^2)$ classical samples~\cite{BrassardEtAl02,Montanaro15}. Building on these two primitives, quantum speedups of Markov chain Monte Carlo have been studied for two decades, mostly for discrete-space Gibbs partition function estimation~\cite{WCNA09,HarrowWei20,ArunachalamEtAl22,CornelissenHamoudi23}. Chakrabarti, Childs, Hung, Li, Wang, and Wu~\cite{CCHLW23} gave the first quantum algorithm for convex-body volume estimation with quantum query complexity $\widetilde{O}(d^3+d^{5/2}/\varepsilon)$, which quantizes the simulated annealing algorithm of Lov\'asz and Vempala~\cite{LovaszVempala06Annealing}. Cornelissen and Hamoudi~\cite{CornelissenHamoudi23} improved its product-estimation step with a nearly unbiased quantum mean estimator and claimed an expected query complexity of\footnote{As pointed out in~\cite{CornelissenApersGribling26}, the $d^3$ terms in these claimed bounds may not hold due to an issue discussed in~\cref{subsec:comparison-previous}.}
\begin{align*}
    \widetilde O\left(d^3+\frac{d^{9/4}}{\varepsilon}\right).
\end{align*}
Both algorithms, however, rely on the hit-and-run walk and its $\widetilde O(d^3)$ warm-start mixing time from two decades ago~\cite{LovaszVempala06Corner,LovaszVempala07}. Since then, progress on the KLS conjecture has translated directly into faster classical algorithms~\cite{JiaEtAl26}. The quantum counterparts, in contrast, have not benefited from this progress. This leads to the following question:
\begin{quote}
\emph{Does the progress on the KLS conjecture also yield more efficient quantum algorithms for convex-body volume estimation?}
\end{quote}

The question was raised explicitly by Cornelissen, Apers, and Gribling~\cite[footnote~1]{CornelissenApersGribling25}, and it is not resolved by the most recent hit-and-run walk analyses: neither the warm-start mixing time from the companion paper~\cite{Zhang26HitRun} nor the spectral gap lower bound in the concurrent work of Kook and Vempala~\cite{kook2026spectral} can be directly applied to improve the quantum algorithms of~\cite{CCHLW23,CornelissenHamoudi23}. We explain this in detail in \cref{subsec:technical-overview}.

\subsection{Main results}\label{subsec:main-result}
In this paper, we answer the question affirmatively. Our main result is a quantum volume algorithm with
\[
 \widetilde O\!\left(d^{5/2}+\frac{d^{3/2}}{\varepsilon}\right)
\]
membership queries (formally defined in~\cref{subsec:membership-oracle-model}):
\begin{align*}
    O_K:\ket{x}\ket{b}
 \longmapsto
 \ket{x}\ket{b\oplus\one_K(x)}.
\end{align*}
This improves the previous quantum upper bounds~\cite{CCHLW23,CornelissenHamoudi23} in both terms. In addition, we prove an $\Omega(d)$ quantum query complexity lower bound for volume estimation with constant accuracy. \Cref{tab:intro-volume-comparison} compares the classical and quantum results.

\begin{table}[t]
\centering
\begin{tabularx}{0.85\linewidth}{@{}l>{\centering\arraybackslash}Xc@{}}
\toprule
 & \textbf{References} & \textbf{Membership queries} \\ \midrule
Classical algorithm & \cite{JiaEtAl26,KookVempala25} & $\widetilde O(d^{7/2}+d^3/\varepsilon^2)$ \\ \midrule
\multirow{3}{*}{Quantum algorithm} & \cite{CCHLW23} & $\widetilde O(d^3+d^{5/2}/\varepsilon)$ \\
 & \cite{CornelissenHamoudi23} & $\widetilde O(d^3+d^{9/4}/\varepsilon)$ \\
 & {\bf This work} (\cref{thm:intro-volume_informal}) & $\widetilde O(d^{5/2}+d^{3/2}/\varepsilon)$ \\ \midrule
Classical lower bound & \cite{rademacher2008dispersion,vempala2026nearly} & $\widetilde{\Omega}(d^2)$ \\ \midrule
\multirow{2}{*}{Quantum lower bound} & \cite{CCHLW23} & $\Omega(\sqrt d)$ \\
 & {\bf This work} (\cref{thm:intro-volume-lower_informal}) & $\Omega(d)$ \\ \bottomrule
\end{tabularx}
\caption{Classical and quantum query bounds as stated in the cited works.}
\label{tab:intro-volume-comparison}
\end{table}

\paragraph{The algorithm.}
We first define the Kannan--Lov\'asz--Simonovits (KLS) constant~\cite{KLS95}. For a probability measure $\lambda$ on $\R^d$, its Cheeger constant $\psi_{\sf ch}(\lambda)$ is
\begin{align*}
    \psi_{\sf ch}(\lambda)=
 \inf_{A\subseteq \R^d:0<\lambda(A)<1}
 \frac{\lambda^+(A)}{\min\{\lambda(A),1-\lambda(A)\}},
\end{align*}
where $\lambda^+(A):=\liminf_{r\rightarrow 0+} \frac{\lambda(A+B(0,r))-\lambda(A)}{r}$ is the boundary measure of $A$. The reciprocal of the KLS constant is defined as
\begin{equation}\label{eq:intro-kls}
 \psi_n:=\inf\{\psi_{\sf ch}(\lambda):
        \lambda\text{ is isotropic and log-concave on }\R^n\},
\end{equation}
where isotropic means $\E_\lambda[X]=0$ and $\Cov_\lambda(X)=\Id$. In 2021, a breakthrough result of Chen~\cite{chen21} proved that $\psi_d^{-1}=d^{o(1)}$. Subsequent improvements~\cite{klartag2022bourgain,jambulapati2022slightly,Klartag23} led to  $\psi_d^{-1}=O(\log^{1/4}d)$ by Letwin~\cite{Letwin26} and $\psi_d^{-1}=O(4^{\log^*(d+2)})$ by Song and Zhang~\cite{SongZhang26KLS}.

The following is an informal statement of our main theorem.
\begin{theorem}[Main theorem, informal version of~\cref{thm:intro-volume}]\label{thm:intro-volume_informal}
Let $K\subset\R^d$ be a convex body given by a quantum membership oracle, along with an initial point $x_0$ and $0<r<R$ such that $B(x_0,r)\subseteq K\subseteq B(x_0,R)$ and $R/r$ is at most polynomial in $d$. For every $0<\varepsilon<1/2$, there is a quantum algorithm that computes $\widehat{V}$ satisfying
\begin{align*}
    (1-\varepsilon)\vol(K)\leq \widehat{V}\leq (1+\varepsilon)\vol(K)
\end{align*}
with probability at least $2/3$ and query complexity
\begin{equation}\label{eq:intro-volume_informal}
 \left(\frac{d^{5/2}}{\psi_d}+\frac{d^{3/2}}{\psi_d\varepsilon}\right)
 \log^{O(1)}\!\left(\frac{dR}{\varepsilon r}\right)
\end{equation}
\end{theorem}

The first term in \eqref{eq:intro-volume_informal} accounts for preparing the qsamples used for covariance estimation and rounding, and the second term is the cost of estimating the product of normalization-constant ratios. Using the aforementioned KLS bound $\psi_{d}^{-1}=O(\log^{1/4}d)$, the query complexity is $\widetilde O(d^{5/2}+d^{3/2}/\varepsilon)$. Compared with the classical bound~\eqref{eq:intro-classical}, this gives a quadratic improvement in the second term, and a $d$-factor improvement in the first term. These improvements are only for the number of membership oracle queries. The gate complexity is polynomial in the problem parameters and input bit length, but we do not optimize its exponent.

Our algorithm combines two main quantum ingredients. First, we develop an \emph{input-dependent spectral analysis} that constructs approximate stationary-state reflections from uniform classical warm-start mixing bounds (\cref{sec:tech_ov_reflections}). Under a fourth-moment condition on the input state, fast mixing implies that the state has little amplitude in the nonstationary eigenspaces with eigenvalues close to $1$. We can therefore leave these eigenspaces unresolved. If the walk mixes every $2$-warm start in $t$ steps to within total-variation distance chosen from the desired reflection tolerance and the input's fourth-moment norm, then the reflection uses $\widetilde O(\sqrt t)$ queries to the block-encoding of its discriminant. Our analysis also provides a corrected replacement for the effective spectral gap arguments in~\cite{CCHLW23}.

Second, we introduce a \emph{quantum estimator with streaming access to qsamples} for products of ratios of normalization constants in simulated annealing (\cref{sec:tech_ov_estimation}). It combines classical bridge sampling~\cite{Bennett76,MengWong96} and phase kickback to accumulate logarithmic contributions in a shared phase register, then estimates their signed sum using quantum phase estimation. For an annealing schedule with $m$ stages, it uses $\widetilde O(m/\varepsilon)$ qsample reflections, saving a factor of $\sqrt m$ over the product estimator of~\cite{CornelissenHamoudi23}. Streaming access means that our quantum estimator only uses one independent qsample at a time and then discards it. Given controlled qsample reflections and rotations for bounded observables whose expectations are bounded below by a positive constant, the estimator applies to general sequences of distributions, including those arising in other Gibbs partition-function problems.

In addition, we construct a finite-state-space Markov chain, the \emph{lattice Metropolis walk}, to implement these techniques for convex-body volume estimation (\cref{sec:tech_ov_walk}). The walk adapts the multiscale Metropolis kernel of the companion paper~\cite[Appendix~F]{Zhang26HitRun}. Intuitively, it can be viewed as a discrete ball walk whose step size is drawn at each step from a finite list of dyadic scales. When its preconditioner approximates the covariance within a constant factor, the lattice Metropolis walk mixes from every $2$-warm start to the discretized log-affine distributions in our annealing schedule in $\widetilde O(d^2\psi_d^{-2})$ steps. Unlike the discretized hit-and-run walk~\cite{CCHLW23}, this construction avoids the two-level discretization and the conditional sampling along a chord, making it particularly convenient for the quantum implementation. 

\paragraph{The lower bound.} Our second result shows that linearly many queries are necessary, even for ellipsoids with a known center.
\begin{theorem}[Linear quantum query lower bound, informal version of~\cref{thm:intro-volume-lower}]\label{thm:intro-volume-lower_informal}
Every quantum algorithm that estimates the volume of a convex body $K\subset\R^d$ to relative error $1/4$ with probability at least $2/3$ makes $\Omega(d)$ membership queries in the worst case, even when $K$ is promised to be an ellipsoid centered at the origin with $B(0,1/2)\subseteq K\subseteq B(0,O(d^2))$.
\end{theorem}

The previous lower bound for constant accuracy is $\Omega(\sqrt d)$~\cite{CCHLW23}. They also prove an $\Omega(1/\varepsilon)$ lower bound in the regime $1/d\le\varepsilon\le1/3$, so the factor $1/\varepsilon$ in the second term of our upper bound \eqref{eq:intro-volume_informal} is necessary in that regime. The restriction on $\varepsilon$ cannot be dropped: for fixed $d$ and $\varepsilon\rightarrow 0$, the quantum query complexity is $\widetilde\Theta(\varepsilon^{-(d-1)/(d+1)})$~\cite{CornelissenApersGribling25}.


\subsection{Technical overview for the algorithm}\label{subsec:technical-overview}

We first provide a high-level overview of our algorithm in~\cref{sec:tech_ov_algorithm}. We then analyze the reflections (\cref{sec:tech_ov_reflections}), the lattice Metropolis walk (\cref{sec:tech_ov_walk}), the product estimator (\cref{sec:tech_ov_estimation}), and the preparation of qsamples (\cref{sec:tech_ov_alter}). Finally, we combine the costs of these components and derive the total query complexity of our algorithm in~\cref{sec:tech_ov_total_cost}. In addition, we provide a detailed comparison of our techniques with the earlier and concurrent work  in~\cref{subsec:comparison-previous}.

\subsubsection{A bird's-eye view of our quantum algorithm}\label{sec:tech_ov_algorithm}

After rescaling, we may assume that $B(0,1)\subseteq K\subseteq B(0,\kappa)$. Our algorithm also uses the classical pencil construction~\cite{LovaszVempala06Annealing} for the annealing schedule, as in~\cite{CCHLW23}. It adds a height coordinate and defines
\[
 K_0=\{(x_0,x'):0\le x_0\le2\kappa,\ x'\in K,\ \|x'\|_2\le x_0\}\subset\R^{d+1}.
\]
The horizontal slice at height $x_0=t$ is $K\cap B(0,t)$. It is a ball near the tip and is exactly $K$ for $\kappa\le t\le2\kappa$ (\cref{fig:pencil}). The annealing distributions have density $\propto e^{-a_ix_0}{\bf 1}_{K_0}$ with decreasing parameters $a_0>\cdots>a_m>0$: a large $a_0$ concentrates mass near the cone at the bottom, while a small $a_m$ gives constant probability to the slab whose cross-section is $K$.

\begin{figure}
 \centering
 \includegraphics[width=0.95\linewidth]{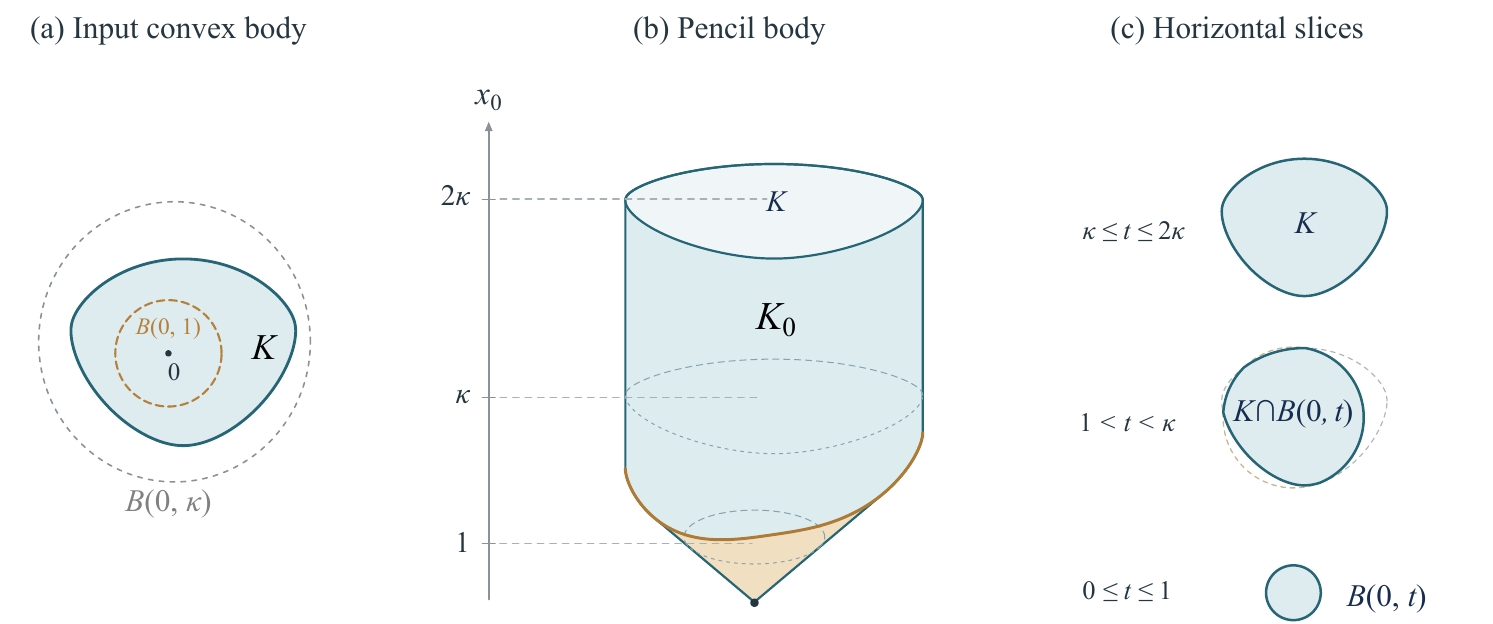}
 \caption{The pencil construction. (a) The input body $K\subset\R^2$ lies between concentric balls of radii $1$ and $\kappa$. (b) The pencil body $K_0\subset\R^3$. (c) Its section at height $t$ is $K\cap B(0,t)$, which equals $B(0,t)$ for $0\le t\le1$ and $K$ for $\kappa\le t\le2\kappa$.}
 \label{fig:pencil}
\end{figure}

The quantum circuit uses the corresponding finite distributions $\pi_i$ on the fixed lattice space $\Omega_\Lambda=K_0\cap h\Z^{d+1}$. Writing $Z_\Lambda(a)=h^{d+1}\sum_{z\in\Omega_\Lambda}e^{-az_0}$ for their normalization constants, we have
\begin{align}\label{eq:tech_ov_telescoping}
 \frac{Z_\Lambda(a_m)}{Z_\Lambda(a_0)}
 =\prod_{i=0}^{m-1}\mu_i,
 \qquad
 \mu_i=\E_{\pi_i}[X_i],
 \qquad
 X_i(z)=e^{(a_i-a_{i+1})z_0}.
\end{align}
This ratio product, together with two endpoint probabilities estimated by the algorithm and explicit normalization factors, determines the volume up to lattice-discretization error (\cref{subsec:pencil-geometry}).

We use a schedule with $m=\widetilde O(\sqrt d)$ stages such that the density ratios between adjacent distributions $\pi_i$ and $\pi_{i+1}$ have bounded second moments in both directions. These conditions give bounded relative variance for the ratio observables and constant overlap between adjacent qsamples (\cref{lem:cooling-moments}).

For a finite distribution $\pi_i$, define its qsample and the reflection about it by
\[
 \ket{\pi_i}=\sum_{z\in\Omega_\Lambda}\sqrt{\pi_i(z)}\ket z,
 \qquad R_{\pi_i}=2\ket{\pi_i}\!\bra{\pi_i}-\Id.
\]
Measuring the qsample produces a classical sample from $\pi_i$, and the quantum state collapses. We use these reflections to convert qsamples from $\ket{\pi_i}$ to $\ket{\pi_{i+1}}$ and to perform quantum estimation. We discuss below how to prepare and use them.

\paragraph{Preparing and advancing the qsamples.}
All copies of qsamples that will be consumed by the algorithm during the cooling schedule are prepared at the beginning in $\ket{\pi_0}$. We obtain them from the explicit cone distribution by amplitude amplification using the membership oracle for the pencil body $K_0$. At each transition, fixed-point state conversion advances the unused copies from $\ket{\pi_i}$ to $\ket{\pi_{i+1}}$ using reflections about the two states~\cite{YoderLowChuang14}. Since their overlap is constant, each qsample conversion needs only $\widetilde{O}(1)$ reflections. Two additional blocks of qsamples remain unmeasured throughout the whole process. They are used to coherently estimate the covariance, thereby preconditioning the random walk to implement these reflections. The two blocks are alternately evolved from $\ket{\pi_0}$ to $\ket{\pi_m}$, each serving as the covariance preconditioner for the other.

\paragraph{Processing the estimates at each stage.}
At stage $i$, we measure a reserved batch of copies of $\ket{\pi_i}$ and use classical sample averages to obtain a constant-factor estimate $c_i$ of $\mu_i$. These measured copies are discarded. Classical bridge sampling uses $c_i$ to express $\mu_i$ through two probabilities bounded below by a positive constant, one for $\pi_i$ and one for $\pi_{i+1}$. Using qsample reflections, we can block-encode a Hamiltonian whose eigenvalues approximate the logarithms of these probabilities. Their controlled evolutions use phase kickback to add the signed logarithmic contributions to a shared quantum register $\mathsf F$. The contribution for $\pi_i$ is processed at stage $i$ and the one for $\pi_{i+1}$ at the next stage, using the stored value $c_i$. Each contribution uses a separate qsample, which is discarded afterward. 

\paragraph{Final readout and query cost.}
After all stages, $\mathsf F$ contains the accumulated phase, including the contributions from the two endpoints. An inverse quantum Fourier transform and measurement estimate the combined logarithmic sum. The known values $c_i$ and the explicit endpoint normalization factors then give the volume estimate by classical computations. \Cref{fig:algorithm-workflow} shows how the quantum algorithm uses the qsamples.

\begin{figure}[!ht]
 \centering
 \begin{tikzpicture}[
 x=1cm,y=1cm,
 every node/.style={font=\normalsize,align=center},
 arrow/.style={-{Latex[length=1.8mm]},line width=.65pt},
 sample/.style={arrow,draw=black!60},
 wire/.style={draw=black!60,line width=.65pt},
 state/.style={circle,draw=black!60,fill=white,inner sep=0pt,minimum size=4pt,line width=.65pt},
 block/.style={rounded corners=1.5pt,line width=.65pt,minimum width=2.5cm,minimum height=.82cm,inner sep=4pt},
 consume/.style={block,draw=blue!45!black,fill=blue!3},
 measure/.style={block,draw=orange!65!black,fill=orange!5},
 classical/.style={arrow,densely dashed,draw=orange!65!black}]

 \node[anchor=west,font=\bfseries] at (0,-.0) {Purpose};
 \foreach \xx/\ii in {3.2/0,7.1/i,10.3/{i+1},14.3/m}{
  \node at (\xx,.8) {Stage $\ii$};
  \node[text=black!65] at (\xx,.24) {$\ket{\pi_{\ii}}$};
  \draw[black!15,densely dotted] (\xx,-.22)--(\xx,-6.92);
 }
 \node at (4.85,.8) {$\cdots$};
 \node at (12.25,.8) {$\cdots$};

 \node[anchor=west,align=left,text width=2.55cm] at (0,-1.1)
   {Estimate $c_i$\\$P$ copies};
 \node[anchor=west,align=left,text width=2.55cm] at (0,-2.35)
   {Estimate $p_i^+$\\$1$ copy};
 \node[anchor=west,align=left,text width=2.55cm] at (0,-3.6)
   {Estimate $p_i^-$\\$1$ copy};
 \node[anchor=west,align=left,text width=2.55cm] at (0,-4.85)
   {Estimate $q_{\rm out}$\\$1$ copy};

 \node[measure] (pilot) at (7.1,-1.1) {Measure batch};
 \node (ci) at (9.1,-1.1) {$c_i$};
 \draw[classical] (pilot.east)--(ci.west);
 \node[consume] (plus) at (7.1,-2.35) {Add phase\\[-1pt]{\small then discard}};
 \node[consume] (minus) at (10.3,-3.6) {Add phase\\[-1pt]{\small then discard}};
 \node[consume] (out) at (14.3,-4.85) {Add phase\\[-1pt]{\small then discard}};

 \foreach \yy in {-1.1,-2.35,-3.6,-4.85}{
  \node[state] at (3.2,\yy) {};
  \draw[wire] (3.28,\yy)--(4.32,\yy);
  \node[text=black!60] at (4.85,\yy) {$\cdots$};
 }
 \draw[sample] (5.38,-1.1)--(pilot.west);
 \draw[sample] (5.38,-2.35)--(plus.west);
 \foreach \yy in {-3.6,-4.85}{
  \draw[sample] (5.38,\yy)--(7.02,\yy);
  \node[state] at (7.1,\yy) {};
 }
 \draw[sample] (7.18,-3.6)--(minus.west);
 \draw[sample] (7.18,-4.85)--(10.22,-4.85);
 \node[state] at (10.3,-4.85) {};
 \draw[wire] (10.38,-4.85)--(11.72,-4.85);
 \node[text=black!60] at (12.25,-4.85) {$\cdots$};
 \draw[sample] (12.78,-4.85)--(out.west);

 \node[anchor=west,align=left,text width=3.5cm] at (-0.5,-6.3)
   {Estimate covariance\\~~~~$2\times L$ copies};
 \foreach \yy in {-6.1,-6.55}{
  \foreach \xx in {3.2,7.1,10.3,14.3}{
   \node[state,draw=green!35!black] at (\xx,\yy) {};
  }
  \draw[wire,draw=green!35!black] (3.28,\yy)--(4.32,\yy);
  \node[text=green!35!black] at (4.85,\yy) {$\cdots$};
  \draw[arrow,draw=green!35!black] (5.38,\yy)--(7.02,\yy);
  \draw[arrow,draw=green!35!black] (7.18,\yy)--(10.22,\yy);
  \draw[wire,draw=green!35!black] (10.38,\yy)--(11.72,\yy);
  \node[text=green!35!black] at (12.25,\yy) {$\cdots$};
  \draw[arrow,draw=green!35!black] (12.78,\yy)--(14.22,\yy);
 }
 \node[font=\small,text=green!35!black] at (8.75,-7.03)
   {Both blocks supply reflections and continue through every stage.};

 \draw[black!15] (0,-7.45)--(15.62,-7.45);
 \node[anchor=west,align=left,text width=2.55cm] at (0,-8.35)
   {Shared $\mathsf F$};
 \node[consume,minimum width=1.55cm] (init) at (3.2,-8.35) {Initialize};
 \node[consume,minimum width=2.5cm] (readout) at (14.3,-8.35)
   {Inverse QFT\\and measurement};
 \draw[arrow,draw=blue!45!black,line width=.85pt] (init.east)--(readout.west);
 \node[font=\small,text=blue!45!black] at (8.55,-7.96)
   {All \textsc{Add} calls act on this register};
 \end{tikzpicture}
 \caption{Qsample consumption along the annealing schedule. The first three rows are repeated for every stage $0\le i<m$: $P=O(\log(m+1))$ copies are measured for $c_i$, and two separate copies for the quantum product estimator. An \textsc{Add} call accumulates the corresponding signed logarithm in $\mathsf F$ and then discards the qsample. The covariance blocks are reserved for preparing qsamples and implementing reflections. The initial endpoint uses a separate cone qsample $\ket{\nu_\Lambda}$, omitted here.}
 \label{fig:algorithm-workflow}
\end{figure}
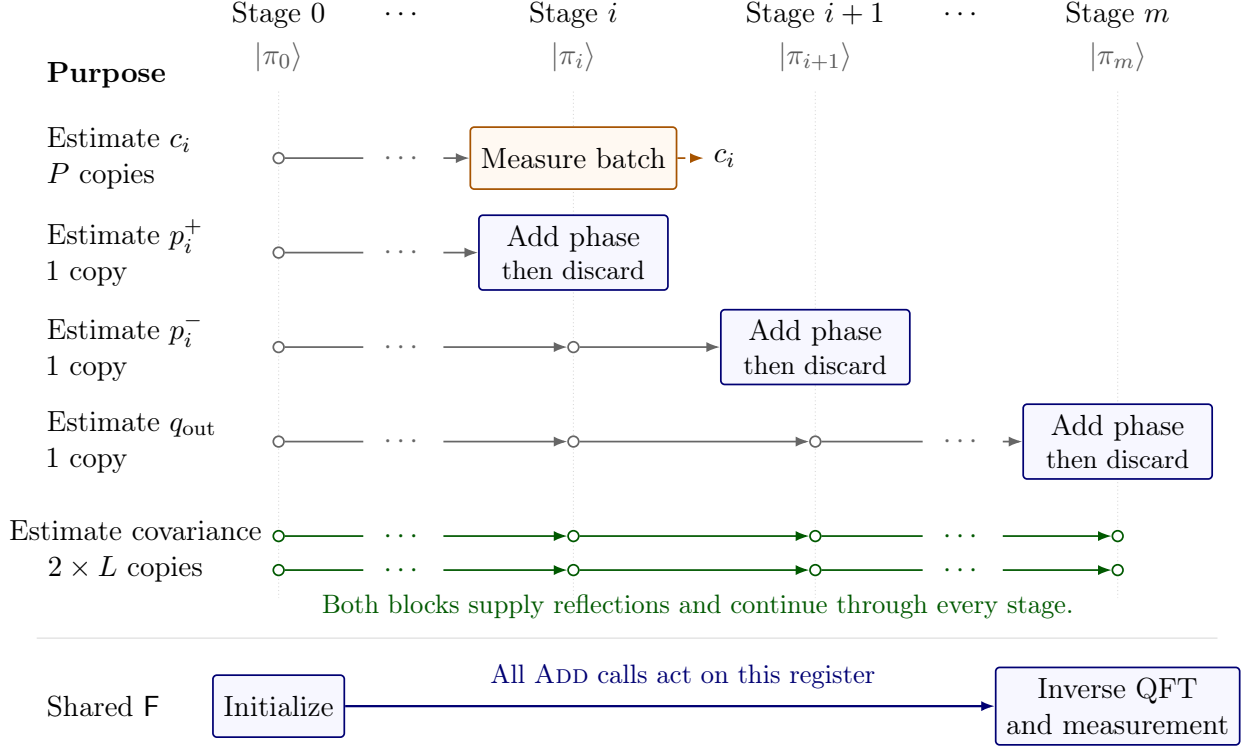

Each qsample reflection is implemented through the lattice Metropolis walk, and the dominant membership query cost is
\[
\begin{array}{c@{\qquad}c@{\qquad}c}
 \text{Number of qsample reflections}
 & \times
 & \text{Cost of one reflection}
 \\[6pt]
 \displaystyle \widetilde O\!\left(
 d^{3/2}+\frac{\sqrt d}{\varepsilon}\right)
 & 
 & \displaystyle \widetilde O\!\left(\frac{d}{\psi_d}\right)
\end{array}
\]

\subsubsection{Reflections from warm-start mixing}\label{sec:tech_ov_reflections}

A standard intuition about quantum walks~\cite{Szegedy04,WocjanAbeyesinghe08} is that a classical mixing time of $t$ suggests a qsample reflection cost of order $\sqrt t$. However, the worst-case mixing time is governed by the Markov chain's global spectral gap, which can be small even when the relevant inputs have little weight in the slow eigenspaces. Our new analysis uses warm-start mixing to bound that weight and resolves the rest of the spectrum to the required accuracy.

One challenge is the mismatch between classical distributions and qsamples. Classical evolution acts on the density $r=d\nu/d\pi$, whereas the qsample has relative amplitude $\sqrt r$. Taking the square root can remove cancellations against slow eigenfunctions, so a fast mixing bound for one specific initial distribution gives insufficient control of its qsample.\footnote{\Cref{app:counterexample} gives a counterexample to the general single-start effective spectral gap used in~\cite{CCHLW23}.} We instead assume mixing uniformly from every \emph{$2$-warm} law, meaning $\nu(z)\le2\pi(z)$ for every $z$.

The resulting reflection guarantee depends on how concentrated the quantum input is relative to $\pi$. For a normalized joint state on a sample register ${\sf S}$ and an arbitrary auxiliary register ${\sf A}$:
\begin{align*}
    \ket{\Psi}_{\sf SA}=\sum_{x}\ket{x}_{\sf S}\otimes \ket{\psi_x}_{\sf A},
\end{align*}
let $\mu_\Psi$ be the law obtained by measuring its sample register ${\sf S}$. The $L^4(\pi)$ norm of its relative amplitudes satisfies
\[
 \|\ket{\Psi}_{\sf SA}\|_{4,\pi}^4
 =\sum_z\frac{\mu_\Psi(z)^2}{\pi(z)}
 =1+\chi^2(\mu_\Psi\|\pi).
\]
Thus, this fourth-moment norm measures concentration through the relative second moment, and it applies equally to states entangled with auxiliary registers. The cooling condition gives a constant initial fourth-moment norm when converting between adjacent qsamples. The pointwise $2$-warm condition is a hypothesis on the Markov chain's initial distribution, while this fourth-moment bound is a condition on the quantum input. 

Uniform mixing connects the two conditions. Because the mixing time bound holds for every $2$-warm start, we can apply it to any bounded perturbation of $\pi$. Take a real function $g$ with $\E_\pi[g]=0$ and $|g(z)|\le1$. Then $\nu(z)=\pi(z)(1+g(z))$ is a $2$-warm probability distribution. Writing $\sP$ for the reversible Markov operator, the density of $\nu\sP^t$ relative to $\pi$ is $1+\sP^t g$. Thus, if the Markov chain starting from $\nu$ mixes within total-variation distance $\eta$ in $t$ steps, then $\|\sP^t g\|_1\le2\eta$, where all function norms below are taken with respect to $\pi$. 

Call an eigenfunction \emph{slow} if its eigenvalue lies in $[1-1/(2t),1)$, and let $g_{\rm slow}$ be the projection of $g$ onto these eigenfunctions. Each slow component shrinks by at most a factor of two in $t$ steps. Then, we have
\begin{align*}
     \frac12\|g_{\rm slow}\|_2^2
 \le\langle g,\sP^t g\rangle_\pi
 \le\|g\|_\infty\|\sP^t g\|_1
 \le2\eta.
\end{align*}
Hence, every such perturbation has a slow component of norm $O(\sqrt\eta)$. Subtracting the mean removes only the stationary component. After rescaling, the estimate becomes $\|f_{\rm slow}\|_2\leq O(\sqrt{\eta}\,\|f\|_\infty)$ for any bounded function $f$. For a state $\ket\Psi=\sum_z\sqrt{\pi(z)}f(z)\ket z$, this $f$ is its relative-amplitude function, and $\|\ket\Psi\|_{4,\pi}=\|f\|_4$.

We now use the fourth-moment condition to control large amplitudes. For a threshold $B>0$, split $f$ into the part where $|f|\le B$ and the remaining tail. The bounded part contributes at most $O(\sqrt\eta\,B)$ to the slow component. The tail has $L^2$ norm at most $\|f\|_4^2/B$, and projection cannot increase that norm. Adding the two contributions gives
\[
 \|f_{\rm slow}\|_2
 \leq O\left(\sqrt\eta\,B+\|f\|_4^2/B\right)\leq O(\eta^{1/4}\|f\|_4)
\]
by choosing $B=\eta^{-1/4}\|f\|_4$ to balance the two terms.

The same estimate extends to vector-valued relative amplitudes, allowing the sample register to be entangled with auxiliary registers. More precisely, let $D_\pi$ be the discriminant matrix of an ergodic reversible Markov chain with nonnegative spectrum. This symmetric matrix has the chain's spectrum and stationary eigenvector $\ket\pi$. If every $2$-warm start mixes to error $\eta$ in $t$ steps, our key estimate is
\begin{equation}\label{eq:intro-warm-projector}
 \|(\Pi_{\rm slow}\otimes\Id)\ket{\Psi}_{\sf SA}\|
 \leq O\left(\eta^{1/4}\|\ket{\Psi}_{\sf SA}\|_{4,\pi}\right),
 \qquad
 \Pi_{\rm slow}=\one_{[1-1/(2t),1)}(D_\pi).
\end{equation}

To implement the reflection, suppose we have access to a $1$-block encoding of the discriminant $D_\pi$. Quantum singular value transformation (QSVT)~\cite{GSLW19} gives a polynomial filter that fixes $\ket\pi$ and approximately negates eigenvectors with eigenvalues at most $1-1/(2t)$. It uses 
\begin{align}\label{eq:tech_ov_reflection_cost}
    O(\sqrt t\log(1/\zeta))
\end{align}
queries to the block-encoding and its inverse. Equation~\eqref{eq:intro-warm-projector} bounds the error from the remaining eigenspaces, giving total error
\[
 O\!\left(\eta^{1/4}\|\ket{\Psi}_{\sf SA}\|_{4,\pi}+\zeta\right)
\]
on the joint state, including fresh implementation workspace (\cref{thm:effective-reflection}). Thus, it suffices to ensure that the input states have polynomially bounded fourth-moment norm. The implementation cost for each reflection depends on the $2$-warm-start mixing time of the Markov chain and the cost of block encoding the discriminant. We next construct a walk with both rapid uniform warm-start mixing and constant query cost for its block-encoding.

\subsubsection{Preconditioning the walk on a fixed lattice}\label{sec:tech_ov_walk}

We use a lattice Metropolis walk on the same finite state space throughout the cooling schedule. Let $\Lambda=h\Z^{d+1}$. At stage $i$, the state space and target distribution are
\[
 \Omega_\Lambda=K_0\cap\Lambda,
 \qquad
 \pi_i(z)\propto e^{-a_i z_0}
 \qquad \forall z\in\Omega_\Lambda.
\]
The main idea is to precondition the proposed displacements while keeping this state space and target distribution fixed. In continuous space, affine preconditioning changes coordinates to make the target distribution well-conditioned. Discretizing each transformed body on a new cubic grid would generally produce a different finite target for each covariance estimate. Our construction instead guarantees that every covariance estimate leads to the same stationary qsample, which is essential when the estimate is computed coherently.

Let $A$ be a positive-definite covariance estimate, and let $T_A$ be an invertible approximation to $A^{1/2}$, as constructed in \cref{prop:finite-affine-branches}. From its current position $x\in\Omega_\Lambda$, the walk chooses a proposal through one of three branches:
\begin{itemize}
 \item \emph{Holding:} With probability $1/2$, set $y=x$.
 \item \emph{Local proposal:} With probability $3/8$, choose a radius $\delta$ from a list of $J$ radii spaced by factors of two, each with probability $\Theta(1/J)$. Draw $U$ uniformly from $B(0,\delta)$, apply $T_A$, and round $T_AU$ coordinatewise to a lattice displacement $v$. Set $y=x+\sigma v$, where $\sigma\in\{-1,1\}$ is an independent uniform sign.
 \item \emph{Connector proposal:} With probability $1/8$, choose a lattice displacement $v$ uniformly from a fixed box containing every difference between two points of $\Omega_\Lambda$. Set $y=x+\sigma v$, again using an independent uniform sign.
\end{itemize}
The independent sign makes the implemented displacement law exactly symmetric. Since this law is also independent of the current position, the Metropolis acceptance probability is
\[
 \alpha(x,y)=\one_{K_0}(y)\min\{1,e^{-a_i(y_0-x_0)}\}.
\]
The walk moves to $y$ if the proposal is accepted and stays at $x$ otherwise. This rule preserves $\pi_i$ for every allowed covariance estimate. The holding branch makes the chain lazy, and the connector branch makes it irreducible. \Cref{fig:lattice-metropolis} illustrates the one-step update of this random walk.

\begin{figure}[!ht]
 \centering
 \begin{tikzpicture}[
 x=1cm,y=1cm,
 every node/.style={font=\normalsize,align=center},
 proposal/.style={-{Latex[length=2mm]},draw=blue!55!black,line width=.9pt,shorten <=3pt,shorten >=2pt},
 rejected/.style={-{Latex[length=2mm]},draw=red!65!black,densely dashed,line width=.9pt,shorten <=3pt,shorten >=3pt},
 scale contour/.style={draw=teal!65!black,line width=.7pt},
 point label/.style={fill=white,inner sep=1.5pt},
 current/.style={circle,draw=white,line width=1pt,fill=blue!55!black,inner sep=0pt,minimum size=7pt},
 destination/.style={circle,draw=blue!55!black,fill=white,line width=.85pt,inner sep=0pt,minimum size=4.5pt}]

 \node[anchor=west,font=\bfseries] at (0,6.05) {(a) Possible moves in each step};
 \def\latticebody{(.55,2.25)--(1.35,4.05)--(3.6,4.95)--(5.85,4.5)--(6.75,2.7)--(5.85,.9)--(3.6,.45)--(1.35,.9)--cycle}
 \path[fill=teal!5,rounded corners=14pt] \latticebody;
 \draw[step=.45cm,black!12,line width=.2pt] (0,0) grid (7.2,5.4);
 \foreach \ix in {0,...,16}{
  \foreach \iy in {0,...,12}{
   \fill[black!18] (.45*\ix,.45*\iy) circle (.7pt);
  }
 }
 \begin{scope}
  \clip[rounded corners=14pt] \latticebody;
  \foreach \ix in {0,...,16}{
   \foreach \iy in {0,...,12}{
    \fill[black!55] (.45*\ix,.45*\iy) circle (1pt);
   }
  }
 \end{scope}
 \draw[draw=teal!55!black,line width=.9pt,rounded corners=14pt] \latticebody;
 \node[point label,text=teal!55!black] at (5.85,3.15) {$K_0$};

 \coordinate (x) at (2.25,2.25);
 \coordinate (near) at (3.15,2.7);
 \coordinate (far) at (4.05,4.05);
 \coordinate (low) at (5.4,1.8);
 \coordinate (outside) at (1.35,4.95);
 \draw[proposal] (x)--(near);
 \draw[proposal] (x)--(far);
 \draw[proposal] (x)--(low);
 \draw[rejected] (x)--(outside);
 \draw[-{Latex[length=1.8mm]},black!70,line width=.75pt,shorten <=3pt,shorten >=3pt]
   (x) .. controls (1.0,2.4) and (1.0,1.05) .. (x);
 \node[point label,font=\small,text=black!75] at (1.65,1.35) {stay};
 \node[point label,font=\small,text=red!65!black,rotate=18] at (1.95,3.25) {reject};
 \node[current] at (x) {};
 \foreach \name in {near,far,low}{\node[destination] at (\name) {};}
 \node[point label,anchor=north west,text=blue!55!black] at (2.28,2.09) {$x$};
 \node[point label,anchor=south,text=blue!55!black] at (4.05,4.17) {$y$};
 \draw[red!65!black,line width=1pt] (1.28,4.88)--(1.42,5.02) (1.28,5.02)--(1.42,4.88);
 \node[point label,anchor=south,text=red!65!black] at (1.35,5.09) {$y_{\rm out}\notin K_0$};


 \draw[black!50,line width=.45pt] (.45,-.06)--(.45,-.26) (.9,-.06)--(.9,-.26);
 \draw[{Latex[length=1mm]}-{Latex[length=1mm]},black!65,line width=.45pt] (.45,-.19)--(.9,-.19);
 \node[font=\small] at (.675,-.49) {$h$};
 \node at (4.0,-.63) {$\Omega_\Lambda=K_0\cap h\mathbb Z^{d+1}$};

 \begin{scope}[shift={(8.5,0)}]
  \node[anchor=west,font=\bfseries] at (0,6.05) {(b) Local proposals at several scales};
  \draw[step=.45cm,black!12,line width=.2pt] (0,0) grid (7.2,5.4);
  \foreach \ix in {0,...,16}{
   \foreach \iy in {0,...,12}{
    \fill[black!30] (.45*\ix,.45*\iy) circle (.8pt);
   }
  }
  \foreach \rr in {.675,1.35,2.7}{
   \draw[scale contour,rotate around={28:(3.15,2.7)}]
     (3.15,2.7) ellipse[x radius=\rr cm,y radius=.5*\rr cm];
  }
  \node[point label,text=teal!65!black] at (2.55,2.05) {$\delta$};
  \node[point label,text=teal!65!black] at (1.78,1.7) {$2\delta$};
  \node[point label,text=teal!65!black] at (.58,1.05) {$4\delta$};

  \coordinate (xb) at (3.15,2.7);
  \foreach \xx/\yy in {3.6/2.7,3.15/2.25,4.95/4.05}{
   \draw[proposal] (xb)--(\xx,\yy);
   \node[destination] at (\xx,\yy) {};
  }
  \node[current] at (xb) {};
  \node[point label,anchor=west,text=blue!55!black] at (5.1,3.65) {$y$};
  \node[font=\small] at (3.6,-.63) {Preconditioned balls};
 \end{scope}

 \end{tikzpicture}
 \caption{A two-dimensional illustration of the lattice Metropolis walk. (a) Dark grid points are feasible points in $\Omega_\Lambda$, and $x$ is the current position. Blue arrows show alternative feasible proposals. The dashed proposal leaves $K_0$ and is rejected. The loop represents both the holding step and rejected proposals. (b) Three consecutive scales of the local proposal centered at $x$, after preconditioning by $T_A$. The complete update is given in \cref{alg:finite-metropolis}.}
 \label{fig:lattice-metropolis}
\end{figure}

For the mixing analysis, let $\Sigma_i$ be the covariance of the continuous density proportional to $e^{-a_i z_0}$ on $K_0$, and fix an estimate $A\asymp\Sigma_i$. Write $W_A=T_A^{-1}$. In the coordinates $u=W_Ax$, the target has comparable variance in every direction:
\[
 W_A\Sigma_iW_A^\top\asymp\Id.
\]
The same change of coordinates maps the original lattice $\Lambda$ to $hW_A\Z^{d+1}$. For sufficiently close points outside a set of small stationary mass, the continuous analysis~\cite{Lovasz99,LovaszVempala06Corner,Zhang26HitRun} shows that their transition laws have substantial overlap and acceptance probabilities. The key step is to preserve these properties when rounding proposals onto the fixed lattice. Let $\mathcal R_h$ denote coordinatewise rounding to $\Lambda$. Since translation by a lattice point commutes with rounding, we have
\[
 x+\mathcal R_h(T_AU)
 =\mathcal R_h\!\left(T_A(W_Ax+U)\right),
 \qquad \forall x\in\Omega_\Lambda.
\]
Thus, the ideal rounded proposals from $x$ and $y$ are obtained by applying the same map $\mathcal R_h\circ T_A$ to uniform points in $B(W_Ax,\delta)$ and $B(W_Ay,\delta)$, respectively. Their common probability mass is sent to the same lattice destinations, so rounding cannot decrease their overlap. 

We also control the effect of rounding on Metropolis acceptance. Write $e=\mathcal R_h(T_AU)-T_AU$. Except for a small set of lattice points near the boundary, a sufficiently small lattice spacing $h$ ensures that $x+2e\in K_0$. And there exists a suitable proposal scale $\delta$ such that for most $U$, $x+2T_AU\in K_0$ and has density weight at least a constant fraction of the weight at $x$. Thus, by the convexity of $K_0$, we get that
\[
 x+\mathcal R_h(T_AU)
 =\tfrac12(x+2T_AU)+\tfrac12(x+2e)\in K_0.
\]
For a feasible proposal $y=x+\mathcal R_h(T_AU)$, the target-weight ratio factors as
\begin{align*}
    e^{-a_i(y_0-x_0)}=e^{-a_i(T_A U)_0} \cdot e^{-a_ie_0},
\end{align*}
where the chosen scale bounds the first factor from below by a constant for most proposals, and a sufficiently small $h$ keeps the second factor close to $1$. For feasible $y$, the acceptance probability is $\min\{1,e^{-a_i(y_0-x_0)}\}$. Applying these estimates at both starting points shows that a constant fraction of their common proposal mass is preserved after the Metropolis step.

For a requested mixing error $\eta$, we choose one sufficiently fine lattice that works for every cooling stage and every covariance estimate accurate to a constant factor. The remaining argument adapts the local-overlap and separator analysis of~\cite[Appendix~F]{Zhang26HitRun} to this lattice, giving a mixing time of
\[
 t_{\rm mix}=O\!\left(
 \frac{J^2d^2}{\psi_d^2}\log^3\frac4\eta
 \right)
 =\widetilde O\left(\frac{d^2}{\psi_d^2}\right)
\]
steps to reach total-variation error $\eta$ from every $2$-warm start, provided $A$ approximates the covariance within a constant factor. The formal guarantees and proof appear in \cref{thm:finite-metropolis,app:finite-state}.

Finally, we can implement a $1$-block-encoding of the discriminant using $O(1)$ membership queries. Together with the qsample-reflection cost~\eqref{eq:tech_ov_reflection_cost}, this yields $\widetilde O(d/\psi_d)$ queries per reflection (\cref{cor:finite-reflection-cost-main}). 

\subsubsection{Estimating the product by phase accumulation}\label{sec:tech_ov_estimation}
We now discuss our quantum estimator for the telescoping product of ratios in~\eqref{eq:tech_ov_telescoping}.
Suppose the required independent qsamples, controlled reflections, and rotations for bounded observables~\eqref{eq:tech_ov_obs_rotation} are supplied. Our estimator uses $\widetilde O(m/\varepsilon)$ qsample-reflections, improving the expected $\widetilde O(m^{3/2}/\varepsilon)$ bound of Cornelissen--Hamoudi~\cite{CornelissenHamoudi23}. The improvement comes from estimating the combined logarithm, rather than estimating each ratio separately to high precision.

For $\mu_i=\E_{\pi_i}[X_i]$ and $\pi_{i+1}=X_i\pi_i/\mu_i$, classical bridge sampling~\cite{Bennett76,MengWong96} gives
\[
 p_i^+=\E_{\pi_i}\!\left[\frac{X_i}{c_i+X_i}\right],
 \qquad
 p_i^-=\E_{\pi_{i+1}}\!\left[\frac{c_i}{c_i+X_i}\right],
 \qquad
 \mu_i=c_i\frac{p_i^+}{p_i^-}.
\]
The identity holds for every $c_i>0$. We take $c_i$ to be an $O(1)$-approximation to $\mu_i$ estimated using  $\widetilde{O}(1)$ classical samples, which are obtained by measuring $\widetilde{O}(1)$ copies of the qsample $\ket{\pi_i}$. Then, $c_i=\Theta(\mu_i)$ ensures $p_i^\pm=\Omega(1)$. Therefore,
\begin{equation}\label{eq:tech_ov_L}
 \log\prod_i\mu_i=\sum_i\log c_i-\mathsf L,
 \qquad
 \mathsf L=\sum_i(-\log p_i^++\log p_i^-).
\end{equation}
Only $\mathsf L$ requires quantum estimation, and our estimator will encode each logarithm as an eigenvalue of the corresponding Hamiltonian. 

Consider a random variable $Y\in [0,1]$ with $p=\E_\pi[Y]=\Omega(1)$. We can apply the controlled-rotation gate
\begin{equation}\label{eq:tech_ov_obs_rotation}
 V_Y\ket{x}_{\mathsf S}\ket0_{\mathsf Q}
 =\ket{x}_{\mathsf S}
 \left(\sqrt{1-Y(x)}\ket0_{\mathsf Q}
       +\sqrt{Y(x)}\ket1_{\mathsf Q}\right)
\end{equation}
to obtain $\ket\phi=V_Y(\ket\pi_{\mathsf S}\ket0_{\mathsf Q})$. Let $Q_\phi=\ket\phi\!\bra\phi$ and $\Pi=\Id_{\mathsf S}\otimes\ket1\!\bra1_{\mathsf Q}$. The identities
\[
 G_\phi:=Q_\phi(\Id-\Pi)Q_\phi=(1-p)Q_\phi,
 \qquad
 H_k:=\sum_{\ell=1}^k\frac{G_\phi^\ell}{\ell}\approx -\log (\Id-G_\phi)
\]
give
\[
 H_k\ket\phi
 =\bigl(-\log p+O(\varepsilon/m)\bigr)\ket\phi
\]
for $k=O(\log(m/\varepsilon))$. Using quantum singular value transformation (QSVT)~\cite{GSLW19}, we can block-encode $H_k$ with $\widetilde{O}(1)$ qsample-reflections and observable rotations, and the Hamiltonian simulation implements the evolution for time $\tau$ using $\widetilde O(1+|\tau|)$ reflections. 

For each of the $2m$ terms in $\mathsf L$, let $H_j\ket{\phi_j}=h_j\ket{\phi_j}$ and choose a sign $\sigma_j\in\{\pm1\}$ so that $\sum_j\sigma_jh_j=\mathsf L+O(\varepsilon)$. The phase register $\mathsf F$ has basis labels $u=0,\ldots,M-1$, as in the standard QPE, where $M=\Theta(m/\varepsilon)$ is a power of two. Each \textsc{Add} uses controlled evolution to implement
\[
 \ket u_{\mathsf F}\ket{\phi_j}
 \longmapsto
 e^{iu\tau_0\sigma_jh_j}\ket u_{\mathsf F}\ket{\phi_j},
 \qquad \tau_0=\Theta(1/m).
\]
Phase kickback adds the contribution to the shared register $\mathsf F$. In the ideal evolution, $\ket{\phi_j}$ remains unchanged and unentangled with $\mathsf F$, so it can be discarded before processing the next contribution.

After all contributions, as in the standard QPE, we apply an inverse Fourier transform and measure the ${\sf F}$ register to obtain an estimate of their signed sum. The total Hamiltonian evolution time is $(M-1)\tau_0=O(1/\varepsilon)$ per contribution, giving $\widetilde O(m/\varepsilon)$ reflections overall (\cref{lem:log-sum}). An illustration of the quantum circuits is shown in~\cref{fig:algorithm-components}. We note that our quantum product estimator needs only \emph{streaming access} to the qsamples: each copy is used for one contribution and then discarded, while $\mathsf F$ carries the accumulated phase through the annealing schedule.

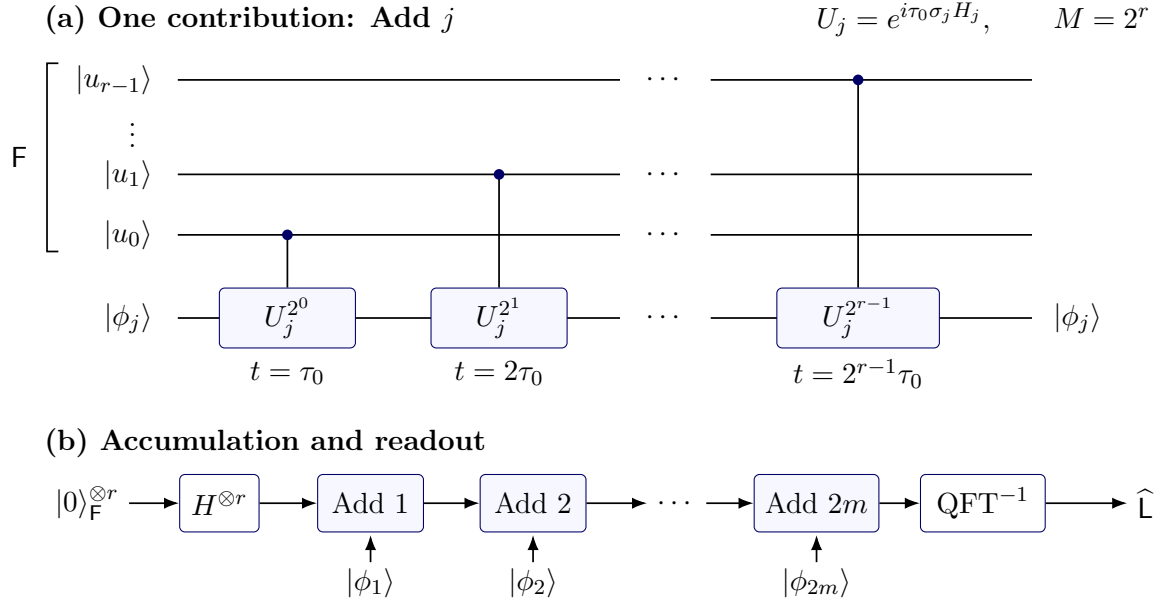
\begin{figure}
 \centering
 \begin{tikzpicture}[
 x=1cm,y=1cm,
 every node/.style={font=\normalsize,align=center},
 gate/.style={draw=blue!40!black,fill=blue!3,rounded corners=2pt,minimum height=.76cm,minimum width=1.8cm},
 wire/.style={line width=.65pt},
 arrow/.style={-{Latex[length=2mm]},line width=.65pt}]

 \node[anchor=west,font=\bfseries] at (0,.8) {(a) One contribution: Add $j$};
 \node[anchor=east] at (14.9,.8)
   {$U_j=e^{i\tau_0\sigma_jH_j},\qquad M=2^r$};

 \draw[wire] (.32,.22)--(.17,.22)--(.17,-2.27)--(.32,-2.27);
 \node at (-.2,-1.02) {$\mathsf F$};
 \node[anchor=east] at (1.7,0) {$\ket{u_{r-1}}$};
 \node at (1.3,-.62) {$\vdots$};
 \node[anchor=east] at (1.7,-1.25) {$\ket{u_1}$};
 \node[anchor=east] at (1.7,-2.05) {$\ket{u_0}$};
 \node[anchor=east] at (1.7,-3.15) {$\ket{\phi_j}$};

 \foreach \yy in {0,-1.25,-2.05,-3.15}{
  \draw[wire] (1.9,\yy)--(7.75,\yy);
  \node at (8.35,\yy) {$\cdots$};
  \draw[wire] (8.95,\yy)--(13.2,\yy);
 }
 \node[gate] (u0) at (3.35,-3.15) {$U_j^{2^0}$};
 \node[gate] (u1) at (6.15,-3.15) {$U_j^{2^1}$};
 \node[gate,minimum width=2.15cm] (ur) at (10.9,-3.15) {$U_j^{2^{r-1}}$};
 \draw[wire] (3.35,-2.05)--(u0.north);
 \draw[wire] (6.15,-1.25)--(u1.north);
 \draw[wire] (10.9,0)--(ur.north);
 \fill[blue!40!black] (3.35,-2.05) circle (2pt);
 \fill[blue!40!black] (6.15,-1.25) circle (2pt);
 \fill[blue!40!black] (10.9,0) circle (2pt);
 \node at (3.35,-3.88) {$t=\tau_0$};
 \node at (6.15,-3.88) {$t=2\tau_0$};
 \node at (10.9,-3.88) {$t=2^{r-1}\tau_0$};
 \node[anchor=west] at (13.35,-3.15) {$\ket{\phi_j}$};

 \node[anchor=west,font=\bfseries] at (0,-4.8) {(b) Accumulation and readout};
 \node (initial) at (.7,-5.6) {$\ket{0}_{\mathsf F}^{\otimes r}$};
 \node[gate,minimum width=1.05cm,fill=white] (had) at (2.45,-5.6) {$H^{\otimes r}$};
 \node[gate,minimum width=1.4cm] (a1) at (4.45,-5.6) {Add $1$};
 \node[gate,minimum width=1.4cm] (a2) at (6.6,-5.6) {Add $2$};
 \node (dots) at (8.5,-5.6) {$\cdots$};
 \node[gate,minimum width=1.65cm] (an) at (10.35,-5.6) {Add $2m$};
 \node[gate,minimum width=1.65cm,fill=white] (qft) at (12.55,-5.6) {$\mathrm{QFT}^{-1}$};
 \node (read) at (14.7,-5.6) {$\widehat{\mathsf L}$};
 \draw[arrow] (initial.east)--(had.west);
 \draw[arrow] (had.east)--(a1.west);
 \draw[arrow] (a1.east)--(a2.west);
 \draw[arrow] (a2.east)--(dots.west);
 \draw[arrow] (dots.east)--(an.west);
 \draw[arrow] (an.east)--(qft.west);
 \draw[arrow] (qft.east)--(read.west);
 \foreach \xc/\j in {4.45/1,6.6/2,10.35/{2m}}{
  \node at (\xc,-6.65) {$\ket{\phi_{\j}}$};
  \draw[arrow] (\xc,-6.38)--(\xc,-6.02);
 }
 \end{tikzpicture}
 \caption{Phase accumulation in the ideal estimator. (a) The procedure \textsc{Add} applies all binary-controlled powers of $U_j=e^{i\tau_0\sigma_jH_j}$ to the supplied eigenstate $\ket{\phi_j}$. It multiplies the amplitude of $\ket u_{\mathsf F}$ by $e^{iu\tau_0\sigma_jh_j}$ while leaving $\ket{\phi_j}$ unchanged. (b) All contributions use the same register $\mathsf F$, with a separate qsample for each \textsc{Add}. The pseudocode is in~\cref{alg:add-log}.}
 \label{fig:algorithm-components}
\end{figure}

\subsubsection{Alternating qsample conversions and  error analysis}\label{sec:tech_ov_alter}

We reserve two blocks $\mathsf B_1,\mathsf B_2$, each containing
$L=\widetilde O(d)$ qsamples, to supply covariance estimates throughout
the cooling schedule. Following the interlacing idea
of~\cite{CCHLW23}, one block supplies the preconditioner used to convert
the other by fixed-point amplitude
amplification~\cite{YoderLowChuang14}, and then the blocks exchange roles:
\[
\begin{aligned}
 \ket{\pi_i}_{\mathsf B_1}^{\otimes L}
 \ket{\pi_i}_{\mathsf B_2}^{\otimes L}
 &\xmapsto{\;\mathsf B_1\text{ controls the conversion of }\mathsf B_2\;}
 \ket{\pi_i}_{\mathsf B_1}^{\otimes L}
 \ket{\pi_{i+1}}_{\mathsf B_2}^{\otimes L}\\
 &\xmapsto{\;\mathsf B_2\text{ controls the conversion of }\mathsf B_1\;}
 \ket{\pi_{i+1}}_{\mathsf B_1}^{\otimes L}
 \ket{\pi_{i+1}}_{\mathsf B_2}^{\otimes L}.
\end{aligned}
\]
These are the intended states before and after the qsample conversions. Adjacent distributions have comparable covariances, so an accurate estimate from either stage preconditions both lattice Metropolis walks. It therefore suffices to implement both $R_{\pi_i}$ and
$R_{\pi_{i+1}}$ by~\cref{cor:finite-reflection-cost-main}.
Each block is converted one qsample at a time.

To use a control block, we compute its empirical covariance \emph{coherently}, use the estimate to select the reflections in the conversion circuit, and uncompute the covariance workspace. In the resulting state, the control block could be entangled with the converted qsample. Fortunately, our lattice Metropolis walk has a fixed stationary distribution that is independent of the covariance estimate. Therefore,  every accurate covariance estimate gives an approximation to the same intended qsample-reflection. We use a \emph{two-level hybrid argument} to show that these errors remain small throughout the algorithm (see~\cref{fig:qsample-error-analysis} for an illustration).

\begin{figure}
\centering
\definecolor{errink}{RGB}{38,47,61}
\definecolor{errblue}{RGB}{41,85,130}
\definecolor{errgray}{RGB}{100,108,119}
\begin{tikzpicture}[
 x=1cm,y=1cm,line width=.8pt,
 every node/.style={
  align=center,
  text=errink
},
 heading/.style={anchor=west,font=\bfseries},
 gate/.style={draw=errblue,fill=errblue!5,rounded corners=1.5pt,
   minimum width=1cm,minimum height=.8cm,inner sep=3pt},
 othergate/.style={gate,draw=errgray!75,fill=errgray!4},
 state/.style={draw=errgray!75,fill=errgray!3,rounded corners=2pt,
   text width=3.05cm,minimum height=1.15cm,inner sep=5pt},
 reference/.style={state,draw=errblue,fill=errblue!4},
 wire/.style={draw=errink},
 operation/.style={-{Latex[length=2mm]},draw=errink},
 comparison/.style={{Latex[length=2mm]}-{Latex[length=2mm]},
   draw=errblue,densely dashed},
 cut/.style={draw=errblue,densely dashed}]

\node[heading] at (0,0)
 {(a) Outer hybrid: one primitive};
\node[state] (actualin) at (1.80,-1.35)
 {actual input\\[3pt]$\widetilde\rho_k$};
\node[reference] (refin) at (8.10,-1.35)
 {reference input\\[3pt]$\rho_k$};
\node[state,minimum height=1.4cm] (actualout) at (1.80,-3.85)
 {actual output\\[4pt]$\widetilde\rho_{k+1}$};
\node[state,minimum height=1.4cm] (implref) at (8.10,-3.85)
 {output from $\rho_k$\\[4pt]
  $\widetilde{\mathcal E}_k(\rho_k)$};
\node[reference,minimum height=1.4cm] (nextref) at (14.40,-3.85)
 {next reference\\[4pt]$\rho_{k+1}$};

\draw[operation] (actualin.south)--node[right=5pt]
 {$\widetilde{\mathcal E}_k$}(actualout.north);
\draw[operation] (refin.south)--node[right=5pt]
 {$\widetilde{\mathcal E}_k$}(implref.north);
\draw[operation] (refin.east)--node[above=5pt]
 {intended primitive}(14.40,-1.35)--(nextref.north);

\draw[comparison] (actualin.east)--node[above=5pt,text=errblue]
 {$E_k$}(refin.west);
\draw[comparison] (actualout.east)--node[above=5pt,text=errblue]
 {$\le E_k$}(implref.west);
\draw[comparison] (implref.east)--node[above=5pt,text=errblue]
 {$e_k$}(nextref.west);
\node[text=errblue] at (4.95,-4.3) {contractivity};
\node[text=errblue] at (11.25,-4.3) {primitive error};

\node[heading] at (0,-5.3)
 {(b) Inner hybrid: reflections within one conversion};
\begin{scope}[yshift=1.1cm]
\node[anchor=east] at (1.30,-7.45) {$\ket{\pi_i}^{\otimes L}$};
\draw[wire] (1.45,-7.45)--(14.05,-7.45);
\node[fill=white,inner sep=4pt,text=errblue] at (8.10,-7.1)
 {control block unchanged};
\node[anchor=west] at (14.15,-7.45) {$\ket{\pi_i}^{\otimes L}$};

\node[anchor=east] at (1.30,-9.05) {$\ket{\pi_i}$};
\draw[wire] (1.45,-9.05)--(14.05,-9.05);
\node[othergate,minimum width=.85cm] at (2.05,-9.05) {$V_0$};
\node[gate] at (3.60,-9.05) {$R_1$};
\node[fill=white,inner sep=4pt] at (4.70,-9.05) {$\cdots$};
\node[othergate] at (5.95,-9.05) {$V_{j-1}$};
\node[gate] at (7.60,-9.05) {$R_j$};
\node[fill=white,inner sep=4pt] at (8.70,-9.05) {$\cdots$};
\node[othergate] at (10.15,-9.05) {$V_{q-1}$};
\node[gate] at (11.85,-9.05) {$R_q$};
\node[othergate,minimum width=.85cm] at (13.10,-9.05) {$V_q$};
\node[anchor=west] at (14.15,-9.05) {$U\ket{\pi_i}$};
\foreach \xx/\jj in {2.90/1,6.90/j,11.15/q}{
 \draw[cut] (\xx,-8.42)--(\xx,-9.68);
 \node[text=errblue,anchor=south] at (\xx,-8.39)
  {$\ket{\psi_{\jj}}$};
}
\end{scope}
\end{tikzpicture}
\caption{Two-level hybrid argument for error analysis. (a) Solid arrows apply operations, and dashed double arrows compare joint states in trace distance. The outer hybrid argument gives $E_{k+1}\leq E_k +e_k$. (b) For one qsample conversion, the outer reference input $\rho_k$ contains $\ket{\pi_i}^{\otimes L}\otimes\ket{\pi_i}$. $\ket{\psi_j}$ are the inner reference states. Each reflection has error at most $\xi+2\sqrt{\varepsilon_{\Cov}}$, and the inner hybrid argument gives $e_k\le\varepsilon_{\rm conv}+q(\xi+2\sqrt{\varepsilon_{\Cov}})$}
\label{fig:qsample-error-analysis}
\end{figure}

\paragraph{Outer level: comparing primitives.}
View the algorithm as a sequence of quantum primitives, such as state conversion and Hamiltonian simulation. Let $\rho_k$ denote the state before the $k$-th primitive in the reference computation, where each primitive performs its intended operation without error. For example, a reference conversion replaces $\ket{\pi_i}$ by $\ket{\pi_{i+1}}$ and leaves the control block unchanged.

Let $\widetilde\rho_k$ be the corresponding implemented state and $\widetilde{\mathcal E}_k$ the implemented primitive. Define
\[
 E_k=\dtr(\widetilde\rho_k,\rho_k),
 \qquad
 e_k=\dtr\!\left(
   \widetilde{\mathcal E}_k(\rho_k),\rho_{k+1}
 \right).
\]
Since applying the same quantum channel to two states cannot increase their trace distance, we prove in~\cref{lem:composition} that
\[
 E_{k+1}=\dtr\!\left(\widetilde{\mathcal E}_k(\widetilde\rho_k),\rho_{k+1}\right)\leq \dtr(\widetilde{\mathcal E}_k(\widetilde\rho_k),\widetilde{\mathcal E}_k(\rho_k))+\dtr(\widetilde{\mathcal E}_k(\rho_k),\rho_{k+1})
 \le E_k+e_k.
\]
Thus, we can bound each primitive's error $e_k$ on the reference input $\rho_k$ and add these bounds. 

\paragraph{Inner level: comparing individual reflections.}
We illustrate how to bound $e_k$ for one qsample conversion. Write
\[
 \ket{b_i}_{\mathsf B_1}
 :=\ket{\pi_i}_{\mathsf B_1}^{\otimes L},
 \qquad
 \pi_i^{\otimes L}(\mathbf x)
 :=\prod_{\ell=1}^{L}\pi_i(x_\ell).
\]
For this primitive, the outer reference input contains $\ket{b_i}_{\mathsf B_1}\otimes\ket{\pi_i}_{\mathsf S}$, where $\mathsf S$ is the qsample being converted. We ignore the auxiliary registers for simplicity.

Let $U=U[R_{\pi_i},R_{\pi_{i+1}}]$ be the finite conversion circuit of~\cref{lem:transport} with exact qsample-reflections, and let $\widetilde{U}_{\mathbf x}=U[\widetilde{R}_{\pi_i}({\bf x}),\widetilde{R}_{\pi_{i+1}}({\bf x})]$ be the same circuit with the qsample-reflections replaced by their approximate implementations using the covariance estimate computed from $\mathbf x$. After uncomputing the covariance
workspace, the implemented operation has the form
\begin{equation}\label{eq:tech_ov_approx_convert}
 \ket{\pi_i}^{\otimes L}_{\mathsf B_1}\otimes\ket{\pi_i}_{\mathsf S}
 \longmapsto
 \sum_{\mathbf x\in\Omega_\Lambda^{\times L}}
 \sqrt{\pi_i^{\otimes L}(\mathbf x)}
 \ket{\mathbf x}_{\mathsf B_1}
 \otimes \widetilde U_{\mathbf x}\ket{\pi_i}_{\mathsf S}.
\end{equation}

The inner reference computation runs $U$ on this same outer reference input $\rho_k$. Let $\ket{\psi_j}$ be its working state immediately before the $j$-th reflection, denoted by $R_j$, where $R_j$ is either $R_{\pi_i}$ or $R_{\pi_{i+1}}$. Because the exact reflections do not depend on the covariance estimate, the reference input to this call is $\ket{\pi_i}^{\otimes L}\otimes\ket{\psi_j}$.

Suppose inaccurate covariance estimates have total probability at most $\varepsilon_{\Cov}$, and every accurate estimate implements $R_j$ with vector-norm error at most $\xi$ on $\ket{\psi_j}$. Let $\widetilde R_j$ denote the coherently selected implementation that acts jointly on the control block and the working register. Orthogonality of $\{\ket{\bf x}\}$ gives
\[
 \left\|
 \left(\widetilde R_j-I_{\mathsf B_1}\otimes R_j\right)
 \left(\ket{b_i}\otimes\ket{\psi_j}\right)
 \right\|
 \le \sqrt{\xi^2+4\varepsilon_{\Cov}}
 \le \xi+2\sqrt{\varepsilon_{\Cov}},
\]
as formalized in~\cref{lem:coherent-reflection-error}.
This bounds the error of one reflection on the full joint state,
including its effect on the control block.

A hybrid argument over the $q$ reflection calls adds these errors. The finite conversion circuit $U$ itself has conversion error at most $\varepsilon_{\rm conv}$, even with exact reflections (\cref{lem:transport}).
Consequently,
\[
 e_k\le
 \varepsilon_{\rm conv}
 +q\bigl(\xi+2\sqrt{\varepsilon_{\Cov}}\bigr).
\]
The same comparison applies when the blocks exchange roles: covariance accuracy is evaluated on the reference product states, while the outer trace-distance bound accounts for the actual state even when its control and working registers are entangled.

Recall that the error guarantee for an approximate qsample reflection depends on its input. To obtain the required reflection accuracy $\xi$, we also bound the fourth-moment norm of each inner reference state $\ket{\psi_j}$. Along the exact conversion circuit, sample-controlled gates 
\begin{align*}
    \sum_{x\in \Omega_\Lambda}\ket{x}\bra{x}\otimes G_{x}
\end{align*}
preserve this norm, and each reflection adds at most a constant to this norm. Thus, \cref{lem:fourth-growth} gives a bound of $O(1+j)$ before the $j$-th reflection. A similar analysis holds along the reference estimation circuits. 

\subsubsection{Total query complexity}\label{sec:tech_ov_total_cost}
Advancing the two covariance blocks through $m$ transitions uses $\widetilde O(Lm)$ qsample-reflections. The $\widetilde O(m)$ additional copies used for classical and quantum estimation are prepared at stage zero, advanced to their designated stages, and discarded after use. Each undergoes at most $m$ conversions, giving $\widetilde O(m^2)$ additional reflections. Including product estimation and the endpoint terms in the telescoping ratio product~\eqref{eq:tech_ov_telescoping}, the total number of qsample-reflections is
\begin{equation}\label{eq:intro-reflection-count}
 N=\widetilde O\!\left(
 \underbrace{Lm}_{\text{covariance blocks}}
 +\underbrace{m^2}_{\text{other sample conversions}}
 +\underbrace{m/\varepsilon}_{\text{estimation}}
 \right)
 =\widetilde O\!\left(
 d^{3/2}+\frac{\sqrt d}{\varepsilon}
 \right),
\end{equation}
where $L=\widetilde O(d)$ and $m=\widetilde O(\sqrt d)$. The polynomial bounds on the number of calls and the fourth-moment norms allow inverse-polynomial error tolerances, with only polylogarithmic overhead. Multiplying by the $\widetilde O(d/\psi_d)$ membership queries per reflection gives~\cref{thm:intro-volume_informal}. All circuit lengths are fixed in advance, so the query complexity holds on every execution. \cref{sec:assembly,app:parameters} bound the initialization cost and the remaining numerical errors, and complete the success-probability analysis.

\subsubsection{Comparison with previous works}\label{subsec:comparison-previous}

\paragraph{Rounding and nondestructive covariance estimation.}
Chakrabarti et al.~\cite{CCHLW23} already propose coherent interlaced rounding. The unresolved point in their analysis is the guarantee that the qsamples remain close to their intended joint state after the covariance information is used. Their argument appeals to concentration on a correct affine transformation, whereas ordinary covariance concentration only guarantees that the estimate belongs to a set of accurate matrices. It need not concentrate on a single encoded matrix. Therefore, neither statistical accuracy nor uncomputing the matrix register alone guarantees small disturbance to the qsample block. Cornelissen, Apers, and Gribling~\cite[Section~1.2.1, footnote~1]{CornelissenApersGribling26} identify this gap in the multivariate nondestructive mean-estimation argument, and the univariate guarantees of~\cite{CornelissenHamoudi23} do not directly resolve it.

Our construction uses the interlacing idea but does not require a classical covariance estimate as output. The joint-state analysis in~\cref{sec:tech_ov_alter} justifies using the covariance coherently to select an approximation to a fixed reflection, even when many different matrix values occur. Thus our rounding bound does not assume a nondestructive multivariate mean estimator. This distinction is separate from the single-start effective spectral gap issue in \cref{app:counterexample}.

To avoid the unresolved quantum-rounding step, \cite[Section~1.2.1, footnotes~1--2]{CornelissenApersGribling26} uses classical rounding and shows the corrected query complexities
\[
 \begin{array}{ll}
  \widetilde O(d^{7/2}+d^{5/2}/\varepsilon)
   &\text{\cite{CCHLW23}},\\
  \widetilde O(d^{7/2}+d^{9/4}/\varepsilon)
   &\text{\cite{CornelissenHamoudi23}}.
 \end{array}
\]
It changes the $d^3$ quantum rounding costs in \cref{tab:intro-volume-comparison} to $d^{7/2}$, which is the cost of replacing quantum rounding by classical preprocessing. Our smaller rounding term follows from the independently justified coherent construction described above.

\paragraph{The concurrent transducer approach.}
Concurrently and independently, Cornelissen, Apers, and Gribling~\cite{CornelissenApersGribling26} use transducers to amortize quantum-walk costs and quantize the speedy-walk approach of Cousins and Vempala~\cite{CousinsVempala18}. Their algorithm uses classical rounding and has membership-query complexity
\[
 \widetilde O\!\left(d^{7/2}+\frac{d^{7/4}}{\varepsilon}\right).
\]
Their amortized analysis addresses the state-dependent cost of implementing speedy-walk updates. Our lattice Metropolis walk introduced in~\cref{sec:tech_ov_walk} uses $O(1)$ membership queries per update and mixes from every $2$-warm start in $\widetilde O(d^2\psi_d^{-2})$ steps under the preconditioning and discretization assumptions above. This matches the nearly quadratic dimension dependence of the speedy-walk analysis under the respective assumptions, while avoiding the amortized analysis.

\paragraph{Relation to hit-and-run analyses.}
The discretization analysis of~\cite[Section~5]{CCHLW23} already compares continuous and discrete stationary flows. For uniform isotropic targets, Chen and Eldan~\cite[Theorem~2]{ChenEldan26} first proved a lazy hit-and-run mixing time that depends nearly quadratically on $d$. However, their mixing time also has an $\eta^{-11}$ dependence on the accuracy, which is too costly for our quantum application. Combining the discretization with sharper continuous hit-and-run bounds having polylogarithmic dependence on mixing accuracy remains a possible alternative to our lattice construction. The global spectral-gap theorem of Kook and Vempala~\cite{kook2026spectral} cannot be applied directly because their result is for the uniform distribution over the convex body, rather than for the log-affine distributions in our pencil schedule. The companion paper~\cite[Appendix F]{Zhang26HitRun} proves a warm-start mixing time for the log-affine distributions. However, this result does not give a global spectral gap lower bound, and the existing effective spectral gap argument~\cite[Proposition~4.2]{CCHLW23} is disproved by our counterexample in \cref{app:counterexample}.

\subsection{Technical overview for the lower bound}
Our construction is inspired by Childs~\cite{Childs26Matrix}, who proved that determining the sign of the determinant of a real $d\times d$ matrix requires $\Omega(d)$ quantum matrix-vector queries. We use determinant sign to define two distributions of ellipsoids whose volumes lie in ranges separated by a constant factor with high probability. At the same time, for any fixed set of fewer than $d$ query points, the product of their membership bits has the same expectation under the two distributions. The quantum polynomial method then gives an $\Omega(d)$ membership-query lower bound.

An ellipsoid's volume depends on the magnitude of the determinant of its defining linear map. The main geometric step is therefore to convert the hidden determinant sign into a gap in these magnitudes, while keeping every ellipsoid between known concentric balls whose radii have a polynomial ratio. We establish this gap using quantitative estimates of determinants.

In even dimension, let $\mu_\pm$ be the Haar measure on the orthogonal group conditioned on the determinant equal to $\pm 1$, respectively. Let $O$ be sampled from $\mu_\pm$, and define an ellipsoid
\[
 K_O=\{x:\|(\Id-aO)x\|_2\le1\},
 \qquad 0<a<1,\quad 1-a=\Theta(d^{-2}).
\]
We can show that each $K_O$ satisfies
\begin{align*}
    B(0,1/2)\subseteq K_O\subseteq B(0,\Theta(d^2))
\end{align*}
The membership query at $x$ depends only on the quadratic form $x^\top Ox$, since
\[
 \|(\Id-aO)x\|_2^2=(1+a^2)\|x\|_2^2-2a\,x^\top Ox.
\]
For any fixed tuple of fewer than $d$ query points, choose an orthogonal reflection $H$ that fixes their span. The map $O\mapsto HO$ flips the sign of $\det O$ and preserves every membership answer in the tuple. On the other hand, if $O$ is sampled from $\mu_+$, then the distribution of $HO$ is $\mu_-$, and vice versa. Thus every product of fewer than $d$ membership bits has the same expectation under the two distributions.

To separate the volumes, we show that the defining determinant is
typically smaller when $\det O=-1$ than when $\det O=1$. Write
\[
 \Delta_a(O)=\det(\Id-aO)>0,
 \qquad
 \vol(K_O)=\frac{\vol(B(0,1))}{\Delta_a(O)}.
\]
Thus, a smaller determinant gives a larger volume.

When $\det O=-1$, we can make the determinant small on average:
\[
 \E_{\mu_-}[\Delta_a(O)]=1-a^d\le d(1-a).
\]
Since $\Delta_a(O)$
is positive, Markov's inequality then bounds the probability that
it exceeds any chosen threshold.

For $\det O=1$, we use the recursive construction of Haar
measure~\cite{BourgadeEtAl07} to bound the expectation of its inverse
fourth root. A very small determinant makes this inverse power large,
so the bound limits how often such small values can occur. In
particular, for a sufficiently small universal constant $c>0$,
\[
 \Pr_{\mu_+}\!\left[\Delta_a(O)\ge\frac{c}{d}\right]\ge\frac78.
\]
We now choose $1-a$ to be a sufficiently small constant multiple of
$d^{-2}$. The expectation bound for $\det O=-1$ then gives
\[
 \Pr_{\mu_-}\!\left[\Delta_a(O)\le\frac{c}{4d}\right]\ge\frac78.
\]
These two determinant bounds give volume ranges separated by a
factor of four, while every outer radius remains $O(d^2)$.
Both probability bounds hold under the original Haar distributions,
so the membership symmetry established above continues to apply.

By the quantum polynomial method~\cite{BealsEtAl01}, a $q$-query algorithm's output probabilities have degree at most $2q$ in the membership bits. Each monomial involves a fixed tuple, so the symmetry makes the two averaged output distributions identical when $2q<d$. A volume estimate with relative error $1/4$ would distinguish the two volume ranges, giving an $\Omega(d)$ query lower bound. 

\subsection{Related work}

\paragraph{Classical volume estimation.}
Dyer, Frieze, and Kannan~\cite{DFK91} gave the first randomized polynomial-time algorithm, which uses $O^*(d^{23})$ membership queries.\footnote{Here $O^*(\cdot)$ hides polylogarithmic factors in $d$ and the dependence on the accuracy.} Improvements in the random walks, the cooling schedules, and the isoperimetric inequalities reduced the query complexity to $O^*(d^4)$~\cite{LovaszSimonovits90,ApplegateKannan91,DyerFrieze91,LovaszSimonovits93,KLS97,LovaszVempala06Annealing}; see~\cite{Simonovits03} for a survey of this period. Cousins and Vempala~\cite{CousinsVempala18} introduced the Gaussian cooling and obtained $O^*(d^3)$ for well-rounded bodies. Jia, Laddha, Lee, and Vempala~\cite{JiaEtAl26} bounded the rounding cost in terms of the KLS constant and combined it with the Cousins--Vempala algorithm, yielding $O^*(d^{3.5})$ for general convex bodies. Another line of work uses algorithmic diffusion. Kook, Vempala, and Zhang~\cite{KookVempalaZhang26} introduced the In-and-Out sampler for convex bodies and obtained guarantees in R\'enyi divergence. Kook and Vempala~\cite{KookVempala25} extended this approach to general log-concave functions through a reduction to exponential sampling in one higher dimension, improving the complexity of sampling, rounding, and integration. For explicitly given polytopes, Riemannian Hamiltonian Monte Carlo (RHMC) leads to faster volume algorithms~\cite{LeeVempala18,pmlr-v247-gatmiry24a}. On the lower-bound side, randomized volume estimation to constant precision requires $\widetilde\Omega(d^2)$ membership queries~\cite{rademacher2008dispersion}, and a nearly quadratic lower bound also holds for linear optimization over convex bodies~\cite{vempala2026nearly}.
  
\paragraph{Quantum algorithms for discrete partition functions.}
Quantum algorithms for partition-function estimation have applications in both statistical physics and approximate counting. Examples include the ferromagnetic Ising model, the monomer--dimer model for matchings, the hard-core model for independent sets, and counting proper graph colorings~\cite{Montanaro15,HarrowWei20,ArunachalamEtAl22,CornelissenHamoudi23}. These algorithms use sampling Markov chains that mix rapidly along an annealing schedule. Representative regimes with such guarantees include sufficiently high temperatures for the Ising model, sufficiently many colors relative to the maximum degree for graph colorings, and graphs of maximum degree at most five for counting independent sets~\cite{CornelissenHamoudi23}. Given coherent implementations of the sampling chains and access to an initial qsample, quantum simulated annealing and mean estimation provide nearly quadratic improvements in the dependence on the inverse spectral gap and the relative-error parameter~\cite{WCNA09,Montanaro15,HarrowWei20,ArunachalamEtAl22}. Cornelissen and Hamoudi~\cite{CornelissenHamoudi23} further improved the estimation step by using nearly unbiased quantum mean estimates to approximate the product of ratios along the schedule.

\paragraph{Quantum algorithms for continuous distributions.}
Beyond the quantum volume algorithms discussed above~\cite{CCHLW23,CornelissenApersGribling25,CornelissenApersGribling26}, Childs, Li, Liu, Wang, and Zhang~\cite{ChildsEtAl22} gave quantum speedups for sampling log-concave distributions and estimating their normalizing constants using evaluation and gradient oracles. Liu, Wang, and Ji~\cite{LiuWangJi26LocalSampling} obtained quantum speedups for log-concave distributions whose potentials decompose into local terms, using coherent evaluation and gradient queries to those terms. Li and Zhang~\cite{LiZhang22} proposed quantum algorithms for approximately log-concave sampling, with applications to approximately convex optimization and the quantum stochastic convex bandit problem. 

For non-log-concave sampling, Ozgul, Li, Mahdavi, and Wang~\cite{OzgulEtAl24} studied quantum sampling and partition-function estimation with a stochastic-gradient oracle; their subsequent work~\cite{OzgulEtAl26Stochastic} accelerated stochastic Langevin and Hamiltonian Monte Carlo under stochastic-gradient or noisy-evaluation access. Motamedi and Ronagh~\cite{MotamediRonagh24} used quantum differential-equation solvers to sample Gibbs distributions of periodic, possibly non-convex potentials and accelerate mean estimation. Given an initial quantum state with constant overlap with the target qsample, Leng, Ding, Chen, and Lin~\cite{LengEtAl26} used a factorization of the Witten Laplacian and singular-value filtering to accelerate Langevin sampling. Their method gives up to quartic improvements in the relevant functional-inequality parameters compared with existing classical Langevin-based bounds. Building on this approach, Olivucci et al.~\cite{OlivucciEtAl26} proved a worst-case quantum--classical query separation for sampling from a class of Gevrey-smooth Gibbs distributions on the torus at constant total-variation accuracy.

\paragraph{Roadmap.}
\Cref{sec:preliminaries} introduces the notation and preliminaries. \Cref{sec:warm-reflection} constructs qsample reflections from uniform warm-start mixing and bounds their accumulated implementation error. \Cref{sec:coherent-product} develops the quantum product estimator with streaming access to qsamples. \Cref{sec:finite-walk} presents the lattice Metropolis walk, coherent preconditioning, and alternating qsample conversions. \Cref{sec:assembly} combines these ingredients to prove the volume-estimation upper bound, and \cref{sec:volume-lower-bound} proves the linear query lower bound. \Cref{app:counterexample} shows why mixing from a single initial distribution is insufficient for the effective spectral gap argument. \Cref{app:finite-state} proves the lattice discretization and mixing guarantees, and \cref{app:implementation} analyzes the bit-complexity and gives the  parameter choices.

\section{Preliminaries}\label{sec:preliminaries}

\paragraph{Notations.} We write $d$ for the input dimension and $D=d+1$ for the dimension of the pencil construction.  A convex body is a compact convex set with nonempty interior.  The Euclidean ball of radius $r$ about $x$ is $B(x,r)$. All logarithms are natural unless a base is specified.  Constants in $\lesssim$, $\gtrsim$, and $\asymp$ are universal.  

Quantum registers are denoted by subscripts on a state, such as $\ket{x}_{\sf S}$. We write $\mathsf{S}$ for the sample register, $\mathsf{A}$ for the existing auxiliary registers, and $\mathsf{E}$ for fresh implementation workspace. Thus $\ket{\Psi}_{\mathsf{SA}}$ is a joint vector on the sample and auxiliary registers. Register subscripts are omitted only when the acting spaces are clear. In a proof involving a distinguished sample, $\mathsf{A}$ collects all other existing registers, including workspace allocated by an outer circuit. The label $\mathsf{E}$ refers only to fresh workspace for the current subroutine. 

A unitary $U$ is an \emph{$(\alpha,a,\varepsilon)$-block-encoding} of $A$ when
\[
 \bigl\|A-\alpha(\bra{0^a}\otimes\Id)U(\ket{0^a}\otimes\Id)\bigr\|_{\op}
 \le\varepsilon.
\]
We also use $\alpha$-block-encoding when the ancilla count and error are specified separately. A block-encoding query is one use of the encoding unitary or its inverse, with any required controls. The notation $\ctrl(U)$ denotes a controlled use of $U$.

\subsection{Norms}
For $1\le p<\infty$, we use
\[
 \|f\|_{L^p(\pi)}=\left(\int|f|^p\,d\pi\right)^{1/p},
 \qquad \langle f,g\rangle_{\pi}=\int\overline f g\,d\pi.
\]
The $L^\infty$ norm is the essential supremum.  When $\pi$ is fixed, we abbreviate the function norms by $\|f\|_p$.  Vector norms without a subscript are Hilbert-space norms, and $\|T\|_{\op}$ is the induced operator norm.  For density operators, we write
\[
 \dtr(\rho,\sigma)=\tfrac12\|\rho-\sigma\|_1
\]
for their trace distance.  We use $\pi(x)$ for either a probability mass or a density, and $dx$ for counting or Lebesgue measure as appropriate.  All quantum circuits act on finite registers.  Only the classical geometric estimates and the spectral lemma are stated more generally.

We will use the following special case of Riesz--Thorin interpolation~\cite[Theorem~21]{Tao09Interpolation}. If a linear operator $T$ is bounded from both $L^2(\pi)$ and $L^\infty(\pi)$ to $L^2(\pi)$, then
\begin{equation}\label{eq:riesz-thorin-prelim}
 \|Tf\|_2\le
 \|T\|_{L^2\to L^2}^{1/2}
 \|T\|_{L^\infty\to L^2}^{1/2}\|f\|_4
 \qquad \forall f\in L^4(\pi).
\end{equation}

\subsection{Convex geometry and covariance comparison}
A full-dimensional probability measure on $\R^D$ is log-concave if its density is proportional to $e^{-V}$ for a convex extended-real-valued function $V$.  It is \emph{isotropic} when its mean is zero and its covariance is $\Id$.  It is \emph{$C_0$-near-isotropic} when
\[
 C_0^{-1}\Id\preceq\Cov(X)\preceq C_0\Id.
\]
Centering is understood when needed in the geometric analysis.  The displacement sampler does not need to know the mean.

For a measure $\lambda$, let
\[
 \lambda^+(A)=\liminf_{r\downarrow0}
 \frac{\lambda\{x:\operatorname{dist}(x,A)<r\}-\lambda(A)}r,
 \qquad
 \psi_{\sf ch}(\lambda)=
 \inf_{0<\lambda(A)<1}
 \frac{\lambda^+(A)}{\min\{\lambda(A),1-\lambda(A)\}}.
\]
The constant $\psi_D$ in \eqref{eq:intro-kls} is the infimum of this quantity over isotropic log-concave measures.  An affine change to isotropic coordinates gives
\[
 \psi_{\sf ch}(\lambda)\ge
 \frac{\psi_D}{\sqrt{\|\Cov_\lambda(X)\|_{\op}}}.
\]
To reuse a preconditioner at the next cooling phase, we need to compare adjacent covariances.  The following standard bound provides that comparison from relative second moments.

\begin{lemma}[Covariance comparison {\cite[Lemma~4.3]{KalaiVempala06}}]
\label{lem:adjacent-covariance-prelim}
Let $\mu,\nu$ be full-dimensional log-concave laws satisfying
\[
 \int\left(\frac{d\nu}{d\mu}\right)^2d\mu\le B,
 \qquad
 \int\left(\frac{d\mu}{d\nu}\right)^2d\nu\le B.
\]
Then
\[
 (16B)^{-1}\Cov_\mu(X)\preceq\Cov_\nu(X)
 \preceq16B\Cov_\mu(X).
\]
\end{lemma}

\subsection{Warm starts, conductance, and the discriminant}
For a Markov kernel $\sP$ with stationary law $\pi$, its operator on functions is
\[
 (\sP f)(x)=\int f(y)\sP(x,dy).
\]
We use the same symbol for the kernel and this operator. The kernel is reversible for $\pi$ when $\pi(dx)\sP(x,dy)$ is symmetric.  Its operator on $L^2(\pi)$ is then self-adjoint.  It is positive when its spectrum lies in $[0,1]$, and ergodic when $\ker(\Id-\sP)=\Span\{\one\}$.  Laziness, $\sP(x,\{x\})\ge1/2$, implies positivity.  We write $E_\sP(I)$ for the spectral projection onto an interval $I$.  This notation also covers continuous-state operators with no eigenbasis.

For a distribution $\pi$, its qsample and the corresponding reflection are
\[
 \ket{\pi}=\int\sqrt{\pi(x)}\ket{x}\,dx,
 \qquad R_\pi=2\ket{\pi}\!\bra{\pi}-\Id.
\]
On finite state-spaces, the integral is a sum. The unitary identification $J_\pi f=\sqrt\pi f$ from $L^2(\pi)$ to the supported position space defines the discriminant
\begin{equation}\label{eq:discriminant}
 D_\pi=J_\pi\sP J_\pi^\dagger.
\end{equation}
It has the same spectrum as $\sP$, with stationary vector $\ket{\pi}$. 

Total variation distance is
\begin{align*}
    \dtv(\mu,\nu)=\sup_A|\mu(A)-\nu(A)|.
\end{align*}
A start $\nu$ is \emph{$M$-warm} for $\pi$ when $d\nu/d\pi\le M$. The mixing time of ${\cal P}$ starting from a distribution $\nu$ is defined as
\begin{align*}
    \tau_{\mix}(\eta,\nu,\pi;\sP):=\inf\left\{t\in \N:\dtv\left(\nu\sP^t,\pi\right)\le\eta\right\}.
\end{align*}
We also define the $2$-warm-start mixing time
\[
 \tau_{\mix,2}(\eta,\pi;\sP)
 =\sup_{\nu:\,d\nu/d\pi\le2}\tau_{\mix}(\eta,\nu,\pi;\sP).
\]

The stationary flow and $s$-conductance of a Markov kernel ${\cal P}$ are defined as
\[
 Q_\sP(A,B)=\int_A\sP(x,B)\,d\pi(x),\qquad
 \Phi_s(\sP)=\inf_{s<\pi(A)\le1/2}
 \frac{Q_\sP(A,A^c)}{\pi(A)-s}.
\]
The cutoff in the $s$-conductance permits a small exceptional set with little flow.  Its probability under an $M$-warm start is at most $M$ times its stationary mass. Therefore, we have the following connection between $s$-conductance and warm-start mixing time.

\begin{lemma}[Warm-start mixing from conductance {\cite{LovaszSimonovits93}}]
\label{lem:warm-mixing-from-conductance}
For a lazy reversible kernel and an $M$-warm start, $0<s<1/4$ implies
\[
 \dtv(\nu\sP^t,\pi)
 \le Ms+M\exp\bigl(-c\Phi_s(\sP)^2t\bigr).
\]
Choosing $s=\eta/(2M)$ makes the first term at most $\eta/2$. Making the exponential term at most $\eta/2$ gives
\[
 \tau_{\mix}(\eta,\nu,\pi;\sP)
 \lesssim\Phi_{\eta/(2M)}(\sP)^{-2}\log\frac{2M}{\eta}.
\]
In particular, the bound is uniform over two-warm starts when $s=\eta/4$.
\end{lemma}

\subsection{Strong coherent membership oracle and computational model}
\label{subsec:membership-oracle-model}

\paragraph{Strong coherent membership.}
For exactly represented dyadic coordinates $x$, the oracle acts as
\[
 O_K:\ket{x}_{\mathsf S}\ket{b}_{\mathsf Q}
 \longmapsto
 \ket{x}_{\mathsf S}\ket{b\oplus\one_K(x)}_{\mathsf Q}.
\]
Membership is answered exactly, including on the boundary, at any finite coordinate precision chosen by the algorithm. Each call counts as one query, independently of precision. A controlled query uses two ordinary calls: compute the membership bit into a fresh qubit, use it under the control, and uncompute it.

\paragraph{Numerical precision.}
We do not assume exact real arithmetic. Elementary functions and matrix functions are evaluated to finite accuracy, and rotations are synthesized over the fixed gate set. The total gate count and bit complexity of our finite-precision implementation are analyzed in~\cref{app:implementation}.
\paragraph{Dyadic rescaling.}
Given $B(x_c,r)\subseteq K\subseteq B(x_c,R)$, choose a dyadic center
$\widetilde x_c$ with $\|\widetilde x_c-x_c\|_2\le r/8$ and dyadic radii
\[
 \widetilde r\in[r/4,r/2],
 \qquad
 \widetilde R\in[R+r/8,\,2(R+r/8)].
\]
Then, we have
\[
 B(\widetilde x_c,\widetilde r)\subseteq K
 \subseteq B(\widetilde x_c,\widetilde R),
 \qquad
 \widetilde R/\widetilde r\le9R/r.
\]
Relabel these data as $x_c,r,R$. A lattice point $hk$ in the rescaled
body, with dyadic $h$ and $k\in\Z^d$, corresponds to the exactly dyadic
query point $x_c+rhk$ in the original coordinates.

\section{Stationary-state reflections from warm-start mixing}
\label{sec:warm-reflection}

A standard reflection circuit uses a spectral gap to distinguish the stationary eigenvalue from every nonstationary eigenvalue. Here, we only need to distinguish eigenvalues that are sufficiently far below $1$. Uniform warm-start mixing bounds the input's remaining slow component. We first prove this bound, then construct the reflection and control its use in the full algorithm.

\subsection{Warm-start mixing bounds the slow component}
For a finite target law $\pi$ that is positive on the sample space, write a joint vector as
\[
 \ket{\Psi}_{\mathsf{SA}}=\sum_x\ket{x}_{\mathsf{S}}\otimes\ket{\psi_x}_{\mathsf{A}}.
\]
The auxiliary vectors need not be normalized. We measure the concentration of the sample amplitudes relative to $\pi$ by the fourth-moment norm
\begin{equation}\label{eq:fourth-moment-norm}
 \|\ket{\Psi}_{\mathsf{SA}}\|_{4,\pi}
 =\left(\sum_x\frac{\|\ket{\psi_x}_{\mathsf{A}}\|^4}{\pi(x)}\right)^{1/4}.
\end{equation}
For a normalized input, measuring the sample register $\mathsf{S}$ gives the law $\mu_\Psi(x)=\|\ket{\psi_x}_{\mathsf{A}}\|^2$. Thus
\[
 \|\ket{\Psi}_{\mathsf{SA}}\|_{4,\pi}^4
 =\sum_x\frac{\mu_\Psi(x)^2}{\pi(x)}
 =1+\chi^2(\mu_\Psi\|\pi),\qquad
 \chi^2(\mu\|\pi):=\sum_x\frac{(\mu(x)-\pi(x))^2}{\pi(x)}.
\]
A bounded fourth-moment norm therefore means that the measured sample law has bounded relative second moment.

Equivalently, write $\ket{\psi_x}_{\mathsf{A}}=\sqrt{\pi(x)}\ket{h(x)}_{\mathsf{A}}$. Then
\[
 \|\ket{\Psi}_{\mathsf{SA}}\|^2=\int\|\ket{h(x)}_{\mathsf{A}}\|^2d\pi(x),
 \qquad \|\ket{\Psi}_{\mathsf{SA}}\|_{4,\pi}^4=\int\|\ket{h(x)}_{\mathsf{A}}\|^4d\pi(x).
\]
This is a homogeneous vector-valued $L^4(\pi)$ norm. In particular, $\|\ket{\pi}_{\mathsf{S}}\otimes\ket{w}_{\mathsf{A}}\|_{4,\pi}=\|\ket{w}_{\mathsf{A}}\|$. For the general measure-space statement below, the same definition is $(\int(d\mu_\Psi/d\pi)^2d\pi)^{1/4}$, where $\mu_\Psi$ has total mass $\|\ket{\Psi}\|^2$; the value is infinite when $\mu_\Psi$ is not absolutely continuous with respect to $\pi$.

\begin{theorem}[Slow-spectrum bound from uniform warm-start mixing]\label{thm:warm-low-spectrum}
Let $\sP$ be a positive, ergodic, reversible Markov operator with stationary law $\pi$, and suppose $\tau_{\mix,2}(\eta,\pi;\sP)\le t$, where $0<\eta<1/2$ and $t\in\N_+$. Let
\[
 \Pi_{\rm slow}=E_\sP\bigl([1-1/(2t),1)\bigr)
\]
be the spectral projection on $L^2(\pi)$, excluding the stationary eigenvalue. For every joint vector $\ket{\Psi}_{\mathsf{SA}}$ with a finite-dimensional or separable auxiliary register $\mathsf{A}$,
\begin{equation}\label{eq:warm-low-spectrum}
 \bigl\|(J_\pi\Pi_{\rm slow}J_\pi^\dagger\otimes\Id_{\mathsf{A}})
                          \ket{\Psi}_{\mathsf{SA}}\bigr\|
 \lesssim\eta^{1/4}\|\ket{\Psi}_{\mathsf{SA}}\|_{4,\pi}.
\end{equation}
\end{theorem}

\begin{proof}
We first obtain an $L^\infty(\pi)\to L^2(\pi)$ estimate from warm-start mixing, interpolate to $L^4(\pi)$, and then allow auxiliary registers. All function norms in this proof are with respect to $\pi$.

Let $g$ be nonzero, real, bounded, and mean zero. The density $r=1+g/\|g\|_\infty$ satisfies $0\le r\le2$ and $\E_\pi[r]=1$, so $d\nu=r\,d\pi$ is two-warm. Reversibility gives, for every measurable $B$,
\[
 (\nu\sP^t)(B)=\langle r,\sP^t\one_B\rangle_\pi
 =\langle\sP^t r,\one_B\rangle_\pi.
\]
Thus $\sP^t r$ is the density of $\nu\sP^t$ relative to $\pi$, and
\begin{align*}
 \|\sP^t g\|_1
 = &~ \|g\|_\infty\|\sP^t r-1\|_1\\*
 = &~ 2\|g\|_\infty\dtv(\nu\sP^t,\pi)
 \le2\eta\|g\|_\infty,
\end{align*}
where the last inequality follows from the warm-start mixing hypothesis. On the slow interval, $\lambda^t\ge(1-1/(2t))^t\ge1/2$. Positivity gives $\lambda^t\ge0$ elsewhere, hence $\Pi_{\rm slow}\preceq2\sP^t$ by spectral calculus. Therefore
\begin{align*}
 \|\Pi_{\rm slow}g\|_2^2
 = &~ \langle g,\Pi_{\rm slow}g\rangle_\pi\\*
 \leq &~ 2\langle g,\sP^t g\rangle_\pi\\*
 \leq &~ 2\|g\|_\infty\|\sP^t g\|_1\\*
 \leq &~ 4\eta\|g\|_\infty^2,
\end{align*}
where the first inequality follows from spectral domination, the second from H\"older's inequality, and the third from the preceding mixing estimate. The same bound holds for $g=0$. Since $\Pi_{\rm slow}\one=0$, centering costs at most a factor of two in $L^\infty$. Applying the real estimate to the real and imaginary parts of a bounded complex $f$ gives
\[
 \|\Pi_{\rm slow}\|_{L^\infty\to L^2}\lesssim\sqrt\eta,
 \qquad \|\Pi_{\rm slow}\|_{L^2\to L^2}\le1.
\]
Riesz--Thorin interpolation, \eqref{eq:riesz-thorin-prelim}, now yields
\begin{equation}\label{eq:scalar-fourth-spectrum}
 \|\Pi_{\rm slow}f\|_2
 \le\|\Pi_{\rm slow}\|_{L^2\to L^2}^{1/2}
       \|\Pi_{\rm slow}\|_{L^\infty\to L^2}^{1/2}\|f\|_4
 \lesssim\eta^{1/4}\|f\|_4.
\end{equation}

For a finite-dimensional auxiliary register, write
\[
 \ket{\Psi}_{\mathsf{SA}}=(J_\pi\otimes\Id_{\mathsf{A}})\sum_jh_j\otimes\ket j_{\mathsf{A}},
 \qquad \|\ket{\Psi}_{\mathsf{SA}}\|_{4,\pi}^4
 =\int\Bigl(\sum_j|h_j(x)|^2\Bigr)^2d\pi(x).
\]
Let $\varepsilon_j$ be independent uniform signs. Their cross terms have zero expectation, so
\begin{align*}
 \sum_j\|\Pi_{\rm slow}h_j\|_2^2
 = &~ \E_\varepsilon\!\left[\Bigl\|\Pi_{\rm slow}\sum_j\varepsilon_jh_j\Bigr\|_2^2\right]\\*
 \lesssim &~ \sqrt\eta\,\E_\varepsilon\!\left[\Bigl\|\sum_j\varepsilon_jh_j\Bigr\|_4^2\right]\\*
 \leq &~ C\sqrt\eta\left(\int\E_\varepsilon\!\left[\Bigl|\sum_j\varepsilon_jh_j(x)\Bigr|^4\right]d\pi(x)\right)^{1/2},
\end{align*}
where the second step follows from \eqref{eq:scalar-fourth-spectrum}, and the last from Jensen's inequality and Tonelli's theorem. For complex numbers $z_j$, expansion and $|\sum_jz_j^2|\le\sum_j|z_j|^2$ give
\[
 \E_\varepsilon\!\left[\Bigl|\sum_j\varepsilon_jz_j\Bigr|^4\right]
 =2\Bigl(\sum_j|z_j|^2\Bigr)^2+\Bigl|\sum_jz_j^2\Bigr|^2-2\sum_j|z_j|^4
 \le3\Bigl(\sum_j|z_j|^2\Bigr)^2.
\]
Consequently,
\[
 \bigl\|(J_\pi\Pi_{\rm slow}J_\pi^\dagger\otimes\Id_{\mathsf{A}})\ket{\Psi}_{\mathsf{SA}}\bigr\|^2
 =\sum_j\|\Pi_{\rm slow}h_j\|_2^2
 \lesssim\sqrt\eta\,\|\ket{\Psi}_{\mathsf{SA}}\|_{4,\pi}^2.
\]
Taking square roots proves the claim. For separable $\mathsf{A}$, project onto its first $r$ basis vectors. The vectors converge in Hilbert norm and the fourth-moment norms converge by monotone convergence; the projected images converge because $\Pi_{\rm slow}$ is a contraction. Passing to the limit gives the same bound. If the fourth-moment norm is infinite, the assertion is immediate.
\end{proof}

\begin{remark}
The exponent $1/4$ is sharp, as shown by the two-state example in \cref{rem:fourth-root-sharp}. Uniformity over warm starts is also essential: \cref{app:counterexample} gives a four-state chain in which mixing from one prescribed start does not control the slow component of its qsample.
\end{remark}

\subsection{The reflection circuit}
\begin{theorem}[Input-dependent stationary-state reflection]\label{thm:effective-reflection}
Let $\sP$ be a finite positive, ergodic, reversible chain with stationary law $\pi$ and $\tau_{\mix,2}(\eta,\pi;\sP)\le t$, where $0<\eta<1/2$ and $t\in\N_+$. Given controlled access to a $(1,b,0)$-block-encoding of $D_\pi$, for $0<\zeta<1/2$ there is a unitary $\widetilde R$ acting on the sample register $\mathsf{S}$ and a fresh register $\mathsf{E}$ of $a=b+O(1)$ ancilla qubits, using $O(\sqrt t\log(1/\zeta))$ block-encoding queries, such that for every joint state $\ket{\Psi}_{\mathsf{SA}}$,
\begin{equation}\label{eq:effective-reflection}
 \begin{aligned}
 \Bigl\|(\widetilde R\otimes\Id_{\mathsf{A}})
         (\ket{0^a}_{\mathsf{E}}\otimes\ket{\Psi}_{\mathsf{SA}})
       -\ket{0^a}_{\mathsf{E}}\otimes
          (R_\pi\otimes\Id_{\mathsf{A}})\ket{\Psi}_{\mathsf{SA}}\Bigr\|\lesssim \eta^{1/4}\|\ket{\Psi}_{\mathsf{SA}}\|_{4,\pi}+\zeta.
 \end{aligned}
\end{equation}
The same cost and bound hold for a controlled reflection, with its control included among the auxiliary registers when evaluating the fourth-moment norm.
\end{theorem}

\begin{proof}
We distinguish the stationary eigenvalue from the fast spectrum and bound the unresolved component using \cref{thm:warm-low-spectrum}. On the sample space, set
\[
 \Pi_{\rm fix}=\ket\pi\!\bra\pi,\qquad
 \widehat\Pi_{\rm slow}=J_\pi\Pi_{\rm slow}J_\pi^\dagger,\qquad
 \Pi_{\rm fast}=\Id-\Pi_{\rm fix}-\widehat\Pi_{\rm slow}.
\]
QSVT singular-value discrimination~\cite[Theorem~32]{GSLW19} distinguishes $\lambda=1$ from $0\le\lambda\le1-1/(2t)$ with error probability at most $\zeta^2/4$, using $O(\sqrt t\log(1/\zeta))$ block-encoding queries and $b+O(1)$ ancillas. Let $V$ be this circuit and let $Z_{\rm dec}$ apply $+1$ to the stationary decision and $-1$ to the other decision. Define
\[
 \widetilde R=V^\dagger Z_{\rm dec}V.
\]
The one-sided endpoint guarantee is exact at $1$, so $\widetilde R(\ket0_{\mathsf{E}}\otimes\ket\pi)=\ket0_{\mathsf{E}}\otimes\ket\pi$. For a fast vector $\ket\phi$, the wrong-decision component has norm at most $(\zeta/2)\|\ket\phi\|$. Since $R_\pi\ket\phi=-\ket\phi$,
\begin{align*}
 \bigl\|(\widetilde R-\Id_{\mathsf{E}}\otimes R_\pi)(\ket0_{\mathsf{E}}\otimes\ket\phi)\bigr\|
 = &~ \bigl\|(Z_{\rm dec}+\Id)V(\ket0_{\mathsf{E}}\otimes\ket\phi)\bigr\|\\*
 \leq &~ \zeta\|\ket\phi\|.
\end{align*}
This uniform fast-subspace estimate remains valid with auxiliary registers. On the slow subspace, the difference of two unitaries has norm at most two. Decomposing the input accordingly gives
\begin{align*}
 &\bigl\|[(\widetilde R-\Id_{\mathsf{E}}\otimes R_\pi)\otimes\Id_{\mathsf{A}}]
                 (\ket0_{\mathsf{E}}\otimes\ket\Psi_{\mathsf{SA}})\bigr\|\\*
 \leq &~ \zeta\|(\Pi_{\rm fast}\otimes\Id_{\mathsf{A}})\ket\Psi\|
          +2\|(\widehat\Pi_{\rm slow}\otimes\Id_{\mathsf{A}})\ket\Psi\|\\*
 \leq &~ \zeta+C\eta^{1/4}\|\ket\Psi\|_{4,\pi},
\end{align*}
where the last inequality follows from \cref{thm:warm-low-spectrum}.

To control the reflection, control only $Z_{\rm dec}$; the inactive branch is exactly $V^\dagger V=\Id$. Write the normalized input as $\ket\Psi_{\mathsf{CSA}}=\ket0_{\mathsf{C}}\ket{\psi_0}_{\mathsf{SA}}+\ket1_{\mathsf{C}}\ket{\psi_1}_{\mathsf{SA}}$. Only the unnormalized vector $\ket{\psi_1}$ contributes error. Pointwise projection onto control one cannot increase the fourth-moment norm, hence
\[
 \text{controlled error}
 \le C\eta^{1/4}\|\ket{\psi_1}\|_{4,\pi}+\zeta\|\ket{\psi_1}\|
 \le C\eta^{1/4}\|\ket\Psi\|_{4,\pi}+\zeta.
\]
No additional block-encoding queries are needed.
\end{proof}

A finite sample register encodes a box of lattice points, only some of which lie in the body. The chain acts on the feasible labels, whereas its circuit must be defined on the entire register. We use a block-encoding of $D_\pi\oplus0$ on the feasible and infeasible subspaces, respectively. The added eigenvalues are zero, so they belong to the fast subspace. The construction consequently approximates $R_\pi\oplus(-\Id_{\rm bad})$, with the stated bound on inputs supported on the feasible space.

\subsection{Fourth-moment norms along the reference circuit}
To apply the input-dependent error bound in~\cref{thm:effective-reflection}, we bound the fourth-moment norm of the input to each reflection in the reference circuit with exact reflections. Gates controlled by the sample register preserve its measurement distribution, whereas qsample reflections can change it. The next lemma bounds the resulting growth of the fourth-moment norm.

\begin{lemma}[Growth of the fourth-moment norm]\label{lem:fourth-growth}
Let $\{\pi_v:v\in I\}$ consist of one law or a pair of mutually absolutely continuous laws, with
\[
 \int\left(\frac{d\pi_v}{d\pi_u}\right)^2d\pi_u\le B^4
 \qquad(u,v\in I).
\]
Start a circuit from the state $\ket{\pi_r}_{\mathsf{S}}\otimes\ket{w}_{\mathsf{A}}$, where $r\in I$. Suppose each gate is sample-controlled or a selective phase through $\ket{\pi_v}_{\mathsf{S}}$, possibly controlled by an auxiliary projector. The sample register $\mathsf{S}$ is not exchanged with another sample register. After $q$ selective phases, the state $\ket{\Psi}_{\mathsf{SA}}$ satisfies
\begin{equation}\label{eq:fourth-growth}
 \|\ket{\Psi}_{\mathsf{SA}}\|_{4,\pi_u}\le B(1+2q)
 \qquad(u\in I).
\end{equation}
The same bound holds when the counted operations are controlled qsample reflections.
\end{lemma}

\begin{proof}
Fix a reference index $u\in I$. We show that sample-controlled gates preserve the fourth-moment norm and each selective phase increases it by at most $2B$.

For $V=\sum_x\ket x\!\bra x_{\mathsf{S}}\otimes V_{\mathsf{A}}(x)$ and $\ket\Psi=\sum_x\ket x_{\mathsf{S}}\otimes\ket{\psi_x}_{\mathsf{A}}$, unitarity gives
\[
 \|V\ket\Psi\|_{4,\pi_u}^4
 =\sum_x\frac{\|V_{\mathsf{A}}(x)\ket{\psi_x}\|^4}{\pi_u(x)}
 =\sum_x\frac{\|\ket{\psi_x}\|^4}{\pi_u(x)}
 =\|\ket\Psi\|_{4,\pi_u}^4.
\]
For a selective phase controlled by an auxiliary projector $\Pi_{\mathsf{A}}$, put
\[
 U_\theta=\Id+(e^{i\theta}-1)\ket{\pi_v}\!\bra{\pi_v}_{\mathsf{S}}\otimes\Pi_{\mathsf{A}},
 \qquad \ket a_{\mathsf{A}}=(\bra{\pi_v}_{\mathsf{S}}\otimes\Pi_{\mathsf{A}})\ket\Psi.
\]
The partial inner product is a contraction, so $\|\ket a\|\le\|\ket\Psi\|=1$. The relative-moment hypothesis then gives
\[
 \|\ket{\pi_v}_{\mathsf{S}}\otimes\ket a_{\mathsf{A}}\|_{4,\pi_u}^4
 =\|\ket a\|^4\int\left(\frac{d\pi_v}{d\pi_u}\right)^2d\pi_u
 \le B^4.
\]
Since $U_\theta\ket\Psi=\ket\Psi+(e^{i\theta}-1)\ket{\pi_v}\otimes\ket a$, the triangle inequality yields
\[
 \|U_\theta\ket\Psi\|_{4,\pi_u}
 \le\|\ket\Psi\|_{4,\pi_u}+|e^{i\theta}-1|B
 \le\|\ket\Psi\|_{4,\pi_u}+2B.
\]
The initial fourth-moment norm is at most $B$; after $q$ stages it is at most $B+2qB$.

A controlled reflection factors as
\begin{align*}
 \Id_{\mathsf{S}}\otimes(\Id_{\mathsf{A}}-\Pi_{\mathsf{A}})+R_{\pi_v}\otimes\Pi_{\mathsf{A}}
 = &~ (\Id_{\mathsf{S}}\otimes(\Id_{\mathsf{A}}-2\Pi_{\mathsf{A}}))
   (\Id-2\ket{\pi_v}\!\bra{\pi_v}_{\mathsf{S}}\otimes\Pi_{\mathsf{A}}).
\end{align*}
The first factor is an auxiliary unitary; the second is $U_\pi$. Thus the same increment bound applies. The reference index $u$ was arbitrary, so the bounds hold for every $u\in I$ without changing reference measures during the proof.
\end{proof}

\subsection{Accumulating implementation errors}\label{sec:composition}
A primitive may approximate its intended task even when all its elementary operations are exact. Implementing those operations adds a second error. We first bound one coherently selected reflection, then accumulate these two errors along a circuit.

\begin{lemma}[Error of a coherently selected reflection]\label{lem:coherent-reflection-error}
Let $R$ be the intended qsample reflection on $\mathsf{SA}$, possibly controlled by qubits in $\mathsf{A}$, and let $\mathsf{E}$ be fresh implementation workspace. Fix states $\ket{\Psi}_{\mathsf{SA}}$ and
\[
 \ket{b}_{\mathsf{B}}=\sum_z\sqrt{p_z}\ket{z}_{\mathsf{B}},\qquad p_z\ge0,\qquad\sum_zp_z=1,
\]
where $\mathsf{B}$ is the control register. For each control label $z$, let $\widetilde R_z$ be a unitary on $\mathsf{ESA}$, and define
\[
 \widetilde R_{\mathsf{B}}=\sum_z\ket{z}\!\bra{z}_{\mathsf{B}}\otimes\widetilde R_z.
\]
Let $\xi\ge0$ and $0\le{\iota}\le1$. Suppose a set $G$ of good labels satisfies
\begin{equation}\label{eq:composition-good-labels}
 \sum_{z\notin G}p_z\le{\iota},\qquad
 \bigl\|(\widetilde R_z-\Id_{\mathsf{E}}\otimes R)(\ket{0}_{\mathsf{E}}\otimes\ket{\Psi}_{\mathsf{SA}})\bigr\|\le\xi\quad(z\in G).
\end{equation}
Then
\begin{equation}\label{eq:composition-single-call}
 \begin{aligned}
 \Bigl\|(\widetilde R_{\mathsf{B}}-\Id_{\mathsf{B}}\otimes\Id_{\mathsf{E}}\otimes R)(\ket{b}_{\mathsf{B}}\otimes\ket{0}_{\mathsf{E}}\otimes\ket{\Psi}_{\mathsf{SA}})\Bigr\|\le\xi+2\sqrt{\iota}.
 \end{aligned}
\end{equation}
The same bound holds in trace distance after discarding workspace.
\end{lemma}

\begin{proof}
Put $\ket{e_z}=(\widetilde R_z-\Id_{\mathsf{E}}\otimes R)(\ket{0}_{\mathsf{E}}\otimes\ket{\Psi}_{\mathsf{SA}})$. The joint error vector is $\sum_z\sqrt{p_z}\ket{z}_{\mathsf{B}}\otimes\ket{e_z}_{\mathsf{ESA}}$. Orthogonality of the labels gives
\[
 \left\|\sum_z\sqrt{p_z}\ket{z}_{\mathsf{B}}\otimes\ket{e_z}_{\mathsf{ESA}}\right\|^2
 =\sum_zp_z\|\ket{e_z}\|^2
 \le\xi^2\sum_{z\in G}p_z+4\sum_{z\notin G}p_z
 \le\xi^2+4{\iota}.
\]
On a bad label, $\|\ket{e_z}\|\le2$ by unitarity. Taking square roots proves \eqref{eq:composition-single-call}. The trace-distance assertion follows from $\dtr(\ket{u}\!\bra{u},\ket{v}\!\bra{v})\le\|\ket{u}-\ket{v}\|$ for normalized vectors and contractivity under partial trace.
\end{proof}

The first lemma needs an independent control block only on its reference input. The next lemma carries forward the error of the actual, possibly entangled state. We use channels so that the same argument includes measurements and discards.

\begin{lemma}[Composition of approximate quantum channels]\label{lem:composition}
For $T\ge1$, let $\mathcal E_0,\ldots,\mathcal E_{T-1}$ be primitives with exact elementary operations, and let $\widetilde{\mathcal E}_k$ be their implementations. Both maps are completely positive and trace preserving. Let $\rho_0,\ldots,\rho_T$ be reference density operators on all registers just before the corresponding quantum channels. Suppose that for each $k$, the primitive approximation and implementation errors satisfy, respectively,
\begin{align}
 \dtr\bigl(\mathcal E_k(\rho_k),\rho_{k+1}\bigr)\leq &~ \varepsilon_k,
 \label{eq:primitive-error}\\*
 \dtr\bigl(\widetilde{\mathcal E}_k(\rho_k),\mathcal E_k(\rho_k)\bigr)\leq &~ \Delta_k,
 \label{eq:primitive-implementation-error}
\end{align}
The actual states evolve as $\widetilde\rho_{k+1}=\widetilde{\mathcal E}_k(\widetilde\rho_k)$. If $\dtr(\widetilde\rho_0,\rho_0)\le e_{\rm init}$, then
\begin{equation}\label{eq:composition}
 \dtr(\widetilde\rho_T,\rho_T)
 \le e_{\rm init}+\sum_{k=0}^{T-1}(\varepsilon_k+\Delta_k).
\end{equation}
\end{lemma}

\begin{proof}
The same implemented channel acts on the actual and reference inputs. Inserting its reference output and the exact primitive's output gives
\begin{align*}
 \dtr(\widetilde\rho_{k+1},\rho_{k+1})
 \leq &~ \dtr\bigl(\widetilde{\mathcal E}_k(\widetilde\rho_k),\widetilde{\mathcal E}_k(\rho_k)\bigr)\\*
 &\quad+\dtr\bigl(\widetilde{\mathcal E}_k(\rho_k),\mathcal E_k(\rho_k)\bigr)
       +\dtr\bigl(\mathcal E_k(\rho_k),\rho_{k+1}\bigr)\\*
 \leq &~ \dtr(\widetilde\rho_k,\rho_k)+\Delta_k+\varepsilon_k,
\end{align*}
where the last inequality follows from channel contractivity and the two local hypotheses, both evaluated only on $\rho_k$. Iteration gives
\[
 \dtr(\widetilde\rho_T,\rho_T)
 \le\dtr(\widetilde\rho_0,\rho_0)+\sum_{k=0}^{T-1}(\Delta_k+\varepsilon_k)
 \le e_{\rm init}+\sum_{k=0}^{T-1}(\Delta_k+\varepsilon_k).
\]
\end{proof}

To bound a primitive's implementation error, follow its exact internal circuit. At this finer level, the reference outputs are exact, so $\varepsilon_k=0$ in \cref{lem:composition}. A reflection replacement contributes $\xi_j+2\sqrt{{\iota_j}}$ by \cref{lem:coherent-reflection-error}; a numerical unitary replacement satisfies
\[
 \dtr(\widetilde U\rho\widetilde U^\dagger,U\rho U^\dagger)
 \le\|\widetilde U-U\|_{\op}.
\]
For a mixed working input, apply the pure-state reflection estimate to an untouched purification included in $\mathsf{A}$. The reference control block remains pure and independent. Summing the internal errors and then composing the primitives gives
\begin{equation}\label{eq:composition-reflection-budget}
 \dtr(\widetilde\rho_T,\rho_T)
 \le e_{\rm init}+e_{\rm prim}+e_{\rm arith}
       +\sum_{j=1}^{N}(\xi_j+2\sqrt{{\iota_j}}),
\end{equation}
where $e_{\rm prim}=\sum_k\varepsilon_k$ and $e_{\rm arith}$ sums the other numerical gate errors. Each error is counted once. Polynomial call counts and fourth-moment norms require only inverse-polynomial mixing and filtering accuracies; \cref{app:parameters} makes this choice explicit.

\section{Quantum Product Estimator}\label{sec:coherent-product}
In this section, we introduce a quantum subroutine to estimate the product of a sequence of expectations (\cref{lem:log-sum}) with streaming access to each distribution. 


For a random variable $Y(x)\in[0,1]$, define $V_Y$ to be the following sample-controlled rotation
\[
 V_Y(\ket{x}_{\mathsf{S}}\otimes\ket{0}_{\mathsf{Q}})
 =\ket{x}_{\mathsf{S}}\otimes\bigl(\sqrt{1-Y(x)}\ket{0}_{\mathsf{Q}}+\sqrt{Y(x)}\ket{1}_{\mathsf{Q}}\bigr),
\]
where $\mathsf{S}$ is the sample register and $\mathsf{Q}$ is an ancilla qubit. 

\begin{lemma}[Quantum product estimator]\label{lem:log-sum}
Let $n\ge1$, let $\pi_1,\ldots,\pi_n$ be probability distributions on finite sets, and let random variables $Y_\ell$ and signs $\sigma_\ell\in\{-1,1\}$ satisfy
\[
 0\le Y_\ell\le1,\qquad \E_{\pi_\ell}[Y_\ell]\ge\frac1{32}
 \qquad \forall 1\le\ell\le n.
\]
Given one independent qsample from each $\pi_\ell$ and access to the controlled $R_{\pi_\ell}$ and $V_{Y_\ell}$, for every $0<\varepsilon<1/2$, there is a quantum algorithm that estimates
\[
 \mathsf{L}=\sum_{\ell=1}^n\sigma_\ell(-\log (\E_{\pi_\ell}[Y_\ell])) = -\log\left(\prod_{\ell=1}^n \left(\E_{\pi_\ell}[Y_\ell]\right)^{\sigma_\ell} \right)
\]
to additive error $\varepsilon$ with probability at least $49/50$, using a total of $\widetilde O(n/\varepsilon)$ calls to the controlled qsample reflections and to $V_{Y_\ell}$ and its inverse. The qsamples may be supplied one at a time, and each is discarded before the next is needed; the phase register is shared between contributions. 
\end{lemma}


We construct and simulate a Hamiltonian for each logarithm in~\cref{sec:encode_log_phase}, then combine the evolutions to obtain the estimator in~\cref{sec:q_est_procedure}. \cref{sec:app_bridge_sampling} applies the estimator to ratios of normalization constants through bridge sampling~\cite{MengWong96}.

\subsection{Encoding the logarithms in eigenphases}\label{sec:encode_log_phase}
Fix one target $\pi$ and one random variable $Y$. Applying $V_Y$ to its qsample and a zero ancilla qubit gives
\[
 \ket{\phi}_{\mathsf{SQ}}=V_Y(\ket{\pi}_{\mathsf{S}}\otimes\ket{0}_{\mathsf{Q}}),
 \qquad \Pi=\Id_{\mathsf{S}}\otimes\ket{1}\!\bra{1}_{\mathsf{Q}},
 \qquad p=\bra\phi\Pi\ket\phi=\E_\pi[Y].
\]
The reflection through $\ket{\phi}$ can be implemented by
\begin{equation}\label{eq:marked-reflection}
 R_\phi=V_Y\bigl(R_\pi\otimes\ket{0}\!\bra{0}_{\mathsf{Q}}
 -\Id_{\mathsf{S}}\otimes\ket{1}\!\bra{1}_{\mathsf{Q}}\bigr)V_Y^\dagger,
\end{equation}
using one controlled call to $R_\pi$ and one call each to $V_Y$ and $V_Y^\dagger$. The same cost holds under an additional control.

Let $Q_\phi=\ket\phi\!\bra\phi_{\mathsf{SQ}}$. Then, we have
\begin{align*}
 G_\phi:=Q_\phi(\Id-\Pi)Q_\phi
 = \langle\phi|(\Id-\Pi)|\phi\rangle Q_\phi
 = (1-p)Q_\phi,
\end{align*}
Thus, $G_\phi^j=(1-p)^jQ_\phi$, and
\begin{equation}\label{eq:logarithm-polynomial}
 H_k=\sum_{j=1}^k\frac{G_\phi^j}{j}=g_k(p)Q_\phi,
 \qquad g_k(p)=\sum_{j=1}^k\frac{(1-p)^j}{j}.
\end{equation}
In particular, $\ket{\phi}_{\sf SQ}$ is an eigenstate of $H_k$ with eigenvalue $g_k(p)$:
\[
 e^{i\tau H_k}\ket\phi_{\mathsf{SQ}}
 =e^{i\tau g_k(p)}\ket\phi_{\mathsf{SQ}}.
\]

For $1\geq p\ge1/32$, the logarithm series has a nonnegative tail and satisfies
\begin{align}\label{eq:log-tail}
 -\log p-g_k(p)
 = \sum_{j=k+1}^\infty\frac{(1-p)^j}{j}
 \leq \frac{(1-p)^{k+1}}{(k+1)p}
 \le\frac{32}{k+1}e^{-(k+1)/32}.
\end{align}
Hence, degree $k=O(\log(1/\eta))$ suffices to give eigenvalue error at most $\eta$. 

Let $\alpha_k=\sum_{j=1}^k1/j$. The following lemma shows the cost of controlled Hamiltonian evolution for $H_k$.

\begin{lemma}[Controlled Hamiltonian evolution]\label{lem:logarithmic-phase}
For $k\ge1$, $\tau\in\R$, and $0<{\varepsilon_{\rm sim}}<1/2$, controlled access to $R_\phi$ and $\Id_{\mathsf{SQ}}-2\Pi$ gives a unitary $\widetilde U_\tau$ approximating $U_\tau=\ctrl_{\mathsf{C}}(e^{i\tau H_k})$. It acts on $\mathsf{C},\mathsf{S},\mathsf{Q}$ and a fresh register $\mathsf{E}$ of $a=O(k)$ qubits, and satisfies
\[
 \Bigl\|(\widetilde U_\tau\otimes\Id_{\mathsf{A}})(\ket{0^a}_{\mathsf{E}}\otimes\ket\Psi_{\mathsf{CSQA}})
 -\ket{0^a}_{\mathsf{E}}\otimes(U_\tau\otimes\Id_{\mathsf{A}})\ket\Psi_{\mathsf{CSQA}}\Bigr\|\le{\varepsilon_{\rm sim}}
\]
for every joint state $\ket\Psi_{\mathsf{CSQA}}$. It uses
\begin{equation}\label{eq:log-phase-cost}
 O\bigl(k(\alpha_k|\tau|+\log(1/{\varepsilon_{\rm sim}}))\bigr)
\end{equation}
controlled reflection queries. 
\end{lemma}

\begin{proof}
We first construct a block-encoding of $ H_k$, then apply QSVT-based Hamiltonian simulation to implement its controlled evolution. 

For a reflection $R$ on $\mathsf{S},\mathsf{Q}$, a Hadamard, a controlled $R$, and a second Hadamard on a fresh qubit give
\[
 W_R=({\rm H}\otimes\Id_{\mathsf{SQ}})\ctrl(R)({\rm H}\otimes\Id_{\mathsf{SQ}}),
 \qquad
 (\bra{0}\otimes\Id_{\mathsf{SQ}})W_R(\ket{0}\otimes\Id_{\mathsf{SQ}})
 =\frac{\Id_{\mathsf{SQ}}+R}{2}.
\]
Taking $R=R_\phi$ and $R=\Id_{\mathsf{SQ}}-2\Pi$ gives block-encodings of $Q_\phi$ and $\Id_{\mathsf{SQ}}-\Pi$. Their product on three separate ancillas is a $(1,3,0)$-block-encoding $W_G$ of $G_\phi=Q_\phi (\Id - \Pi) Q_\phi$ using $O(1)$ reflection calls. Then, we can block-encode $G_\phi^j$ for $1\le j\le k$ using $3k$ ancilla qubits, and let $W_j$ denote the resulting unitary:
\[
 (\bra{0^{3k}}\otimes\Id_{\mathsf{SQ}})W_j
                    (\ket{0^{3k}}\otimes\Id_{\mathsf{SQ}})=G_\phi^j,\qquad \forall 1\leq j\leq k.
\]
We can apply the linear combination of unitaries (LCU) to block-encode the Hamiltonian $H_k$. Specifically, let 
\begin{align*}
    \alpha_k=\sum_{j=1}^k\frac{1}{j}\le1+\log k,
\end{align*}
and prepare $\sum_{j=1}^k(j\alpha_k)^{-1/2}\ket{j}_{\mathsf{J}}$ in an index register $\mathsf{J}$ of $\lceil \log_2 k\rceil$ ancilla qubits. Then, apply $\sum_j\ket{j}\!\bra{j}_{\mathsf{J}}\otimes W_j$, and uncompute the register ${\sf J}$. This gives an $(\alpha_k, O(k),0)$-block-encoding for $H_k$ using $O(k)$ reflections~\cite[Lemmas~52--53]{GSLW19}. 

Now, we can apply the QSVT Hamiltonian simulation~\cite[Corollary~60]{GSLW19} to $\ket{1}\bra{1}_{\sf C}\otimes H_k$ and obtain a unitary $\widetilde{U}_\tau$ that approximates
\[
 U_\tau=\ket{0}\!\bra{0}_{\mathsf{C}}\otimes\Id_{\mathsf{SQ}}
       +\ket{1}\!\bra{1}_{\mathsf{C}}\otimes e^{i\tau H_k}.
\]
It uses $O(\alpha_k|\tau|+\log(1/{\varepsilon_{\rm sim}}))$ queries to the block-encoding of $H_k$, and satisfies the following error guarantee:
\[
 B_\tau=(\bra{0^a}_{\mathsf{E}}\otimes\Id_{\mathsf{CSQ}})\widetilde U_\tau
                          (\ket{0^a}_{\mathsf{E}}\otimes\Id_{\mathsf{CSQ}}),
 \qquad \|B_\tau-U_\tau\|_{\op}\le{\varepsilon_{\rm sim}}^2/2.
\]
It also implies a full-state output error guarantee: for any state $\ket{\Psi}_{\sf CSQA}$,
\begin{align*}
 &\Bigl\|(\widetilde U_\tau\otimes\Id_{\mathsf{A}})
          (\ket{0^a}_{\mathsf{E}}\otimes\ket{\Psi}_{\mathsf{CSQA}})
       -\ket{0^a}_{\mathsf{E}}\otimes(U_\tau\otimes\Id_{\mathsf{A}})
                              \ket{\Psi}_{\mathsf{CSQA}}\Bigr\|^2\\*
 = &~ 2-2\operatorname{Re}\bra{\Psi}
                    (U_\tau^\dagger B_\tau\otimes\Id_{\mathsf{A}})\ket{\Psi}\\*
 = &~ 2\operatorname{Re}\bra{\Psi}
          \bigl(U_\tau^\dagger(U_\tau-B_\tau)\otimes\Id_{\mathsf{A}}\bigr)
                    \ket{\Psi}\\*
 \leq &~ 2\|B_\tau-U_\tau\|_{\op}\le{\varepsilon_{\rm sim}}^2,
\end{align*}
where the first step follows by expanding the squared norm and using unitarity, the second step follows from $U_\tau^\dagger U_\tau=\Id$, and the third step uses the operator norm and the approximation error of the Hamiltonian simulation. 

Finally, since each block-encoding query of $H_k$ costs $O(k)$ reflections, in total the number of reflections is $O(k(\alpha_k|\tau|+\log(1/{\varepsilon_{\rm sim}})))$, which proves \eqref{eq:log-phase-cost}.
\end{proof}

\subsection{Quantum estimation procedure}\label{sec:q_est_procedure}
In this subsection, we present the quantum algorithm for estimating the logarithm of the expectation products (\cref{alg:add-log}) and finish the proof of~\cref{lem:log-sum}.

For $n$ distributions and target additive error $\varepsilon$, choose the parameters as follows:
\begin{equation}\label{eq:log-clock-parameters}
 \begin{gathered}
 k=\left\lceil32\log\frac{256n}{\varepsilon}\right\rceil,
 \qquad M=2^{\lceil\log_2(2^{16}n/\varepsilon)\rceil},\\
 \tau_j=\frac{2\pi2^j}{64n},\qquad 0\le j<\log_2M.
 \end{gathered}
\end{equation}

\cref{alg:add-log} uses quantum phase estimation~\cite{kitaev1995quantum}, with each controlled evolution split across the $n$ contributions. The \textsc{Add} calls process the qsamples one at a time. We initialize $\mathsf F$ with a known phase offset of $1/4$ to keep the signed sum away from the wraparound point; the final output subtracts this offset. The caller fixes $n,k,M$ once; every call to \textsc{Add} uses these same parameters.

\begin{algorithm}[H]
\caption{Quantum product estimator}
\label{alg:add-log}
\begin{algorithmic}[1]
\Procedure{Add}{$\mathsf{F}$, $\mathsf{S}$; $R_\pi$, $V_Y$, $\sigma$}\Comment{${\sf S}$ contains an independent qsample $\ket{\pi}$, $\sigma\in \{\pm 1\}$}
 \State Append $\ket0_{\mathsf{Q}}$ and apply $V_Y$ to $\mathsf{SQ}$ to obtain $\ket\phi_{\mathsf{SQ}}$.
 \For{$j=0,\ldots,\log_2M-1$}
  \State Set $\tau_j=2\pi2^j/(64n)$.
  \State Apply the unitary of $\ctrl_{\mathsf{F}_j}(e^{i\sigma\tau_j H_k})$ using fresh ancillas $\mathsf E$ and error $1/(100n\log_2M)$.\Comment{\cref{lem:logarithmic-phase}}
  \State Discard $\mathsf{E}$.
 \EndFor
 \State Discard $\mathsf{S}$ and $\mathsf{Q}$.
\EndProcedure
\Procedure{Estimation}{$n$, $\{{\sf S}_\ell, R_{\pi_\ell}, V_{Y_\ell},\sigma_\ell\}_{\ell\in [n]}$, $\varepsilon$}
\State Set $k$ and $M$ as~\eqref{eq:log-clock-parameters}.
\State Prepare the following state in a fresh register ${\sf F}$:
\begin{align*}
    \frac{1}{\sqrt{M}}\sum_{z=0}^{M-1}e^{2\pi iz/4}\ket z_{\mathsf{F}}
\end{align*}
\For{$\ell=1,\ldots,n$}
    \State \textsc{Add}(${\sf F}$, ${\sf S}_\ell$; $R_{\pi_\ell}$, $V_{Y_\ell}$, $\sigma_\ell$).
\EndFor
\State Apply the inverse QFT to ${\sf F}$
\State Measure ${\sf F}$ and let $Y\in \{0,\dots,M-1\}$ be the outcome
\State \Return $64n(Y/M-1/4)$
\EndProcedure
\end{algorithmic}
\end{algorithm}

\begin{proof}[Proof of \cref{lem:log-sum}]
Write $p_\ell=\E_{\pi_\ell}[Y_\ell]$. We first analyze an ideal case where the Hamiltonian evolutions in~\cref{alg:add-log} are implemented perfectly. On a state $\ket{z}_{\sf F}$ with $z=\sum_jz_j2^j$, the procedure \textsc{Add} acts as
\begin{align*}
 \ket z_{\mathsf{F}}\ket{\phi_\ell}_{\mathsf{SQ}}
 \longmapsto &~ \exp\!\left(i\sigma_\ell g_k(p_\ell)\sum_jz_j\tau_j\right)
 \ket z_{\mathsf{F}}\ket{\phi_\ell}_{\mathsf{SQ}}\\*
 = &~ \exp\left(\frac{2\pi iz\sigma_\ell g_k(p_\ell)}{64n}\right)\ket z_{\mathsf{F}}\ket{\phi_\ell}_{\mathsf{SQ}}.
\end{align*}
By linearity, it only changes the phases in $\mathsf{F}$, leaving $\ket{\phi_\ell}_{\sf SQ}$ unchanged. Thus,  discarding the ${\sf SQ}$ registers does not affect the accumulated phase. After all $n$ applications, the ideal ${\sf F}$ register is
\begin{equation}\label{eq:accumulated-clock}
 \frac1{\sqrt M}\sum_{z=0}^{M-1}e^{2\pi iz\theta}\ket z_{\mathsf{F}},
 \qquad \theta=\frac14+\frac{\mathsf{L}_k}{64n},
 \qquad \mathsf{L}_k=\sum_{\ell=1}^n\sigma_\ell g_k(p_\ell).
\end{equation}
Under the assumption that $p_\ell\geq 1/32$, it holds that $0\le g_k(p_\ell)\le-\log p_\ell\le\log32$. By~\eqref{eq:log-tail} and the choice of $k$ in~\eqref{eq:log-clock-parameters}, we have
\begin{align*}
 \theta\in &~ \left[\frac14-\frac{\log32}{64},\frac14+\frac{\log32}{64}\right]
 \subseteq[1/8,3/8],\\
 |\mathsf{L}_k-\mathsf{L}|
 \leq &~\sum_{\ell=1}^n|g_k(p_\ell)+\log p_\ell|
 \leq \frac{32n}{k+1}e^{-(k+1)/32}\le\frac\varepsilon8.
\end{align*}
Apply the inverse quantum Fourier transform to $\mathsf{F}$ and measure it. Let $Y\in \{0,\dots,M-1\}$ be its measurement outcome. By the standard analysis of QPE, we get that with probability at least $99/100$, $|Y/M-\theta|_\circ \leq 64/M$, where $|a|_\circ:=\min_{k\in \mathbb{Z}}\{|a+k|\}$ is the wrap-around distance. Since $\theta\in [1/8,3/8]$ and $M>512$ as in~\eqref{eq:log-clock-parameters}, this implies that $|Y/M-\theta|\leq 64/M$. Then, let
\begin{align*}
    \widehat{\mathsf{L}}=64n\left(\frac{Y}{M}-\frac{1}{4}\right)
\end{align*}
be the output of~\textsc{Estimation}. Since $M\ge2^{16} n/\varepsilon$, we have
\begin{align*}
 \left|\widehat{\mathsf{L}}-\mathsf{L}\right|\leq  64n\left|\frac{Y}{M}-\theta\right|+|\mathsf{L}_k-\mathsf{L}|
 \leq \frac{4096n}{M}+\frac\varepsilon8
 \le\frac{3\varepsilon}{16}<\varepsilon.
\end{align*}

We simulate the Hamiltonian evolutions in~\cref{alg:add-log} using~\cref{lem:logarithmic-phase}. Each of the $n\log_2M$ evolutions has error at most $(100n\log_2M)^{-1}$. By~\cref{lem:composition}, the trace distance between the simulated and ideal final states of $\mathsf F$ is at most
\begin{align*}
 \sum_{\ell=1}^n\sum_{j=0}^{\log_2M-1}\frac1{100n\log_2M}
 = \frac1{100},
\end{align*}
so simulation decreases the success probability by at most $1/100$. Therefore,
\[
 \Pr\left[\left|\widehat{\mathsf{L}}-\mathsf{L}\right|\le\varepsilon\right]
 \ge\frac{99}{100}-\frac1{100}=\frac{49}{50}.
\]

For the complexity of~\cref{alg:add-log}, each \textsc{Add} has total Hamiltonian evolution time
\[
 T_{\rm total}=\sum_{j=0}^{\log_2M-1}|\tau_j|
 =\frac{2\pi(M-1)}{64n}=O(1/\varepsilon).
\]
By~\eqref{eq:marked-reflection}, each reflection about $\ket\phi_{\mathsf{SQ}}$ uses one call each to $R_\pi$, $V_Y$, and $V_Y^\dagger$. Thus, by~\cref{lem:logarithmic-phase}, the \textsc{Estimation} procedure uses
\begin{align*}
    O(nk(\alpha_k T_{\rm total} + \log(100 n \log_2M))\cdot \log_2 M)=\widetilde{O}(n/\varepsilon)
\end{align*}
calls to $R_\pi$, $V_Y$, and $V_Y^\dagger$ in total, which follows from $k=O(\log (n/\varepsilon))$, $\alpha_k=O(\log k)$, $M=O(n/\varepsilon)$. 

%
\end{proof}

\subsection{Application to ratios of normalization constants}\label{sec:app_bridge_sampling}
We apply the product estimator to ratios of normalization constants along an annealing schedule. 

Suppose we have designed an $m$-stage cooling schedule with distributions $\pi_0,\ldots,\pi_m$ on a common finite set and positive observables $X_i$ satisfying
\begin{equation}\label{eq:bridge-family}
 \pi_{i+1}(x)=\frac{X_i(x)}{\mu_i}\pi_i(x),\qquad
 \mu_i=\E_{\pi_i}[X_i].
\end{equation}
The goal is to estimate
\begin{align*}
    R=\prod_{i=0}^{m-1}\mu_i.
\end{align*}
We also assume the following relative second moment bounds:
\begin{align}\label{eq:sec_moment_assump}
 \E_{\pi_i}\!\left[\left(\frac{X_i}{\mu_i}\right)^2\right]
 =\E_{\pi_i}\!\left[\left(\frac{\pi_{i+1}}{\pi_i}\right)^2\right]\le4
 \qquad\forall 0\le i<m.
\end{align}


The observables $X_i$ need not lie in $[0,1]$, as required by~\cref{lem:log-sum}. We therefore replace each $X_i$ by two bounded observables, using a positive scale $c_i$ to be chosen below:
\begin{equation}\label{eq:bridge-marks}
 \begin{aligned}
 Y_i^+(x)= &~ \frac{X_i(x)}{c_i+X_i(x)},\qquad &
 p_i^+= &~ \E_{\pi_i}[Y_i^+],\\
 Y_i^-(x)= &~ \frac{c_i}{c_i+X_i(x)},\qquad &
 p_i^-= &~ \E_{\pi_{i+1}}[Y_i^-].
 \end{aligned}
\end{equation}
The expectations use different distributions. Using~\eqref{eq:bridge-family}, we have
\begin{align*}
 p_i^-
 = \sum_x\pi_{i+1}(x)\frac{c_i}{c_i+X_i(x)}
 =  \sum_x\pi_i(x)\frac{X_i(x)}{\mu_i}\frac{c_i}{c_i+X_i(x)}
 = \frac{c_i}{\mu_i}p_i^+,
\end{align*}
which gives the bridge-sampling identity~\cite{MengWong96}:
\begin{equation}\label{eq:bridge-ratio-identity}
 p_i^-=(c_i/\mu_i)p_i^+,\qquad
 \mu_i=c_i\frac{p_i^+}{p_i^-}.
\end{equation}
Note that both $p_i^\pm$ lie in $[0,1]$, regardless of the size of $X_i$. 

Taking logarithms of~\eqref{eq:bridge-ratio-identity} gives
\begin{equation}\label{eq:bridge-log-sum}
 \log R=\sum_{i=0}^{m-1}\log c_i-\mathsf{L},\qquad
 \mathsf{L}=\sum_{i=0}^{m-1}(-\log p_i^++\log p_i^-).
\end{equation}
The first term in $\log R$ can be computed classically. To apply \cref{lem:log-sum} to estimate ${\sf L}$, we also need to guarantee that $p_i^\pm\geq 1/32$. 

We claim that taking $c_i\in[\mu_i/2,2\mu_i]$ suffices. Indeed,
\begin{align*}
 \mu_i^2
 = \left(\E_{\pi_i}\!\left[\sqrt{\frac{X_i}{c_i+X_i}}\sqrt{X_i(c_i+X_i)}\right]\right)^2
 \leq  p_i^+\,\E_{\pi_i}[X_i(c_i+X_i)]
 =  p_i^+\bigl(c_i\mu_i+\E_{\pi_i}[X_i^2]\bigr),
\end{align*}
where the first step follows from Cauchy--Schwarz, and the second step uses $\E_{\pi_i}[X_i]=\mu_i$. Hence, by~\eqref{eq:sec_moment_assump}, we have
\begin{align*}
 p_i^+\geq &~ \frac{\mu_i^2}{c_i\mu_i+\E_{\pi_i}[X_i^2]}
 \ge\frac1{c_i/\mu_i+4}\ge\frac16,\\*
 p_i^-= &~ \frac{c_i}{\mu_i}p_i^+
 \ge\frac{c_i/\mu_i}{c_i/\mu_i+4}\ge\frac19.
\end{align*}

The following lemma obtains preliminary constant-factor estimates $c_i$, called pilot estimates, using the median-of-means estimator with $\widetilde O(1)$ samples per stage.

\begin{lemma}[Pilot estimates]\label{lem:pilots}
If $\E_{\pi_i}[(X_i/\mu_i)^2]\le4$ for every $0\le i<m$, then $P=O(\log(m+1))$ independent samples per stage $i$ suffice to find $c_i\in[\mu_i/2,2\mu_i]$ simultaneously for all $i<m$, with probability at least $99/100$. The same conclusion holds when each sampled value $X_i(z)$ is replaced by $\widetilde X_i(z)$ satisfying $|\widetilde X_i(z)-X_i(z)|\le |X_i(z)|/16$.
\end{lemma}
\begin{proof}
Fix any $i\in \{0,\dots,m-1\}$.  Divide the independent samples into batches of size $256$.  For batch $j$, let
\[
 M_j=\frac1{256}\sum_{r=1}^{256}X_i(x_{j,r}),\qquad
 \widetilde M_j=\frac1{256}\sum_{r=1}^{256}\widetilde X_i(x_{j,r}),
\]
where each $x_{j,r}$ is an independent sample from $\pi_i$. By the second-moment hypothesis, we have
\[
 \E[M_j]=\mu_i,\qquad
 \Var(M_j)=\frac{\Var_{\pi_i}(X_i)}{256}
 \le\frac{3\mu_i^2}{256}.
\]
Chebyshev's inequality therefore yields
\[
 \Pr\bigl[|M_j-\mu_i|>\mu_i/4\bigr]
 \le\frac{16\Var(M_j)}{\mu_i^2}
 \le\frac3{16}<\frac14.
\]
On the complementary event, finite evaluation preserves the desired constant-factor interval, because
\[
 \frac{\mu_i}{2}<\frac{15}{16}\frac{3\mu_i}{4}
 \le\widetilde M_j\le\frac{17}{16}\frac{5\mu_i}{4}<2\mu_i.
\]

Take the smallest odd number $t$ of batches with $t\ge8\log(100m)$, and let $c_i$ be their median.  Each batch is good with probability at least $3/4$, independently.  A median can leave $[\mu_i/2,2\mu_i]$ only when at least half of the batches are bad. Thus, by Hoeffding's inequality,
\[
 \Pr[c_i\notin[\mu_i/2,2\mu_i]]
 \le\exp\bigl(-2t(1/2-1/4)^2\bigr)
 =e^{-t/8}\le\frac1{100m}.
\]
The number of samples is $P=256t=O(\log(m+1))$. 

Finally, by a union bound over all $i\in \{0,\dots,m-1\}$, we prove the lemma.
\end{proof}

To implement the qsample reflections in~\cref{alg:add-log} using~\cref{thm:effective-reflection}, we need a polynomial bound on the fourth-moment norm of each reference input.

\begin{remark}[Growth of the fourth-moment norm in~\cref{alg:add-log}]\label{subsec:estimator-moments}
Consider one \textsc{Add} call in the reference computation, with exact qsample reflections and an independent input qsample $\ket\pi$ in $\mathsf S$. Let $\mathsf A$ contain all other registers. After purifying their state, we can write the input as 
\begin{align*}
    \ket{\Psi_0}_{\mathsf{SA}}
 =\ket\pi_{\mathsf S}\otimes\ket w_{\mathsf A},
 \qquad
 \|\ket{\Psi_0}_{\mathsf{SA}}\|_{4,\pi}^4
 =\sum_{x:\pi(x)>0}\frac{\pi(x)^2}{\pi(x)}
 =1.
\end{align*}
The procedure uses the following kinds of operations:
\begin{itemize}
    \item The sample-controlled rotation $V_Y$ is controlled by ${\sf S}$ and acts on ${\sf Q}$ (contained in ${\sf A}$).
    \item $R_\phi$ has one controlled qsample-reflection $R_\pi$ on ${\sf S}$ (see~\eqref{eq:marked-reflection}).
    \item $\Id_{\sf SQ}-2\Pi$ acts only on ${\sf Q}$ (contained in ${\sf A}$).
    \item The block-encoding for $H_k$ and the Hamiltonian simulation combine the above operations with elementary gates on ${\sf A}$.
\end{itemize}
Sample-controlled operations of the form $\sum_x\ket{x}\bra{x}_{\sf S}\otimes U_x$ preserve the fourth-moment norm. Applying~\cref{lem:fourth-growth} to the single target $\pi$ with $B=1$ therefore gives, after $q$ qsample reflections,
\begin{equation}\label{eq:estimator-internal-moment}
 \|\ket{\Psi_q}_{\mathsf{SA}}\|_{4,\pi}
 \le1+2q\le\widetilde{O}(1/\varepsilon),
\end{equation}
where the last step uses the total number of queries to $R_\pi$ in~\textsc{Add}. This count restarts for each reference \textsc{Add} call because it receives an independent qsample. Any discrepancy of the implemented input, including entanglements with other registers, is carried forward by the trace-distance bound in~\cref{lem:composition}.
\end{remark}

\section{Lattice Metropolis walks and coherent preconditioning}\label{sec:finite-walk}

We construct a walk on one fixed lattice, block-encode its discriminant, and prove a uniform warm-start mixing bound. We then show how covariance blocks supply the preconditioners coherently and how the blocks advance along the cooling schedule.

\subsection{Lattice discretization and the random walk}\label{sec:lattice_metropolis}
This section treats a general $D$-dimensional body; in the volume algorithm we apply the construction to $K_0$ with tilt vectors $b_i=a_i e_0$. Let $K\subset\R^D$ be a convex body with known rational data $x_c,r_0,R_0$ satisfying $B(x_c,r_0)\subseteq K\subseteq B(x_c,R_0)$, and let $b_0,\ldots,b_m$ be rational tilt vectors in $\mathbb Q^D$. Choose one dyadic spacing $h>0$ and define the lattice
\[
 \Lambda=h\Z^D,\qquad\Omega_\Lambda=K\cap\Lambda.
\]
The continuous target distribution and its discretization are
\begin{equation}\label{eq:main-finite-logaffine-target}
 \begin{aligned}
 Z_{b_i}= &~ \int_K e^{-\langle b_i,x\rangle}\,dx,&
 \pi_{b_i}(dx)= &~ Z_{b_i}^{-1}e^{-\langle b_i,x\rangle}\one_K(x)\,dx,\\
 Z_{b_i,\Lambda}= &~ h^D\sum_{x\in K\cap h\Z^D}e^{-\langle b_i,x\rangle},&
 \pi_{i,\Lambda}(x)= &~ h^De^{-\langle b_i,x\rangle}/Z_{b_i,\Lambda}.
 \end{aligned}
\end{equation}
The factor $h^D$ is the volume of a lattice cell, so $Z_{b_i,\Lambda}$ approximates the continuous integral $Z_{b_i}$ as the grid becomes finer.

To obtain fast mixing, choose a positive-definite rational preconditioner $A$ approximating the continuous covariance $\Sigma_i=\Cov_{\pi_{b_i}}(X)$. Indeed,
\[
 A\asymp \Sigma_i
 \quad\Longrightarrow\quad
 \Cov_{\pi_{b_i}}(A^{-1/2}X)=A^{-1/2}\Sigma_i A^{-1/2}\asymp \Id.
\]
Here $\asymp$ denotes a two-sided constant-factor comparison in the Loewner order. We implement this preconditioning through the proposal distribution in each Markov chain update, while keeping the lattice and target distribution fixed. For the finite implementation, \Cref{prop:finite-affine-branches} shows that we can compute a rational whitening matrix $W_A$ and its inverse $T_A=W_A^{-1}$ such that
\[
 \|(W_A-A^{-1/2})A^{1/2}\|_{\op}\le1/10,
\]
and for an accurate covariance estimate $A$, $W_AX$ has covariance comparable to $\Id$.

In the local proposal of our \emph{lattice Metropolis walk} (lines~\ref{ln:local_proposal_1}-\ref{ln:local_proposal_2} in~\cref{alg:finite-metropolis}), we will first choose one of $J$ dyadic radii $\delta_j$ with probability $w_j\in[1/(2J),2/J]$. The finite-seed sampler in~\cref{lem:finite-direction-preparation} produces a rational vector $U\in B(0,\delta_j)$ whose rounded displacement approximates the ideal proposal from a uniform ball sample. We set
\[
 v=h\left\lfloor h^{-1}T_AU+\tfrac12\one\right\rfloor.
\]
The proposed displacement is $\sigma v$, where  $\sigma\in\{-1,1\}$ is an independent uniform sign. This makes the finite displacement law exactly symmetric, including at rounding ties. 


\begin{algorithm}[H]
\caption{One step of the lattice Metropolis walk $\sP_{i,\Lambda,A}$}
\label{alg:finite-metropolis}
\begin{algorithmic}[1]
\Require Feasible point $x$; tilt $b_i$; fixed lattice, radii, and a preconditioner $A$.
\State With probability $1/2$, return $x$.\label{ln:holding_branch}
\State Otherwise choose a \textcolor{red}{local} proposal with probability $3/4$, or a \textcolor{blue}{connector} proposal with probability $1/4$.
\If{the proposal is \textcolor{red}{local}}\label{ln:local_proposal_1}
 \State Sample a radius $\delta_j$, use the finite-seed sampler of~\cref{lem:finite-direction-preparation} to obtain a rational $U\in B(0,\delta_j)$, and set $v=h\lfloor h^{-1}T_AU+\frac12\one\rfloor$.\label{ln:local_proposal_2}
\Else
 \State Sample $n$ uniformly from a fixed integer box containing the coefficients of the difference of every two feasible lattice points, and set $v=hn$.\label{ln:connector_proposal}
\EndIf
\State Sample $\sigma\in\{-1,1\}$ uniformly and propose $y=x+\sigma v$.
\State Return $y$ with probability $\one_K(y)\min\{1,e^{-\sigma\langle b_i,v\rangle}\}$, and return $x$ otherwise.\label{ln:metropolis_step}
\end{algorithmic}
\end{algorithm}

Note that the connector proposal (line~\ref{ln:connector_proposal}) ensures that every feasible point has positive probability of moving to any other feasible point, so that the Markov chain is irreducible.

\subsection{Block-encoding of the discriminant}\label{sec:block_encode_discriminant}
The following lemma constructs the block-encoding of the discriminant of the lattice Metropolis walk using a coherently prepared finite random seed.

\begin{lemma}[Random seed-preserving Metropolis block-encoding]
\label{lem:metropolis-block-encoding}
Let $\Omega$ be a finite feasible subset of a sample register, with positive weights proportional to $e^{-\langle b,x\rangle}$. Suppose a finite random seed $s$ with a state-independent implementable law $p_s$ computes a displacement $v_s$, and proposals use an independent uniform sign.  The discriminant $D_\pi$ has a $1$-block-encoding that uses $O(1)$ random-seed preparations, reversible displacement computations, and membership queries.  The block-encoding can represent $D_\pi\oplus0$ on the full encoded sample register.
\end{lemma}

\begin{proof}
Include all branch choices in~\cref{alg:finite-metropolis} in the seed $s$, with $v_s=0$ on a holding seed. For a feasible $x$, let
\[
 a(x,s,\sigma)=\one_\Omega(x+\sigma v_s)
                   \min\{1,e^{-\sigma\langle b,v_s\rangle}\}.
\]
Set $a=0$ for an infeasible $x$ or a proposal outside the encoded coordinate range. 

Let an operator $T$ prepare the random seed, an independent uniform sign, and coherently flip a coin qubit:
\begin{equation}\label{eq:metropolis-isometry}
 T\ket{x}_{\mathsf{S}}=\sum_{s,\sigma}\sqrt{p_s/2}\ket{s,\sigma}_{\sf E}\otimes 
 \Bigl(\sqrt{1-a(x,s,\sigma)}\ket{0}_{\mathsf{E}}
       +\sqrt{a(x,s,\sigma)}\ket{1}_{\mathsf{E}}\Bigr)\otimes\ket{x}_{\mathsf{S}}.
\end{equation}
The operator $T$ may use an additional arithmetic workspace in the ${\sf E}$ register, which will be uncomputed after the coin flip. This is an isometry because the square amplitudes of the random seed sum to one.

Let an involution $F$ leave rejected proposals fixed and reverse accepted proposals:
\begin{equation}\label{eq:metropolis-involution}
 \begin{aligned}
 F(\ket{s,\sigma,0}_{\mathsf{E}}\otimes\ket{x}_{\mathsf{S}})= &~ \ket{s,\sigma,0}_{\mathsf{E}}\otimes\ket{x}_{\mathsf{S}},\\
 F(\ket{s,\sigma,1}_{\mathsf{E}}\otimes\ket{x}_{\mathsf{S}})= &~ \ket{s,-\sigma,1}_{\mathsf{E}}\otimes\ket{x+\sigma v_s}_{\mathsf{S}}.
 \end{aligned}
\end{equation}
Addition is modular on the full register. The overflow moves will be rejected since they have $a=0$. Thus, every accepted move coincides with its intended physical displacement, while sign reversal ensures $F^2=\Id$ for every encoded input.

For distinct feasible $x,y$, we have
\begin{align}\label{eq:metropolis-compressed-entry}
 \langle y|T^\dagger FT|x\rangle
 = &~ \sum_{s,\sigma:\,y=x+\sigma v_s}\frac{p_s}{2}
          \sqrt{a(x,s,\sigma)a(y,s,-\sigma)}\notag\\*
 = &~ \sqrt{\frac{\pi(x)}{\pi(y)}}
     \sum_{s,\sigma:\,y=x+\sigma v_s}\frac{p_s}{2}a(x,s,\sigma)\notag\\*
 = &~ \sqrt{\frac{\pi(x)}{\pi(y)}}\sP(x,y)=D_\pi(y,x),
\end{align}
where the second equality follows from $\pi(x)a(x,s,\sigma)=\pi(y)a(y,s,-\sigma)$, and the last equality from the definition of the discriminant and reversibility. 

For a feasible diagonal entry $x$, we have
\begin{align*}
    \langle x|T^\dagger FT|x\rangle= &~ \sum_{s,\sigma}\frac{p_s}{2} (1-a(x,s,\sigma))+\sum_{s,\sigma:v_s=0} \frac{p_s}{2} a(x,s,\sigma)\\
    = &~ \sum_{s,\sigma}\frac{p_s}{2} (1-a(x,s,\sigma)) + \sum_{s:v_s=0} p_s\\
    = &~ {\cal P}(x,x)=D_\pi(x,x),
\end{align*}
where the second step follows from $a(x,s,\sigma)=1$ when $v_s=0$.

Thus, we get that $T^\dagger F T = D_\pi$ on the feasible subspace. If $x$ is infeasible, all proposals will be rejected, and $FT\ket{x}=T\ket{x}$. Hence, for the unitary $U_T$ that implements $T$, i.e., $U_T(\ket{0}_{\sf E}\otimes \ket{x}_{\sf S})=T\ket{x}$, we have
\begin{align*}
(\bra0_{\mathsf{E}}\otimes\Id)U_T^\dagger FU_T(\ket0_{\mathsf{E}}\otimes\Id)
 =D_\pi\oplus\Id_{\rm bad}.    
\end{align*}
Let $P_\Omega=\sum_{x\in\Omega}\ket{x}\!\bra{x}$ be the projector onto feasible sample labels. A membership query implements the unitary $V_\Omega$ given by
\[
 V_\Omega\ket{f}\ket{x}
 =\ket{f\oplus(1-\one_\Omega(x))}\ket{x},
 \qquad
 (\bra0\otimes\Id)V_\Omega(\ket0\otimes\Id)=P_\Omega.
\]
Apply $V_\Omega$ before and after $U_T^\dagger F U_T$, using a separate flag qubit for each application. Projecting both flags and $\mathsf E$ onto zero gives
\[
 P_\Omega(D_\pi\oplus\Id_{\rm bad})P_\Omega
 =D_\pi\oplus0_{\rm bad}.
\]
Thus the resulting unitary is a $1$-block-encoding of the zero-extended discriminant.

The total number of membership oracle queries is $O(1)$, and the numerical implementation error is bounded in~\cref{lem:acceptance-arithmetic}.

\end{proof}


\subsection{Mixing and the cost of qsample-reflection}\label{sec:mixing_time_guarantee}
The following theorem establishes that we can choose a lattice spacing $h$ that depends inverse-polynomially on the input parameters such that the lattice Metropolis walk mixes fast and the normalization constant in~\eqref{eq:main-finite-logaffine-target} has small discretization errors on this lattice. Its proof is deferred to~\cref{app:finite-state}.

\begin{theorem}[One finite walk for each covariance branch]
\label{thm:finite-metropolis}
Fix the body and the finite list of tilts in \eqref{eq:main-finite-logaffine-target}, rational bounds $0<\lambda_-\le\lambda_+$, $0<\psi\le\psi_D$, and accuracies $0<\eta,\zeta<1/2$. Fix also a covariance-estimation failure budget $0<\varepsilon_{\Cov}<1/2$ and set $L=\lceil CD(1+\log(4/\varepsilon_{\Cov}))^2\rceil$, where $C$ is a sufficiently large universal constant. There is a common dyadic spacing $h$ and common radius, seed-length, and precision bounds for which the following hold:
\begin{enumerate}[label=(\roman*)]
\item For every rational $A$ with $\lambda_-\Id\preceq A\preceq\lambda_+\Id$, the lattice Metropolis walk in~\cref{alg:finite-metropolis} is lazy, positive, irreducible, and reversible for $\pi_{i,\Lambda}$.
\item If also $A\asymp\Cov_{\pi_{b_i}}(X)$ with fixed comparison constants, then $\tau_{\mix,2}(\eta,\pi_{i,\Lambda};\sP_{i,\Lambda,A}) \le\widetilde O(D^2/\psi^2)$, where $\sP_{i,\Lambda,A}$ denotes the Markov kernel of the walk.
\item Its zero-extended discriminant has a $1$-block-encoding.  For any requested error $0<\varepsilon<1/2$, an implementation is within operator-norm error $\varepsilon$ of an exact block-encoding unitary, using $O(1)$ coherent membership queries.
\item Every listed normalization constant satisfies $|Z_{b_i,\Lambda}/Z_{b_i}-1|\le\zeta$.
\item For every listed tilt, let $\Sigma_i=\Cov_{\pi_{b_i}}(X)$. The empirical covariance $\widehat\Sigma_i$ of $L$ independent samples from $\pi_{i,\Lambda}$ satisfies
\[
 \Pr\!\left[\tfrac14\Sigma_i\preceq\widehat\Sigma_i\preceq4\Sigma_i\right]\ge1-\varepsilon_{\Cov}.
\]
\end{enumerate}
After rescaling the inner radius of the convex body to one, $1/h$ is at most polynomial in the geometric bounds, largest tilt norm, regularization bounds, and their reciprocals, $1/\eta$, $1/\zeta$, $1/\varepsilon_{\Cov}$, and $1/\psi$.
\end{theorem}

As a corollary, we can efficiently implement the corresponding qsample-reflection.

\begin{corollary}[Lattice qsample reflections]
\label{cor:finite-reflection-cost-main}
On a good branch $A\asymp\Cov_{\pi_{b_i}}(X)$, for every $M_4\ge1$ and $0<{\varepsilon_{\rm ref}}<1/2$, a common sufficiently fine lattice permits one circuit approximating $R_{\pi_{i,\Lambda}}$ to error ${\varepsilon_{\rm ref}}$ on every joint state $\ket{\Psi}_{\mathsf{SA}}$ supported on the feasible sample space and satisfying $\|\ket{\Psi}_{\mathsf{SA}}\|_{4,\pi_{i,\Lambda}}\le M_4$. The error is measured as in \eqref{eq:effective-reflection}, including fresh workspace $\mathsf{E}$ and leaving $\mathsf{A}$ untouched. Each call, including a controlled one, uses $\widetilde O(D/\psi)$ membership queries.
\end{corollary}
\begin{proof}
Choose $\eta=({\varepsilon_{\rm ref}}/(CM_4))^4$, with $C$ sufficiently large, and use \cref{thm:finite-metropolis} to fix one lattice at this accuracy. On each good branch,
\[
 t=\widetilde O(D^2/\psi^2),\qquad
 N_{\rm be}=O(\sqrt t\log(1/{\varepsilon_{\rm ref}}))=\widetilde O(D/\psi).
\]
Apply \cref{thm:effective-reflection} with an exact block-encoding unitary $U$ and filter error ${\varepsilon_{\rm ref}}/4$. Its clean-input error is at most
\begin{align*}
    e_{\rm ideal}\leq C_0\eta^{1/4}M_4+{\varepsilon_{\rm ref}}/4\le{\varepsilon_{\rm ref}}/2.
\end{align*}

Next implement $U$ by $\widetilde U$ with $\|\widetilde U-U\|_{\op}\le{\varepsilon_{\rm ref}}/(4N_{\rm be})$. Inverses and controlled versions have the same operator error, so telescoping over all the block-encoding queries in the reflection circuit gives that
\[
 e_{\rm be}
 \le N_{\rm be}\|\widetilde U-U\|_{\op}\le{\varepsilon_{\rm ref}}/4.
\]
We can suppose that the remaining elementary gates have total error at most ${\varepsilon_{\rm ref}}/4$. On any promised joint input, the total error is therefore
\begin{align*}
 e_{\rm ideal} + e_{\rm be} +{\varepsilon_{\rm ref}}/4 \leq {\varepsilon_{\rm ref}}.
\end{align*}
Each block-encoding query uses $O(1)$ membership queries by \cref{lem:metropolis-block-encoding}. The total query complexity is
\begin{align*}
    N_{\rm be}\cdot O(1) = \widetilde O(D/\psi).
\end{align*}
\end{proof}

\subsection{Coherent covariance estimation}\label{sec:coherent-rounding}
We use the interlaced-rounding construction of~\cite{CCHLW23}. The reflection error bound depends on the input state. We bound this error on reference inputs with independent covariance blocks and propagate it using~\cref{lem:composition}.

Let $\pi_0,\ldots,\pi_m$ be finite laws with constant adjacent qsample overlap.  Let $\Sigma_i$ be positive-definite reference matrices with $\Sigma_i\asymp\Sigma_{i+1}$.  In the application, $\Sigma_i$ is the covariance of the continuous log-affine target and $\pi_i$ is its finite lattice law.

For a tuple $\mathbf x=(x_1,\ldots,x_L)$, define
\[
 \overline x=\frac1L\sum_{\ell=1}^Lx_\ell,\qquad
 \widehat\Sigma(\mathbf x)
 =\frac1L\sum_{\ell=1}^L(x_\ell-\overline x)(x_\ell-\overline x)^\top.
\]
Fix a covariance-estimation failure budget $0<\varepsilon_{\Cov}<1/2$, and assume that independent samples from each finite law satisfy
\begin{equation}\label{eq:good-covariance}
 \Pr_{\mathbf x\sim\pi_i^{\otimes L}}
 \left[\tfrac14\Sigma_i\preceq\widehat\Sigma(\mathbf x)
                         \preceq4\Sigma_i\right]\ge1-{\varepsilon_{\Cov}}.
\end{equation}
Let $G_i$ be the set of tuples ${\bf x}$ satisfying this covariance approximation. For the reference qsample block,
\[
 \ket{\pi_i^{\otimes L}}=\sum_{\mathbf x=(x_1,\dots,x_L)}\prod_{j=1}^L\sqrt{\pi_i(x_j)}\ket{\mathbf x},
 \qquad
 \sum_{\mathbf x\notin G_i}\pi_i^{\otimes L}(\mathbf x)\le{\varepsilon_{\Cov}}.
\]
Thus the classical failure probability is exactly the bad-label weight in \cref{lem:coherent-reflection-error}. \Cref{lem:covariance-confidence,lem:finite-block-transfer} establish \eqref{eq:good-covariance} for the lattice targets used later.

The preconditioner must be invertible on every encoded tuple. Choose a known rational ridge $0<\lambda_-\le\min_i\lambda_{\min}(\Sigma_i)$ and set
\begin{equation}\label{eq:ridge-covariance}
 A(\mathbf x)=\widehat\Sigma(\mathbf x)+\lambda_-\Id.
\end{equation}
The empirical covariance is positive semidefinite even on inaccurate or infeasible tuples, so $A(\mathbf x)\succeq\lambda_-\Id$ everywhere. If all encoded positions lie in $B(x_c,R_{\rm enc})$, then for any unit vector $u$,
\begin{align*}
 u^\top\widehat\Sigma(\mathbf x)u
 = &~ \frac1L\sum_\ell\langle u,x_\ell-x_c\rangle^2
                  -\langle u,\overline x-x_c\rangle^2\\*
 \leq &~ \frac1L\sum_\ell\|x_\ell-x_c\|_2^2
 \le R_{\rm enc}^2.
\end{align*}
Hence, we obtain the following branch-independent bounds:
\[
 \lambda_-\Id\preceq A(\mathbf x)
 \preceq(R_{\rm enc}^2+\lambda_-)\Id.
\]

For a good tuple from stage $i$, the ridge satisfies $\lambda_-\Id\preceq\Sigma_i$, so
\[
 \tfrac14\Sigma_i\preceq A(\mathbf x)
 \preceq4\Sigma_i+\lambda_-\Id\preceq5\Sigma_i.
\]
Together with $\Sigma_i\asymp\Sigma_{i+1}$, this gives $A(\mathbf x)\asymp\Sigma_j$ for both $j=i$ and $j=i+1$. 

We now use a block of samples without measuring it. Let $\mathsf{B}$ contain the control tuple, $\mathsf{S}$ be the sample register, $\mathsf{A}$ the existing auxiliary registers, and $\mathsf{E}$ the fresh implementation workspace. For each tuple $\mathbf x$, let $F_{j,\mathbf x}$ approximate the reflection $R_{\pi_j}$ on $\sf S$ by computing the preconditioner $A(\mathbf x)$, applying the reflection circuit (\cref{thm:effective-reflection}) on $\mathsf{E},\mathsf{S}$, and uncomputing $A({\bf x})$. The resulting controlled operation is
\[
 F_j=\sum_{\mathbf x}\ket{\mathbf x}\!\bra{\mathbf x}_{\mathsf{B}}
                              \otimes F_{j,\mathbf x}.
\]
The tuple label is unchanged, so the matrix computation can be reversed even when the target becomes entangled with it. For a good tuple ${\bf x}\in G_i$, suppose that
\[
 \bigl\|[(F_{j,\mathbf x}-\Id_{\mathsf{E}}\otimes R_{\pi_j})\otimes\Id_{\mathsf{A}}](\ket{0}_{\mathsf{E}}\otimes\ket{\Psi}_{\mathsf{SA}})\bigr\|\le\xi.
\]
Applying \cref{lem:coherent-reflection-error} gives
\begin{align*}
 \|\bigl[(F_j-\Id_{\mathsf{B}}\otimes\Id_{\mathsf{E}}\otimes R_{\pi_j})\otimes\Id_{\mathsf{A}}\bigr](\ket{\pi_i^{\otimes L}}_{\mathsf{B}}\otimes\ket{0}_{\mathsf{E}}\otimes\ket{\Psi}_{\mathsf{SA}})\Bigr\|\le\xi+2\sqrt{\varepsilon_{\Cov}},
\end{align*}
where ${\varepsilon_{\Cov}}$ is defined as in~\eqref{eq:good-covariance}.

\subsection{Alternating conversions of qsamples}\label{subsec:rounding-conversion}
The error bound for $F_j$, for $j=i$ or $j=i+1$, uses a reference control block containing $L$ independent copies of $\ket{\pi_i}$. We maintain two covariance blocks $\mathsf B_1,\mathsf B_2$ that alternate between controlling the conversion and being converted. Their reference states evolve as follows:
\begin{align*}
 \ket{\pi_i^{\otimes L}}_{\mathsf{B}_1}\otimes\ket{\pi_i^{\otimes L}}_{\mathsf{B}_2}
 &\xrightarrow{\ \mathsf{B}_1\text{ controls}\ }
 \ket{\pi_i^{\otimes L}}_{\mathsf{B}_1}\otimes\ket{\pi_{i+1}^{\otimes L}}_{\mathsf{B}_2}\\*
 &\xrightarrow{\ \mathsf{B}_2\text{ controls}\ }
 \ket{\pi_{i+1}^{\otimes L}}_{\mathsf{B}_1}\otimes\ket{\pi_{i+1}^{\otimes L}}_{\mathsf{B}_2}.
\end{align*}
First keep $\mathsf{B}_1$ fixed while converting $\mathsf{B}_2$ and all unused copies needed at later phases. Then use the updated $\mathsf{B}_2$ to convert $\mathsf{B}_1$. In the reference computation, estimation acts only on its reserved copies and the shared phase register $\mathsf F$.

\begin{lemma}[Qsample conversion~{\cite{YoderLowChuang14}}]\label{lem:transport}
Let $\pi$ and $\nu$ be probability distributions on a common finite set $\Omega$, and suppose that $|\langle\pi|\nu\rangle|^2\ge f_0$ for a known constant $0<f_0\le1$. Given exact controlled access to $R_\pi=2\ket{\pi}\!\bra{\pi}-\Id$ and $R_\nu=2\ket{\nu}\!\bra{\nu}-\Id$, for every $0<\varepsilon<1/2$, there is a circuit that maps a supplied qsample $\ket{\pi}_{\mathsf S}$ to a state $\ket{\widetilde\nu}_{\mathsf S}$ with trace-distance error at most $\varepsilon$ using $O(f_0^{-1/2}\log(1/\varepsilon))$ controlled qsample reflections.
\end{lemma}
\begin{proof}
The fixed-point amplitude amplification~\cite{YoderLowChuang14} uses $O(f_0^{-1/2}\log(1/\varepsilon))$ selective phases (as in~\eqref{eq:selected_phase}) and returns a state $\ket{\widetilde\nu}_{\mathsf{S}}$ with
\[
 |\langle\nu|\widetilde\nu\rangle|^2\ge1-\varepsilon^2,
\]
which implies a trace-distance error bound:
\[
 \dtr(\ket{\widetilde\nu}\!\bra{\widetilde\nu},
                           \ket{\nu}\!\bra{\nu})
 =\sqrt{1-|\langle\nu|\widetilde\nu\rangle|^2}\le\varepsilon.
\]

Each selective phase uses two controlled qsample reflections. For a state $\ket{\chi}_{\mathsf{S}}$ being $\ket{\pi}$ or $\ket{\nu}$, let
\[
 M_\chi=(H_{\mathsf{C}}\otimes\Id_{\mathsf{S}})\ctrl_{\mathsf{C}}(R_\chi)(H_{\mathsf{C}}\otimes\Id_{\mathsf{S}}).
\]
For every vector $\ket{v}_{\mathsf{S}}$,
\[
 M_\chi(\ket{0}_{\mathsf{C}}\otimes\ket{v}_{\mathsf{S}})
 =\ket{0}_{\mathsf{C}}\otimes\ket{\chi}_{\mathsf{S}}\langle\chi|v\rangle
   +\ket{1}_{\mathsf{C}}\otimes(\Id_{\mathsf{S}}-\ket{\chi}\!\bra{\chi})\ket{v}_{\mathsf{S}}.
\]
Apply phase $e^{i\theta}$ on control zero and uncompute. This gives
\begin{align}\label{eq:selected_phase}
 M_\chi^\dagger
     [(e^{i\theta}\ket{0}\!\bra{0}_{\mathsf{C}}+\ket{1}\!\bra{1}_{\mathsf{C}})
           \otimes\Id_{\mathsf{S}}]M_\chi(\ket{0}_{\mathsf{C}}\otimes\ket{v}_{\mathsf{S}})
 =  \ket{0}_{\mathsf{C}}\otimes
        \left(\Id_{\mathsf{S}}+(e^{i\theta}-1)\ket{\chi}\!\bra{\chi}\right)\ket{v}_{\mathsf{S}},
\end{align}
where the control qubit is returned to zero and can be reused later. Thus, the total number of reflections is $O(f_0^{-1/2}\log(1/\varepsilon))$.
\end{proof}

Now consider a sequence of finite targets $\pi_0,\ldots,\pi_m$, where $m$ is the number of transitions in the schedule. Assume
\[
 |\langle\pi_i|\pi_{i+1}\rangle|^2\ge f_0,\qquad  \forall 0\le i<m.
\]
A single-copy conversion evolves one sample register from $\ket{\pi_i}$ to $\ket{\pi_{i+1}}$. Suppose the algorithm uses $T\ge1$ such conversions. At each conversion boundary, the reference state contains an independent copy of the source qsample; the implemented state may already be correlated. We set the trace-distance error of each conversion to be at most $\beta/T$ in~\cref{lem:transport}. Let $\mathcal C_i$ be the single-copy state-conversion channel from $\ket{\pi_i}$ to $\ket{\pi_{i+1}}$ using exact reflections. For any state $\sigma_{\mathsf A}$ of the remaining registers, we have
\begin{align*}
 &\dtr\Bigl(
 (\mathcal C_i\otimes\operatorname{id}_{\mathsf A})
       (\ket{\pi_i}\!\bra{\pi_i}_{\mathsf S}\otimes\sigma_{\mathsf A}),
 \ket{\pi_{i+1}}\!\bra{\pi_{i+1}}_{\mathsf S}\otimes\sigma_{\mathsf A}
 \Bigr)\\
 = &~ \dtr\Bigl(
       \mathcal C_i(\ket{\pi_i}\!\bra{\pi_i}),
       \ket{\pi_{i+1}}\!\bra{\pi_{i+1}}
       \Bigr)\\
 \leq &~ \frac{\beta}{T},
\end{align*}
where the first step follows from invariance of trace distance under tensoring with the same density operator. Let $\rho_{\rm ref}$ be the final joint state of the reference computation and let $\widetilde\rho$ be the final joint state obtained using ${\cal C}_0,\dots,{\cal C}_{m-1}$. Starting from the same initial state and keeping all other operations unchanged, \cref{lem:composition} gives
\[
 \dtr(\widetilde\rho,\rho_{\rm ref})
 \le\sum_{j=1}^T\frac{\beta}{T}
 =\beta.
\]
By~\cref{lem:transport}, the total number of controlled qsample reflections is
\[
 O\!\left(\frac{T}{\sqrt{f_0}}\log\frac{T}{\beta}\right),
\]
which is $O(T\log(T/\beta))$ for constant $f_0$. Errors from implementing the reflection oracles are bounded separately by \cref{lem:composition}.

In the two-block construction, each of $\mathsf B_1$ and $\mathsf B_2$ contains $L$ sample registers. Advancing both blocks through all $m$ transitions requires $2Lm$ single-copy conversions. Let $n_i$ be the number of additional qsample copies that must evolve from $\ket{\pi_i}$ to $\ket{\pi_{i+1}}$. The total number of conversions is therefore
\[
 T=\sum_{i=0}^{m-1}(2L+n_i)
  =2Lm+\sum_{i=0}^{m-1}n_i.
\]
For example, an additional qsample prepared at stage zero and used at stage $r$ contributes $r$ conversions to this count.

Finally, relative second moments give the required adjacent overlap. Let $r=d\mu/d\nu$ and suppose $\E_\nu[r^2]\le B$. H\"older's inequality, with exponents $3/2$ and $3$, gives
\[
 1=\E_\nu[r]
 =\E_\nu\bigl[(\sqrt r)^{2/3}(r^2)^{1/3}\bigr]
 \le(\E_\nu[\sqrt r])^{2/3}(\E_\nu[r^2])^{1/3}.
\]
Thus, we get that
\begin{equation}\label{eq:overlap-from-moment}
 |\langle\mu|\nu\rangle|^2
 =(\E_\nu[\sqrt r])^2
 \ge\frac1{\E_\nu[r^2]}\ge\frac1B.
\end{equation}
In the volume estimation, \cref{lem:cooling-moments} shows that both adjacent relative second moments are at most $4$. Hence, we can take $f_0=1/4$ in the above calculations. 

\section{Quantum volume estimation by pencil cooling}\label{sec:assembly}

In this section, we present our quantum algorithm for estimating the volume of convex bodies. Its guarantees are stated in the formal version of the main theorem as follows.

\begin{theorem}[Main theorem]\label{thm:intro-volume}
Let $K\subset\R^d$ be a convex body given by a strong coherent membership oracle, and suppose that we are given a rational point $x_0$, rational numbers $0<r\le R$, and a certified rational bound $0<\psi\le\psi_{d+1}$ such that
\[
 B(x_0,r)\subseteq K\subseteq B(x_0,R),
 \qquad \kappa:=R/r\le\poly(d).
\]
Then, for every rational $0<\varepsilon<1/2$, there is a quantum algorithm on finitely many qubits that outputs $\widehat V\in[(1-\varepsilon)\vol(K),(1+\varepsilon)\vol(K)]$ with probability at least $2/3$ using at most
\begin{equation}
 \left(\frac{d^{5/2}}{\psi}+\frac{d^{3/2}}{\psi\varepsilon}\right)
 \log^{O(1)}\!\left(\frac{d\kappa}{\varepsilon\psi}\right)
\end{equation}
coherent membership queries on every execution. The remaining computation is polynomial in $d$, $\kappa$, $\psi^{-1}$, $\varepsilon^{-1}$, and the bit length of the input.
\end{theorem}

The pencil construction~\cite{LovaszVempala06Annealing} expresses the volume through a telescoping product. We choose a cooling schedule with bounded adjacent relative second moments, then combine qsample conversion with the product estimator to prove~\cref{thm:intro-volume}.

\subsection{The volume estimation identity}\label{subsec:pencil-geometry}
Apply the dyadic adjustment of \cref{subsec:membership-oracle-model} and rescale the input body.  We may assume
\[
 B(0,1)\subseteq K\subseteq B(0,\kappa),\qquad
 \vol(K_{\rm in})=r^d\vol(K),\qquad D=d+1.
\]
Every membership query after rescaling uses one original query. For $d=1$, binary searches from the given interior point locate the two interval endpoints to accuracy $\varepsilon r/4$ using $O(\log(\kappa/\varepsilon))$ queries.  Since its length is at least $2r$, this settles that case.  Assume $d\ge2$ below.

Following the classical pencil construction~\cite{LovaszVempala06Annealing}, define
\begin{equation}\label{eq:assembly-pencil}
 K_0=\left\{(x_0,x'):\ 0\le x_0\le2\kappa,\ x'\in K,
                                  \ \|x'\|_2\le x_0\right\}.
\end{equation}
Membership query for $K_0$ requires one query to $K$. Note that for any $\kappa\leq x_0\leq 2\kappa$, the cross-section of $K_0$ is exactly $K$. It satisfies the known support bounds
\begin{equation}\label{eq:pencil-support-balls}
 B((1/2,0),1/4)\subseteq K_0\subseteq B((1/2,0),3\kappa).
\end{equation}
The first inclusion follows from $B(0,1)\subseteq K$ and $\|x'\|_2\le1/4\le x_0$ in the displayed ball; the second follows from $x_0\le2\kappa$ and $\|x'\|_2\le\kappa$.

Define a continuous log-affine distribution over $K_0$ parameterized by $a\in \mathbb{R}$:
\begin{align*}
    Z(a)=\int_{K_0}e^{-ax_0}\,dx,\qquad \pi_a(dx)=Z(a)^{-1}e^{-ax_0}\one_{K_0}(x)\,dx.
\end{align*}
Consider a decreasing sequence $a_0=4D>\cdots>a_m>0$ with $a_m\kappa\le1$. 
On the common dyadic lattice to be chosen later, define the discrete distribution
\begin{equation}\label{eq:assembly-finite-target}
 Z_\Lambda(a)=h^D\sum_{z\in K_0\cap h\Z^D}e^{-az_0},\qquad
 \pi_i(z)=\frac{h^De^{-a_i z_0}}{Z_\Lambda(a_i)},\qquad 0\leq i\leq m.
\end{equation}
The observables
\begin{equation}\label{eq:assembly-ratio-definition}
 X_i(z)=e^{(a_i-a_{i+1})z_0},\qquad
 \mu_i=\E_{\pi_i}[X_i]=\frac{Z_\Lambda(a_{i+1})}{Z_\Lambda(a_i)}
\end{equation}
satisfy the exact adjacent-law relation in \eqref{eq:bridge-family}. 

Let $C=\{(x_0,x'):0\le x_0\le2\kappa,\ \|x'\|_2\le x_0\}$ be the truncated round cone.  Its finite normalization constant and law at $a_0$ are
\begin{equation}\label{eq:finite-cone-anchor}
 Z_{C,\Lambda}=h^D\sum_{z\in C\cap h\Z^D}e^{-a_0z_0},\qquad
 \nu_\Lambda(z)=h^De^{-a_0z_0}/Z_{C,\Lambda}.
\end{equation}
Define the initial and terminal probabilities
\begin{equation}\label{eq:endpoint-probabilities}
 q_{\rm in}=\Pr_{z\sim \nu_\Lambda}[z'\in K],\qquad
 q_{\rm out}=\Pr_{z\sim \pi_m}[z_0\in[\kappa,2\kappa]].
\end{equation}
We have the following telescoping equation:
\begin{equation}\label{eq:exact-volume-telescope}
\begin{aligned}
 Z_{C,\Lambda}\,q_{\rm in}
 \left(\prod_{i=0}^{m-1}\mu_i\right)q_{\rm out}= &~ \left(h^D\sum_{z\in K_0\cap h\Z^D} e^{-a_0z_0}\right) \frac{Z_\Lambda(a_m)}{Z_\Lambda(a_0)}\left(\sum_{z\in([\kappa,2\kappa]\times K)\cap h\Z^D}\frac{h^D e^{-a_mz_0}}{Z_\Lambda(a_m)}\right)\\
 = &~ h^D\sum_{z\in([\kappa,2\kappa]\times K)\cap h\Z^D}e^{-a_mz_0}.
\end{aligned}
\end{equation}
The corresponding continuous slab mass is
\[
 \int_{[\kappa,2\kappa]\times K}e^{-a_mx_0}dx
 =\left(\int_Kdx'\right)
   \left(\int_\kappa^{2\kappa}e^{-a_mt}dt\right)
 =\vol(K)I_m,
\]
where $I_m:=\int_\kappa^{2\kappa}e^{-a_mt}\,dt$.
The slab contains $B((3\kappa/2,0),1/2)$ and has outer radius $O(\kappa)$, so by~\cref{lem:full-grid-quadrature}, the discrete sum can approximate this integral to any requested relative accuracy. Hence, \eqref{eq:exact-volume-telescope} can be used to estimate the volume of $K$.

We apply the quantum product estimator~\cref{lem:log-sum} to estimate~\eqref{eq:exact-volume-telescope}. For any positive bridge scales $c_i$, let $p_i^+,p_i^-$ be the probabilities in \eqref{eq:bridge-marks} for these finite targets. Define
\begin{equation}\label{eq:volume-log-sum}
 \mathsf L=\sum_{i=0}^{m-1}(-\log p_i^++\log p_i^-)
                   -\log q_{\rm in}-\log q_{\rm out}.
\end{equation}
The bridge identities and \eqref{eq:exact-volume-telescope} motivate the following estimator: after obtaining $\widehat{\mathsf L}$, take
\begin{equation}\label{eq:volume-output}
 \log\widehat V=d\log r+\log Z_{C,\Lambda}
      +\sum_{i=0}^{m-1}\log c_i-\widehat{\mathsf L}-\log I_m,
\end{equation}
where $r$ is the rescaling factor for the input convex body $K_{\rm in}$, and $\widehat{V}$ is our estimator for $\vol(K_{\rm in})$. The remaining terms are evaluated classically using the procedure for $Z_{C,\Lambda}$ in~\cref{app:initialization-details} and the numerical bounds in~\cref{app:mark-arithmetic}.

We next choose a schedule with bounded adjacent relative second moments and verify that both endpoint probabilities $q_{\rm in}$ and $q_{\rm out}$ are at least $1/32$, as required by~\cref{lem:log-sum}.


\subsection{The cooling schedule}\label{sec:cool_schedule}
Choose the rational schedule
\begin{equation}\label{eq:cooling-schedule}
 a_0=4D,\qquad
 a_{i+1}=\left(1-\frac1{16\lceil\sqrt d\rceil}\right)a_i,
 \qquad
 m=\min\{i:a_i\kappa\le1\}.
\end{equation}
Thus, $m=O(\sqrt D\log(D\kappa))=\widetilde O(\sqrt D)$. 

\begin{lemma}[Second moments along cooling]\label{lem:cooling-moments}
For the cooling schedule in~\eqref{eq:cooling-schedule}, for any $0\leq i<m$, $\pi_{a_i}$ and $\pi_{a_{i+1}}$ have continuous relative second moments $\E_{\pi_{a_i}}[(\pi_{a_{i+1}}/\pi_{a_i})^2]$ and $\E_{\pi_{a_{i+1}}}[(\pi_{a_{i}}/\pi_{a_{i+1}})^2]$ at most $2$.  

Furthermore, we can choose a dyadic lattice such that at $a\in \{a_i-2\Delta,a_i-\Delta,a_i,a_i+\Delta\}$ with $\Delta=a_i/(16\lceil\sqrt d\rceil)$,
\begin{align*}
    \frac{99}{100}Z(a)\leq Z_{\Lambda}(a) \leq \frac{101}{100} Z(a).
\end{align*}
Then both finite relative second moments are at most $4$.
\end{lemma}
\begin{proof}
Relative second moments are ratios of normalization constants, so we bound second differences of $F(a)=\log Z(a)$. The following bound for the second derivative of $F$ is proved in~\cite{LovaszVempala06Annealing}:  
\begin{align*}
    0\le F''(a)\le D/a^2.
\end{align*}
Indeed, since $0\le x_0\le2\kappa$, we can differentiate under the integral and get that
\[
 F'(a)=-\E_{\pi_a}[x_0],\qquad
 F''(a)=\E_{\pi_a}[x_0^2]-(\E_{\pi_a}[x_0])^2=\Var_{\pi_a}(x_0)\ge0.
\]
For the upper bound, the set $\{(a,y):a>0,\ y\in aK_0\}$ is convex: if $y_j=a_jx_j$ with $x_j\in K_0$, $a_j>0$, and $j\in \{1,2\}$, then we have
\[
 \frac{ty_1+(1-t)y_2}{ta_1+(1-t)a_2}
 =\frac{ta_1}{ta_1+(1-t)a_2}x_1
  +\frac{(1-t)a_2}{ta_1+(1-t)a_2}x_2\in K_0.
\]
Pr\'ekopa's theorem~\cite{Prekopa73} applied to $e^{-y_0}\one_{\{y\in aK_0\}}$ shows that the function
\[
 a\longmapsto\int_{aK_0}e^{-y_0}\,dy=a^DZ(a)
\]
is log-concave. Taking two derivatives yields
\[
 0\ge(D\log a+F(a))''=-D/a^2+F''(a).
\]

For any $0\leq i< m$, let $\Delta=a_i/(16\lceil\sqrt d\rceil)$. We have
\begin{align*}
 \E_{\pi_{a_i}}\!\left[\left(\frac{d\pi_{a_{i+1}}}{d\pi_{a_i}}\right)^2\right]= &~ \int_{K_0}\frac{Z(a_{i}-\Delta)^{-2}e^{-2(a_i-\Delta)x_0}}{Z(a_{i})^{-1}e^{-a_ix_0}}\,dx
 =  \frac{Z(a_i)Z(a_i-2\Delta)}{Z(a_i-\Delta)^2},\\*
 \E_{\pi_{a_{i+1}}}\!\left[\left(\frac{d\pi_{a_i}}{d\pi_{a_{i+1}}}\right)^2\right] =  &~ \int_{K_0}\frac{Z(a_{i})^{-2}e^{-2a_ix_0}}{Z(a_{i}-\Delta)^{-1}e^{-(a_i-\Delta)x_0}}\,dx
 = \frac{Z(a_i-\Delta)Z(a_i+\Delta)}{Z(a_i)^2}.
\end{align*}
For either $u=a_i-\Delta$ or $u=a_i$, take the logarithm and we get that
\begin{equation}\label{eq:log_normalizer_prod}
\begin{aligned}
 F(u-\Delta)+F(u+\Delta)-2F(u)
 = &~ \int_{-\Delta}^{\Delta}(\Delta-|v|)F''(u+v)\,dv\\*
 \leq &~ \frac{D\Delta^2}{(u-\Delta)^2}
 \le\frac{D\Delta^2}{(a_i-2\Delta)^2}<\frac1{16},
\end{aligned}
\end{equation}
where the first inequality follows from $F''(u+v)\le D/(u-\Delta)^2$, and the last step follows from $a_i/\Delta=16\lceil \sqrt{d}\rceil$, $D=d+1$, and $d\geq 2$. Both continuous moments are therefore at most $e^{1/16}<2$.

The finite relative second moments of $\pi_{i}$ and $\pi_{i+1}$ can also be expressed as the ratios between the lattice normalization constants. For each $a\in\{a_i-2\Delta,a_i-\Delta,a_i,a_i+\Delta\}$,
\[
 |Z_\Lambda(a)/Z(a)-1|\le1/100
 \quad\Longrightarrow\quad
 |\log Z_\Lambda(a)-\log Z(a)|\le2/100.
\]
Thus, the additive error to the LHS of~\eqref{eq:log_normalizer_prod} is at most $8/100$. And each finite relative second moment is bounded by $e^{1/16+8/100}<4$.

Finally, write $\Delta_i=a_i-a_{i+1}$. We can choose a common lattice such that the above approximation guarantee holds for all $a$ in the following finite list
\[
\bigcup_{i=0}^{m-1}
 \{a_i-2\Delta_i,\ a_i-\Delta_i,\ a_i,\ a_i+\Delta_i\}\subset [a_m/2,2a_0].
\]
This list is determined by the cooling schedule, so we can apply~\cref{thm:finite-metropolis}~\textnormal{(iv)} to choose the lattice spacing $h$.

\end{proof}

The finite second-moment bounds in~\cref{lem:cooling-moments} give constant overlap between adjacent qsamples by~\eqref{eq:overlap-from-moment}. Applying~\cref{lem:adjacent-covariance-prelim} to the continuous laws gives $\Sigma_i\asymp\Sigma_{i+1}$, allowing one preconditioner to serve both stages.



The following lemma proves that $q_{\rm in}$ and $q_{\rm out}$ as defined in~\eqref{eq:endpoint-probabilities} are bounded away from zero, and also shows the quantum implementation costs for the initialization stage of the cooling schedule.
\begin{lemma}\label{lem:endpoint-resources}
For a sufficiently fine lattice, $q_{\rm in}\ge1/2$ and $q_{\rm out}\ge1/(2(e+1))$. The quantity $\log Z_{C,\Lambda}$ can be computed to any prescribed accuracy with no membership queries.  The qsample $\ket{\nu_\Lambda}$ and its controlled reflection can also be prepared to prescribed accuracy without membership queries. 

Moreover, given access to $\ket{\nu_\Lambda}$ and its controlled reflection, we can prepare $\ket{\pi_0}$ to trace-distance error ${\varepsilon_{\rm init}}$ using $O(\log(1/{\varepsilon_{\rm init}}))$ membership queries. 
\end{lemma}
\begin{proof}
We bound the two event probabilities
\begin{align*}
    q_{\rm in}=\Pr_{z\sim \nu_\Lambda}[z'\in K],\qquad
 q_{\rm out}=\Pr_{z\sim \pi_m}[z_0\in[\kappa,2\kappa]].
\end{align*}
by first estimating the corresponding quantities with the continuous distributions and then applying the lattice quadrature errors. 

Recall that
\begin{align*}
    C=\{(x_0,x'):0\le x_0\le2\kappa,\ \|x'\|_2\le x_0\}\supset K_0.
\end{align*}
Let $\nu$ be the continuous distribution with density
\begin{align*}
    \nu(x)=Z_C^{-1}e^{-a_0x_0}\one_C(x),\qquad Z_C=\int_Ce^{-a_0x_0}\,dx.
\end{align*}
For $x_0=t\in[0,2\kappa]$, the cross-section of $C$ is the $d$-dimensional ball $B(0,t)$, whose volume is $v_dt^d$. Therefore,
\begin{align*}
 \E_{\nu}[x_0]
 = \frac{\int_0^{2\kappa}t^{d+1}e^{-a_0t}\,dt}
            {\int_0^{2\kappa}t^de^{-a_0t}\,dt}
 = -\frac{(2\kappa)^{d+1}e^{-2\kappa a_0}}
        {a_0\int_0^{2\kappa}t^de^{-a_0t}\,dt}+\frac{d+1}{a_0}
 \leq \frac{D}{a_0}=\frac14,
\end{align*}
where the first step follows by integrating over the horizontal cross-sections, the second step follows from integration by parts, and the last step uses $D=d+1$ and $a_0=4D$. If $x_0\le1$, then $\|x'\|_2\le 1$, so $x'\in B(0,1)\subseteq K$ and $x\in K_0$. Thus, we have
\begin{align*}
 \frac{Z(a_0)}{Z_C} = \frac{\int_{K_0} e^{-a_0x_0}\,dx}{Z_C} 
 =  \Pr_{x\sim \nu}[x\in K_0]
 \geq \Pr_{x\sim \nu}[x_0\le 1]
 \geq  1-\E_{\nu}[x_0]
 \geq \frac34.
\end{align*}
For the final target, every cross-section of $K_0$ is contained in $K$ and those on $x_0\in [\kappa,2\kappa]$ exactly equal $K$. Hence, we have
\[
 \Pr_{x\sim \pi_{a_m}}[x_0\in[\kappa,2\kappa]]
 \ge\frac{\vol(K)\int_\kappa^{2\kappa}e^{-a_mt}\,dt}
          {\vol(K)\int_0^{2\kappa}e^{-a_mt}\,dt}
 =\frac1{1+e^{a_m\kappa}}\ge\frac1{e+1},
\]
where the last step follows from $a_m\kappa \leq 1$.

Note that
\begin{align*}
    q_{\rm in}=\Pr_{z\sim \nu_\Lambda}[z'\in K] = \frac{Z_{\Lambda}(a_0)}{Z_{C,\Lambda}}\,,
\end{align*}
and by choosing a sufficiently small lattice spacing $h$, we have
\begin{align*}
    |Z_{\Lambda}(a_0)/Z(a_0)-1|\leq \frac{1}{100},\qquad |Z_{C,\Lambda}/Z_C-1|\leq \frac{1}{100}.
\end{align*}
Then, we have
\begin{align*}
    q_{\rm in}\geq \frac{\frac{99}{100}Z(a_0)}{\frac{101}{100}Z_C}\ge\frac{99}{101}\frac34>\frac{1}{2}.
\end{align*}
Similarly, let
\begin{align}\label{eq:def_Vm}
    V_m=\int_{[\kappa,2\kappa]\times K}e^{-a_mx_0}\,dx,
 \qquad
 V_{m,\Lambda}=h^D\sum_{z\in([\kappa,2\kappa]\times K)\cap h\Z^D}e^{-a_mz_0}.
\end{align}
We have
\begin{align*}
    q_{\rm out}=\Pr_{z\sim \pi_m}[z_0\in[\kappa,2\kappa]] = \frac{V_{m,\Lambda}}{Z_\Lambda(a_m)},\qquad \Pr_{x\sim \pi_{a_m}}[x_0\in[\kappa,2\kappa]]=\frac{V_m}{Z(a_m)}.
\end{align*}
Suppose the lattice also satisfies
\begin{align*}
    |Z_{\Lambda}(a_m)/Z(a_m)-1|\leq \frac{1}{100},\qquad |V_{m,\Lambda}/V_m-1|\leq \frac{1}{100}.
\end{align*}
Then, 
\begin{align*}
    q_{\rm out}\geq \frac{\frac{99}{100}V_m}{\frac{101}{100}Z(a_m)}\ge\frac{99}{101}\frac1{e+1}>\frac{1}{2(e+1)}.
\end{align*}

Finally, we consider the quantum implementation costs. Using the procedure in~\cref{app:initialization-details}, we can evaluate $Z_{C,\Lambda}$ and construct a unitary $U_C$ with no membership queries that $U_C\ket{0}=\ket{\nu_\Lambda}$. Then, the reflection can be implemented by $R_{\nu_\Lambda}=U_C(2\ket{0}\bra{0}-\Id)U_C^\dagger$.

To prepare $\ket{\pi_0}$, note that
\[
 \frac{\nu_\Lambda(z)\one_K(z')}{q_{\rm in}}
 =\frac{h^De^{-a_0z_0}}{Z_{C,\Lambda}q_{\rm in}}\one_{K_0}(z)=\frac{h^De^{-a_0z_0}}{Z_\Lambda(a_0)}\one_{K_0}(z)
 =\pi_0(z).
\]
Let $\Pi_{\rm in}$ be the membership projector defined by $\Pi_{\rm in}\ket z_{\mathsf S}=\one_K(z')\ket z_{\mathsf S}$. Then
\begin{align*}
    \Pi_{\rm in}\ket{\nu_\Lambda}
 =\sqrt{q_{\rm in}}\ket{\pi_0}.
\end{align*}
Since $q_{\rm in}\geq 1/2$, we can apply the fixed-point amplification~\cite{YoderLowChuang14} to prepare $\ket{\pi_0}$ from $\ket{\nu_\Lambda}$ within trace-distance error ${\varepsilon_{\rm init}}$ using  $O(\log(1/{\varepsilon_{\rm init}}))$ controlled-reflections $R_{\nu_\Lambda}$ and $2\Pi_{\rm in}-\Id$, which cost $O(\log(1/{\varepsilon_{\rm init}}))$ membership queries in total. 
\end{proof}


\subsection{The volume estimation algorithm}\label{sec:vol_est_algorithm}


We present the main volume estimation algorithm (\cref{alg:volume}) in this subsection.

For each $0\leq i<m$, let $\mathsf{S}_i^+$ and $\mathsf{S}_i^-$ be the reserved sample registers for $p_i^+$ and $p_i^-$. They are used at stages $i$ and $i+1$, respectively. Let $\mathsf{S}_{\rm in}$ and $\mathsf{S}_{\rm out}$ be the registers used for $q_{\rm in}$ and $q_{\rm out}$, respectively. The corresponding endpoint observables are
\[
 Y_{\rm in}(z)=\one_K(z'),\qquad
 Y_{\rm out}(z)=\one_{[\kappa,2\kappa]}(z_0).
\]
Lines~\ref{ln:begin_prod_est}--\ref{ln:end_prod_est} of~\cref{alg:volume} implement the \textsc{Estimation} procedure in~\cref{alg:add-log} with $n=2m+2$ random variables. The stage loop interleaves bridge estimation with qsample conversions.

\begin{algorithm}[H]
\caption{Quantum convex-body volume estimation}
\label{alg:volume}
\begin{algorithmic}[1]
\Procedure{VolumeEstimation}{${\cal O}_{K}$, $x_c,r,R$, $\psi$, $\varepsilon$}
\State \Comment{Initialization}
\State Make the dyadic rescaling, form $K_0$ and $C$, and choose the schedule \eqref{eq:cooling-schedule}.
\State Fix all counts, one lattice, and precisions as in \cref{app:parameters}. 
\State Set $n=2m+2$ and choose $k,M$ from \eqref{eq:log-clock-parameters} with accuracy $\varepsilon/8$.
\State Prepare $\ket{\pi_0}^{\otimes L}_{\mathsf{B}_1}$ and $\ket{\pi_0}^{\otimes L}_{\mathsf{B}_2}$
\State For each $i<m$, prepare $\ket{\pi_0}^{\otimes P}$ for pilot estimation, $\ket{\pi_0}_{\mathsf{S}_i^+}$, and $\ket{\pi_0}_{\mathsf{S}_i^-}$,
\State Prepare $\ket{\nu_\Lambda}_{\mathsf{S}_{\rm in}}$ and $\ket{\pi_0}_{\mathsf{S}_{\rm out}}$
\State \Comment{Product estimation}
\State Prepare $M^{-1/2}\sum_{z=0}^{M-1}e^{2\pi iz/4}\ket z_{\mathsf{F}}$.\label{ln:begin_prod_est}
\State \Call{Add}{$\mathsf{F},\mathsf{S}_{\rm in};R_{\nu_\Lambda},V_{Y_{\rm in}},+1$}.\Comment{\cref{alg:add-log}}
\For{$i=0,\ldots,m$}
 \If{$i>0$}
  \State \Call{Add}{$\mathsf{F},\mathsf{S}_{i-1}^-;R_{\pi_i},V_{Y_{i-1}^-},-1$}.
 \EndIf
 \If{$i<m$}
  \State \textcolor{blue}{Measure the reserved $\ket{\pi_i}^{\otimes P}$ for stage $i$ and compute $c_i$.}
  \State \Call{Add}{$\mathsf{F},\mathsf{S}_i^+;R_{\pi_i},V_{Y_i^+},+1$}.
  \State \textcolor{blue}{With $\mathsf{B}_1$ fixed, convert the qsamples in $\mathsf{B}_2$ and all unused copies to $\ket{\pi_{i+1}}$; then use the updated $\mathsf{B}_2$ to convert $\mathsf{B}_1$.}
 \Else
  \State \Call{Add}{$\mathsf{F},\mathsf{S}_{\rm out};R_{\pi_m},V_{Y_{\rm out}},+1$}.
 \EndIf
\EndFor
\State Apply the inverse Fourier transform to $\mathsf{F}$, measure $Y$, and compute $\widehat{\mathsf{L}}=64n(Y/M-1/4)$.\label{ln:end_prod_est}
\State \Comment{Classical postprocessing}
\State Compute $\log \widehat{V}$ in~\eqref{eq:volume-output} with total numerical error at most $\varepsilon/8$.
\State \Return $\widehat{V}$.
\EndProcedure
\end{algorithmic}
\end{algorithm}



\subsection{Proof of the main theorem}\label{subsec:assembly-supply}
We prove~\cref{thm:intro-volume} by bounding the query complexity of~\cref{alg:volume} in~\cref{lem:query_complexity_main} and analyzing its estimation error and success probability in~\cref{lem:est_error_succ_prob_main}.

\begin{lemma}[Query complexity]\label{lem:query_complexity_main}
Let ${\varepsilon_{\rm init}}$ be the error of preparing $\ket{\pi_0}$ from $\ket{\nu_\Lambda}$ (\cref{lem:endpoint-resources}), $\xi$ be the reflection error for each accurate covariance estimate on reference inputs satisfying the fourth-moment norm bound (\cref{cor:finite-reflection-cost-main}), and $\varepsilon_{\rm conv}$ be the target error of each qsample conversion (\cref{lem:transport}). Suppose ${\varepsilon_{\rm init}}$, $\xi$, and $\varepsilon_{\rm conv}$ are inverse-polynomial in $D$ and $1/\varepsilon$. Then, the total number of membership queries in~\cref{alg:volume} is
\begin{align*}
    \widetilde O\left(\frac{D^{5/2}}\psi
                     +\frac{D^{3/2}}{\psi\varepsilon}\right).
\end{align*}
\end{lemma}
\begin{proof}
We first count the number of qsamples needed throughout the algorithm:
\begin{itemize}
    \item Two blocks ${\sf B}_1,{\sf B}_2$, each containing $L$ qsamples, are converted alternately from $\ket{\pi_0}^{\otimes L}$ to $\ket{\pi_m}^{\otimes L}$.
    \item For $0\leq i<m$, estimating $\log p_i^+$ and $\log p_i^-$ uses $\ket{\pi_i}$ and $\ket{\pi_{i+1}}$, respectively, and estimating $c_i$ costs $P=O(\log(m+1))$ qsamples $\ket{\pi_i}^{\otimes P}$ by~\cref{lem:pilots}.
    \item Estimating $q_{\rm out}$ uses one qsample $\ket{\pi_m}$.
    \item Estimating $q_{\rm in}$ uses one qsample $\ket{\nu_\Lambda}$.
\end{itemize}
Thus, the total number of independent qsamples is
\begin{align*}
    2L+(P+2)m+2.
\end{align*}
Each qsample is prepared from $\ket{\nu_\Lambda}$ at the initial stage, then evolved to its designated stage, and discarded after use. Since preparing one copy of $\ket{\pi_0}$ from $\ket{\nu_\Lambda}$ only takes $O(\log (1/{\varepsilon_{\rm init}}))=\widetilde{O}(1)$ membership queries by~\cref{lem:endpoint-resources}, it takes 
\begin{align*}
    \widetilde{O}(2L+(P+2)m+1)=\widetilde{O}(L+m)
\end{align*}
queries in total. Then, we count the total number of qsample conversions from $\ket{\pi_i}$ to $\ket{\pi_{i+1}}$ for $0\leq i<m$. A rough upper bound is
\begin{align*}
    m\cdot (2L + (P+2)m + 1)=\widetilde{O}(Lm+m^2).
\end{align*}
Since the overlap between $\ket{\pi_i}$ and $\ket{\pi_{i+1}}$ is at least $1/4$ by~\eqref{eq:overlap-from-moment} and~\cref{lem:cooling-moments}, one single-copy qsample conversion with error $\varepsilon_{\rm conv}$ uses $\widetilde{O}(1)$ controlled-reflections by~\cref{lem:transport}. Thus, we need $\widetilde{O}(Lm+m^2)$ qsample reflections to prepare those qsamples. 

In addition, the quantum estimation for ${\sf L}$ uses $\widetilde O(m/\varepsilon)$ qsample reflections by~\cref{lem:log-sum}. Summing the two costs and using $L=\widetilde{O}(D)$ by~\cref{lem:covariance-confidence} and $m=\widetilde O(\sqrt D)$ by~\eqref{eq:cooling-schedule}, we obtain that the total number of qsample reflections used in~\cref{alg:volume} is
\begin{equation}\label{eq:total-reflection-count}
 \begin{aligned}
 N=\widetilde O\left(Lm+m^2+\frac m\varepsilon\right)
 = \widetilde O\left(D^{3/2}+\frac{\sqrt D}{\varepsilon}\right).
 \end{aligned}
\end{equation}

To apply~\cref{cor:finite-reflection-cost-main}, we verify the covariance-estimation guarantee and bound the fourth-moment norm of each reference input.

Let $\Sigma_i=\Cov_{\pi_{a_i}}(X)$ be the covariance of the continuous target at stage $i$, and let ${\varepsilon_{\Cov}}$ be the failure probability for covariance estimation. By~\cref{lem:covariance-confidence}, we can choose $L$ so that the empirical covariance $\widehat{\Sigma}({\bf x})$ of $L$ independent continuous samples satisfies
\begin{align*}
    \Pr_{{\bf x}\sim \pi_{a_i}^{\otimes L}}\left[\frac{1}{2}\Sigma_i\preceq \widehat{\Sigma}({\bf x})\preceq 2\Sigma_i\right]\geq 1-{\varepsilon_{\Cov}}/2.
\end{align*}
Then, by \cref{lem:finite-block-transfer}, we can choose one sufficiently small lattice spacing $h$ so that, for every stage $i$, 
\[
 \Pr_{\mathbf x\sim\pi_i^{\otimes L}}
 \left[\frac14\Sigma_i\preceq\widehat\Sigma(\mathbf x)\preceq4\Sigma_i\right]
 \ge1-{\varepsilon_{\Cov}}.
\]
As discussed in~\cref{sec:coherent-rounding}, this ensures that $A({\bf x})\asymp \Sigma_i$ for those good ${\bf x}$. In addition, by~\cref{lem:adjacent-covariance-prelim} and~\cref{lem:cooling-moments}, we have $\Sigma_i\asymp \Sigma_{i+1}$. Thus, the covariance condition in~\cref{cor:finite-reflection-cost-main} is satisfied.

We bound the fourth-moment norm separately within each subroutine with exact reflections, starting from its reference input: 
\begin{itemize}
    \item A single-copy conversion uses reflections about two adjacent targets $\ket{\pi_i}$ and $\ket{\pi_{i+1}}$, whose relative second moments in both directions are at most $4$ by~\cref{lem:cooling-moments}.
    \item An \textsc{Add} call uses reflections about a single target and starts with an independent qsample, whose measurement law has relative second moment $1$ (\cref{subsec:estimator-moments}).
\end{itemize}
Thus, in either case, after $q$ preceding qsample-reflection queries within this subroutine,
\[
 \|\ket{\Psi_q}_{\mathsf{SA}}\|_{4,\pi}
 \le\sqrt2(1+2q),
\]
where $\pi$ is the target of the reflection being applied. Let $q_{\max}$ be a deterministic upper bound on the number of qsample-reflection queries in any one such subroutine. A single-copy conversion uses $O(\log(1/\varepsilon_{\rm conv}))$ queries by \cref{lem:transport}, where $\varepsilon_{\rm conv}$ is its prescribed trace-distance error. One call to \textsc{Add} uses $\widetilde O(1/\varepsilon)$ queries by \cref{subsec:estimator-moments}. Since $\varepsilon_{\rm conv}$ is inverse-polynomial in the input parameters, we may therefore take
\[
 q_{\max}=\widetilde O(1/\varepsilon),
 \qquad
 M_4:=\sqrt2(1+2q_{\max})=\widetilde O(1/\varepsilon).
\]
Consequently, every reference input to a qsample-reflection satisfies
\[
 \|\ket{\Psi}_{\mathsf{SA}}\|_{4,\pi}\le M_4.
\]
The estimate is applied separately within each subroutine.

By~\cref{cor:finite-reflection-cost-main}, each accurate covariance estimate gives a reflection circuit using $\widetilde O(D/\psi)$ membership queries, with error at most $\xi$ on these reference inputs. The total query complexity of~\cref{alg:volume} can be bounded by
\[
 \underbrace{\widetilde O\left(D^{3/2}+\frac{\sqrt D}{\varepsilon}\right)
 \cdot\widetilde O\left(\frac D\psi\right)}_{\text{qsample reflections}}+\underbrace{\widetilde{O}(L+m)}_{\text{preparing}~\ket{\pi_0}} + \underbrace{\widetilde{O}(1/\varepsilon)}_{\text{estimating}~q_{\rm in}}
 =\widetilde O\left(\frac{D^{5/2}}\psi
                     +\frac{D^{3/2}}{\psi\varepsilon}\right).
\]

\end{proof}

\begin{lemma}[Estimation error and success probability]\label{lem:est_error_succ_prob_main}
With probability at least $2/3$, \cref{alg:volume} can estimate $\vol(K_{\rm in})$ within relative error at most $\varepsilon$.    
\end{lemma}
\begin{proof}
We first analyze the reference computation in which every required qsample is supplied exactly and the reflection oracles are exact. We then bound the trace distance introduced by preparing and converting the qsamples and implementing the reflections, which bounds the decrease in success probability.

\paragraph{Estimation error.} Let $V_m$ and $V_{m,\Lambda}$ be defined as~\eqref{eq:def_Vm}. By~\eqref{eq:exact-volume-telescope},
\begin{align*}
    \log Z_{C,\Lambda} + \sum_{i=0}^{m-1}\log c_i - {\sf L} = \log V_{m,\Lambda}.
\end{align*}
Then, by~\eqref{eq:volume-output} and including the numerical error $e_{\rm num}$, we have 
\begin{align}\label{eq:vol_est_error_decomp}
    \log\frac{\widehat{V}}{\vol(K_{\rm in})}= &~ d\log r+\log Z_{C,\Lambda}
      +\sum_{i=0}^{m-1}\log c_i-\widehat{\mathsf L}-\log I_m +e_{\rm num}- (d\log r + \log \vol(K)) \notag\\
      = &~ \log V_{m,\Lambda} + ({\sf L} - \widehat{\sf L}) - \log (I_m \vol(K))+e_{\rm num}\notag\\
      = &~ \log\frac{V_{m,\Lambda}}{V_m} + ({\sf L} - \widehat{\sf L})+e_{\rm num}.
\end{align}

For the first term in~\eqref{eq:vol_est_error_decomp}, \cref{lem:full-grid-quadrature} ensures that
\[
 \left|\frac{V_{m,\Lambda}}{\vol(K)I_m}-1\right|
 \le\frac{\varepsilon}{16}.
\]
Since $|\log u|\le2|u-1|$ whenever $|u-1|\le1/2$, it follows that
\[
 \left|\log\frac{V_{m,\Lambda}}{\vol(K)I_m}\right|
 \le\frac{\varepsilon}{8}.
\]

For the second term in~\eqref{eq:vol_est_error_decomp}, conditioned on every $c_i$ being a 2-approximation of $\mu_i$ (\cref{lem:pilots}), the bridge bounds in~\cref{sec:coherent-product} and~\cref{lem:cooling-moments} imply that all probabilities in~\eqref{eq:volume-log-sum} are at least $1/32$. Applying \cref{lem:log-sum} with additive accuracy $\varepsilon/8$ gives
\[
 \left|\widehat{\mathsf L}-\mathsf L\right|\le\frac{\varepsilon}{8}
\]
with probability at least $49/50$. 

For the third term in~\eqref{eq:vol_est_error_decomp}, we can choose the scalar arithmetic precision as in~\cref{app:mark-arithmetic} so that
\begin{align*}
    |e_{\rm num}|\leq \frac{\varepsilon}{8}.
\end{align*}

Therefore, by the triangle inequality, we have
\begin{align*}
    \left|\log\frac{\widehat V}{\vol(K_{\rm in})}\right|
 \leq &~ |\mathsf L-\widehat{\mathsf L}|
 +\left|\log\frac{V_{m,\Lambda}}{\vol(K)I_m}\right|
 +|e_{\rm num}|\\
 \leq &~ \frac{\varepsilon}{8}
       +\frac{\varepsilon}{8}
       +\frac{\varepsilon}{8}
 =\frac{3\varepsilon}{8}<\frac{\varepsilon}{2}.
\end{align*}
Since $0<\varepsilon<1/2$, exponentiating gives
\[
 1-\varepsilon\le e^{-\varepsilon/2}
 \le\frac{\widehat V}{\vol(K_{\rm in})}
 \le e^{\varepsilon/2}\le1+\varepsilon.
\]

\paragraph{Success probability.} By~\cref{lem:pilots,lem:log-sum}, the failure probability of the reference computation is bounded by
\begin{align*}
    \frac{1}{100} + \frac{1}{50} = \frac{3}{100}.
\end{align*}
Consider the implemented circuit. Let $\beta=1/100$. Choose the parameters as described in~\cref{app:parameters} so that
\begin{align*}
    {\varepsilon_{\rm init}}\leq \frac{\beta}{4N},\qquad \xi\leq \frac{\beta}{8N},\qquad {\varepsilon_{\Cov}}\leq \left(\frac{\beta}{16N}\right)^2,\qquad \varepsilon_{\rm conv}\leq \frac{\beta}{4N},\qquad e_{\rm arith}\leq \frac{\beta}{4},
\end{align*}
Here ${\varepsilon_{\rm init}},\xi,\varepsilon_{\rm conv}$ are defined in~\cref{lem:query_complexity_main}, ${\varepsilon_{\Cov}}$ bounds the probability of an inaccurate covariance tuple as in~\eqref{eq:good-covariance}, and $e_{\rm arith}$ is the total error from the arithmetic gates. Then, by~\cref{lem:composition} and~\eqref{eq:composition-reflection-budget}, we have
\begin{equation}\label{eq:assembly-implementation-error}
\begin{aligned}
 \dtr(\rho_{\rm impl},\rho_{\rm ref})\leq &~ \underbrace{N{\varepsilon_{\rm init}}}_{e_{\rm init}} +\underbrace{N\varepsilon_{\rm conv}}_{e_{\rm prim}}+ N(\xi+2\sqrt{{\varepsilon_{\Cov}}}) + e_{\rm arith}\\
 \le &~ \beta/4+\beta/4+\beta/4+\beta/4=\beta,
\end{aligned}
\end{equation}
where $\rho_{\rm impl}$ is the final state of the implemented circuit, and $\rho_{\rm ref}$ is the final state of the reference computation. Therefore, the success probability of~\cref{alg:volume} is at least
\begin{align*}
    1-\frac{3}{100}-\beta=\frac{24}{25}>\frac{2}{3}. 
\end{align*}

\end{proof}

\section{A linear quantum query lower bound}\label{sec:volume-lower-bound}

We prove an $\Omega(d)$ quantum query lower bound~\cref{thm:intro-volume-lower_informal} by constructing two distributions of ellipsoids. Their volumes lie in disjoint constant-factor ranges with high probability, but a quantum algorithm that uses fewer than half as many queries as the dimension cannot distinguish between the two distributions.

\subsection{The ellipsoids and their membership queries}

Let $n=2\lfloor d/2\rfloor$, so $n$ is even and $d-n\in\{0,1\}$. Let 
\begin{align*}
    \mathrm O(n)=\left\{O\in\R^{n\times n}:O^\top O=\Id_n\right\}
\end{align*}
be the orthogonal group, and let $\mu_+$ and $\mu_-$ be the corresponding Haar probability measure conditioned on $\det O=1$ and $\det O=-1$, respectively. Equivalently, $\mu_+$ is the uniform probability measure on the rotation group $\mathrm{SO}(n)$, and multiplication by any fixed orthogonal matrix of determinant $-1$ sends $\mu_+$ to $\mu_-$.

For $0<a<1$, decompose a point as $x=(u,v)\in\R^n\times\R^{d-n}$ and define
\begin{align}\label{eq:def_KO}
 K_O=\left\{(u,v):\|(\Id_n-aO)u\|_2^2+v^2\le1\right\},
 \qquad \Delta_a(O)=\det(\Id_n-aO).
\end{align}
The $v$ coordinate is omitted when $n=d$. Since $0<a<1$ and $O$ is orthogonal,
\[
 (1-a)\|u\|_2\le\|(\Id_n-aO)u\|_2\le(1+a)\|u\|_2.
\]
The lower bound implies that $\Id_n-aO$ is invertible. Hence, $K_O$ is the inverse image of the unit ball under the invertible linear map $(u,v)\mapsto((\Id_n-aO)u,v)$, and is therefore an origin-centered ellipsoid. The same inequalities give
\[
 B(0,1/2)\subseteq K_O\subseteq B(0,(1-a)^{-1}).
\]
The eigenvalues of $O$ are $1$, $-1$, or conjugate pairs on the unit circle. Their contributions to $\Delta_a(O)$ are $1-a$, $1+a$, or $|1-ae^{i\theta}|^2$, which are all positive. Thus, writing $v_d=\vol(B(0,1))$, the linear change of variables gives
\[
 \vol(K_O)=\frac{v_d}{\Delta_a(O)}.
\]

A membership query depends on $O$ only through one quadratic form:
\begin{align}\label{eq:membership_KO}
 \one_{K_O}(u,v)
 =\one\!\left\{(1+a^2)\|u\|_2^2-2a\,u^\top Ou+v^2\le1\right\}.
\end{align}
The algorithm can only access $K_O$ through the membership oracle. The above formula allows us to compare the two oracle distributions without approximating the ellipsoid's boundary.

\begin{lemma}[Indistinguishability of the two orthogonal components]\label{lem:orthogonal-membership-indistinguishability}
Fix $0<a<1$. For any quantum algorithm making at most $q$ membership queries to $K_O$, if $2q<n$, its output distributions averaged over $O\sim\mu_+$ and $O\sim\mu_-$ are identical. This holds for every finite choice of coordinate registers and their precision.
\end{lemma}
\begin{proof}
We first compare membership at any fixed set of $r<n$ points $x_j=(u_j,v_j)$. Choose a unit vector $w$ orthogonal to $u_1,\ldots,u_r$, and put $H=\Id_n-2ww^\top$. Then, $\det H=-1$ and $Hu_j=u_j$, so
\[
 u_j^\top HOu_j=(Hu_j)^\top Ou_j=u_j^\top Ou_j,
 \qquad
 \one_{K_{HO}}(x_j)=\one_{K_O}(x_j).
\]
The first equality uses $H^\top=H$, and the second uses the membership formula~\eqref{eq:membership_KO}. Multiplication by $H$ exchanges $\mu_+$ and $\mu_-$, while preserving all these answers. In particular,
\[
 \E_{O\sim\mu_+}\!\left[\prod_{j=1}^r\one_{K_O}(x_j)\right]
 =\E_{O\sim\mu_-}\!\left[\prod_{j=1}^r\one_{K_O}(x_j)\right].
\]

Fix the algorithm's finite set of encoded query points. By the polynomial method~\cite[Lemma~4.2]{BealsEtAl01}, the probability of any specified output event is a multilinear polynomial of degree at most $2q$ in the membership bits at these points. If $2q<n$, the preceding equality applies to every monomial of this polynomial. Its expectations under $\mu_+$ and $\mu_-$ are therefore equal. The argument permits intermediate measurements and adaptive queries by storing the measurement outcomes in auxiliary registers. The equality of membership answers is exact, including on the boundary, so no restriction on coordinate precision is needed.
\end{proof}

\subsection{The determinant estimates}

It remains to separate the volumes while keeping $(1-a)^{-1}$ polynomial in $d$. In even dimension, every orthogonal matrix with determinant $-1$ fixes a nonzero vector, giving $K_O$ a semiaxis of length $(1-a)^{-1}$. A Haar-random orthogonal matrix with determinant $1$ has no fixed nonzero vector with probability one. The following moment bounds quantify this distinction after replacing $O$ by $aO$.

\begin{lemma}[Two determinant moments]\label{lem:orthogonal-determinant-moments}
For even $n\ge2$ and $0<a<1$,
\[
 \E_{O\sim\mu_-}[\Delta_a(O)]=1-a^n\le n(1-a),
 \qquad
 \E_{O\sim\mu_+}[\det(\Id_n-O)^{-1/4}]\le2n^{1/4}.
\]
\end{lemma}
\begin{proof}
Expand the determinant into principal minors:
\[
 \Delta_a(O)=\sum_{S\subseteq\{1,\ldots,n\}}(-a)^{|S|}\det O_{S,S}.
\]
For a nonempty proper subset $S$, choose $i\in S$ and $j\notin S$. Negating rows $i$ and $j$ preserves each of $\mu_+,\mu_-$, because the row operation has determinant $1$, but it negates $\det O_{S,S}$. Hence this minor has expectation zero. Only the empty and full minors remain. Since $n$ is even,
\[
 \E_{\mu_-}[\Delta_a(O)]
 =1+a^n\E_{\mu_-}[\det O]
 =1-a^n\le n(1-a).
\]
The last inequality is $1-a^n=(1-a)\sum_{j=0}^{n-1}a^j\le n(1-a)$.

For the second estimate, we use the recursive construction of Haar measure; see \cite[Section~1 and the remark in Section~2.2]{BourgadeEtAl07}. We give the determinant calculation to specify the resulting independent factors. In dimension $j\ge2$, for any uniformly random vector $\omega=(z,w)$ on the unit sphere, define
\[
 H_\omega=\Id_j-\frac{(e_1-\omega)(e_1-\omega)^\top}{1-z}.
\]
Outside the probability-zero event $\omega=e_1$, $H_\omega$ is a reflection satisfying
\begin{align*}
    H_\omega^2 = \Id_j, \qquad H_\omega e_1=\omega,\qquad \det H_\omega=-1.
\end{align*}

Let $O$ be a Haar-random orthogonal matrix conditioned on $\det O=s$, where $s\in\{\pm1\}$, and let $\omega=Oe_1$ be its first column. By the rotational invariance of the Haar measure, $\omega$ is a uniformly random vector on the unit sphere. Since $H_\omega\omega=H_\omega H_\omega e_1=e_1$, the orthogonal matrix $H_\omega O$
fixes $e_1$ and therefore has the form
\[
 H_\omega O=\begin{pmatrix}
     1 & \\
     & O'
 \end{pmatrix},
 \qquad O'\in\mathrm O(j-1).
\]
Moreover, $\det O'=\det(H_\omega)\det O=-s$.

For every $V\in\mathrm{SO}(j-1)$, right multiplication by
$\operatorname{diag}(1,V)$ preserves the law of $O$ and its first
column, while replacing $O'$ by $O'V$. Conditional on $\omega$,
the matrix $O'$ is therefore Haar distributed on the component
with determinant $-s$. This conditional law does not depend on
$\omega$, so $O'$ and $\omega$ are independent. Consequently,
\[
 O=H_\omega\begin{pmatrix}
     1 & \\
     & O'
 \end{pmatrix}
\]
gives the required recursive decomposition. Consider
\begin{align*}
    \Id_j - O = \begin{pmatrix}
 1-z&-w^\top O'\\
 -w&\Id_{j-1}-O'+(1-z)^{-1}ww^\top O'
 \end{pmatrix}.
\end{align*}
The Schur complement of $1-z$ is
\begin{align*}
    \Id_{j-1}-O'+(1-z)^{-1}ww^\top O' - (-w)(1-z)^{-1} (-w^\top O')=\Id_{j-1}-O'.
\end{align*}
Thus, we have
\[
 \det(\Id_j-O)=(1-z)\det(\Id_{j-1}-O').
\]
At each recursive step, the dimension decreases by one, and the determinant of the orthogonal matrix changes sign. Starting with $\det O=1$ in even dimension $n$, the final one-dimensional matrix therefore has determinant $(-1)^{n-1}=-1$. Therefore, 
\[
 \det(\Id_n-O)\ \stackrel{\rm law}{=}\ 2\prod_{j=2}^n(1-Z_j),
\]
where the $Z_j$ are independent and $Z_j$ is the first coordinate of a uniform point on the unit sphere in $\R^j$. All factors are positive almost surely.

For $j\ge4$, the density of $Z_j$ is proportional to $(1-z^2)^{(j-3)/2}$ on $[-1,1]$. A beta integral and concavity of $x^{1/4}$ give
\[
 \E[(1-Z_j)^{-1}]=\frac{j-2}{j-3},
 \qquad
 \E[(1-Z_j)^{-1/4}]\le\left(\frac{j-2}{j-3}\right)^{1/4}.
\]
For the two remaining dimensions, $Z_3$ is uniform on $[-1,1]$ and $Z_2$ has the law of $\cos\theta$ for uniform $\theta\in[0,2\pi]$. Thus
\begin{align*}
 \E[(1-Z_3)^{-1/4}]
 = &~ \frac12\int_{-1}^1(1-z)^{-1/4}\,dz
 =\frac43\,2^{-1/4},\\
 \E[(1-Z_2)^{-1/4}]
 = &~ \frac{2^{-1/4}}\pi\int_0^\pi(\sin u)^{-1/2}\,du
 \le2^{3/4}.
\end{align*}
The last inequality uses symmetry about $\pi/2$ and $\sin u\ge2u/\pi$ on $[0,\pi/2]$. Independence now yields, for $n\ge4$,
\begin{align*}
 \E_{\mu_+}[\det(\Id_n-O)^{-1/4}]
 = &~ 2^{-1/4}\prod_{j=2}^n\E[(1-Z_j)^{-1/4}]\\
 \leq &~ \frac43\,2^{1/4}
 \left(\prod_{j=4}^n\frac{j-2}{j-3}\right)^{1/4}\\
 = &~ \frac43\,2^{1/4}(n-2)^{1/4}\le2n^{1/4}.
\end{align*}
For $n=2$, the bound follows directly from the estimate for $Z_2$.
\end{proof}

\subsection{From volume estimation to distinguishing the components}
\begin{theorem}[Linear quantum query lower bound]\label{thm:intro-volume-lower}
There is a universal constant $C>0$ such that, for every $d\ge2$, every quantum algorithm that estimates the volume of a convex body $K\subset\R^d$ to relative error $1/4$ with success probability at least $2/3$ makes $\Omega(d)$ membership queries in the worst case. This holds even when $K$ is promised to be an ellipsoid centered at the origin with $B(0,1/2)\subseteq K\subseteq B(0,Cd^2)$, for queries of arbitrary finite precision, and with unrestricted computation between queries.
\end{theorem}
\begin{proof}[Proof of \cref{thm:intro-volume-lower}]
Choose a sufficiently small universal dyadic constant $c_*>0$, and set
\[
 n=2\lfloor d/2\rfloor,
 \qquad \tau=\frac{c_*}{n},
 \qquad \gamma=\frac{c_*}{64n^2},
 \qquad a=1-\gamma.
\]
We take $K=K_O$, which satisfies
\[
 B(0,1/2)\subseteq K_O\subseteq B(0,64d^2/c_*).
\]

On the negative component $\mu_-$, positivity of $\Delta_a$ and~\cref{lem:orthogonal-determinant-moments} give
\begin{align*}
 \Pr_{\mu_-}[\Delta_a(O)>\tau/4]
 \leq \frac{4\E_{\mu_-}[\Delta_a(O)]}{\tau}
 \leq \frac{4n\gamma}{\tau}=\frac1{16},
\end{align*}
where the first step follows from Markov's inequality. 

On the positive component $\mu_+$, each eigenvalue $\lambda$ of $O$ satisfies
\[
 |1-a\lambda|^2=(1-a)^2+a|1-\lambda|^2\ge a|1-\lambda|^2.
\]
Multiplying over the eigenvalues and using positivity gives
\[
 \Delta_a(O)\ge a^{n/2}\det(\Id_n-O)
 \ge\frac12\det(\Id_n-O),
\]
where $a^{n/2}=(1-\gamma)^{n/2}\ge1-n\gamma/2\ge1/2$. Applying Markov's inequality to the inverse fourth root, we obtain
\begin{align*}
 \Pr_{\mu_+}[\Delta_a(O)<\tau]
 \leq &~ \Pr_{\mu_+}[\det(\Id_n-O)<2\tau]\\
 \leq &~ (2\tau)^{1/4}\E_{\mu_+}[\det(\Id_n-O)^{-1/4}]\\
 \leq &~ 2(2n\tau)^{1/4}
 =2(2c_*)^{1/4}\le\frac18,
\end{align*}
where the last step fixes $c_*$ sufficiently small. Therefore, for the known number $V_*=v_d/\tau$,
\[
 \Pr_{\mu_+}[\vol(K_O)\le V_*]\ge\frac78,
 \qquad
 \Pr_{\mu_-}[\vol(K_O)\ge4V_*]\ge\frac78.
\]

Suppose an algorithm makes $q<n/2$ membership queries and estimates the volume of every promised body to relative error $1/4$ with probability at least $2/3$. Decide that $O$ has negative determinant when its estimate exceeds $2V_*$. On a successful estimation event, the two volume ranges give
\[
 \vol(K_O)\le V_*\ \Longrightarrow\ \widehat V\le\frac54V_*<2V_*,
 \qquad
 \vol(K_O)\ge4V_*\ \Longrightarrow\ \widehat V\ge3V_*>2V_*.
\]
Since the estimation guarantee holds for every fixed $O$, the resulting decision succeeds under each distribution with probability at least
\[
 \frac78\cdot\frac23=\frac7{12}>\frac12.
\]
This is impossible by~\cref{lem:orthogonal-membership-indistinguishability}: identical averaged output distributions have average success probability exactly $1/2$ when the two components have equal prior probability. Thus $q\ge n/2=\lfloor d/2\rfloor$, proving the linear lower bound.
\end{proof}

\section*{Acknowledgments}
\addcontentsline{toc}{section}{Acknowledgments}
The author thanks Shouvanik Chakrabarti, Arjan Cornelissen, Yassine Hamoudi, Tongyang Li, and Xinzhao Wang for helpful discussions during the early stages of this project.
\section*{AI Methodology}
\addcontentsline{toc}{section}{AI Methodology}

\begin{flushleft}
\renewcommand{\arraystretch}{1.3}
\begin{tabular}{@{}l@{\qquad}c@{}}
\textbf{Activity} & \textbf{AI used} \\
Asking the research question & \texttt{No} \\
Coming up with the approach & \texttt{Yes} \\
Proof development & \texttt{Yes} \\
Writing and exposition & \texttt{Yes} \\
Checking for bugs & \texttt{No} \\
Other supporting tasks & \texttt{Yes} \\
\end{tabular}
\end{flushleft}

\begin{sloppypar}
\paragraph{Details.}
The author began thinking about how progress on the KLS conjecture could improve the quantum volume algorithm of~\cite{CCHLW23} in 2023. The main obstacles were the slower known mixing bounds for the hit-and-run walk compared with the ball walk, a gap in the earlier effective spectral-gap argument, and the difficulty of implementing the speedy walk in quantum. After obtaining the results in the companion paper~\cite{Zhang26HitRun}, the author recognized that these techniques could lead to improved quantum volume algorithms.

With the author's guidance, ChatGPT 5.5 Pro developed a way to use uniform warm-start mixing bounds to replace the earlier effective spectral-gap argument based on a single initial distribution. The resulting algorithm reduced the first term in the reported query bound from $d^3$ to $d^{5/2}$, while keeping the second term of~\cite{CornelissenHamoudi23}. This initial version used a lattice discretization of hit-and-run. Further interactions with ChatGPT 5.6 Pro and GPT 5.6 Sol led to an improved quantum product estimator, yielding the query complexity stated in~\cref{thm:intro-volume_informal}. The $\Omega(d)$ lower bound was obtained by asking ChatGPT 6 Pro to improve the quantum volume estimation lower bound using the techniques in~\cite{Childs26Matrix}. However, after multiple rounds of conversations, ChatGPT 6 Pro and GPT-6 Astra did not obtain a superlinear lower bound. 

The resulting proof was long and technically involved. The author spent considerable effort over months checking its correctness and simplifying it, with assistance from ChatGPT 5.6 Pro, ChatGPT 6 Pro, and Claude Fable 5.1. The lattice Metropolis walk used in the present paper was developed during this process. 

The author verified all AI-generated material before incorporating it into the manuscript and wrote the final proofs and text, with AI assistance for polishing and creating figures. The author takes full responsibility for the paper's correctness and presentation.
\end{sloppypar}

\clearpage
\appendix
\crefalias{section}{appendix}
\crefalias{subsection}{appendix} 

\section{Why one-start mixing is insufficient}\label{app:counterexample}

Classical mixing controls the density ratio of a starting law, whereas its qsample contains the square root of that ratio.  Taking square roots can destroy cancellation against a slow eigenfunction. The following finite example disproves the general implication used in \cite[Proposition~4.2]{CCHLW23}.

\begin{proposition}[A four-state counterexample]\label{prop:counterexample}
For every integer $t\ge1$, there is a lazy, reversible, ergodic four-state chain with uniform stationary law $\pi$, a nonstationary eigenvalue $\lambda=1-1/(100t)$, and a two-warm start $a$ such that $\sum_x\pi(x)^2/a(x)\le2$ and
\[
 \dtv(a\sP^j,\pi)=\frac3{16}2^{-j}\quad(j\ge0),
\]
but the projection of $\ket{a}$ on the $\lambda$-eigenspace has a fixed positive norm independent of $t$.
\end{proposition}
\begin{proof}
Let $\one=(1,1,1,1)^\top$ and $s=(1,1,-1,-1)^\top$, and set
\[
 \sP=\tfrac12\Id+\tfrac18\one\one^\top
                    +\tfrac{2\lambda-1}{8}ss^\top.
\]
The entries have the explicit form
\[
 \sP(x,y)=\tfrac12\one_{\{x=y\}}
                    +\frac{1+(2\lambda-1)s_xs_y}{8}>0.
\]
Since $0<2\lambda-1<1$, positivity holds for both possible signs of $s_xs_y$.  Also $s^\top\one=0$, so $\sP\one=\one$. The matrix is symmetric and its diagonal is at least $1/2$. Thus it defines a lazy, reversible, irreducible chain with uniform stationary law.

The orthogonality relations and $\|s\|_2^2=\|\one\|_2^2=4$ give
\[
 \sP\one=\one,\qquad
 \sP s=\left(\tfrac12+\tfrac{2\lambda-1}{2}\right)s=\lambda s,
 \qquad
 \sP v=\tfrac12v\quad(v\perp\one,s).
\]
Hence its spectrum is $1,\lambda,1/2,1/2$. Its discriminant is the same matrix because $\pi$ is uniform; $\ket{v}_{\mathsf{S}}=\frac12\sum_xs_x\ket{x}_{\mathsf{S}}$ is a stationary-orthogonal eigenstate with eigenvalue $\lambda$.

Choose $g=(1,-1,2,-2)^\top$ and $a_x=(1+g_x/4)/4$. Since $g\perp\one,s$, for every $j\ge0$,
\[
 a^\top\sP^j=\pi^\top+2^{-j}g^\top/16.
\]
Therefore
\[
 \dtv(a\sP^j,\pi)
 =\tfrac12\sum_x|2^{-j}g_x/16|
 =\frac{2^{-j}}{32}(1+1+2+2)
 =\frac3{16}2^{-j}.
\]
The density ratio and reverse moment satisfy
\[
 \tfrac12\le a_x/\pi_x\le\tfrac32,
 \qquad
 \sum_x\frac{\pi_x^2}{a_x}
 =\sum_x\pi_x\frac{\pi_x}{a_x}\le2.
\]
These bounds are independent of $t$.

The coherent amplitude has a different slow coefficient:
\[
 \langle v|a\rangle
 =\tfrac14\bigl(\sqrt{1+1/4}+\sqrt{1-1/4}
                       -\sqrt{1+1/2}-\sqrt{1-1/2}\bigr)>0.
\]
Indeed,
\[
 \frac{d}{du}\bigl(\sqrt{1+u}+\sqrt{1-u}\bigr)
 =\frac1{2\sqrt{1+u}}-\frac1{2\sqrt{1-u}}<0
 \quad(0<u<1).
\]
This positive coefficient does not depend on $t$, whereas the one-start mixing error $(3/16)2^{-t}$ tends to zero.  Thus, that mixing error cannot bound the slow qsample component by a quantity tending to zero.
\end{proof}

For a slow eigenfunction $f$, classical mixing controls $\int f(a/\pi)\,d\pi$, not $\int f\sqrt{a/\pi}\,d\pi$. A pointwise bound on $\sqrt{\pi/a}$ cannot compare these signed integrals while preserving cancellation.  This is the issue with the coefficient comparison in \cite[Equation~(4.40)]{CCHLW23}.

\paragraph{Connection to the Szegedy walk.}
To compare with the walk formulation in~\cite{CCHLW23}, let $V_1,V_2$ be the Szegedy isometries with $V_1^\dagger V_2=D_\pi$, and put $W=(2V_2V_2^\dagger-\Id)(2V_1V_1^\dagger-\Id)$. For a real state $\ket{v}_{\mathsf{S}}$ satisfying $D_\pi\ket{v}_{\mathsf{S}}=\lambda\ket{v}_{\mathsf{S}}$, where $0<\lambda<1$, an orthonormal basis of the corresponding invariant plane is
\[
 \ket{e_1}=V_1\ket{v}_{\mathsf{S}},
 \qquad
 \ket{e_2}=\frac{V_2\ket{v}_{\mathsf{S}}-\lambda V_1\ket{v}_{\mathsf{S}}}
                           {\sqrt{1-\lambda^2}}.
\]
For any vector $\ket{\psi}_{\mathsf{S}}$, the identities $V_1^\dagger V_1=\Id$ and $V_1^\dagger V_2\ket{v}_{\mathsf{S}} =\lambda\ket{v}_{\mathsf{S}}$ give
\[
 \bra{e_1}V_1\ket{\psi}_{\mathsf{S}}=\langle v|\psi\rangle,
 \qquad
 \bra{e_2}V_1\ket{\psi}_{\mathsf{S}}
 =\frac{\lambda\langle v|\psi\rangle
                 -\lambda\langle v|\psi\rangle}
                  {\sqrt{1-\lambda^2}}=0.
\]
In this basis the walk is
\[
 W\big|_{\Span\{\ket{e_1},\ket{e_2}\}}=
 \begin{pmatrix}
 2\lambda^2-1&-2\lambda\sqrt{1-\lambda^2}\\
 2\lambda\sqrt{1-\lambda^2}&2\lambda^2-1
 \end{pmatrix}.
\]
Its eigenstates $(\ket{e_1}\mp i\ket{e_2})/\sqrt2$ have eigenphases $\pm2\arccos\lambda$. Therefore each corresponding coefficient of $V_1\ket{\psi}_{\mathsf{S}}$ has magnitude $|\langle v|\psi\rangle|/\sqrt2$. For the counterexample, take $\ket{v}_{\mathsf{S}}=\frac12\sum_xs_x\ket{x}_{\mathsf{S}}$ and $\ket{\psi}_{\mathsf{S}}=\ket{a}_{\mathsf{S}}$. The computed overlap is nonzero, so both slow-walk phases have nonzero amplitudes.

\begin{remark}[Sharpness of the fourth-root dependence]
\label{rem:fourth-root-sharp}
Take $\pi_p=(p,1-p)$ and $\sP_p=\frac34\Id+\frac14\one\pi_p^\top$, with $0<p\le1/4$. Every two-warm start is $(q,1-q)$ for $0\le q\le2p$; hence the largest one-step distance from stationarity is $\eta=3p/4$. The centered function $f_p=\one_{\{1\}}-p$ has eigenvalue $3/4$, which belongs to the slow interval $[1/2,1)$ for $t=1$. Its norms satisfy
\[
 \|f_p\|_2^2=p(1-p),\qquad
 \|f_p\|_4^4=p(1-p)^4+(1-p)p^4.
\]
Set $\ket{\Psi_p}_{\mathsf{S}}=J_{\pi_p}f_p$, which is a possibly unnormalized vector in the slow subspace. The displayed moments give
\[
 \frac{\|\ket{\Psi_p}_{\mathsf{S}}\|}{\|\ket{\Psi_p}_{\mathsf{S}}\|_{4,\pi_p}}
 =\frac{\|f_p\|_2}{\|f_p\|_4}\asymp p^{1/4}.
\]
Since $\eta=3p/4$, no larger exponent of $\eta$ can replace $1/4$ in \eqref{eq:warm-low-spectrum}.
\end{remark}

\section{Geometry and mixing of the lattice Metropolis walk}\label{app:finite-state}
This appendix proves \cref{thm:finite-metropolis}. We first state some geometric tools from the companion paper~\cite{Zhang26HitRun} in~\cref{subsec:classical-logaffine}. In~\cref{subsec:cubic-grid}, we bound the quadrature error of the lattice-discretized normalization constants and analyze affine preconditioning on the fixed cubic lattice $h\Z^D$. We then prove local overlap for the lattice Metropolis walk in~\cref{subsec:finite-proposals} and combine it with the separator estimate to obtain the conductance and warm-start mixing bounds in~\cref{subsec:lattice-conductance}. In \cref{subsec:common-mesh}, we transfer empirical covariance concentration to the lattice target and bound the continuous covariance using the known geometric parameters. Finally, we combine these estimates to uniformly select the lattice spacing and the lattice Metropolis walk parameters that work for all tilts and allowed preconditioners, completing the proof of \cref{thm:finite-metropolis}.

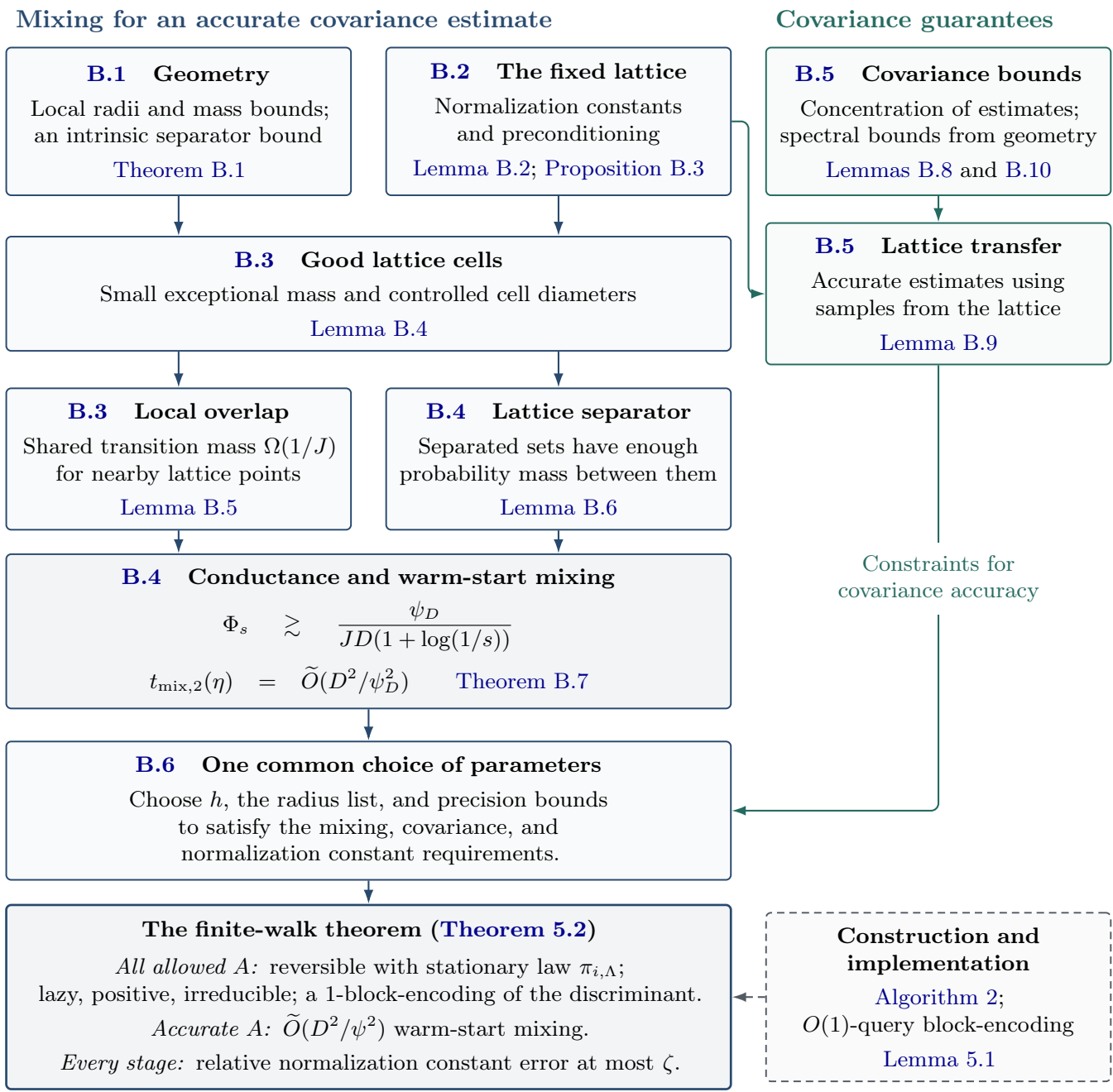
\begin{figure}[!ht]
 \centering
 \begingroup
 \definecolor{proofblue}{RGB}{40,73,110}
 \definecolor{proofteal}{RGB}{36,106,96}
 \definecolor{proofgray}{RGB}{85,91,100}
 \makebox[\textwidth][c]{%
 \resizebox{1.08\textwidth}{!}{%
 \begin{tikzpicture}[
  x=1cm,y=1cm,
  every node/.style={font=\fontsize{9}{11}\selectfont,align=center},
  proofbox/.style={draw=proofblue,fill=proofblue!3,rounded corners=2pt,
    line width=.65pt,inner xsep=5pt,inner ysep=6pt,
    text width=4.53cm,minimum width=4.9cm,minimum height=1.85cm},
  widebox/.style={proofbox,text width=9.93cm,minimum width=10.3cm},
  supportbox/.style={proofbox,draw=proofteal,fill=proofteal!3},
  dependency/.style={-{Latex[length=2mm,width=1.4mm]},
    draw=proofblue,line width=.7pt},
  support/.style={dependency,draw=proofteal},
  imported/.style={dependency,draw=proofgray,densely dashed}]

  \node[anchor=west,text=proofblue,font=\fontsize{10}{11}\selectfont\bfseries]
    at (0,1.42) {Mixing for an accurate covariance estimate};
  \node[anchor=west,text=proofteal,font=\fontsize{10}{11}\selectfont\bfseries]
    at (10.8,1.42) {Covariance guarantees};

  \node[proofbox,minimum height=2.1cm] (geometry) at (2.45,0) {
    \textbf{\ref{subsec:classical-logaffine}\quad Geometry}\\[3pt]
    Local radii and mass bounds;\\
    an intrinsic separator bound\\[3pt]
    \cref{thm:classical-log-affine-input}};

  \node[proofbox,minimum height=2.1cm] (grid) at (7.85,0) {
    \textbf{\ref{subsec:cubic-grid}\quad The fixed lattice}\\[3pt]
    Normalization constants\\
    and preconditioning\\[3pt]
    \cref{lem:full-grid-quadrature};~\cref{prop:finite-affine-branches}};

  \node[supportbox,minimum height=2.1cm] (covariance) at (13.25,0) {
    \textbf{\ref{subsec:common-mesh}\quad Covariance bounds}\\[3pt]
    Concentration of estimates;\\
    spectral bounds from geometry\\[3pt]
    \cref{lem:covariance-confidence,lem:log-affine-covariance-floor}};

  \node[widebox,minimum height=1.5cm] (cells) at (5.15,-2.45) {
    \textbf{\ref{subsec:finite-proposals}\quad Good lattice cells}\\[3pt]
    Small exceptional mass and controlled cell diameters\\[3pt]
    \cref{lem:mesh-collar}};

  \node[supportbox,minimum height=1.65cm] (transfer) at (13.25,-2.45) {
    \textbf{\ref{subsec:common-mesh}\quad Lattice transfer}\\[3pt]
    Accurate estimates using\\
    samples from the lattice\\[3pt]
    \cref{lem:finite-block-transfer}};

  \node[proofbox,minimum height=1.7cm] (overlap) at (2.45,-4.8) {
    \textbf{\ref{subsec:finite-proposals}\quad Local overlap}\\[3pt]
    Shared transition mass $\Omega(1/J)$\\
    for nearby lattice points\\[3pt]
    \cref{lem:finite-overlap}};

  \node[proofbox,minimum height=1.7cm] (separator) at (7.85,-4.8) {
    \textbf{\ref{subsec:lattice-conductance}\quad Lattice separator}\\[3pt]
    Separated sets have enough\\
    probability mass between them\\[3pt]
    \cref{lem:cell-separator}};

  \node[widebox,fill=proofblue!6,minimum height=1.95cm] (mixing) at (5.15,-7.25) {
    \textbf{\ref{subsec:lattice-conductance}\quad Conductance and warm-start mixing}\\[4pt]
    $\displaystyle \Phi_s\gtrsim\frac{\psi_D}{JD(1+\log(1/s))}$\\[5pt]
    $t_{\mathrm{mix},2}(\eta)=\widetilde O(D^2/\psi_D^2)$\qquad
    \cref{thm:finite-mixing}};

  \node[widebox,minimum height=1.65cm] (parameters) at (5.15,-9.8) {
    \textbf{\ref{sec:B_proof}\quad One common choice of parameters}\\[3pt]
    Choose $h$, the radius list, and precision bounds\\
    to satisfy the mixing, covariance, and normalization constant requirements.};

  \node[widebox,line width=.9pt,fill=proofblue!6,minimum height=2.3cm] (result) at (5.15,-12.45) {
    \textbf{The finite-walk theorem (\cref{thm:finite-metropolis})}\\[4pt]
    \emph{All allowed $A$:} reversible with stationary law $\pi_{i,\Lambda}$;\\
    lazy, positive, irreducible; a $1$-block-encoding of the discriminant.\\[3pt]
    \emph{Accurate $A$:} $\widetilde O(D^2/\psi^2)$ warm-start mixing.\\[3pt]
    \emph{Every stage:} relative normalization constant error at most $\zeta$.};

  \node[proofbox,draw=proofgray,fill=black!2,densely dashed,
    minimum height=2.3cm] (implementation) at (13.25,-12.45) {
    \textbf{Construction and implementation}\\[3pt]
    \cref{alg:finite-metropolis};\\
    $O(1)$-query block-encoding\\[3pt]
    \cref{lem:metropolis-block-encoding}};

  \draw[dependency] (geometry.south)--([xshift=-2.7cm]cells.north);
  \draw[dependency] (grid.south)--([xshift=2.7cm]cells.north);
  \draw[support] (covariance.south)--(transfer.north);
  \draw[support,rounded corners=3pt]
    (grid.east)--(10.55,0)--(10.55,-2.45)--(transfer.west);

  \draw[dependency] ([xshift=-2.7cm]cells.south)--(overlap.north);
  \draw[dependency] ([xshift=2.7cm]cells.south)--(separator.north);
  \draw[dependency] (overlap.south)--([xshift=-2.7cm]mixing.north);
  \draw[dependency] (separator.south)--([xshift=2.7cm]mixing.north);
  \draw[dependency] (mixing.south)--(parameters.north);
  \draw[dependency] (parameters.south)--(result.north);

  \draw[support,rounded corners=3pt]
    (transfer.south)--(13.25,-9.8)--(parameters.east);
  \node[fill=white,text=proofteal,inner sep=5pt,text width=4.1cm]
    at (13.25,-6.5) {Constraints for\\covariance accuracy};
  \draw[imported] (implementation.west)--(result.east);
 \end{tikzpicture}
 }
 }
 \endgroup
 \caption{Proof structure of~\cref{app:finite-state}. The left branch proves mixing for a fixed accurate preconditioner $A\asymp\operatorname{Cov}_\pi(X)$. The right branch transfers covariance concentration to lattice samples. \cref{sec:B_proof} chooses common parameters for all cooling stages and allowed preconditioners and completes the proof of the main theorem for the lattice Metropolis walk (\cref{thm:finite-metropolis}).}
 \label{fig:appendix-b-proof-map}
\end{figure}

\subsection{Continuous geometry}\label{subsec:classical-logaffine}
Let $K\subset\R^D$ be a convex body. For a continuous log-affine distribution
\[
 \pi_b(dx)=\frac{f_b(x)}{Z_b}\,dx,\qquad f_b(x)=e^{-\langle b,x\rangle}\one_K(x),\qquad Z_b=\int_K f_b(x)\,dx,
\]
define the local radius
\begin{equation}\label{eq:density-local-radius}
 \rho_b(x)=\sup\left(\{0\}\cup\left\{r>0:\frac{\vol(B(x,r)\cap\{y:f_b(y)\ge\tfrac34f_b(x)\})}{\vol(B(0,r))}\ge\frac{63}{64}\right\}\right),\qquad x\in K.
\end{equation}
Thus, at any radius $r\leq \rho_b(x)$, 
\begin{align*}
    \Pr_{U\sim B(0,r)}\left[x+U\in K~\wedge~e^{-\langle b,U\rangle}\geq \frac{3}{4}\right]\geq \frac{63}{64}.
\end{align*}

For $x,y\in\operatorname{int}K$, define
\begin{equation}\label{eq:density-radius-metric}
 d_{\rho_b}(x,y)=\inf_\gamma\int_0^1\frac{\|\gamma'(t)\|_2}{\rho_b(\gamma(t))}\,dt,
\end{equation}
where the infimum is over absolutely continuous paths in $\operatorname{int}K$ joining $x$ and $y$. For subsets $S,T\subseteq K$, set $d_{\rho_b}(S,T)=\inf\{d_{\rho_b}(x,y):x\in S\cap\operatorname{int}K,\ y\in T\cap\operatorname{int}K\}$, with infimum $+\infty$ for an empty set. The following estimates are from the companion paper.

\begin{theorem}[{\cite[Lemmas~F.2, F.3, and F.5; Corollary~F.4; Theorem~F.6]{Zhang26HitRun}}]\label{thm:classical-log-affine-input}
Suppose $\pi_b$ is $C_0$-near-isotropic. There is a constant $C$ that depends only on $C_0$ such that the following holds.
\begin{enumerate}[label=(\roman*)]
\item The zero extension of $\rho_b$ is $C\sqrt D$-Lipschitz, is positive in $\operatorname{int}K$, and vanishes on $\partial K$. In particular, $\rho_b(x)\le C\sqrt D\operatorname{dist}(x,\partial K)$.
\item For every $u>0$, $\pi_b\{x:\rho_b(x)<u\}\le\min\{1,C\sqrt D\,u\}$.
\item For interior $x,y$,
\[
 \left|\log\frac{\rho_b(y)}{\rho_b(x)}\right|\le C\sqrt D\,d_{\rho_b}(x,y),\qquad
 \|x-y\|_2\le\rho_b(x)e^{C\sqrt D\,d_{\rho_b}(x,y)}d_{\rho_b}(x,y).
\]
\item If measurable $S,T\subseteq K$ have smaller mass $p=\min\{\pi_b(S),\pi_b(T)\}>0$ and $d_{\rho_b}(S,T)\ge a/\sqrt D$ for a fixed $a>0$, then
\[
 \pi_b(K\setminus(S\cup T))\ge c_a\frac{\psi_D}{D(1+\log(1/p))}\,p.
\]
For sufficiently small fixed $a$, $c_a$ can be chosen proportional to $a$.
\end{enumerate}
\end{theorem}

\subsection{The cubic grid and affine preconditioning}\label{subsec:cubic-grid}
Fix $B(x_c,r_0)\subseteq K\subseteq B(x_c,R_0)$ and a tilt $b$. Write $\pi=\pi_b$ and $Z=Z_b$. The discretization lattice and its cells are
\[
 \Lambda=h\Z^D,\qquad Q_\Lambda=h[-1/2,1/2)^D.
\]
Choose a known dyadic cell-radius bound 
\begin{align}\label{eq:rho_Lambda_bound}
    h\sqrt D/2\le\varrho_\Lambda\le h\sqrt D.
\end{align}
In particular, $\|e\|_2\le\varrho_\Lambda$ for every $e\in Q_\Lambda$. The lattice-discretized target distribution is
\begin{equation}\label{eq:full-grid-target}
 \Omega_\Lambda=K\cap\Lambda,\qquad
 Z_\Lambda=h^D\sum_{z\in\Omega_\Lambda}e^{-\langle b,z\rangle},\qquad
 \pi_\Lambda(z)=\frac{h^De^{-\langle b,z\rangle}}{Z_\Lambda}.
\end{equation}

The following lemma bounds the quadrature error to the normalization constant $Z$ due to the lattice discretization. It also compares a sample from $\pi_\Lambda$ plus an independent uniform offset in $Q_\Lambda$ and a sample from $\pi$ in total-variation distance. 

\begin{lemma}[Quadrature error]\label{lem:full-grid-quadrature}
Suppose $\varrho_\Lambda\le r_0/4$, and for a sufficiently large constant $C$, let 
\begin{equation}\label{eq:quadrature-error}
 \varepsilon_{\rm quad}=C\left(\frac{(D+\|b\|_2R_0)\varrho_\Lambda}{r_0}+\|b\|_2\varrho_\Lambda\right)\leq 1.
\end{equation}
Then
\begin{equation}\label{eq:full-grid-quadrature}
 |Z_\Lambda-Z|\le\varepsilon_{\rm quad}Z.
\end{equation}
Moreover, for $X_\Lambda\sim\pi_\Lambda$ and $E$ an independent uniform point in $Q_\Lambda$, the law $\widetilde\pi_\Lambda$ of $X_\Lambda+E$ satisfies
\begin{equation}\label{eq:finite-jitter-tv}
 \dtv(\widetilde\pi_\Lambda,\pi)\le\varepsilon_{\rm quad}.
\end{equation}
These bounds hold simultaneously for finitely many tilts when $\|b\|_2$ is replaced in \eqref{eq:quadrature-error} by their largest norm.
\end{lemma}
\begin{proof}
Let
\begin{align*}
    U_\Lambda=\bigcup_{z\in\Omega_\Lambda}(z+Q_\Lambda).
\end{align*}
We extend the discrete density to $\mathbb{R}^D$ by defining $f_\Lambda(y)=e^{-\langle b,z\rangle}$ for any $y$ in a cell $z+Q_\Lambda$ for $z\in\Omega_\Lambda$, and zero elsewhere. Then, $f_\Lambda$ is supported on $U_\Lambda$ and
\begin{align*}
    \int_{U_\Lambda} f_\Lambda(y)\,dy=\sum_{z\in \Omega_\Lambda}\int_{z+Q_\Lambda}e^{-\langle b,z\rangle}\,dy=h^D \sum_{z\in \Omega_\Lambda} e^{-\langle b,z\rangle}=Z_\Lambda. 
\end{align*}
On each cell $z+Q_\Lambda$,
\begin{align}\label{eq:in_cell_error_bound}
 e^{-\|b\|_2\varrho_\Lambda}\le\frac{f_\Lambda(y)}{e^{-\langle b,y\rangle}}=e^{-\langle b, z-y\rangle}\le e^{\|b\|_2\varrho_\Lambda}.
\end{align}

We first bound the quadrature error $|Z_\Lambda-Z|$, which can be decomposed as follows:
\begin{align}\label{eq:quadrature_decomposition}
    |Z_\Lambda-Z|=&~ \left|\int_{U_\Lambda} f_\Lambda(y)\,dy - \int_Kf_b(y)\,dy\right|\leq \|f_\Lambda-f_b\|_{1}\notag\\
    = &~ \int_{K\cap U_\Lambda} |f_\Lambda(y)-f_b(y)|\,dy + \int_{U_\Lambda\backslash K} f_\Lambda(y)\,dy + \int_{K\backslash U_\Lambda} f_b(y)\,dy\notag\\
    \leq &~ (e^{\|b\|_2\varrho_\Lambda}-1)\int_{K\cap U_\Lambda} e^{-\langle b,y\rangle}\,dy + e^{\|b\|_2\varrho_\Lambda}\int_{U_\Lambda\backslash K} e^{-\langle b,y\rangle}\,dy+\int_{K\backslash U_\Lambda} e^{-\langle b,y\rangle}\,dy\notag\\
    = &~ \underbrace{\int_{U_\Lambda\triangle K}e^{-\langle b,y\rangle}\,dy}_{\text{support error}} + \underbrace{(e^{\|b\|_2\varrho_\Lambda}-1)\int_{U_\Lambda}e^{-\langle b,y\rangle}\,dy}_{\text{in-cell error}},
\end{align}
where the third line follows from~\eqref{eq:in_cell_error_bound}. 

To bound the above two parts of the quadrature error, we show some geometric bounds for $U_\Lambda$ and $U_\Lambda\triangle K$. For $t=\varrho_\Lambda/r_0\leq 1/4$, define
\[
 K_t^\pm=x_c+(1\pm t)(K-x_c).
\]
The inner ball and convexity give
\[
 K_t^-+\varrho_\Lambda B(0,1)\subseteq K,\qquad K+\varrho_\Lambda B(0,1)\subseteq K_t^+.
\]
Every point of $K_t^-+\varrho_\Lambda B(0,1)$ has the form $x_c+(1-t)(x-x_c)+tr_0u$ for some $x\in K$ and $u\in B(0,1)$, and 
\begin{align*}
    x_c + (1-t)(x-x_c) + tr_0u=(1-t)x + t(x_c+r_0 u)\in K,
\end{align*}
where the last step follows from $B(x_c,r_0)\subseteq K$. On the other hand, $x+tr_0u\in K+\varrho_\Lambda B(0,1)$ and
\begin{align*}
    \frac{x + tr_0 u - x_c}{1+t} + x_c=\frac{1}{1+t} x + \frac{t}{1+t} (x_c+r_0u)\in K\quad \Longrightarrow \quad x+tr_0 u\in K_t^+.
\end{align*}
Since every point in $\mathbb{R}^D$ is within $\varrho_\Lambda$ of its cell center, we have
\begin{equation}\label{eq:parallel-sandwich}
 K_t^-\subseteq U_\Lambda\subseteq K_t^+,\qquad U_\Lambda\triangle K\subseteq K_t^+\setminus K_t^-.
\end{equation}

Then, the first term in~\eqref{eq:quadrature_decomposition} can be bounded by
\begin{align}\label{eq:integral_U_K_diff}
 \int_{U_\Lambda\triangle K}e^{-\langle b,y\rangle}\,dy \leq &~ \int_{K_t^+} e^{-\langle b,y\rangle}\,dy - \int_{K_t^-} e^{-\langle b,y\rangle}\,dy\notag\\
 = &~  (1+t)^D\int_{K} e^{-\langle b,x\rangle}e^{-t\langle b, x-x_c\rangle}\,dx - (1-t)^D\int_{K} e^{-\langle b,x\rangle}e^{t\langle b, x-x_c\rangle}\,dx\notag\\
 \leq &~ (1+t)^D\int_{K} e^{-\langle b,x\rangle}e^{t\|b\|_2R_0}\,dx - (1-t)^D\int_{K} e^{-\langle b,x\rangle}e^{-t\|b\|_2R_0}\,dx\notag\\
 = &~ (1+t)^De^{t\|b\|_2R_0}Z-(1-t)^De^{-t\|b\|_2R_0}Z\notag\\
 \leq &~ e^{(D+\|b\|_2R_0)t}Z - e^{-2(D+\|b\|_2R_0)t}Z\notag\\
 \leq &~ 3e(D+\|b\|_2R_0)tZ,
\end{align}
Here \eqref{eq:parallel-sandwich} reduces the integral to the two dilates of $K$. The substitutions $y=x_c+(1\pm t)(x-x_c)$ contribute Jacobians $(1\pm t)^D$, and $|\langle b,x-x_c\rangle|\le\|b\|_2R_0$ bounds the change in density. For the final estimates, use $\log(1+t)\le t$, $\log(1-t)\ge-2t$, and $e^u-e^{-2u}\le3eu$ for $u=(D+\|b\|_2R_0)t\le1$.

Similarly, the second term in~\eqref{eq:quadrature_decomposition} can be bounded by
\begin{align}\label{eq:integral_U}
    \int_{U_\Lambda}e^{-\langle b,y\rangle}\,dy\leq &~ \int_{K_t^+} e^{-\langle b,y\rangle}\,dy\leq (1+t)^D\int_{K}e^{-\langle b,x\rangle} e^{t\|b\|_2R_0}\,dx\leq e Z.
\end{align}

Combining them together, we get that
\begin{align*}
 |Z_\Lambda-Z|\leq \|f_\Lambda-f_b\|_{1}
 \leq &~ (e^{\|b\|_2\varrho_\Lambda}-1) e Z + 3e (D+\|b\|_2R_0)tZ\\
 \leq &~ 3e\bigl(\|b\|_2\varrho_\Lambda+(D+\|b\|_2R_0)t\bigr)Z\\*
 \leq &~ \varepsilon_{\rm quad}Z,
\end{align*}
This follows by substituting \eqref{eq:integral_U_K_diff} and \eqref{eq:integral_U} into \eqref{eq:quadrature_decomposition}, using $e^u-1\le eu$ for $u=\|b\|_2\varrho_\Lambda\le1$, and taking $C\ge3e$ in \eqref{eq:quadrature-error}.

Finally, since $f_\Lambda/Z_\Lambda$ is the density of $X_\Lambda+E$, it holds that
\begin{align*}
 2\dtv(\widetilde\pi_\Lambda,\pi)
 = \left\|\frac{f_\Lambda}{Z_\Lambda}-\frac{f_b}{Z}\right\|_1\leq \left\|\frac{f_\Lambda}{Z_\Lambda}-\frac{f_\Lambda}{Z}\right\|_1 + \left\|\frac{f_\Lambda}{Z}-\frac{f_b}{Z}\right\|_1
 \leq \frac{|Z-Z_\Lambda|}{Z}+\frac{\|f_\Lambda-f_b\|_1}{Z}
 \le2\varepsilon_{\rm quad}.
\end{align*}

\end{proof}

\paragraph{Affine preconditioning.} The mixing argument will use a positive-definite matrix $A$ as a preconditioner for the target distribution. Let $W_A$ be a nonsingular matrix and $T_A=W_A^{-1}$. Let $X\sim \pi$ and $\pi_A$ be the law of $W_A X$. Then, under the change of variable $y=W_Ax$, we have
\begin{align*}
    x=T_Ay,\qquad dx=|\det(T_A)|\,dy=\frac{dy}{|\det(W_A)|}, \qquad \langle b, x\rangle = \langle T_A^\top b, y\rangle.
\end{align*}
Thus,
\begin{align*}
    \pi_A(dy)=\frac{e^{-\langle T_A^\top b, y\rangle}}{|\det(W_A)|Z}{\bf 1}_{W_AK}(y)=\frac{e^{-\langle b_A, y\rangle}}{Z_A}{\bf 1}_{K_A}(y),
\end{align*}
where we define $K_A=W_AK$, $b_A=T_A^\top b$, and $Z_A=|\det(W_A)|Z$. Hence, the transformed distribution is still a log-affine distribution over a convex body $K_A$.

The lattice is also transformed from $h\Z^D$ to $hW_A\Z^D$. The new lattice normalization constant is
\begin{align*}
    Z_{A,\Lambda}=\vol(W_AQ_\Lambda)\sum_{y\in W_A\Omega_\Lambda}e^{-\langle b_A,y\rangle}=h^D|\det(W_A)|\sum_{x\in \Omega_\Lambda}e^{-\langle b,x\rangle}=|\det(W_A)|Z_\Lambda.
\end{align*}
Therefore, the affine transformation will not incur additional quadrature error:
\begin{align*}
    \left|\frac{Z_{A,\Lambda}}{Z_A}-1\right|=\left|\frac{|\det(W_A)|Z_{\Lambda}}{|\det(W_A)|Z}-1\right|=\left|\frac{Z_{\Lambda}}{Z}-1\right|\leq \varepsilon_{\rm quad}.
\end{align*}
If $E$ is uniform in $Q_\Lambda$ and independent of $X_\Lambda$, then $W_AE$ is uniform in $W_AQ_\Lambda$. Applying the invertible map $W_A$ preserves the total-variation bound in~\cref{lem:full-grid-quadrature}.

Let $\rho_A(x)=\rho_{b_A}(W_Ax)$ be the local radius from \eqref{eq:density-local-radius} for the transformed target. Equivalently, $\rho_A(x)$ is the largest $r$ such that
\begin{align*}
    \frac{63}{64}\leq \Pr_{U\sim B(0,r)}\left[y+U\in K_A~\wedge~e^{-\langle b_A, U\rangle}\geq \frac{3}{4}\right]=\Pr_{U\sim B(0,r)}\left[x+T_AU\in K~\wedge~e^{-\langle b, T_AU\rangle}\geq \frac{3}{4}\right].
\end{align*}
The intrinsic metric becomes
\begin{equation}\label{eq:preconditioned-metric}
 d_A(x,y)=\inf_\gamma\int_0^1\frac{\|W_A\gamma'(t)\|_2}{\rho_A(\gamma(t))}\,dt=d_{\rho_{b_A}}(W_Ax,W_Ay),
\end{equation}
where the infimum is over paths $\gamma$ from $x$ to $y$ in $K$.

The following proposition shows that on a good branch of the coherent covariance estimation where $A\asymp \Cov_\pi(X)$, the transformed target distribution $\pi_A$ is nearly-isotropic.

\begin{proposition}[Affine preconditioning on the fixed grid]\label{prop:finite-affine-branches}
Suppose $\lambda_-\Id\preceq A\preceq\lambda_+\Id$. A rational nonsingular matrix $W_A$ and its rational inverse $T_A=W_A^{-1}$ can be computed so that
\[
 \|W_AA^{1/2}-\Id\|_{\op}\le\frac1{10}.
\]
If $A\asymp\Cov_\pi(X)$, then $W_A\Cov_\pi(X)W_A^\top\asymp\Id$. 

Furthermore, dyadic upper bounds for the radius of the transformed cell $W_AQ_\Lambda$ and the outer radius of $K_A=W_AK$ can be taken to be
\begin{equation}\label{eq:affine-lattice-uniform-bounds}
 \varrho_A=\frac{Ch\sqrt D}{\sqrt{\lambda_-}},\qquad
 R_A=\frac{CR_0}{\sqrt{\lambda_-}},
\end{equation}
\end{proposition}
\begin{proof}
Since $\|W_AA^{1/2}-\Id\|_{\op}\le 1/10$,  we have
\begin{align}\label{eq:WA_TA_op_norm}
 \|W_A\|_{\op}\le \frac{11}{10\sqrt{\sigma_{\min}(A)}}\leq \frac{11}{10\sqrt{\lambda_-}},\qquad
 \|T_A\|_{\op}\le\frac{10\sqrt{\sigma_{\max}(A)}}{9}\leq \frac{10\sqrt{\lambda_+}}9.
\end{align}
Thus,
\begin{align*}
    \|W_A\|_{\op} \varrho_\Lambda\leq \frac{11h\sqrt{D}}{10\sqrt{\lambda_-}}, \qquad \|W_A\|_{\op} R_0\leq \frac{11R_0}{9\sqrt{\lambda_-}},
\end{align*}
which imply~\eqref{eq:affine-lattice-uniform-bounds} by rounding these bounds to dyadic values.

If $c\Cov_\pi(X)\preceq A\preceq C\Cov_\pi(X)$ for some constants $c\leq C$, then we have
\begin{align*}
    C^{-1}\Id\preceq A^{-1/2}\Cov_\pi(X) A^{-1/2}\preceq c^{-1}\Id.
\end{align*}
Using 
\begin{align*}
    W_A\Cov_\pi(X) W_A^\top
 = &~ (W_AA^{1/2})(A^{-1/2}\Cov_\pi(X) A^{-1/2})(W_AA^{1/2})^\top
\end{align*}
and $\|W_A A^{1/2}-\Id\|_{\op}\leq 1/10$, we get that
\begin{align*}
 \frac{(9/10)^2}{C}\Id
 \preceq &~ W_A\Cov_\pi(X) W_A^\top\preceq\frac{(11/10)^2}{c}\Id,
\end{align*}

The rational computation of $W_A$ is given in \cref{app:whitening-arithmetic}, which has polynomial bit complexity in the supplied numerical parameters and input bit length.
\end{proof}

\cref{prop:finite-affine-branches} also implies that on a good covariance estimation branch, \cref{thm:classical-log-affine-input} applies to the transformed distribution. In particular,
\begin{align}\label{eq:preconditioned-radius-bounds}
|\rho_A(x)-\rho_A(y)| = |\rho_{b_A}(W_Ax)-\rho_{b_A}(W_Ay)| \leq C\sqrt D\,\|W_A(x-y)\|_2,
\end{align}
where the second step follows from~\cref{thm:classical-log-affine-input}~\textnormal{(i)}. And for $X\sim \pi$ and $u>0$,
\begin{align}\label{eq:preconditioned-radius-bounds_2}
    \pi\{x:\rho_A(x)<u\}=\Pr_{X\sim \pi}[\rho_{b_A}(W_AX)<u]=\Pr_{Y\sim \pi_A}[\rho_{b_A}(Y)<u]\leq C\sqrt{D}u,
\end{align}
where the third step follows from~\cref{thm:classical-log-affine-input}~\textnormal{(ii)}. 

\subsection{Local overlap of the lattice walk}\label{subsec:finite-proposals}
Fix a good covariance estimate $A$ for this subsection and the conductance proof. We first bound the mass of lattice centers with small local radius. We then specify the walk and prove that transition laws from nearby good centers share probability mass $\Omega(1/J)$ (\cref{lem:finite-overlap}).

Define
\begin{equation}\label{eq:mesh-depth}
 r_{\rm mesh}=C_1D^3\varrho_A,\quad
 \cG_\Lambda=\{z\in\Omega_\Lambda:\rho_A(z)\ge64r_{\rm mesh}\},\quad
 \cB_\Lambda=\Omega_\Lambda\setminus\cG_\Lambda,\quad U_{\cG}=\bigcup_{z\in\cG_\Lambda}(z+Q_\Lambda)
\end{equation}
where $C_1$ is a sufficiently large fixed power of two. 

\begin{lemma}[Good cells and their complement]\label{lem:mesh-collar}
Suppose the hypotheses of \cref{lem:full-grid-quadrature} hold. Every cell centered in $\cG_\Lambda$ lies in $\operatorname{int}K$ and satisfies that for any $x$ in the cell,
\begin{equation}\label{eq:good-cell-depth}
 \rho_A(x)\ge32r_{\rm mesh}.
\end{equation}
Its diameter in $d_A$ is at most
\begin{equation}\label{eq:good-cell-intrinsic-diameter}
 C\varrho_A/r_{\rm mesh}.
\end{equation}
For $\varepsilon_{\rm mesh}=C(\sqrt D\,r_{\rm mesh}+\varepsilon_{\rm quad})$, we have
\begin{equation}\label{eq:bad-grid-mass}
 \pi_\Lambda(\cB_\Lambda)\le\varepsilon_{\rm mesh},\qquad
 \pi(K\setminus U_{\cG})\le\varepsilon_{\rm mesh}.
\end{equation}
\end{lemma}
\begin{proof}
If $x$ lies in the cell centered at $z\in\cG_\Lambda$, then $\|W_A(x-z)\|_2\le\varrho_A$. By~\eqref{eq:preconditioned-radius-bounds}, we have
\begin{align*}
 \rho_A(x)
 \geq  \rho_A(z)-C\sqrt D\,\varrho_A
 \geq 64r_{\rm mesh}-C\sqrt D\,\varrho_A
 \ge32r_{\rm mesh}>0.
\end{align*}
The zero extension vanishes outside $\operatorname{int}K$, so the cell lies inside the body. 

For $x,y$ in that cell, convexity of the cube and \eqref{eq:preconditioned-metric} give
\[
 d_A(x,y)\le\int_0^1\frac{\|W_A(y-x)\|_2}{\rho_A((1-t)x+ty)}\,dt\le\frac{2\varrho_A}{32r_{\rm mesh}},
\]
which proves the diameter bound~\eqref{eq:good-cell-intrinsic-diameter}.

Let $Z\sim\pi_\Lambda$ and let $E$ be an independent uniform cell offset. If $Z\in\cB_\Lambda$, then \eqref{eq:preconditioned-radius-bounds} gives
\begin{align*}
    \rho_A(Z+E)\leq \rho_A(Z)+C\sqrt{D}\varrho_A<64r_{\rm mesh} + C\sqrt{D}\varrho_A\leq Cr_{\rm mesh}
\end{align*}
for a sufficiently large constant $C$. Thus,
\begin{align*}
 \pi_\Lambda(\cB_\Lambda)
 \leq &~ \Pr_{Z\sim \pi_\Lambda,E\sim Q_\Lambda}[\rho_A(Z+E)<Cr_{\rm mesh}]\\*
 \leq &~ \pi\{x:\rho_A(x)<Cr_{\rm mesh}\}+\varepsilon_{\rm quad}\\*
 \leq &~ C\sqrt D\,r_{\rm mesh}+\varepsilon_{\rm quad}\leq \varepsilon_{\rm mesh},
\end{align*}
where the second step follows from~\eqref{eq:finite-jitter-tv}, and the third follows from~\eqref{eq:preconditioned-radius-bounds_2}. We also have
\begin{align*}
 \pi(K\setminus U_{\cG})
 \leq \Pr[Z+E\notin U_{\cG}]+\varepsilon_{\rm quad}
 = \pi_\Lambda(\cB_\Lambda)+\varepsilon_{\rm quad},
\end{align*}
where the equality follows because the cells partition space up to null boundaries. Increasing the constant in $\varepsilon_{\rm mesh}$ proves~\eqref{eq:bad-grid-mass}.
\end{proof}

\paragraph{The lattice Metropolis walk.}
We first specify the local proposals in~\cref{alg:finite-metropolis}. Let $\mathcal R_h(u)=h\left\lfloor h^{-1}u+\tfrac12\one\right\rfloor$ be the coordinatewise rounding to the discretization lattice. Define the ideal displacement law
\begin{equation}\label{eq:lattice-rounded-ball}
 q_{A,\delta}(v)=\Pr_{U_\star \sim B(0,\delta)}[\mathcal R_h(T_AU_\star)=v],\qquad v\in\Lambda.
\end{equation}
This law is symmetric: $T_AU_\star$ has a centrally symmetric continuous law and rounding ties have probability zero. The implemented finite sampler in \cref{lem:finite-direction-preparation} produces a rational $U$ and $v=\mathcal R_h(T_AU)$. Its law $q^0_{A,\delta}$ approximates $q_{A,\delta}$ within $\varepsilon_{\rm dir}$ in total variation. The independent sign in \cref{alg:finite-metropolis} gives the exactly symmetric law
\begin{equation}\label{eq:symmetric-proposal}
 \widehat q_{A,\delta}(v)=\tfrac12\bigl(q^0_{A,\delta}(v)+q^0_{A,\delta}(-v)\bigr),\qquad
 \dtv(\widehat q_{A,\delta},q_{A,\delta})\le\varepsilon_{\rm dir}.
\end{equation}
Only this error guarantee is needed for our mixing analysis in this section, and the implemented sampler's precision analysis is deferred to \cref{app:proposal-arithmetic}.

The radius $\delta$ is sampled from the following dyadic list
\begin{equation}\label{eq:finite-scales}
 \delta_j=2^j(4r_{\rm mesh}),\quad 0\le j<J,\qquad \delta_{J-1}\ge C\sqrt D\,R_A,
\end{equation}
where $J\ge1$ is the smallest integer satisfying the last inequality. Since $\rho_A(z)\le C\sqrt D\,R_A$ by~\eqref{eq:preconditioned-radius-bounds}, this list covers the local radii of all good centers. Its length satisfies
\begin{equation}\label{eq:number-digital-scales}
 J=O\!\left(1+\log\left(2+\frac{\sqrt D\,R_A}{r_{\rm mesh}}\right)\right).
\end{equation}
To select a scale $\delta_j$ using a fixed number of unbiased random bits, let $q=\lceil\log_2J\rceil$, draw $s$ uniformly from $\{0,\ldots,2^q-1\}$, and choose $\delta_j$ with $j=s\bmod J$. Since $J\le2^q<2J$, each scale has either one or two seed values mapped to it. Its probability therefore satisfies
\[
 \frac1{2J}\le\frac1{2^q}\le w_j\le\frac2{2^q}\le\frac2J.
\]
This process can be implemented coherently using reversible gates:
\begin{align*}
    2^{-q/2}\sum_{s=0}^{2^q-1}\ket{s}_{\sf R}\ket{0}_{\sf J}\longmapsto 2^{-q/2}\sum_{s=0}^{2^q-1}\ket{s}_{\sf R}\ket{s \bmod J}_{\sf J}
\end{align*}

Then, we specify the connector proposal in~\cref{alg:finite-metropolis}. Choose a power of two $M_{\rm out}>2R_0/h$ and sample the coefficients of $v$ uniformly from $\{-M_{\rm out},\ldots,M_{\rm out}-1\}^D$. Every feasible difference is included. Append the same independent sign and use Metropolis acceptance. 

If $\sP_j$ and $\sP_{\rm conn}$ denote these Metropolis kernels, the implemented walk is
\begin{equation}\label{eq:finite-chain-mixture}
 \sP_{\Lambda,A}=\tfrac12\Id+\tfrac38\sum_{j=0}^{J-1}w_j\sP_j+\tfrac18\sP_{\rm conn}.
\end{equation}

The following lemma shows that transition laws from good lattice points at intrinsic distance less than $c/\sqrt D$ share probability mass $\Omega(1/J)$. This overlap will give a lower bound on conductance.

\begin{lemma}[Local overlap]\label{lem:finite-overlap}
Suppose $\varepsilon_{\rm dir}\le1/100$ and
\begin{equation}\label{eq:mesh-radius-assumptions}
 \|b\|_2\cdot \varrho_\Lambda\le\frac12\log(4/3).
\end{equation}
There are constants $c,\alpha_0>0$ such that, for $x,y\in\cG_\Lambda$,
\begin{equation}\label{eq:finite-overlap}
 d_A(x,y)<c/\sqrt D
 \quad\Longrightarrow\quad
 \dtv(\sP_{\Lambda,A}(x,\cdot),\sP_{\Lambda,A}(y,\cdot))\le1-\alpha_0/J.
\end{equation}
\end{lemma}
\begin{proof}
To bound the total-variation distance between the one-step transition law from $x$ and $y$, we first have
\begin{align*}
    1-\dtv(\sP_{\Lambda,A}(x,\cdot),\sP_{\Lambda,A}(y,\cdot))=\sum_{u\in \Lambda} \min\{{\cal P}_{\Lambda, A}(x,u),{\cal P}_{\Lambda, A}(y,u)\}.
\end{align*}
For any $0\leq j<J$, let ${\cal P}_j(z,\cdot)$ denote the transition law from $z$ conditioned on the local proposal with scale $\delta_j$ in~\cref{alg:finite-metropolis}. The probability of this conditioning event is 
\begin{align*}
    \frac{3}{8}w_j\geq \frac{3}{16J}.
\end{align*}
Since the other branches of the update have nonnegative transition probabilities, we have
\begin{align*}
    1-\dtv(\sP_{\Lambda,A}(x,\cdot),\sP_{\Lambda,A}(y,\cdot)) \geq &~  \frac{3}{16J}\sum_{u\in \Lambda} \min\{{\cal P}_{j}(x,u),{\cal P}_{j}(y,u)\}\\
    = &~ \frac{3}{16J} (1-\dtv(\sP_{j}(x,\cdot),\sP_{j}(y,\cdot))).
\end{align*}

It therefore suffices to choose a scale $\delta_j$ at which ${\cal P}_j(x,\cdot)$ and ${\cal P}_j(y,\cdot)$ share a constant amount of probability mass. Applying \cref{thm:classical-log-affine-input}~\textnormal{(iii)} in the transformed coordinates gives, for sufficiently small $c$,
\[
 \max\{\rho_A(x),\rho_A(y)\}\le2\min\{\rho_A(x),\rho_A(y)\},\qquad
 \|W_A(x-y)\|_2\le Cc\rho_A(x)/\sqrt D.
\]
The list \eqref{eq:finite-scales} therefore contains a radius $\delta=\delta_j$ such that
\[
 \frac{\max\{\rho_A(x),\rho_A(y)\}}{32}\le\delta\le\frac{\min\{\rho_A(x),\rho_A(y)\}}8.
\]
Indeed, the interval has ratio at least two, its upper endpoint is at least $8r_{\rm mesh}$, and the list begins at $4r_{\rm mesh}$. Therefore, for this $\delta$, it holds that
\begin{align}\label{eq:rho_x_y_WA_bound}
 8\delta\le\rho_A(x),\rho_A(y)\le32\delta,\qquad
 \|W_A(x-y)\|_2\le Cc\delta/\sqrt D.
\end{align}

Let $z\in\{x,y\}$ be any of the starting points, and we consider the next-step move. Let $U_\star$ be uniform in $B(0,\delta)$, $v=\mathcal R_h(T_AU_\star)$, and $e=v-T_AU_\star$. Rounding gives
\[
 \|e\|_2\le\varrho_\Lambda,\qquad \|W_Ae\|_2\le\varrho_A.
\]
Since $z\in {\cal G}_\Lambda$, the boundary-distance estimate in \cref{thm:classical-log-affine-input}~\textnormal{(i)} gives
\[
 \operatorname{dist}(W_Az,\partial K_A)\ge\frac{\rho_A(z)}{C\sqrt D}\ge\frac{64r_{\rm mesh}}{C\sqrt D}\ge\frac{C_1D^{5/2}}{C}\varrho_A>2\varrho_A.
\]
Thus, $z+2e\in K$. Since $2U_\star$ is uniform in $B(0,2\delta)$ and $2\delta<\rho_A(z)$, the definition of $\rho_A$ implies that
\[
 \Pr\left[z+2T_AU_\star\in K,\ e^{-\langle b,2T_AU_\star\rangle}\ge\frac34\right]\ge\frac{63}{64}.
\]
Conditioned on this event, by the convexity of $K$, we have
\begin{equation}\label{eq:convex-rounding-stability}
 z+v=z+T_AU_\star+e=\tfrac12(z+2T_AU_\star)+\tfrac12(z+2e)\in K.
\end{equation}
Moreover, the density satisfies
\begin{align*}
 \frac{f_b(z+v)}{f_b(z)}
 = e^{-\langle b,T_AU_\star\rangle} e^{-\langle b,e\rangle}
 \geq \sqrt{\frac{3}{4}}\,e^{-\|b\|_2\varrho_\Lambda}
 \ge\frac34,
\end{align*}
where the last step follows from~\eqref{eq:mesh-radius-assumptions}. Hence, we get that
\begin{align*}
    \Pr_{v\sim q_{A,\delta}}\left[z+v\notin K~\vee~f_b(z+v)<\frac{3}{4}f_b(z)\right]\leq \frac{1}{64}.
\end{align*}
By the total-variation distance bound~\eqref{eq:symmetric-proposal}, we have
\begin{align}\label{eq:good_delta_bad_destin_prob}
 \Pr_{v\sim\widehat q_{A,\delta}}\left[z+v\notin K~\vee~f_b(z+v)<\frac{3}{4}f_b(z)\right]
 \le\frac1{64}+\varepsilon_{\rm dir}\leq\frac18.
\end{align}
Let $q_x,q_y$ denote the proposal destination laws at the scale $\delta$ before Metropolis acceptance, and define the set of good destinations:
\[
 G_z:=\{u\in\Omega_\Lambda:f_b(u)\ge3f_b(z)/4\},
 \qquad z\in\{x,y\}.
\]
Then, \eqref{eq:good_delta_bad_destin_prob} implies that
\begin{align*}
    q_x(G_x^c)\leq \frac{1}{8},\qquad q_y(G_y^c)\leq \frac{1}{8}.
\end{align*}
Note that for any $u\in G_z$, its Metropolis acceptance probability is at least $3/4$. Thus, we have
\begin{align*}
    \sum_{u\in \Lambda} \min\{{\cal P}_{j}(x,u),{\cal P}_{j}(y,u)\} \geq \sum_{u\in G_x\cap G_y} \frac{3}{4}\min\{q_x(u),q_y(u)\}
\end{align*}
Without loss of generality, we may assume that $f_b(x)\geq f_b(y)$. Then, $G_x\subseteq G_y$, and the above equation becomes
\begin{align*}
    \sum_{u\in \Lambda} \min\{{\cal P}_{j}(x,u),{\cal P}_{j}(y,u)\} \geq &~ \sum_{u\in G_x} \frac{3}{4}\min\{q_x(u),q_y(u)\}\\
    = &~ \frac{3}{4}\left(\sum_{u\in \Lambda} \min\{q_x(u),q_y(u)\} - \sum_{u\in G_x^c} \min\{q_x(u),q_y(u)\}\right)\\
    \geq &~ \frac{3}{4}\left(\sum_{u\in \Lambda} \min\{q_x(u),q_y(u)\} - q_x(G_x^c)\right)\\
    = &~ \frac{3}{4}\left(1-\dtv(q_x,q_y) - q_x(G_x^c)\right)\\
    \geq &~ \frac{3}{4}\left(\frac{7}{8}-\dtv(q_x,q_y)\right).
\end{align*}

For $z\in \{x,y\}$, let $q_z^\star$ denote the law of $z+v$ for $v\sim q_{A,\delta}$. We first bound $\dtv(q_x^\star, q_y^\star)$. Since $z\in h\Z^D$ and $T_AW_A=\Id$, we can swap rounding and translation:
\[
 \mathcal R_h(T_A(W_Az+U_\star))
 =\mathcal R_h(z+T_AU_\star)
 =z+\mathcal R_h(T_AU_\star).
\]
Thus, $q_z^\star$ is the image of the uniform law on $B(W_Az,\delta)$ under the map $\mathcal R_h\circ T_A$. By the data-processing inequality,
\begin{align*}
    \dtv(q_x^\star, q_y^\star)= &~ \dtv(({\cal R}_h\circ T_A) \operatorname{Unif}(B(W_Ax,\delta)), ({\cal R}_h \circ T_A) \operatorname{Unif}(B(W_Ay,\delta)))\\
    \leq &~ \dtv(\operatorname{Unif}(B(W_Ax,\delta)),\operatorname{Unif}(B(W_Ay,\delta)))\\
    = &~ \frac{\vol(B(W_Ax,\delta)\setminus B(W_Ay,\delta))}{v_D\delta^D},
\end{align*}
where $v_D$ denotes the volume of the unit ball in $\R^D$. Slice the two balls along lines parallel to the direction of $W_A(x-y)$. Each pair of slices consists of equal-length intervals shifted by $\|W_A(x-y)\|_2$, so the length of the non-overlap part of each interval is at most $\|W_A(x-y)\|_2$. The perpendicular projection is a $(D-1)$-dimensional ball of radius $\delta$. Hence,
\[
 \vol(B(W_Ax,\delta)\setminus B(W_Ay,\delta))
 \le v_{D-1}\delta^{D-1}\|W_A(x-y)\|_2.
\]
Thus,
\begin{align*}
    \dtv(q_x^\star, q_y^\star)\leq  \frac{v_{D-1}}{v_D}\frac{\|W_A(x-y)\|_2}{\delta}
    \le  O(\sqrt D)\,\frac{\|W_A(x-y)\|_2}{\delta}
    \leq  O(c),
\end{align*}
where the last step follows from~\eqref{eq:rho_x_y_WA_bound}. Then, by~\eqref{eq:symmetric-proposal} and triangle inequality, we get that
\begin{align*}
    \dtv(q_x,q_y)
 \leq &~ \dtv(q_x,q_x^\star)+\dtv(q_x^\star,q_y^\star)+\dtv(q_y^\star,q_y)\\
 \leq &~ 2\varepsilon_{\rm dir}+O(c)\leq \frac{1}{50}+O(c)\leq \frac{1}{2},
\end{align*}
after choosing $c$ sufficiently small in the assumption~\eqref{eq:finite-overlap}.

Therefore, we conclude that 
\begin{align*}
    1-\dtv(\sP_{\Lambda,A}(x,\cdot),\sP_{\Lambda,A}(y,\cdot)) \geq \frac{3}{16J} \cdot \frac{3}{4}\left(\frac{7}{8}-\dtv(q_x,q_y)\right)\geq \frac{27}{512J}.
\end{align*}
\end{proof}

\subsection{The mixing time bound}\label{subsec:lattice-conductance}
In this subsection, we prove a warm-start mixing time bound for the lattice Metropolis walk (\cref{thm:finite-mixing}). To lower bound the $s$-conductance, we adapt the continuous space mixing proof for the hit-and-run walk in~\cite{LovaszVempala06Corner,Zhang26HitRun} to the lattice setting by considering the union of the cells centered at a subset of lattice points. For $S\subseteq\cG_\Lambda$, write $U(S)=\bigcup_{z\in S}(z+Q_\Lambda)$.

\begin{lemma}[Separator on the cubic grid]\label{lem:cell-separator}
Suppose
\begin{equation}\label{eq:finite-error-budget}
 \varepsilon_{\rm mesh}\le c\min\left\{s,\frac{s\psi_D}{D(1+\log(4/s))}\right\}.
\end{equation}
If $S_-,S_+\subseteq\cG_\Lambda$ satisfy $d_A(S_-,S_+)\ge c_0/\sqrt D$ and $\pi_\Lambda(S_-),\pi_\Lambda(S_+)\ge s/2$ for a fixed $c_0>0$, then
\begin{equation}\label{eq:cell-separator}
 \pi_\Lambda(\cG_\Lambda\setminus(S_-\cup S_+))
 \ge \gamma_s\min\{\pi_\Lambda(S_-),\pi_\Lambda(S_+)\},\qquad \gamma_s=c\frac{\psi_D}{D(1+\log(4/s))}.
\end{equation}
The constant in the mesh requirement is chosen sufficiently small depending on $c_0$ and the fixed covariance-comparison constants.
\end{lemma}
\begin{proof}
We will apply~\cref{thm:classical-log-affine-input}~\textnormal{(iv)} to show that the separator between $U(S_-)$ and $U(S_+)$ has large $\pi$-mass, which will imply~\eqref{eq:cell-separator}. 

We first need to check that these two subsets are well-separated in $d_{A}$-distance. By~\cref{lem:mesh-collar}, the good cells lie inside $K$, and their $d_A$-diameter is at most 
\begin{align*}
    C\frac{\varrho_A}{r_{\rm mesh}}\leq C\frac{\varrho_A}{C_1D^3\varrho_A}\leq \frac{c_0}{4\sqrt{D}}
\end{align*}
by increasing $C_1$ in~\eqref{eq:mesh-depth}. Thus, for $x\in U(S_-)$ and $y\in U(S_+)$ with cell centers $z_-,z_+$, by the triangle inequality, we have
\[
 d_A(x,y)\ge d_A(z_-,z_+)-d_A(z_-,x)-d_A(z_+,y)\ge\frac{c_0}{2\sqrt D}.
\]
Hence, $d_A(U(S_-),U(S_+))\geq \frac{c_0}{2\sqrt{D}}$.

For a good center $z\in {\cal G}_\Lambda$, the Euclidean radius of the cell is upper bounded by $\varrho_\Lambda$. Thus, for any $x\in z+Q_\Lambda$, $e^{-\|b\|_2\varrho_\Lambda}\leq \frac{\pi(x)}{\pi(z)}\leq e^{\|b\|_2\varrho_\Lambda}$, and we have
\[
 e^{-\|b\|_2\varrho_\Lambda}\frac{Z_\Lambda}{Z}\pi_\Lambda(z)
 \le\pi(z+Q_\Lambda)
 \le e^{\|b\|_2\varrho_\Lambda}\frac{Z_\Lambda}{Z}\pi_\Lambda(z).
\]
The quadrature bound in~\cref{lem:full-grid-quadrature} makes these multipliers lie in $[1/2,2]$. Summing over any $S\subseteq\cG_\Lambda$ gives
\begin{equation}\label{eq:cell-mass-comparison}
 \frac12\pi_\Lambda(S)\le\pi(U(S))\le2\pi_\Lambda(S).
\end{equation}
Let $p=\min\{\pi_\Lambda(S_-),\pi_\Lambda(S_+)\}$ and $p_c=\min\{\pi(U(S_-)),\pi(U(S_+))\}$. Then, \eqref{eq:cell-mass-comparison} implies that
\begin{align*}
    p_c\ge p/2\ge s/4.
\end{align*}
By~\cref{thm:classical-log-affine-input}~\textnormal{(iv)}, we get that
\begin{align}\label{eq:mass_of_separator}
 \pi(K\setminus(U(S_-)\cup U(S_+)))
 \geq c\frac{\psi_Dp_c}{D(1+\log(1/p_c))}
 \geq c'\frac{\psi_Dp}{D(1+\log(4/s))}.
\end{align}

Therefore, 
\begin{align*}
 \pi_\Lambda(\cG_\Lambda\setminus(S_-\cup S_+))
 \geq &~ \frac{1}{2}\pi(U(\cG_\Lambda\setminus(S_-\cup S_+)))\\
 = &~ \frac{1}{2}\left(\pi(K\setminus(U(S_-)\cup U(S_+)))-\pi(K\setminus U_{\cal G})\right)\\*
 \geq &~ \frac{1}{2}\left(c'\frac{\psi_Dp}{D(1+\log(4/s))}-\varepsilon_{\rm mesh}\right)\\*
 \geq &~ \frac{c'}{4}\frac{\psi_Dp}{D(1+\log(4/s))}=:\gamma_s p,
\end{align*}
We use \eqref{eq:cell-mass-comparison} to compare lattice mass with cell mass, then apply the separator bound \eqref{eq:mass_of_separator} and subtract the exceptional mass bounded in \eqref{eq:bad-grid-mass}. The error budget \eqref{eq:finite-error-budget}, together with $p\ge s/2$, absorbs this loss.

\end{proof}

\begin{theorem}[Conductance and warm-start mixing]\label{thm:finite-mixing}
Let $0<s<1/4$. Suppose $A\asymp\Cov_\pi(X)$, $\varepsilon_{\rm dir}\le1/100$, and the lattice satisfies~\eqref{eq:quadrature-error}, \eqref{eq:mesh-radius-assumptions}, and \eqref{eq:finite-error-budget}. Then
\begin{equation}\label{eq:finite-conductance}
 \Phi_s(\sP_{\Lambda,A})\ge c\frac{\psi_D}{JD(1+\log(4/s))}.
\end{equation}
For $M\ge1$ and $0<\eta<1/2$, if these conditions hold at $s=\eta/(2M)$, every $M$-warm start $\nu$ satisfies
\begin{equation}\label{eq:finite-mixing}
 \tau_{\mix}(\eta,\nu,\pi_\Lambda;\sP_{\Lambda,A})
 \le C\frac{J^2D^2}{\psi_D^2}(1+\log(8M/\eta))^2\log(2M/\eta).
\end{equation}
\end{theorem}
\begin{proof}
Write $\sP=\sP_{\Lambda,A}$, and let $c_0$ and $\alpha_0$ be the distance and overlap constants in \cref{lem:finite-overlap}. Let $\gamma_s$ be as in~\eqref{eq:cell-separator}, with its absolute constant chosen so that $\gamma_s\le1$. Choose the mesh constant so that $\varepsilon_{\rm mesh}\le\gamma_s s/32$.

Fix a cut $S$ with $p:=\pi_\Lambda(S)\in (s,1/2]$ and let $Q=\int_S\sP(x,S^c)\,d\pi_\Lambda(x)$ be the ergodic flow. Define its low-escape subsets
\[
 S_-:=\{x\in S\cap\cG_\Lambda:\sP(x,S^c)<\alpha_0/(4J)\},\qquad
 S_+:=\{y\in S^c\cap\cG_\Lambda:\sP(y,S)<\alpha_0/(4J)\}.
\]
Then, for any $x\in S_-$ and $y\in S_+$,
\[
 \dtv(\sP(x,\cdot),\sP(y,\cdot))\ge\sP(x,S)-\sP(y,S)>1-\frac{\alpha_0}{2J}.
\]
Thus, the contrapositive of~\cref{lem:finite-overlap} implies that
\begin{align*}
    d_A(S_-,S_+)\ge \frac{c_0}{\sqrt{D}}.
\end{align*}
Markov's inequality and reversibility also give
\begin{equation}\label{eq:abstract-low-escape-mass}
 \pi_\Lambda(S\setminus S_-),\ \pi_\Lambda(S^c\setminus S_+)
 \le\varepsilon_{\rm mesh}+\frac{4JQ}{\alpha_0}.
\end{equation}
Indeed, we have
\begin{align*}
    Q\geq \sum_{x\in (S\cap {\cal G}_\Lambda)\setminus S_-} {\cal P}(x,S^c) \pi_\Lambda(x)\geq \frac{\alpha_0}{4J}  \sum_{x\in (S\cap {\cal G}_\Lambda)\setminus S_-} \pi_\Lambda(x)=\frac{\alpha_0}{4J}\pi_\Lambda((S\cap {\cal G}_\Lambda)\setminus S_-),
\end{align*}
which implies that
\begin{align*}
    \pi_\Lambda(S\setminus S_-)= \pi_\Lambda(S\setminus {\cal G}_\Lambda) + \pi_\Lambda((S\cap {\cal G}_\Lambda)\setminus S_-) \leq \pi_\Lambda({\cal B}_\Lambda) + \frac{4JQ}{\alpha_0}\leq \varepsilon_{\rm mesh}+\frac{4JQ}{\alpha_0}.
\end{align*}
The other inequality follows from the reversibility that the flow out of $S^c$ is also equal to $Q$.

Suppose the flow is small:
\begin{align*}
    Q<\frac{\alpha_0 p}{32J} \gamma_s,
\end{align*}
where $\gamma_s$ is defined in~\eqref{eq:cell-separator}. Then by~\eqref{eq:abstract-low-escape-mass}, we have
\begin{align*}
    \pi_\Lambda(S_-)= \pi_\Lambda(S)-\pi_\Lambda(S\setminus S_-)\geq &~p-\varepsilon_{\rm mesh} - \frac{4JQ}{\alpha_0}\\
    > &~p-\varepsilon_{\rm mesh}-\frac{p \gamma_s}{8} \\
    \geq &~ p-\frac{5\gamma_s}{32}p\geq \frac{27}{32}p>s/2,
\end{align*}
where the third line uses $\varepsilon_{\rm mesh}\le\gamma_s p/32$ and $\gamma_s\leq 1$. Similarly, using $\pi_\Lambda(S^c)\ge p$, we also have
\begin{align*}
    \pi_\Lambda(S_+)\geq \frac{27}{32}p>s/2.
\end{align*}
Then, by~\cref{lem:cell-separator}, we get that
\begin{align*}
 \gamma_s\frac{27p}{32}
 \leq \pi_\Lambda(\cG_\Lambda\setminus(S_-\cup S_+))
 \leq 2\varepsilon_{\rm mesh}+\frac{8JQ}{\alpha_0}
 < \frac{\gamma_s p}{16}+\frac{\gamma_s p}{4}
 =\frac{5\gamma_s p}{16},
\end{align*}
where the second step uses \eqref{eq:abstract-low-escape-mass}. This is a contradiction, and therefore, 
\begin{align*}
    Q\geq \frac{\alpha_0 p}{32J}\gamma_s,
\end{align*}
which implies that
\begin{align*}
    \frac{Q}{p-s}\geq \frac{\alpha_0 p}{32 J(p-s)}\gamma_s\geq \frac{\alpha_0}{32 J}\gamma_s.
\end{align*}
Taking the infimum over $S$ proves the $s$-conductance lower bound~\eqref{eq:finite-conductance}.

The mixing time bound~\eqref{eq:finite-mixing} immediately follows from~\cref{lem:warm-mixing-from-conductance} with $s=\eta/(2M)$. 
\end{proof}
\subsection{Covariance estimates and the common mesh}\label{subsec:common-mesh}
We recall empirical covariance concentration for continuous log-concave distributions~\cite{AdamczakEtAl10,LovaszVempala07} and transfer it to lattice samples in~\cref{lem:finite-block-transfer}. We also derive upper and lower spectral bounds for the continuous covariance from the supplied geometric parameters (\cref{lem:log-affine-covariance-floor}). These bounds determine the range of preconditioners used to implement the lattice walk.

\begin{lemma}[Empirical covariance concentration~{\cite{AdamczakEtAl10,LovaszVempala07}}]\label{lem:covariance-confidence}
Let $X_1,\ldots,X_L$ be independent samples from a log-concave distribution on $\R^D$ with covariance $\Sigma$. Define
\[
 \overline X=\frac1L\sum_{\ell=1}^L X_\ell,
 \qquad
 \widehat\Sigma=\frac1L\sum_{\ell=1}^L
 (X_\ell-\overline X)(X_\ell-\overline X)^\top.
\]
There is a universal constant $C$ such that, for every $0<{\varepsilon_{\Cov}}<1/2$, if
\[
 L\ge CD\left(1+\log\frac2{\varepsilon_{\Cov}}\right)^2,
\]
then
\[
\Pr\left[\frac12\Sigma\preceq\widehat\Sigma\preceq2\Sigma\right]\ge1-{\varepsilon_{\Cov}}.
\]
\end{lemma}

\begin{proof}
We first make the log-concave distribution isotropic by defining
\begin{align*}
    Z_\ell = \Sigma^{-1/2}(X_\ell-\E[X_\ell])\quad \forall \ell\in [L],\qquad 
    S_L=\frac1L\sum_{\ell=1}^L Z_\ell Z_\ell^\top,
 \qquad
 \overline Z=\frac1L\sum_{\ell=1}^L Z_\ell.
\end{align*}
By direct calculation, we have 
\[
 \Sigma^{-1/2}\widehat\Sigma\Sigma^{-1/2}
 =S_L-\overline Z\,\overline Z^\top.
\]
We bound the two terms separately.

By \cite[Theorem~4.1]{AdamczakEtAl10} with $t=1$ and $\varepsilon=1/4$, there are universal constants $C_0,c_0>0$ such that, for independent isotropic log-concave vectors $V_1,\ldots,V_L$ in $\R^r$, if $L\ge C_0r$, then
\begin{align}\label{eq:isotropic_cov_concentration}
 \Pr\left[
 \left\|\frac1L\sum_{\ell=1}^L V_\ell V_\ell^\top-\Id_r\right\|_{\op}
 >\frac14
 \right]
 \le e^{-c_0\sqrt r}.
\end{align}
To obtain failure probability ${\varepsilon_{\Cov}}/2$, set
\[
 r
 =\max\left\{D,\left\lceil c_0^{-2}\log^2(2/{\varepsilon_{\Cov}})\right\rceil\right\}\leq C_1D(1+\log(2/{\varepsilon_{\Cov}}))^2.
\]
Append $r-D$ independent standard Gaussian coordinates to each $Z_\ell$. The resulting vectors are isotropic and log-concave, and their empirical second-moment matrix has the block structure $\begin{bmatrix}
        S_L & *\\
        * & *
    \end{bmatrix}$.
Then, we have that $L\ge C_0r$ after increasing $C$, and
\begin{align*}
    \|S_L-\Id_D\|_{\rm op} \leq \left\|\begin{bmatrix}
        S_L & *\\
        * & *
    \end{bmatrix} -\Id_r\right\|_{\rm op}.
\end{align*}
Therefore, by~\eqref{eq:isotropic_cov_concentration}, we obtain that
\[
 \Pr\left[\|S_L-\Id_D\|_{\op}>\frac14\right]
 \le\frac{\varepsilon_{\Cov}}2.
\]

For the second term, we know that $\overline Z$ is log-concave since it is the sum of independent log-concave random vectors. Since $Z_1,\dots,Z_L$ are isotropic and independent, we have 
\[
 \E\left[\left\|\overline Z\right\|_2^2\right]=\frac DL.
\]
By~\cite[Lemma~5.17]{LovaszVempala07}, we have
\[
 \Pr\left[\|\overline Z\|_2>r\sqrt{D/L}\right]\le e^{1-r}
 \qquad\forall r>1.
\]
Since $L\ge4D(1+\log(2/{\varepsilon_{\Cov}}))^2$, taking $r=\frac12\sqrt{L/D}>1$ gives
\begin{align*}
    \Pr\left[\left\|\overline Z\right\|_2>\frac12\right]\leq \exp\left(1-\frac12\sqrt{\frac LD}\right)\leq \frac{\varepsilon_{\Cov}}2.
\end{align*}

By a union bound, both estimates hold with probability at least $1-{\varepsilon_{\Cov}}$. On this event, we have
\[
 \frac12\Id_D
 \preceq S_L-\overline Z\,\overline Z^\top
 \preceq S_L
 \preceq\frac54\Id_D.
\]
Therefore,
\begin{align*}
    \frac{1}{2}\Sigma\preceq \widehat{\Sigma}\preceq \frac{5}{4}\Sigma\preceq 2\Sigma.
\end{align*}
\end{proof}

\begin{lemma}[Empirical covariance transfer]\label{lem:finite-block-transfer}
Let $\Sigma=\Cov_\pi(X)$ be the covariance of the continuous log-affine target $\pi$. Fix $0<\varepsilon_{\Cov}<1/2$ and $L\geq CD(1+\log(4/\varepsilon_{\Cov}))^2$. Choose the lattice so that
\begin{equation}\label{eq:covariance-transfer-conditions}
 L\varepsilon_{\rm quad}\le{\varepsilon_{\Cov}}/2,\qquad
 \varrho_\Lambda^2\le\frac{\lambda_{\min}(\Sigma)}{64}.
\end{equation}
Let $\widehat{\Sigma}$ be the empirical covariance estimated from $L$ independent samples from $\pi_\Lambda$. Then
\begin{align*}
    \Pr\left[\frac{1}{4}\Sigma \preceq \widehat{\Sigma}\preceq 4\Sigma\right]\geq 1-{\varepsilon_{\Cov}}.
\end{align*}
\end{lemma}
\begin{proof}
Let $Z_1,\ldots,Z_L$ be independent lattice samples and $E_1,\ldots,E_L$ be independent uniform offsets in $Q_\Lambda$. By \cref{lem:full-grid-quadrature} and tensorization,
\[
 \dtv\bigl(\operatorname{Law}((Z_\ell+E_\ell)_{\ell=1}^L),\pi^{\otimes L}\bigr)\le L\varepsilon_{\rm quad}\le{\varepsilon_{\Cov}}/2.
\]
Hence there is a coupling with independent $X_1,\ldots,X_L\sim\pi$ such that $X_\ell=Z_\ell+E_\ell$ for every $\ell$, except with probability ${\varepsilon_{\Cov}}/2$.

On this coupling event, compare empirical standard deviations in an arbitrary direction $u$. The centered observation vectors satisfy that for any $\ell\in [L]$,
\[
 \langle u,X_\ell-\overline X\rangle=\langle u,Z_\ell-\overline Z\rangle+\langle u,E_\ell-\overline E\rangle.
\]
Thus, we have
\begin{align*}
 \left|\sqrt{u^\top\widehat\Sigma(X)u}-\sqrt{u^\top\widehat\Sigma(Z)u}\right| = &~ \left|\sqrt{\frac{1}{L}\sum_{\ell=1}^L (\langle u,Z_\ell-\overline{Z}\rangle+\langle u,E_\ell-\overline{E}\rangle)^2}-\sqrt{\frac{1}{L}\sum_{\ell=1}^L \langle u,Z_\ell-\overline{Z}\rangle^2}\right|\\
 \leq &~ \sqrt{\frac{1}{L}\sum_{\ell=1}^L \langle u,E_\ell-\overline{E}\rangle^2}\\*
 \leq &~ \left(\frac1L\sum_{\ell=1}^L\langle u,E_\ell\rangle^2\right)^{1/2}\\*
 \leq &~ \varrho_\Lambda\|u\|_2\le\frac18\sqrt{u^\top\Sigma u},
\end{align*}
The reverse triangle inequality bounds the change in empirical standard deviation by that of the offsets. Centering reduces their empirical second moment, and $\|E_\ell\|_2\le\varrho_\Lambda$ bounds it by $\varrho_\Lambda^2\|u\|_2^2$. Condition~\eqref{eq:covariance-transfer-conditions} gives the final bound.

By~\cref{lem:covariance-confidence} and the choice of $L$, we have
\begin{align*}
    \Pr\left[\frac12\Sigma\preceq\widehat\Sigma(X)\preceq2\Sigma\right]=\Pr\left[\frac{1}{\sqrt{2}}\sqrt{u^\top \Sigma u}\leq \sqrt{u^\top\widehat{\Sigma}(X)u}\leq \sqrt{2}\sqrt{u^\top \Sigma u}~~\forall u\right]\ge1-{\varepsilon_{\Cov}}/2.
\end{align*}
By a union bound, with probability at least $1-{\varepsilon_{\Cov}}$, for every $u$,
\[
 \left(\frac1{\sqrt2}-\frac18\right)\sqrt{u^\top\Sigma u}
 \le\sqrt{u^\top\widehat\Sigma(Z)u}
 \le\left(\sqrt2+\frac18\right)\sqrt{u^\top\Sigma u},
\]
which implies that $\Sigma/4\preceq\widehat\Sigma(Z)\preceq4\Sigma$.
\end{proof}

\cref{lem:finite-block-transfer} requires the minimum eigenvalue of the covariance matrix of the continuous distribution. The following lemma provides a lower bound on it using the known input data.

\begin{lemma}[Covariance spectral bounds]\label{lem:log-affine-covariance-floor}
For $B(x_c,r_0)\subseteq K\subseteq B(x_c,R_0)$ and a density $\pi\propto e^{-\langle b,x\rangle}\one_K(x)$,
\begin{equation}\label{eq:log-affine-covariance-floor}
 \frac{c r_0^2}{(D+\|b\|_2(R_0+r_0))^2}\Id\preceq\Cov(X)\preceq R_0^2\Id.
\end{equation}
\end{lemma}
\begin{proof}
Since translation does not change covariance, we may assume $x_c=0$. Fix a unit vector $u\in\R^D$ and let $m=\E[\langle u,X\rangle]$. We have
\begin{align*}
    u^\top \Cov(X)u=\Var[\langle u, X\rangle].
\end{align*}

We first prove a lower bound on $\Var(\langle u,X\rangle)$. Without loss of generality, we may assume that $m\geq 0$. 
Since $B(0,r_0)\subseteq K$, we have $-r_0u\in K$. Define a contraction map towards $-r_0u$:
\begin{align*}
    F_s(x)=(1-s)x-sr_0u,\qquad s=\frac{1}{2(D+\|b\|_2(R_0+r_0))}.
\end{align*}
Since $0<s\le1/2$ and $K$ is convex, $F_s(K)$ is contained in $K$. Then, we have
\begin{align*}
    \Var(\langle u,X\rangle)
 = &~ \frac1Z\int_K(\langle u,x\rangle-m)^2e^{-\langle b,x\rangle}\,dx\\
\geq &~ \frac1Z\int_{F_s(K)}(\langle u,x\rangle-m)^2e^{-\langle b,x\rangle}\,dx\\
= &~ \frac1Z\int_K(\langle u,F_s(y)\rangle-m)^2 e^{-\langle b,F_s(y)\rangle}(1-s)^D\,dy\\
= &~ \frac1Z\int_K(\langle u,F_s(y)\rangle-m)^2 e^{-\langle b, y\rangle} \cdot e^{s\langle b,y+r_0u\rangle+D\log(1-s)}\,dy\\
\geq &~ \frac1Z\int_K(\langle u,F_s(y)\rangle-m)^2 e^{-\langle b, y\rangle} \cdot e^{-s\|b\|_2(R_0+r_0)-2sD}\,dy\\
\geq &~ \frac{1}{eZ}\int_K(\langle u,F_s(y)\rangle-m)^2 e^{-\langle b, y\rangle}\,dy,
\end{align*}
The substitution $x=F_s(y)$ has Jacobian $(1-s)^D$. We bound the resulting weight using $\log(1-s)\ge-2s$, $\|y+r_0u\|_2\le R_0+r_0$, and the choice of $s$. Thus,
\begin{align*}
    \Var(\langle u,X\rangle)\geq &~ \frac{1}{e}\E\left[(\langle u,F_s(X)\rangle-m)^2\right] \geq \frac{1}{e} \left(\E\left[\langle u,F_s(X)\rangle-m\right]\right)^2\\
    = &~ \frac{s^2(m+r_0)^2}{e} \geq \frac{r_0^2}{4e(D+\|b\|_2(R_0+r_0))^2},
\end{align*}
Jensen's inequality bounds the second moment below by the squared mean. The definition of $F_s$ gives $\E[\langle u,F_s(X)\rangle-m]=-s(m+r_0)$, and $m\ge0$ yields the last inequality.

For the upper bound, since $K\subseteq B(0,R_0)$, we have
\begin{align*}
    \Var(\langle u,X\rangle)\leq \E\left[\langle u,X\rangle^2\right]\leq R_0^2.
\end{align*}

The lower and upper bounds for $u^\top \Cov(X)u$ hold for arbitrary unit vector $u$. Therefore,    \eqref{eq:log-affine-covariance-floor} is proved.
\end{proof}

\subsection{Proof of~\cref{thm:finite-metropolis}}\label{sec:B_proof}
We now wrap up and prove the main result for the lattice Metropolis walk. 
\begin{proof}[Proof of \cref{thm:finite-metropolis}]
We first choose the lattice spacing $h$. Let $\{b_1,b_2,\dots\}$ be the given list of tilts, and $B_*:=\max_i\|b_i\|_2$. By~\eqref{eq:rho_Lambda_bound} and~\cref{prop:finite-affine-branches}, the lattice parameters $\varrho_\Lambda$, $\varrho_A$, and $R_A$ have upper bounds that are independent of the covariance estimate $A$:
\[
 \varrho_\Lambda\le h\sqrt D,\qquad
 \varrho_A=\frac{Ch\sqrt D}{\sqrt{\lambda_-}},\qquad
 R_A=\frac{CR_0}{\sqrt{\lambda_-}}.
\]
They also upper bound the relevant discretization errors defined in this section:
\begin{align*}
 \varepsilon_{\rm quad}=&~ C\left(\frac{(D+\|b\|_2R_0)\varrho_\Lambda}{r_0}+\|b\|_2\varrho_\Lambda\right)
 \leq Ch\sqrt D\left(\frac{D+B_*R_0}{r_0}+B_*\right)\tag{\cref{lem:full-grid-quadrature}}\\*
 r_{\rm mesh}=&~ CD^3\varrho_A \leq  \frac{CD^{7/2}h}{\sqrt{\lambda_-}}\tag{\eqref{eq:mesh-depth}}\\*
 \varepsilon_{\rm mesh} = &~ C(\sqrt D\,r_{\rm mesh}+\varepsilon_{\rm quad})\leq \frac{CD^4h}{\sqrt{\lambda_-}}+C\varepsilon_{\rm quad}\tag{\cref{lem:mesh-collar}}
\end{align*}
Take the block size $L=\lceil CD(1+\log(4/\varepsilon_{\Cov}))^2\rceil$ fixed in the theorem. The mixing and covariance estimation arguments require the following constraints:
\begin{itemize}
    \item By~\cref{lem:full-grid-quadrature}, $\varepsilon_{\rm quad}\leq 1$, $\varrho_\Lambda\leq r_0/4$.
    \item By~\cref{lem:finite-overlap}, $\varrho_\Lambda\le\frac{\log(4/3)}{2B_*}$.
    \item By~\cref{lem:cell-separator}, $\varepsilon_{\rm mesh}\le c\min\left\{s,\frac{s\psi}{D(1+\log(4/s))}\right\}$, where $\psi$ is a known lower bound for $\psi_D$, and $s=\eta/4$ since we only need to bound the mixing time of $2$-warm-starts by~\cref{thm:finite-mixing}.
    \item By~\cref{lem:finite-block-transfer,lem:log-affine-covariance-floor}, $\varepsilon_{\rm quad}\leq \varepsilon_{\Cov}/(2L)$ and $\varrho_\Lambda\leq \frac{cr_0}{D+B_*(R_0+r_0)}$.
    \item By~\cref{thm:finite-metropolis}, $\varepsilon_{\rm quad}\leq \zeta$.
\end{itemize}
Combining them together, we get that for a sufficiently small constant $c_h\in (0,1)$, we can take $h$ to be the largest dyadic number such that
\begin{align}\label{eq:h_bound}
    h\leq c_h\cdot \frac{\psi\eta\varepsilon_{\Cov}\zeta}{D\log(1/\eta)\log^2(1/\varepsilon_{\Cov})(D^4/\sqrt{\lambda_-}+\sqrt{D}(D+B_*R_0)/r_0)},
\end{align}
which will satisfy all the above constraints. 

The radius list $\{\delta_0,\dots, \delta_{J-1}\}$ is determined by $r_{\rm mesh}$ and $R_A$~\eqref{eq:number-digital-scales}, which can also be chosen independent of $A$. Under our choice of parameters, \cref{prop:finite-affine-branches} holds, and the finite-seed sampler in~\cref{lem:finite-direction-preparation} has a seed length independent of $A$. Therefore, the parameters for implementing the lattice Metropolis walk in~\cref{alg:finite-metropolis} can be determined without knowing $A$.

Then, we check every bullet in~\cref{thm:finite-metropolis}. 

For~\textnormal{(i)}, for every $A$ with $\lambda_-\Id\preceq A \preceq \lambda_+\Id$ and every $b_i$ in the list, the symmetric Metropolis (line~\ref{ln:metropolis_step}) guarantees that the stationary distribution is $\pi_{i,\Lambda}$; the holding branch (line~\ref{ln:holding_branch}) ensures the Markov kernel ${\cal P}_{i,\Lambda, A}$ is lazy and positive, and the connector proposal (line~\ref{ln:connector_proposal}) ensures that ${\cal P}_{i,\Lambda, A}$ is irreducible.

For~\textnormal{(ii)}, by~\cref{thm:finite-mixing}, if $A$ is a good estimate of the corresponding covariance, then the 2-warm-start mixing time of ${\cal P}_{i,\Lambda, A}$ is bounded by $\widetilde{O}(D^2/\psi^2)$.

For~\textnormal{(iii)}, \cref{lem:metropolis-block-encoding,lem:acceptance-arithmetic} implement the $1$-block-encoding of the discriminant.

For~\textnormal{(iv)}, by~\cref{lem:full-grid-quadrature} and $\varepsilon_{\rm quad}\leq \zeta$, each normalization constant $Z_{b_i}$ has relative error at most $\zeta$ from the lattice discretization.

For~\textnormal{(v)}, the choice of $L$, the bound $L\varepsilon_{\rm quad}\le\varepsilon_{\Cov}/2$, and the covariance floor ensure the hypotheses of~\cref{lem:finite-block-transfer} for every listed tilt. That lemma gives the stated failure probability. Since $L=O(D(1+\log(4/\varepsilon_{\Cov}))^2)$, the largest dyadic spacing satisfying~\eqref{eq:h_bound} has inverse polynomial in the parameters listed in the theorem, including $1/\varepsilon_{\Cov}$. 

\end{proof}

\section{Finite implementation and parameter choices}\label{app:implementation}
This appendix supplies the finite circuits used in \cref{alg:finite-metropolis,alg:volume}. We use standard numerical algorithms and prove only the additional statements needed here: rounding a ball sample on the fixed grid, computing the cone's conditional weights, and controlling the full block-encoding unitary. We then choose the circuit sizes, accuracies, and mesh in that order.

\paragraph{Numerical conventions.}
Finite rational computations can be made reversible by~\cite{Bennett73}. We use standard finite-precision circuits for elementary functions~\cite{BhaskarEtAl16,HaenerEtAl18} and standard synthesis of one-qubit gates~\cite{DawsonNielsen06}. The arithmetic precision and gate-synthesis accuracy are fixed in advance to satisfy the error bounds below. On the specified input ranges, the evaluation circuits have cost polynomial in the input bit length and the number of precision bits.

We will use the following fact about the implementation error of rotation gates.
\begin{fact}\label{fac:rotation_gate}
For $0\le p\le1$, let $R(p)$ be the following rotation gate:
\begin{align*}
    R(p)=\begin{pmatrix}\sqrt{1-p}&-\sqrt p\\ \sqrt p&\sqrt{1-p}\end{pmatrix}.
\end{align*}
If $p$ is estimated within additive error $\nu^2/8$ and rotation-gate synthesis error is at most $\nu/2$, then the implemented $R(p)$ has operator norm error at most $\nu$. 
\end{fact}
\begin{proof}
Let $q\in[0,1]$ be such that $|p-q|\leq \nu^2/8$. Then, we have
\begin{equation}\label{eq:rotation-stability}
 \|R(p)-R(q)\|_{\op}^2
 =(\sqrt p-\sqrt q)^2+(\sqrt{1-p}-\sqrt{1-q})^2
 \le2|p-q|\leq \nu^2/4,
\end{equation}
where the first step follows from calculating $(R(p)-R(q))^\top (R(p)-R(q))$, and the second step uses $(\sqrt a-\sqrt b)^2\le|a-b|$. Hence, for the implemented gate $\widetilde{R}(q)$,
\begin{align*}
    \|\widetilde{R}(q)-R(p)\|_{\rm op}\leq \nu/2 + \nu/2 = \nu.
\end{align*}
\end{proof}
\subsection{Finite proposal preparation}\label{app:proposal-arithmetic}
The proposal requires a finite seed and a displacement on the fixed physical grid. Fix the rational matrices $W_A,T_A=W_A^{-1}$, and write $\mathcal R_h(u)=h\lfloor h^{-1}u+\tfrac12\one\rfloor$. For $U_\star$ uniform in $B(0,\delta)$, let $q_{A,\delta}$ be the law of $\mathcal R_h(T_AU_\star)$.

\begin{lemma}[Rounded ball sampling; cf.~{\cite[Section~5]{Jerrum06}}]\label{lem:finite-direction-preparation}
For rational $\delta>0$ and $0<\varepsilon_{\rm dir}<1/100$, there is a reversible circuit with a uniform binary seed that outputs a rational $U\in B(0,\delta)$ and $V=\sigma \mathcal R_h(T_AU)$ for $\sigma\sim \{-1,1\}$ such that $V$ and its law $\widehat{q}_{A,\delta}$ satisfy
\begin{equation}\label{eq:direction-tv}
 \dtv(\widehat{q}_{A,\delta},q_{A,\delta})\le\varepsilon_{\rm dir},\qquad
 \|W_AV\|_2\le\delta+\varrho_A\quad\text{on every seed},
\end{equation}
The seed length is
\[
 O\!\left(D\log\left(\frac {D(2+\delta \|T_A\|_{\rm op}/h)}{\varepsilon_{\rm dir}}\right)\right).
\]
In particular, if the condition of~\cref{prop:finite-affine-branches} holds, then it suffices to take a seed length 
\begin{align*}
    O\!\left(D\log\left(\frac{D(2+\delta \sqrt{\lambda_+}/h)}{\varepsilon_{\rm dir}}\right)\right).
\end{align*}
The additional computation is polynomial in these quantities and the rational input bit length, and uses no membership queries. 
\end{lemma}
\begin{proof}
We use the Gaussian and radial sampling construction analyzed in~\cite[Section~5]{Jerrum06}:
\[
 U_\star=\delta \cdot S^{1/D}\cdot \frac G{\|G\|_2},\qquad
 G\sim N(0,\Id_D),\quad S\sim\operatorname{Unif}[0,1],
\]
with independent inputs. Gaussian pairs can be generated from uniform inputs by the Box--Muller transform~\cite{BoxMuller58}. 

For $0<\tau<\delta/4$, use fresh random seeds of length
\[
 q=O\!\left(\log\left(\frac {D\delta}{\varepsilon_{\rm dir}\tau}\right)\right)
\]
bits each to discretize independent $\operatorname{Unif}([0,1])$ inputs. The standard calculation shows that with probability at least $1-\varepsilon_{\rm dir}/2$, the map from uniform random numbers to $U_\star$ has Lipschitz constant $\delta\,\mathrm{poly}(D/\varepsilon_{\rm dir})$. On this event, the rational approximation $\widetilde U$ satisfies $\|\widetilde U-U_\star\|_2\le\tau/8$. Let $U=(1-\tau/(2\delta))\widetilde{U}$, and test $\|U\|_2^2\le\delta^2$. If the result fails, set $U=0$. On the same event, the test passes because
\[
 \|U\|_2\le\left(1-\frac{\tau}{2\delta}\right)
                  \left(\delta+\frac{\tau}{8}\right)<\delta.
\]
Moreover,
\begin{align*}
 \|U-U_\star\|_2
 &\le\left(1-\frac{\tau}{2\delta}\right)\|\widetilde U-U_\star\|_2
       +\frac{\tau}{2\delta}\|U_\star\|_2\\
 &\le\frac{\tau}{8}+\frac{\tau}{2}
       =\frac{5\tau}{8}<\tau.
\end{align*}
Thus,
\begin{equation}\label{eq:finite-ball-coupling}
 \Pr[\|U-U_\star\|_2>\tau]\le\varepsilon_{\rm dir}/2,
 \qquad \text{and}\qquad \|U\|_2\le\delta~~\text{on every random seed}.
\end{equation}
The total number of random bits is
\begin{align}\label{eq:seed_length_qD}
    O(qD) = O\!\left(D\log\left(\frac {D\delta}{\varepsilon_{\rm dir}\tau}\right)\right).
\end{align}

Then, we consider the rounding ${\cal R}_h$. Let 
\begin{align*}
    \ell_i(u)=(h^{-1}e_i^\top T_A)u,\qquad {\cal R}_h(T_Au)_i=h\operatorname{round}(\ell_i(u))\qquad \forall i\in [D]
\end{align*}
with $\|\ell_i\|_2=\|h^{-1}e_i^\top T_A\|_2\leq h^{-1}\|T_A\|_{\rm op}$. On the event $\|U-U_\star\|_2\leq \tau$, ${\cal R}_h(T_A U)\ne {\cal R}_h(T_A U_\star)$ only if some $\ell_i(U_\star)$ is within $\|\ell_i\|_2\tau$ of a half-integer. For any $i\in [D]$, by the rotational symmetry, $\ell_i(U_\star)\overset{d}{=}\|\ell_i\|_2(U_\star)_1$, whose density is
\begin{align*}
    f_i(t)=\frac{v_{D-1}}{v_D\,\delta\|\ell_i\|_2}\left(1-\frac{t^2}{\delta^2\|\ell_i\|_2^2}\right)^{(D-1)/2},\qquad |t|<\delta\|\ell_i\|_2.
\end{align*}
We have 
\begin{align*}
    \sup_t f_i(t)\leq \frac{O(\sqrt{D})}{\delta\|\ell_i\|_2}. 
\end{align*}
Since there are at most 
\begin{align*}
    1+2(\delta+\tau)\|\ell_i\|_2\leq O(1+\delta\|\ell_i\|_2)
\end{align*}
half-integers in $[-(\delta+\tau)\|\ell_i\|_2,(\delta+\tau)\|\ell_i\|_2]$, we have
\begin{equation}\label{eq:rounding-boundary-probability}
\begin{aligned}
 \Pr\left[\mathcal R_h(T_AU)\ne\mathcal R_h(T_AU_\star),\ \|U-U_\star\|_2\le\tau\right]
 \leq &~ \sum_{i=1}^DO\left(\frac{\sqrt{D}}{\delta\|\ell_i\|_2} \cdot 2\tau \|\ell_i\|_2\cdot    (1+\delta\|\ell_i\|_2)\right)\\
 \leq &~ O\left(\frac{D^{3/2}\tau}{\delta}\left(\delta \max_i\|\ell_i\|_2+2\right)\right)\\*
 \leq &~ O\left(\frac{D^{3/2}\tau}{\delta}\chi_A(\delta)\right),
\end{aligned}
\end{equation}
where the last step follows from $\|\ell_i\|_2\le\|T_A\|_{\op}/h$ and $\chi_A(\delta)= 2+\delta\|T_A\|_{\rm op}/h$. Hence, if we choose $\tau=c\varepsilon_{\rm dir} \delta / (D^{3/2}\chi_A(\delta))$ with a sufficiently small constant $c$, then
\begin{align*}
    \Pr\left[\mathcal R_h(T_AU)\ne\mathcal R_h(T_AU_\star),\ \|U-U_\star\|_2\le\tau\right]\leq \frac{\varepsilon_{\rm dir}}{2}.
\end{align*}
And by the union bound, we have
\begin{align*}
    \Pr[\mathcal R_h(T_AU)\ne\mathcal R_h(T_AU_\star)]\leq \Pr[\|U-U_\star\|_2>\tau] + \frac{\varepsilon_{\rm dir}}{2}\leq \varepsilon_{\rm dir}.
\end{align*}
Let $q_{A,\delta}^0$ be the law of $\mathcal R_h(T_AU)$. We get that
\begin{align*}
    \dtv(q_{A,\delta}^0,q_{A,\delta})\leq \Pr[\mathcal R_h(T_AU)\ne\mathcal R_h(T_AU_\star)]\leq \varepsilon_{\rm dir}.
\end{align*}
By the symmetry of $q_{A,\delta}$, adding a uniform sign to $\mathcal R_h(T_AU)$ will not increase the total variation distance. Therefore,
\begin{align*}
    \dtv(\widehat{q}_{A,\delta},q_{A,\delta})\leq \dtv(q_{A,\delta}^0,q_{A,\delta})\leq \varepsilon_{\rm dir}.
\end{align*}
Moreover, on every random seed, it holds that
\begin{align*}
    \|W_AV\|_2=&~ \|W_A {\cal R}_h(T_AU)\|_2\leq \|W_A T_A U\|_2+\|W_A\|_{\op}\|{\cal R}_h(T_AU)-T_AU\|_2\\
    \le &~ \|U\|_2+ \|W_A\|_{\op} \frac{h\sqrt{D}}{2} \leq \delta + \frac{11h\sqrt{D}}{20\sqrt{\lambda_-}}
    \le\delta+\varrho_A,
\end{align*}
where the third step follows from $T_A=W_A^{-1}$ and  ${\cal R}_h$ only shifting each coordinate by at most $h/2$, the fourth step follows from ~\eqref{eq:finite-ball-coupling} and~\eqref{eq:WA_TA_op_norm}, and the last step follows from the definition of $\varrho_{A}$.

Finally, by~\eqref{eq:seed_length_qD}, the seed length can be bounded by
\begin{align*}
    O\!\left(D\log\left(\frac {D\delta}{\varepsilon_{\rm dir}\tau}\right)\right)
    =O\!\left(D\log\left(\frac {D\chi_A(\delta)}{\varepsilon_{\rm dir}}\right)\right).
\end{align*}
Under the condition of~\cref{prop:finite-affine-branches}, by~\eqref{eq:WA_TA_op_norm}, we have $\|T_A\|_{\rm op}\leq O(\sqrt{\lambda_+})$. Hence, 
\begin{align*}
    O\!\left(D\log\left(\frac {D\chi_A(\delta)}{\varepsilon_{\rm dir}}\right)\right)
    =O\!\left(D\log\left(\frac {D(2+\delta\sqrt{\lambda_+}/h)}{\varepsilon_{\rm dir}}\right)\right).
\end{align*}

\end{proof}

\subsubsection{Construction of the affine preconditioning matrices}\label{app:whitening-arithmetic}
$W_A$ and $T_A$ in \cref{prop:finite-affine-branches} can be obtained by the standard inverse-square-root Newton iteration~\cite[Lemmas~2.1 and~2.3]{GuoHigham06}. Choose a positive dyadic $a$ with $\lambda_+\le a^2<4\lambda_+$, and compute in rational arithmetic
\[
 W_0=a^{-1}\Id,\qquad W_{j+1}=\tfrac12W_j(3\Id-AW_j^2).
\]
The iterates are positive definite and commute with $A$. The cited residual estimate gives
\[
 0\preceq\Id-AW_j^2,\qquad
 \|\Id-AW_j^2\|_{\op}\le(1-\lambda_-/a^2)^{2^j}.
\]
We can take $j=O(1+\log(\lambda_+/\lambda_-))$ iterations so that 
\[
 \|W_jA^{1/2}-\Id\|_{\op}
 \le\|AW_j^2-\Id\|_{\op}\le1/20,
\]
where the first step uses $|\sqrt u-1|\le|u-1|$ for $u\ge0$ on each eigenvalue. Set $W_A=W_j$ and compute $T_A=W_A^{-1}$ exactly. Rational bit lengths grow by at most a constant factor per iteration, so the cost is polynomial in the input bit length and $\lambda_+/\lambda_-$. In particular, the same iteration count works for any $A$ that satisfies $\lambda_-\Id\preceq A\preceq \lambda_+\Id$. 

\subsection{The Metropolis block-encoding}\label{app:acceptance-arithmetic}
The following lemma shows the numerical implementation of the Metropolis acceptance step in the block-encoding of the discriminant.

\begin{lemma}[Numerical Metropolis block-encoding]\label{lem:acceptance-arithmetic}
In the setting of \cref{lem:metropolis-block-encoding}, for every $0<{\varepsilon_{\rm be}}<1/2$, there are an exact unitary $U$ block-encoding $D_\pi\oplus0$ and an implementation $\widetilde U$ on the same registers such that
\[
 \|\widetilde U-U\|_{\op}\le{\varepsilon_{\rm be}}.
\]
The implementation uses $O(1)$ membership queries and computation polynomial in the displacement-circuit size, rational input bit length, and $\log(1/{\varepsilon_{\rm be}})$. The estimate is uniform over the allowed covariance controls and also holds for inverse and controlled queries.
\end{lemma}

\begin{proof}
We first approximate the acceptance operation to operator-norm error ${\varepsilon_{\rm be}}/2$, then bound its effect on the block-encoding.

Compute feasibility, overflow, and $t=\max\{0,\sigma\langle b,v_s\rangle\}$ exactly. For an admissible proposal, the acceptance probability is $p=e^{-t}$; otherwise, set $p=0$. Choose a rational threshold ${T_{\rm cut}}\in[\log(32/{\varepsilon_{\rm be}}^2),\log(32/{\varepsilon_{\rm be}}^2)+1]$. For $t\ge {T_{\rm cut}}$, set $\widetilde p=0$, so
\[
 |\widetilde p-p|\le e^{-t}\le e^{-{T_{\rm cut}}}\le{\varepsilon_{\rm be}}^2/32.
\]
For $t<{T_{\rm cut}}$, evaluate $e^{-t}$ to absolute error ${\varepsilon_{\rm be}}^2/32$ and clamp the result to $[0,1]$. Inadmissible proposals have $\widetilde p=p=0$. Thus, uniformly over all encoded controls,
\[
 |\widetilde p-p|\le{\varepsilon_{\rm be}}^2/32,
 \qquad
 \|R(\widetilde p)-R(p)\|_{\op}
 \le\sqrt{2|\widetilde p-p|}
 \le{\varepsilon_{\rm be}}/4,
\]
where the second inequality follows from~\cref{fac:rotation_gate}. Exponential evaluation is needed only on an interval of length $O(\log(1/{\varepsilon_{\rm be}}))$. 

To implement the rotation, compute a binary angle $\widetilde\theta$ satisfying $|\widetilde\theta-\arcsin\sqrt{\widetilde p}|\le{\varepsilon_{\rm be}}/8$. Each angle bit controls a fixed rotation of the acceptance qubit; then uncompute all the probability and angle arithmetic. Let $\overline V$ denote this circuit with those fixed rotations performed exactly, and write $R_\theta$ for the real rotation through angle $\theta$. On zero arithmetic workspace its rotation error is at most
\[
 \|R_{\widetilde\theta}-R(p)\|_{\op}
 \le\|R_{\widetilde\theta}-R(\widetilde p)\|_{\op}
    +\|R(\widetilde p)-R(p)\|_{\op}
 \le{\varepsilon_{\rm be}}/8+{\varepsilon_{\rm be}}/4=3{\varepsilon_{\rm be}}/8.
\]
The subspace $\mathcal H_0$ of zero arithmetic workspace is invariant under $\overline V$ for either input value of the acceptance qubit: all arithmetic depends on unchanged control labels, and the angle register is used only as a control. Unitarity therefore also makes $\mathcal H_0^\perp$ invariant. Define $V$ to apply the exact rotation $R(p)$ on $\mathcal H_0$ and to agree with $\overline V$ on $\mathcal H_0^\perp$. Orthogonality of the control labels then gives the full-register bound
\[
 \|\overline V-V\|_{\op}
 =\sup_{x,s,\sigma,A}
   \|R_{\widetilde\theta_{x,s,\sigma,A}}-R(p_{x,s,\sigma,A})\|_{\op}
 \le3{\varepsilon_{\rm be}}/8.
\]
Synthesize the fixed controlled rotations with total operator error at most ${\varepsilon_{\rm be}}/8$, obtaining $\widetilde V$. This error bound holds on arbitrary workspace inputs, since it is a sum of gate errors. Hence
\[
 \|\widetilde V-V\|_{\op}
 \le\|\widetilde V-\overline V\|_{\op}
    +\|\overline V-V\|_{\op}
 \le{\varepsilon_{\rm be}}/8+3{\varepsilon_{\rm be}}/8={\varepsilon_{\rm be}}/2.
\]

Using $V$, \cref{lem:metropolis-block-encoding} supplies an exact $1$-block-encoding $U$ of $D_\pi\oplus0$. Let $\widetilde U$ be the same construction with $V$ replaced by $\widetilde V$, and let $U_T,\widetilde U_T$ be the corresponding seed-and-acceptance preparation unitaries. All other operations, including the membership-flag unitaries encoding $P_\Omega$, are exact. Therefore,
\begin{align*}
 \|\widetilde U-U\|_{\op}
 \leq \|\widetilde U_T^\dagger F\widetilde U_T-U_T^\dagger FU_T\|_{\op}
 \leq 2\|\widetilde U_T-U_T\|_{\op}
 \leq 2\|\widetilde V-V\|_{\op}
 \le{\varepsilon_{\rm be}}.
\end{align*}

The construction uses $O(1)$ membership queries. The remaining computation is polynomial in the displacement-circuit size, rational input bit length, and $\log(1/{\varepsilon_{\rm be}})$ by the finite-precision arithmetic and gate-synthesis routines stated above. 

Finally,
\[
 \|\widetilde U^\dagger-U^\dagger\|_{\op}
 =\|\widetilde U-U\|_{\op},
 \qquad
 \|\ctrl(\widetilde U)-\ctrl(U)\|_{\op}
 =\|\widetilde U-U\|_{\op},
\]
so inverse and controlled queries satisfy the same error bound.
\end{proof}

\subsection{Preparing the initial samples}\label{app:initialization-details}

\begin{lemma}[Initial qsample preparation]\label{lem:finite-cone-preparation}
Let $\nu_\Lambda$ and $Z_{C,\Lambda}$ be the finite cone law and its normalization constant defined in \eqref{eq:finite-cone-anchor}. For every $0<{\varepsilon_{\rm cone}}<1/2$, we can compute $\log Z_{C,\Lambda}$ to additive error ${\varepsilon_{\rm cone}}$ and implement a unitary $\widetilde U_C$ satisfying
\[
 \left\|\widetilde U_C\ket0_{\mathsf{ES}}
 -\ket0_{\mathsf E}\ket{\nu_\Lambda}_{\mathsf S}\right\|\le{\varepsilon_{\rm cone}}.
\]
The corresponding reflection
\[
 \widetilde R_C
 :=\widetilde U_C(2\ket0\!\bra0_{\mathsf{ES}}-\Id_{\mathsf{ES}})\widetilde U_C^\dagger
\]
satisfies, for every unit vector $\ket\psi_{\mathsf S}$,
\[
 \left\|\widetilde R_C(\ket0_{\mathsf E}\ket\psi_{\mathsf S})
 -\ket0_{\mathsf E}R_{\nu_\Lambda}\ket\psi_{\mathsf S}\right\|\le2{\varepsilon_{\rm cone}}.
\]
These constructions use no membership queries and have cost polynomial in $D$, $1+\kappa/h$, the rational input bit length, and $\log(1/{\varepsilon_{\rm cone}})$.
\end{lemma}

\begin{proof}
We first compute the conditional weights of the cone law, then use Grover--Rudolph preparation~\cite{GroverRudolph02} and bound the numerical error.

At height $kh$, a cone point has the form $(kh,hu)$ with $u\in\Z^d$ and $\|u\|_2\le k$. Put $k_{\max}=\lfloor2\kappa/h\rfloor$. For integers $s\ge0$, let $C_j(s)$ count the vectors in $\Z^j$ of squared norm at most $s$. Conditioning on the first coordinate gives
\begin{equation}\label{eq:cone-counts}
 C_0(s)=1,\qquad
 C_j(s)=\sum_{v=-\lfloor\sqrt s\rfloor}^{\lfloor\sqrt s\rfloor}C_{j-1}(s-v^2).
\end{equation}
The table
\[
 \{C_j(s):0\le j\le d,\ 0\le s\le k_{\max}^2\}
\]
has $O(d(1+k_{\max})^2)$ entries and requires $O(d(1+k_{\max})^3)$ integer additions. Each entry has $O(d\log(2+k_{\max}))$ bits, since it is at most $(2k_{\max}+1)^d$.

The height marginal is obtained from
\begin{equation}\label{eq:cone-normalizer-sum}
 Z_{C,\Lambda}
 =h^D\sum_{k=0}^{k_{\max}}C_d(k^2)e^{-a_0hk},
 \qquad
 \Pr[k]=\frac{h^DC_d(k^2)e^{-a_0hk}}{Z_{C,\Lambda}}.
\end{equation}
Conditional on $k$ and a feasible coordinate prefix $u_1,\ldots,u_j$, let $s=k^2-\sum_{r=1}^ju_r^2$. The next coordinate has distribution
\[
 \Pr[u_{j+1}=v\mid k,u_1,\ldots,u_j]
 =\frac{C_{d-j-1}(s-v^2)}{C_{d-j}(s)}
 \qquad\bigl(|v|\le\lfloor\sqrt s\rfloor\bigr),
\]
and probability zero otherwise. The denominator is the sum of the numerators by \eqref{eq:cone-counts}. For every feasible full vector $u$, the conditional probabilities telescope:
\begin{align*}
 \Pr[k,u]
 = &~ \Pr[k]\prod_{j=0}^{d-1}
 \Pr[u_{j+1}\mid k,u_1,\ldots,u_j]\\
 = &~ \frac{h^DC_d(k^2)e^{-a_0hk}}{Z_{C,\Lambda}}
       \frac{C_0(k^2-\|u\|_2^2)}{C_d(k^2)}\\
 = &~ \frac{h^De^{-a_0hk}}{Z_{C,\Lambda}}
 =\nu_\Lambda((kh,hu)),
\end{align*}
where the last step uses $C_0(s)=1$ for $s\ge0$.

These conditional weights also determine the binary-prefix probabilities required by Grover--Rudolph preparation. Height-prefix weights are sums of the terms in \eqref{eq:cone-normalizer-sum}; coordinate-prefix weights are sums of the corresponding completion counts. Prefixes with no feasible completion have zero weight and are assigned a fixed continuation. The counting table and its lookups can be implemented by ordinary reversible circuits of polynomial size.

We next explain how to evaluate the positive height sums. For a nonempty set of heights $\mathcal I$, let $k_0=\min\mathcal I$ and factor
\[
 \sum_{k\in\mathcal I}C_d(k^2)e^{-a_0hk}
 =e^{-a_0hk_0}
  \sum_{k\in\mathcal I}C_d(k^2)e^{-a_0h(k-k_0)}.
\]
The scaled sum lies between $1$ and $(k_{\max}+1)(2k_{\max}+1)^d$: its $k_0$ term is at least one, every exponential is at most one, and the integer coefficients satisfy the preceding count bound. Thus finite-precision evaluation gives any prescribed relative accuracy using polynomially many bits. 

We evaluate the logarithm of the normalization constant as
\[
 \log Z_{C,\Lambda}
 =D\log h+
 \log\left(\sum_{k=0}^{k_{\max}}C_d(k^2)e^{-a_0hk}\right).
\]
The sum is at least one because its $k=0$ term equals one. Relative error $\xi\le1/2$ in this sum changes its logarithm by at most $2\xi$. Choosing the sum and scalar-evaluation accuracies accordingly gives additive error at most ${\varepsilon_{\rm cone}}$ in $\log Z_{C,\Lambda}$.

For state preparation, the height and each spatial coordinate use $O(\log(2+k_{\max}))$ bits. Grover--Rudolph preparation uses one conditional-rotation stage per output bit, so the number of stages is
\[
 N_{\rm rot}=O(D\log(2+k_{\max})).
\]
At each stage, compute the conditional rotation from the preceding bits, apply it to the next output qubit, and uncompute its workspace. By \eqref{eq:rotation-stability}, choose the numerical precision uniformly over all prefixes so that each complete stage, including gate synthesis, has operator error at most ${\varepsilon_{\rm cone}}/N_{\rm rot}$. Exact unitary completions are chosen as in \cref{lem:acceptance-arithmetic}.

Let $U_C$ be the exact prefix circuit and $\widetilde U_C$ its implementation, with respective stages $U_j$ and $\widetilde U_j$. The probability calculation above gives $U_C\ket0_{\mathsf{ES}}=\ket0_{\mathsf E}\ket{\nu_\Lambda}_{\mathsf S}$, and
\begin{align*}
 \left\|\widetilde U_C\ket0_{\mathsf{ES}}
 -\ket0_{\mathsf E}\ket{\nu_\Lambda}_{\mathsf S}\right\|
 = \left\|(\widetilde U_C-U_C)\ket0_{\mathsf{ES}}\right\|
 \leq  \sum_{j=1}^{N_{\rm rot}}\|\widetilde U_j-U_j\|_{\op}
 \leq  N_{\rm rot}\frac{{\varepsilon_{\rm cone}}}{N_{\rm rot}}
 ={\varepsilon_{\rm cone}},
\end{align*}
where the first inequality follows by telescoping the unitary product. The arithmetic and synthesis require a number of precision bits polynomial in the input and count bit lengths and $\log(N_{\rm rot}/{\varepsilon_{\rm cone}})$.

Finally, write
\[
 \ket u_{\mathsf{ES}}=\widetilde U_C\ket0_{\mathsf{ES}},
 \qquad
 \ket v_{\mathsf{ES}}=\ket0_{\mathsf E}\ket{\nu_\Lambda}_{\mathsf S}.
\]
Since $\|\ket u-\ket v\|\le{\varepsilon_{\rm cone}}$,
\begin{align*}
 \|\widetilde R_C-(2\ket v\!\bra v-\Id_{\mathsf{ES}})\|_{\op}
 = 2\|\ket u\!\bra u-\ket v\!\bra v\|_{\op}
 = 2\sqrt{1-|\langle u|v\rangle|^2}
 \leq 2\|\ket u-\ket v\|
 \le2{\varepsilon_{\rm cone}}.
\end{align*}
On zero workspace,
\[
 (2\ket v\!\bra v-\Id_{\mathsf{ES}})
 (\ket0_{\mathsf E}\ket\psi_{\mathsf S})
 =\ket0_{\mathsf E}R_{\nu_\Lambda}\ket\psi_{\mathsf S}.
\]
This proves the reflection guarantee. Tensoring with an identity or adding an external control preserves the operator-norm bound.
\end{proof}

\subsection{Numerical evaluation of the rotations and output}\label{app:mark-arithmetic}

The algorithm must implement the sample-controlled gate $V_{Y_i^\pm}$:
\begin{align*}
V_{Y_i^{\pm}}(\ket{x}_{\mathsf{S}}\otimes\ket{0}_{\mathsf{Q}})
 =\ket{x}_{\mathsf{S}}\otimes\left(\sqrt{1-{Y_i^{\pm}}(x)}\ket{0}_{\mathsf{Q}}+\sqrt{{Y_i^{\pm}}(x)}\ket{1}_{\mathsf{Q}}\right)=\ket{x}_{\sf S}\otimes R(Y_i^\pm(x))\ket{0}_{\sf Q},
\end{align*}
and evaluate the final expression \eqref{eq:volume-output}. We specify their numerical accuracies below. For all encoded positions $z$, extend $X_i$ by
\[
 X_i(z)=\exp\!\left((a_i-a_{i+1})\min\{2\kappa,\max\{0,z_0\}\}\right).
\]
This agrees with \eqref{eq:assembly-ratio-definition} on the pencil.

\begin{lemma}[Numerical evaluation]\label{lem:marks-and-output}
Let $c_i$ be the median of the batch averages in \cref{lem:pilots}, computed from numerical values $\widetilde X_i(z)$ with $1/16$ relative error for every $z$. For $0<\nu,\varepsilon<1/2$, the following hold.
\begin{enumerate}[label=(\roman*)]
\item Each $V_{Y_i^\pm}$ has an implementation $\widetilde V_{Y_i}^\pm$ with operator norm error at most $\nu$. 

\item Given $c_0,\ldots,c_{m-1}$ and the phase-estimation output $\widehat{\mathsf L}$, a classical computation returns $\widehat V>0$ satisfying
\begin{equation}\label{eq:scalar-output-error}
 \begin{aligned}
 \log\widehat V=  d\log r+\log Z_{C,\Lambda}
       +\sum_{i=0}^{m-1}\log c_i
       -\widehat{\mathsf L}-\log I_m+e_{\rm num}
 \end{aligned}
\end{equation}
with $|e_{\rm num}|\leq \varepsilon/8$.
\end{enumerate}
\end{lemma}

\begin{proof}
We first bound the values used in the numerical computations. Since the cooling schedule~\eqref{eq:cooling-schedule} satisfies $0<a_i-a_{i+1}\le a_0=4D$, we have
\[
 1\le X_i(z)\le e^{8D\kappa}, \qquad 
 \frac{15}{16}X_i(z)\le\widetilde X_i(z)\le\frac{17}{16}X_i(z).
\]
Thus, the median-of-means estimator $c_i$ satisfies
\[
 \frac{15}{16}\le c_i\le\frac{17}{16}e^{8D\kappa},
 \qquad
 |\log c_i|\le8D\kappa+1.
\]

We next implement $V_{Y_i^\pm}$. Recall that
\[
 Y_i^+(z)=\frac{X_i(z)}{c_i+X_i(z)},
 \qquad
 Y_i^-(z)=\frac{c_i}{c_i+X_i(z)}.
\]
For a fixed computed value $c_i$, define
\[
 t=\log X_i(z)-\log c_i,
 \qquad
 g(t)=\frac1{1+e^{-t}}.
\]
Then $Y_i^+=g(t)$ and $Y_i^-=1-g(t)$. Moreover,
\[
 |g'(t)|\le\frac14,
 \qquad
 \min\{g(t),1-g(t)\}\le e^{-|t|}.
\]
Compute $\widetilde t$ with $|\widetilde t-t|\le\nu^2/16$, and approximate $g(\widetilde t)$ to absolute error $\nu^2/16$. Clamp the numerical result to $[0,1]$, denote it by $\widetilde Y_i^+$, and set $\widetilde Y_i^-=1-\widetilde Y_i^+$. Then, we have
\begin{align*}
 \left|\widetilde Y_i^+-Y_i^+\right|
 \leq \left|\widetilde Y_i^+-g(\widetilde t)\right|
          +\left|g(\widetilde t)-g(t)\right|
 \leq \frac{\nu^2}{16}+\frac14\left|\widetilde t-t\right|
 \leq \frac{5\nu^2}{64}\le\frac{\nu^2}{8}.
\end{align*}
The same estimate holds for $Y_i^-$. Thus, by~\cref{fac:rotation_gate}, we get that
\[
 \left\|R(\widetilde Y_i^\pm)-R(Y_i^\pm)\right\|_{\op}
 \le\sqrt{2|\widetilde Y_i^\pm-Y_i^\pm|}
 \le\frac\nu2.
\]
Allocate another $\nu/2$ to gate synthesis. This proves \textnormal{(i)}.

For~\textnormal{(ii)}, we take the computed $c_i$ and $\widehat{\mathsf L}$ and only bound the numerical error in evaluating this expression. The quantity $\log Z_{C,\Lambda}$ can be computed by \cref{lem:finite-cone-preparation}. Also,
\[
 \log I_m=-a_m\kappa+\log(1-e^{-a_m\kappa})-\log a_m,
 \qquad
 15/16<a_m\kappa\le1,
\]
Thus $1-e^{-a_m\kappa}$ is bounded away from zero. The readout definition in~\cref{sec:q_est_procedure}, $\widehat{\mathsf L}=64n(Y/M-1/4)$ with $0\le Y<M$, gives $-16n\le\widehat{\mathsf L}<48n$ on every outcome. Together with the preceding bounds on $c_i$, these estimates allow finite-precision arithmetic to evaluate the output with total logarithmic error at most $\varepsilon/8$, including the final exponentiation and rounding. This proves \eqref{eq:scalar-output-error}.
\end{proof}

\subsection{Choosing the parameters}\label{app:parameters}

We choose the parameters to satisfy the implementation-error bound in~\cref{lem:composition}. Recall that $D=d+1$ is the pencil dimension, $\kappa$ is the input aspect ratio, $\varepsilon$ is the relative volume accuracy, and $\psi$ is the supplied lower bound on $\psi_D$. Put
\[
 Q=\max\{4,D\kappa/(\varepsilon\psi)\},\qquad \beta=1/100.
\]
Let $\rho_{\rm impl}$ and $\rho_{\rm ref}$ be the final joint states before measurement. The reference computation uses perfect qsamples at the qsample conversions and the prescribed finite-estimation circuits with exact qsample-reflections. Our goal is to guarantee that
\[
 \dtr(\rho_{\rm impl},\rho_{\rm ref})
 \le e_{\rm init}+e_{\rm conv}+N\bigl(\xi+2\sqrt{{\varepsilon_{\Cov}}}\bigr)+e_{\rm gate}
 \le\beta.
\]
Here, $N$ bounds the number of qsample-reflections, $\xi$ bounds one implemented reflection's error when \eqref{eq:good-covariance} holds, and ${\varepsilon_{\Cov}}$ bounds the probability that this covariance estimation condition fails. The other three terms are the total errors from initialization, qsample conversions with exact reflections, and the remaining numerical operations. The targeted error for each part is listed in~\cref{tab:implementation-error-budget}.

\begin{table}[ht]
\centering
\small
\renewcommand{\arraystretch}{1.35}
\setlength{\tabcolsep}{4pt}
\begin{tabularx}{\linewidth}{@{}>{\raggedright\arraybackslash}X c >{\raggedright\arraybackslash}p{0.40\linewidth} c@{}}
\toprule
Operation & Count & Error per operation & Total \\
\midrule
Initial qsample preparation & $n_0$ & Trace distance $\beta/(4n_0)$ & $\beta/4$ \\
Single-copy conversion\newline with exact reflections & $T$ & Trace distance $\beta/(4T)$ & $\beta/4$ \\
Reflections about $\pi_i$ & $\le N$ & $\xi+2\sqrt{{\varepsilon_{\Cov}}}\le\beta/(4N)$ & $\beta/4$ \\
Other numerical operations & $\le N_{\rm gate}$ & Operator error $\beta/(4N_{\rm gate})$ & $\beta/4$ \\
\bottomrule
\end{tabularx}
\caption{Allocation of implementation error. The counts $n_0,T,N,N_{\rm gate}$ refer to the total numbers of operations in the corresponding rows.}
\label{tab:implementation-error-budget}
\end{table}

\Cref{fig:parameter-dependencies} shows the order of choosing the parameters. The qsample-reflection count $N$ determines the required error per call; the fourth-moment norm bound $M_4$ then determines the mixing accuracy. These requirements, together with the known geometry, determine the lattice space $\Omega_\Lambda=K_0\cap h\Z^D$.

\begin{figure}[ht]
\centering
\begingroup
\setlength{\unitlength}{1pt}
\footnotesize
\begin{picture}(460,135)
 \put(0,95){\framebox(100,37){\shortstack{Reflection count\\$N$}}}
 \put(117,95){\framebox(110,37){\shortstack{Reflection error $\xi$,\\Covariance failure ${\varepsilon_{\Cov}}$}}}
 \put(244,95){\framebox(100,37){\shortstack{Samples per block\\$L$}}}
 \put(0,43){\framebox(100,36){\shortstack{Fourth-moment\\norm $M_4$}}}
 \put(117,43){\framebox(110,36){\shortstack{Mixing accuracy $\eta$,\\QSVT filter error $\zeta_{\rm filter}$}}}
 \put(244,43){\framebox(100,36){\shortstack{Lattice spacing $h$\\   of $\Omega_\Lambda$}}}
 \put(361,43){\framebox(99,36){\shortstack{Register sizes\\and numerical\\precision}}}
 \put(0,0){\framebox(227,28){\shortstack{Known geometric parameters $\lambda_-,\lambda_+$,\\ volume estimation accuracy $\varepsilon$}}}
 \put(100,113.5){\vector(1,0){17}}
 \put(227,113.5){\vector(1,0){17}}
 \put(172,95){\vector(0,-1){16}}
 \put(100,61){\vector(1,0){17}}
 \put(227,61){\vector(1,0){17}}
 \put(294,95){\vector(0,-1){16}}
 \put(344,61){\vector(1,0){17}}
 \put(227,14){\line(1,0){67}}
 \put(294,14){\vector(0,1){29}}
\end{picture}
\endgroup
\caption{Order of parameter selection.}
\label{fig:parameter-dependencies}
\end{figure}
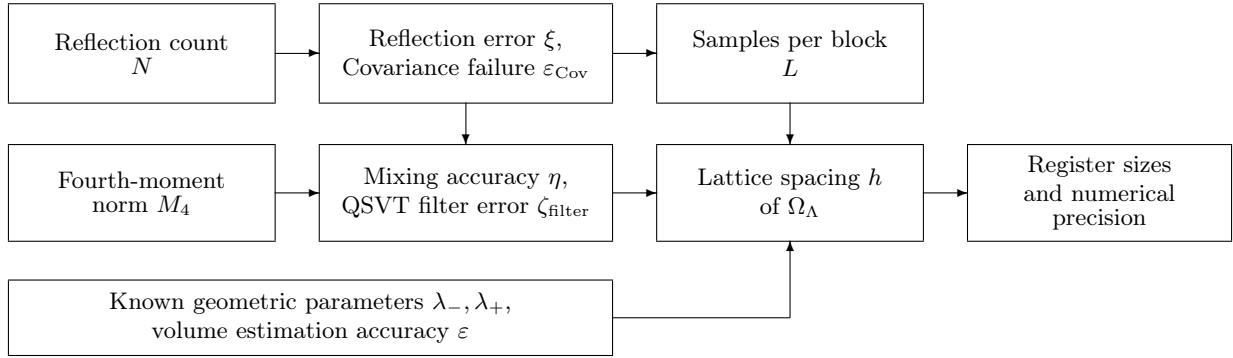

\paragraph{Reflection and sample counts.}
The cooling schedule \eqref{eq:cooling-schedule} has $m=O(\sqrt D\log Q)$ transitions. \cref{lem:pilots} uses $P=O(\log(m+1))$ samples per stage to compute $c_i$. The logarithmic sum \eqref{eq:volume-log-sum} has $n=2m+2$ contributions. Use \eqref{eq:log-clock-parameters} with accuracy $\varepsilon/8$ to choose the polynomial degree $k$ and phase-estimation register $\mathsf F$, so that $k,\log_2M=O(\log Q)$.

Choose a proposed upper bound on the total qsample-reflection count in qsample conversions and product estimation, including queries inside the Hamiltonian-simulation circuits:
\begin{equation}\label{eq:deterministic-reflection-bound}
 N=\left\lceil C_N\left(D^{3/2}+\frac{\sqrt D}{\varepsilon}\right)(1+\log Q)^4\right\rceil.
\end{equation}
To meet the targeted qsample-reflection error in~\cref{tab:implementation-error-budget}, set $\xi:=\beta/(8N)$ and choose the covariance failure probability and block size as
\begin{equation}\label{eq:precision-covariance}
 {\varepsilon_{\Cov}}:=\frac{\xi^2}{4}=\left(\frac{\beta}{16N}\right)^2,
 \qquad L:=\left\lceil CD\left(1+\log\frac4{{\varepsilon_{\Cov}}}\right)^2\right\rceil.
\end{equation}
Indeed, $N(\xi+2\sqrt{{\varepsilon_{\Cov}}})=2N\xi=\beta/4$. Here, $L$ is the number of samples in each block $\mathsf B_1,\mathsf B_2$. By~\cref{lem:covariance-confidence}, the continuous covariance estimate fails with probability at most ${\varepsilon_{\Cov}}/2$. We allocate the remaining ${\varepsilon_{\Cov}}/2$ to the total-variation error when transferring this guarantee to lattice samples, as in~\cref{lem:finite-block-transfer}.

As in the proof of~\cref{lem:query_complexity_main}, we have
\[
 n_0=2L+(P+2)m+2,\qquad
 T=2Lm+\frac{Pm(m-1)}2+m^2+m.
\]
Since adjacent qsamples have squared overlap at least $1/4$, \cref{lem:transport} uses $O(\log(4T/\beta))$ reflections per conversion.

For a fixed $C_N$, we have $\log N=O(\log Q)$, which gives that
\begin{align*}
    L=O(D(1+\log Q)^2),\qquad T=O(D^{3/2}(1+\log Q)^3).
\end{align*}
Thus, the total number of qsample-reflections used for qsample conversions is
\begin{equation}\label{eq:ref_count_1}
 O(T\log(4T/\beta))=O(D^{3/2}(1+\log Q)^4).
\end{equation}
For product estimation, \cref{lem:logarithmic-phase} bounds the number of qsample-reflections in one \textsc{Add} procedure (\cref{alg:add-log}) by
\begin{equation}\label{eq:refs_per_add}
 O\left(k\left[\alpha_k\sum_{j=0}^{\log_2M-1}|\tau_j|+(\log_2M)\log(100n\log_2M)\right]\right)
 =O\left(\frac{(1+\log Q)^3}{\varepsilon}\right),
\end{equation}
where the equality uses $\sum_j|\tau_j|=O(1/\varepsilon)$, $\alpha_k\le1+\log k$, and $k,\log_2M=O(\log Q)$. Multiplying by $n=2m+2=O(\sqrt D\log Q)$ gives the total number of qsample-reflections for the product estimation
\begin{equation}\label{eq:ref_count_2}
 O\left(\frac{\sqrt D}{\varepsilon}(1+\log Q)^4\right).
\end{equation}
The sum of \eqref{eq:ref_count_1} and \eqref{eq:ref_count_2} is at most $N$ for sufficiently large $C_N$: their constants depend on $C_N$ only through fixed powers of $1+\log C_N$, whereas \eqref{eq:deterministic-reflection-bound} is linear in $C_N$. This verifies the choice of $N$.

The fourth-moment norm argument in~\cref{lem:query_complexity_main} applies separately within one \textsc{Add} or one single-copy qsample conversion. Using \eqref{eq:refs_per_add}, which is the dominant term, it gives
\begin{equation}\label{eq:local-reference-fourth-bound}
 \|\ket\Psi_{\mathsf{SA}}\|_{4,\pi}\le M_4,
 \qquad M_4:=\frac{C_4(1+\log Q)^3}{\varepsilon},
\end{equation}
for every reference input to a reflection about $\pi$, with $C_4$ sufficiently large. 

\paragraph{Reflection accuracy and lattice space.}
We next ensure reflection error at most $\xi$ whenever \eqref{eq:good-covariance} holds. Suppose the numerical implementation has error at most $\xi/2$ as in~\cref{cor:finite-reflection-cost-main}. The remaining error bound from \cref{thm:effective-reflection} is $C\eta^{1/4}M_4+\zeta_{\rm filter}$, where $C$ is the constant in that theorem, $\eta$ is the classical mixing accuracy, and $\zeta_{\rm filter}$ is the filter error. Set
\begin{equation}\label{eq:precision-mixing}
 \eta:=\left(\frac{\xi}{4CM_4}\right)^4,
 \qquad \zeta_{\rm filter}:=\frac{\xi}{4}.
\end{equation}
The contributions from unresolved slow eigenvalues, polynomial filtering, and numerical implementation satisfy
\[
 C\eta^{1/4}M_4+\zeta_{\rm filter}+\frac\xi2
 =\frac\xi4+\frac\xi4+\frac\xi2=\xi,
 \qquad \log(1/\eta)=O(\log Q).
\]

To apply the mixing time theorem, we need lower and upper bounds for the spectrum of $A({\bf x})$ that work for every encoded sample tuple. The radii bounds for the convex body in~\eqref{eq:pencil-support-balls} and~\cref{lem:log-affine-covariance-floor} give
\[
 \lambda_{\min}(\Sigma_i)\ge\underline\lambda
 :=\frac{c(1/4)^2}{[D+2a_0(3\kappa+1/4)]^2},
 \qquad \Sigma_i=\Cov_{\pi_{a_i}}(X).
\]
The tilt bound $2a_0$ also covers the additional tilt parameters used in \cref{lem:cooling-moments}. Choose the largest dyadic $\lambda_-\le\underline\lambda/2$, and a fixed coordinate box containing the pencil and cone with enclosing radius $R_{\rm enc}=O(\kappa\sqrt D)$. Then
\[
 \underline\lambda/4<\lambda_-\le\underline\lambda/2,
 \qquad \lambda_+:=R_{\rm enc}^2+\lambda_-.
\]
As proved in \cref{sec:coherent-rounding}, these choices give $\lambda_-\Id\preceq A(\mathbf x)\preceq\lambda_+\Id$ for every tuple, and preserve the covariance comparison whenever \eqref{eq:good-covariance} holds.

Set
\begin{equation}\label{eq:precision-quadrature}
 \zeta:=\min\left\{\frac{\varepsilon}{16},
                       \frac{{\varepsilon_{\Cov}}}{2L}\right\}.
\end{equation}
Choose $h$ to be the largest dyadic spacing satisfying \eqref{eq:h_bound}, with the covariance failure budget $\varepsilon_{\Cov}$ from~\eqref{eq:precision-covariance}, the mixing accuracy $\eta$ from \eqref{eq:precision-mixing}, and normalization constant accuracy $\zeta$. Use geometric bounds valid for the pencil, the initial cone, and the terminal slab, and include all tilts required by \cref{lem:cooling-moments}. The proof of \cref{thm:finite-metropolis} then supplies the mixing bound and the quadrature and covariance-transfer conditions. In particular, $\varepsilon_{\rm quad}\le\zeta\le\varepsilon/16$ gives at most $\varepsilon/8$ discretization error in the logarithm of the volume, while \cref{lem:finite-block-transfer} gives \eqref{eq:good-covariance} with failure probability ${\varepsilon_{\Cov}}$. Since all geometric parameters and inverse accuracies are polynomial in $Q$, this choice satisfies
\[
 1/h\le\poly(Q),
 \qquad J=O(\log Q).
\]
\cref{cor:finite-reflection-cost-main} therefore implements the required reflection with error at most $\xi$ using $\widetilde O(D/\psi)$ membership queries.

\paragraph{Finite implementation and total error.}
Since only constant accuracy is needed for the displacement distribution, we can take $\varepsilon_{\rm dir}=1/1000$ in~\cref{lem:finite-direction-preparation}. We have
\[
 \chi_A(\delta_j)=2+\frac{\delta_j\|T_A\|_{\op}}h\le\poly(Q).
\]
The displacement seed and rescaled position registers use $O(D\log Q)$ qubits. Covariance representations and arithmetic workspace have polynomial size in $Q$ and the original input bit length. 

The circuits in \cref{lem:logarithmic-phase,lem:transport}, together with the counts established above, give a polynomial upper bound $N_{\rm gate}$ on the remaining numerical operations. Apply~\cref{lem:marks-and-output}~\textnormal{(i)} with $\nu=\beta/(4N_{\rm gate})$ for each $V_{Y_i^\pm}$. For each $\ket{\nu_\Lambda}$-reflection, apply \cref{lem:finite-cone-preparation} with preparation error ${\varepsilon_{\rm cone}}=\beta/(8N_{\rm gate})$, giving reflection error at most $2{\varepsilon_{\rm cone}}=\beta/(4N_{\rm gate})$. The remaining one-qubit gates can be synthesized to the same operator accuracy~\cite{DawsonNielsen06}. Thus, the targeted numerical implementation error in the last row of \cref{tab:implementation-error-budget} is achievable with polynomial computation and logarithmic additional bit complexity.

Combining \cref{lem:finite-cone-preparation,lem:endpoint-resources}, prepare each initial qsample to trace-distance error at most $\beta/(4n_0)$, including numerical implementation. Apply \cref{lem:transport} with conversion error $\beta/(4T)$ and exact reflection oracles. By \cref{lem:composition}, the resulting bounds are
\[
 e_{\rm init}\le n_0\frac{\beta}{4n_0}=\frac{\beta}{4},
 \qquad
 e_{\rm conv}\le T\frac{\beta}{4T}=\frac{\beta}{4},
\]
which establish the first two rows of \cref{tab:implementation-error-budget}.

Hence, \cref{lem:coherent-reflection-error,lem:composition} give
\begin{align*}
 \dtr(\rho_{\rm impl},\rho_{\rm ref})
 \leq &~ e_{\rm init}+e_{\rm conv}+N\bigl(\xi+2\sqrt{{\varepsilon_{\Cov}}}\bigr)+e_{\rm gate}\\
 \leq &~ \frac\beta4+\frac\beta4+\frac\beta4+\frac\beta4=\beta.
\end{align*}
Therefore, the finite-precision implementation decreases the success probability by at most $\beta=1/100$, with polynomial additional computation.

\phantomsection 
\addcontentsline{toc}{section}{References}
\bibliographystyle{alpha}
\bibliography{references}

@article{Simonovits03,
  title={How to compute the volume in high dimension?},
  author={Simonovits, Mikl{\'o}s},
  journal={Mathematical programming},
  volume={97},
  number={1},
  pages={337--374},
  year={2003},
  publisher={Springer}
}

@inproceedings{ApplegateKannan91,
  title={Sampling and integration of near log-concave functions},
  author={Applegate, David and Kannan, Ravi},
  booktitle={Proceedings of the twenty-third annual ACM symposium on Theory of computing},
  pages={156--163},
  year={1991}
}

@inproceedings{KookVempala25,
  title={Sampling and integration of logconcave functions by algorithmic diffusion},
  author={Kook, Yunbum and Vempala, Santosh S},
  booktitle={Proceedings of the 57th Annual ACM Symposium on Theory of Computing},
  pages={924--932},
  year={2025}
}

@article{BlumeKohout10,
  title={Optimal, reliable estimation of quantum states},
  author={Blume-Kohout, Robin},
  journal={New Journal of Physics},
  volume={12},
  number={4},
  pages={043034},
  year={2010}
}

@article{BrightwellWinkler91,
  title={Counting linear extensions},
  author={Brightwell, Graham and Winkler, Peter},
  journal={Order},
  volume={8},
  number={3},
  pages={225--242},
  year={1991},
  publisher={Springer}
}

@article{ShangEtAl13,
  title={Optimal error regions for quantum state estimation},
  author={Shang, Jiangwei and Ng, Hui Khoon and Sehrawat, Arun and Li, Xikun and Englert, Berthold-Georg},
  journal={New Journal of Physics},
  volume={15},
  number={12},
  pages={123026},
  year={2013},
  publisher={IOP Publishing}
}

@article{ShangEtAl15,
  title={Monte Carlo sampling from the quantum state space. I},
  author={Shang, Jiangwei and Seah, Yi-Lin and Ng, Hui Khoon and Nott, David John and Englert, Berthold-Georg},
  journal={New Journal of Physics},
  volume={17},
  number={4},
  pages={043017},
  year={2015},
  publisher={IOP Publishing}
}

@article{BertsimasVempala04,
  title={Solving convex programs by random walks},
  author={Bertsimas, Dimitris and Vempala, Santosh},
  journal={Journal of the ACM (JACM)},
  volume={51},
  number={4},
  pages={540--556},
  year={2004},
  publisher={ACM New York, NY, USA}
}

@article{DyerKannanMount97,
  title={Sampling contingency tables},
  author={Dyer, Martin and Kannan, Ravi and Mount, John},
  journal={Random Structures \& Algorithms},
  volume={10},
  number={4},
  pages={487--506},
  year={1997},
  publisher={Wiley Online Library}
}

@article{GelmanMeng98,
  title={Simulating normalizing constants: From importance sampling to bridge sampling to path sampling},
  author={Gelman, Andrew and Meng, Xiao-Li},
  journal={Statistical science},
  pages={163--185},
  year={1998},
  publisher={JSTOR}
}

@incollection{Genz92,
  title={Numerical computation of multivariate normal probabilities using bivariate conditioning},
  author={Genz, Alan and Trinh, Giang},
  booktitle={Monte Carlo and Quasi-Monte Carlo Methods: MCQMC, Leuven, Belgium, April 2014},
  pages={289--302},
  year={2016},
  publisher={Springer}
}

@article{CCHLW23,
  author  = {Shouvanik Chakrabarti and Andrew M. Childs and Shih-Han Hung and Tongyang Li and Chunhao Wang and Xiaodi Wu},
  title   = {Quantum Algorithm for Estimating Volumes of Convex Bodies},
  journal = {ACM Transactions on Quantum Computing},
  volume  = {4},
  number  = {3},
  pages   = {20:1--20:60},
  year    = {2023},
  doi     = {10.1145/3588579},
  eprint  = {1908.03903},
  archivePrefix = {arXiv}
}

@article{JiaEtAl26,
  title={Reducing isotropy and volume to KLS: Faster rounding and volume algorithms},
  author={Jia, He and Laddha, Aditi and Lee, Yin Tat and Vempala, Santosh},
  journal={Journal of the ACM},
  volume={73},
  number={2},
  pages={1--21},
  year={2026},
  publisher={ACM New York, NY}
}

@article{Bennett76,
  title={Efficient estimation of free energy differences from Monte Carlo data},
  author={Bennett, Charles H},
  journal={Journal of Computational Physics},
  volume={22},
  number={2},
  pages={245--268},
  year={1976},
  publisher={Elsevier}
}

@article{Childs26Matrix,
  title={Quantum lower bounds for convex optimization and real matrix-vector query problems},
  author={Childs, Andrew M},
  journal={arXiv preprint arXiv:2609.05679},
  year={2026}
}

@article{BourgadeEtAl07,
  title={The characteristic polynomial on compact groups with Haar measure: some equalities in law},
  author={Bourgade, Paul and Nikeghbali, Ashkan and Rouault, Alain},
  journal={Comptes Rendus. Math{\'e}matique},
  volume={345},
  number={4},
  pages={229--232},
  year={2007}
}

@article{BealsEtAl01,
  title={Quantum lower bounds by polynomials},
  author={Beals, Robert and Buhrman, Harry and Cleve, Richard and Mosca, Michele and De Wolf, Ronald},
  journal={Journal of the ACM (JACM)},
  volume={48},
  number={4},
  pages={778--797},
  year={2001},
  publisher={ACM New York, NY, USA}
}

@inproceedings{CornelissenHamoudi23,
  author    = {Arjan Cornelissen and Yassine Hamoudi},
  title     = {A Sublinear-Time Quantum Algorithm for Approximating Partition Functions},
  booktitle = {Proceedings of the 2023 Annual ACM--SIAM Symposium on Discrete Algorithms},
  pages     = {1245--1264},
  year      = {2023},
  doi       = {10.1137/1.9781611977554.ch46},
  eprint    = {2207.08643},
  archivePrefix = {arXiv}
}

@article{DFK91,
  author  = {Martin Dyer and Alan Frieze and Ravi Kannan},
  title   = {A Random Polynomial-Time Algorithm for Approximating the Volume of Convex Bodies},
  journal = {Journal of the ACM},
  volume  = {38},
  number  = {1},
  pages   = {1--17},
  year    = {1991},
  doi     = {10.1145/102782.102783}
}

@article{LovaszSimonovits93,
  author  = {L{\'a}szl{\'o} Lov{\'a}sz and Mikl{\'o}s Simonovits},
  title   = {Random Walks in a Convex Body and an Improved Volume Algorithm},
  journal = {Random Structures \& Algorithms},
  volume  = {4},
  number  = {4},
  pages   = {359--412},
  year    = {1993},
  doi     = {10.1002/rsa.3240040402}
}

@article{KLS97,
  author  = {Ravi Kannan and L{\'a}szl{\'o} Lov{\'a}sz and Mikl{\'o}s Simonovits},
  title   = {Random Walks and an {$O^*(n^5)$} Volume Algorithm for Convex Bodies},
  journal = {Random Structures \& Algorithms},
  volume  = {11},
  number  = {1},
  pages   = {1--50},
  year    = {1997},
  doi     = {10.1002/(SICI)1098-2418(199708)11:1<1::AID-RSA1>3.0.CO;2-X}
}

@article{Lovasz99,
  author  = {L{\'a}szl{\'o} Lov{\'a}sz},
  title   = {Hit-and-Run Mixes Fast},
  journal = {Mathematical Programming},
  volume  = {86},
  number  = {3},
  pages   = {443--461},
  year    = {1999},
  doi     = {10.1007/s101070050100}
}

@article{LovaszVempala06Corner,
  author  = {L{\'a}szl{\'o} Lov{\'a}sz and Santosh Vempala},
  title   = {Hit-and-Run from a Corner},
  journal = {SIAM Journal on Computing},
  volume  = {35},
  number  = {4},
  pages   = {985--1005},
  year    = {2006},
  doi     = {10.1137/S009753970544727X}
}

@inproceedings{LovaszVempala06Fast,
  author    = {L{\'a}szl{\'o} Lov{\'a}sz and Santosh Vempala},
  title     = {Fast Algorithms for Logconcave Functions: Sampling, Rounding, Integration and Optimization},
  booktitle = {47th Annual IEEE Symposium on Foundations of Computer Science},
  pages     = {57--68},
  year      = {2006},
  doi       = {10.1109/FOCS.2006.28}
}

@article{LovaszVempala07,
  author  = {L{\'a}szl{\'o} Lov{\'a}sz and Santosh Vempala},
  title   = {The Geometry of Logconcave Functions and Sampling Algorithms},
  journal = {Random Structures \& Algorithms},
  volume  = {30},
  number  = {3},
  pages   = {307--358},
  year    = {2007},
  doi     = {10.1002/rsa.20135}
}

@article{CousinsVempala18,
  author  = {Ben Cousins and Santosh Vempala},
  title   = {Gaussian Cooling and {$O^*(n^3)$} Algorithms for Volume and Gaussian Volume},
  journal = {SIAM Journal on Computing},
  volume  = {47},
  number  = {3},
  pages   = {1237--1273},
  year    = {2018},
  doi     = {10.1137/15M1054250}
}

@article{ChenEldan26,
  title={Hit-and-run mixing via localization schemes},
  author={Chen, Yuansi and Eldan, Ronen},
  journal={Discrete \& Computational Geometry},
  volume={75},
  number={3},
  pages={747--794},
  year={2026},
  publisher={Springer}
}

@article{Klartag23,
  author  = {Bo'az Klartag},
  title   = {Logarithmic Bounds for Isoperimetry and Slices of Convex Sets},
  journal = {Ars Inveniendi Analytica},
  volume  = {2023},
  number  = {4},
  pages   = {1--17},
  year    = {2023},
  doi     = {10.15781/jsjy-0b06},
  eprint  = {2303.14938},
  archivePrefix = {arXiv}
}

@article{KalaiVempala06,
  title={Simulated annealing for convex optimization},
  author={Kalai, Adam Tauman and Vempala, Santosh},
  journal={Mathematics of Operations Research},
  volume={31},
  number={2},
  pages={253--266},
  year={2006},
  publisher={INFORMS}
}

@article{KLS95,
  title={Isoperimetric problems for convex bodies and a localization lemma},
  author={Kannan, Ravi and Lov{\'a}sz, L{\'a}szl{\'o} and Simonovits, Mikl{\'o}s},
  journal={Discrete \& Computational Geometry},
  volume={13},
  number={3},
  pages={541--559},
  year={1995},
  publisher={Springer}
}

@article{kook2026spectral,
  title={Spectral Gaps of Hit-and-Run and Coordinate Hit-and-Run},
  author={Kook, Yunbum and Vempala, Santosh S},
  journal={arXiv preprint arXiv:2608.16878},
  year={2026}
}

@article{chen21,
  title={An almost constant lower bound of the isoperimetric coefficient in the KLS conjecture},
  author={Chen, Yuansi},
  journal={Geometric and Functional Analysis},
  volume={31},
  number={1},
  pages={34--61},
  year={2021},
  publisher={Springer}
}

@article{CornelissenApersGribling26,
    author = {Arjan Cornelissen and Simon Apers and Sander Gribling},
    title = {Improved quantum volume estimation with transducers and amortized quantum walks},
    journal = {arXiv preprint arXiv:2609.40217},
    year = {2026}
}

@article{jambulapati2022slightly,
  title={A slightly improved bound for the KLS constant},
  author={Jambulapati, Arun and Lee, Yin Tat and Vempala, Santosh S},
  journal={arXiv preprint arXiv:2208.11644},
  year={2022}
}

@article{klartag2022bourgain,
  title={Bourgain’s slicing problem and KLS isoperimetry up to polylog},
  author={Klartag, Bo’az and Lehec, Joseph},
  journal={Geometric and functional analysis},
  volume={32},
  number={5},
  pages={1134--1159},
  year={2022},
  publisher={Springer}
}

@article{vempala2026nearly,
  title={A Nearly Quadratic Lower Bound for Linear Optimization over Convex Bodies in the Membership Oracle Model},
  author={Vempala, Santosh S},
  journal={arXiv preprint arXiv:2609.30215},
  year={2026}
}

@article{rademacher2008dispersion,
  title={Dispersion of mass and the complexity of randomized geometric algorithms},
  author={Rademacher, Luis and Vempala, Santosh},
  journal={Advances in Mathematics},
  volume={219},
  number={3},
  pages={1037--1069},
  year={2008},
  publisher={Elsevier}
}

@InProceedings{CornelissenApersGribling25,
  author =	{Cornelissen, Arjan and Apers, Simon and Gribling, Sander},
  title =	{{How to Compute the Volume in Low Dimension?}},
  booktitle =	{52nd International Colloquium on Automata, Languages, and Programming (ICALP 2025)},
  pages =	{61:1--61:18},
  series =	{Leibniz International Proceedings in Informatics (LIPIcs)},
  ISBN =	{978-3-95977-372-0},
  ISSN =	{1868-8969},
  year =	{2025},
  volume =	{334},
  editor =	{Censor-Hillel, Keren and Grandoni, Fabrizio and Ouaknine, Jo\"{e}l and Puppis, Gabriele},
  publisher =	{Schloss Dagstuhl -- Leibniz-Zentrum f{\"u}r Informatik},
  address =	{Dagstuhl, Germany},
  URL =		{https://drops.dagstuhl.de/entities/document/10.4230/LIPIcs.ICALP.2025.61},
  URN =		{urn:nbn:de:0030-drops-234381},
  doi =		{10.4230/LIPIcs.ICALP.2025.61}
}

@inproceedings{LovaszSimonovits90,
  title={The mixing rate of Markov chains, an isoperimetric inequality, and computing the volume},
  author={Lov{\'a}sz, L{\'a}szl{\'o} and Simonovits, Mikl{\'o}s},
  booktitle={Proceedings [1990] 31st annual symposium on foundations of computer science},
  pages={346--354},
  year={1990},
  organization={IEEE}
}

@article{KookVempalaZhang26,
  title={In-and-Out: Algorithmic Diffusion for Sampling Convex Bodies},
  author={Kook, Yunbum and Vempala, Santosh S and Zhang, Matthew S},
  journal={Random Structures \& Algorithms},
  volume={68},
  number={3},
  pages={e70061},
  year={2026},
  publisher={Wiley Online Library}
}

@InProceedings{OzgulEtAl26Stochastic,
  author =	{Ozgul, Guneykan and Li, Xiantao and Mahdavi, Mehrdad and Wang, Chunhao},
  title =	{{Quantum Speedups for Sampling and Non-Convex Optimization with Stochastic Oracles}},
  booktitle =	{21st Conference on the Theory of Quantum Computation, Communication and Cryptography (TQC 2026)},
  pages =	{8:1--8:25},
  series =	{Leibniz International Proceedings in Informatics (LIPIcs)},
  ISBN =	{978-3-95977-439-0},
  ISSN =	{1868-8969},
  year =	{2026},
  volume =	{389},
  editor =	{Arnon, Rotem and Harrow, Aram W.},
  publisher =	{Schloss Dagstuhl -- Leibniz-Zentrum f{\"u}r Informatik},
  address =	{Dagstuhl, Germany},
  URL =		{https://drops.dagstuhl.de/entities/document/10.4230/LIPIcs.TQC.2026.8},
  URN =		{urn:nbn:de:0030-drops-273053},
  doi =		{10.4230/LIPIcs.TQC.2026.8}
}

@article{LiuWangJi26LocalSampling,
  title={Quantum Speedups for Log-Concave Sampling from Local Structure},
  author={Liu, Chenghua and Wang, Qisheng and Ji, Zhengfeng},
  journal={arXiv preprint arXiv:2609.20253},
  year={2026}
}

@article{SongZhang26KLS,
      title={{An $O(4^{\log^* n})$ Bound for the KLS Constant}}, 
      author={Zhao Song and Xinzhi Zhang},
      year={2026},
      journal={arXiv preprint arXiv:2610.01447},
      url={https://arxiv.org/abs/2610.01447}
}

@article{OlivucciEtAl26,
  title={Provable Quantum--Classical Separation for Continuous Gibbs Sampling},
  author={Olivucci, Enrico and Sobchuk, Mariia and Hoque, Sehmimul and Hnybida, Jeffrey and Kim, Kyungho W and Shayeghi, Ala and Ronagh, Pooya},
  journal={arXiv preprint arXiv:2608.24527},
  year={2026}
}

@InProceedings{MotamediRonagh24,
  title = 	 {{G}ibbs Sampling of Continuous Potentials on a Quantum Computer},
  author =       {Motamedi, Arsalan and Ronagh, Pooya},
  booktitle = 	 {Proceedings of the 41st International Conference on Machine Learning},
  pages = 	 {36322--36371},
  year = 	 {2024},
  editor = 	 {Salakhutdinov, Ruslan and Kolter, Zico and Heller, Katherine and Weller, Adrian and Oliver, Nuria and Scarlett, Jonathan and Berkenkamp, Felix},
  volume = 	 {235},
  series = 	 {Proceedings of Machine Learning Research},
  month = 	 {21--27 Jul},
  publisher =    {PMLR},
  url = 	 {https://proceedings.mlr.press/v235/motamedi24a.html}
}

@article{LengEtAl26,
  title={Operator-level quantum acceleration of non-logconcave sampling},
  author={Leng, Jiaqi and Ding, Zhiyan and Chen, Zherui and Lin, Lin},
  journal={Proceedings of the National Academy of Sciences},
  volume={123},
  number={8},
  pages={e2512789123},
  year={2026},
  publisher={National Academy of Sciences}
}

@InProceedings{OzgulEtAl24,
  title = 	 {Stochastic Quantum Sampling for Non-Logconcave Distributions and Estimating Partition Functions},
  author =       {Ozgul, Guneykan and Li, Xiantao and Mahdavi, Mehrdad and Wang, Chunhao},
  booktitle = 	 {Proceedings of the 41st International Conference on Machine Learning},
  pages = 	 {38953--38982},
  year = 	 {2024},
  editor = 	 {Salakhutdinov, Ruslan and Kolter, Zico and Heller, Katherine and Weller, Adrian and Oliver, Nuria and Scarlett, Jonathan and Berkenkamp, Felix},
  volume = 	 {235},
  series = 	 {Proceedings of Machine Learning Research},
  month = 	 {21--27 Jul},
  publisher =    {PMLR},
  url = 	 {https://proceedings.mlr.press/v235/ozgul24a.html}
}

@article{LiZhang22,
  title={Quantum speedups of optimizing approximately convex functions with applications to logarithmic regret stochastic convex bandits},
  author={Li, Tongyang and Zhang, Ruizhe},
  journal={Advances in Neural Information Processing Systems},
  volume={35},
  pages={3152--3164},
  year={2022}
}

@InProceedings{pmlr-v247-gatmiry24a,
  title = 	 {Sampling Polytopes with Riemannian HMC: Faster Mixing via the Lewis Weights Barrier},
  author =       {Gatmiry, Khashayar and Kelner, Jonathan and Vempala, Santosh S.},
  booktitle = 	 {Proceedings of Thirty Seventh Conference on Learning Theory},
  pages = 	 {1796--1881},
  year = 	 {2024},
  editor = 	 {Agrawal, Shipra and Roth, Aaron},
  volume = 	 {247},
  series = 	 {Proceedings of Machine Learning Research},
  month = 	 {30 Jun--03 Jul},
  publisher =    {PMLR},
  url = 	 {https://proceedings.mlr.press/v247/gatmiry24a.html}
}

@inproceedings{LeeVempala18,
  title={Convergence rate of Riemannian Hamiltonian Monte Carlo and faster polytope volume computation},
  author={Lee, Yin Tat and Vempala, Santosh S},
  booktitle={Proceedings of the 50th Annual ACM SIGACT Symposium on Theory of Computing},
  pages={1115--1121},
  year={2018}
}

@article{DyerFrieze91,
  title={Computing the volume of convex bodies: a case where randomness provably helps},
  author={Dyer, Martin and Frieze, Alan},
  journal={Probabilistic combinatorics and its applications},
  volume={44},
  number={123-170},
  pages={0754--68052},
  year={1991},
  publisher={American Mathematical Society Providence, RI}
}

@INPROCEEDINGS{ArunachalamEtAl22,
  author={Arunachalam, Srinivasan and Havlicek, Vojtech and Nannicini, Giacomo and Temme, Kristan and Wocjan, Pawel},
  booktitle={2021 IEEE International Conference on Quantum Computing and Engineering (QCE)}, 
  title={Simpler (Classical) and Faster (Quantum) Algorithms for Gibbs Partition Functions}, 
  year={2021},
  volume={},
  number={},
  pages={112-122},
  doi={10.1109/QCE52317.2021.00027}
}

@article{ChildsEtAl22,
  title={Quantum algorithms for sampling log-concave distributions and estimating normalizing constants},
  author={Childs, Andrew M and Li, Tongyang and Liu, Jin-Peng and Wang, Chunhao and Zhang, Ruizhe},
  journal={Advances in Neural Information Processing Systems},
  volume={35},
  pages={23205--23217},
  year={2022}
}

@inproceedings{HarrowWei20,
  author       = {Aram W. Harrow and
                  Annie Y. Wei},
  editor       = {Shuchi Chawla},
  title        = {Adaptive Quantum Simulated Annealing for Bayesian Inference and Estimating
                  Partition Functions},
  booktitle    = {Proceedings of the 2020 {ACM-SIAM} Symposium on Discrete Algorithms,
                  {SODA} 2020, Salt Lake City, UT, USA, January 5-8, 2020},
  pages        = {193--212},
  publisher    = {{SIAM}},
  year         = {2020},
  url          = {https://doi.org/10.1137/1.9781611975994.12},
  doi          = {10.1137/1.9781611975994.12}
}

@article{WCNA09,
  title={Quantum algorithm for approximating partition functions},
  author={Wocjan, Pawel and Chiang, Chen-Fu and Nagaj, Daniel and Abeyesinghe, Anura},
  journal={Physical Review A—Atomic, Molecular, and Optical Physics},
  volume={80},
  number={2},
  pages={022340},
  year={2009},
  publisher={APS}
}

@article{Richter07,
  title={Quantum speedup of classical mixing processes},
  author={Richter, Peter C},
  journal={Physical Review A—Atomic, Molecular, and Optical Physics},
  volume={76},
  number={4},
  pages={042306},
  year={2007},
  publisher={APS}
}

@inproceedings{aharonov2003adiabatic,
  title={Adiabatic quantum state generation and statistical zero knowledge},
  author={Aharonov, Dorit and Ta-Shma, Amnon},
  booktitle={Proceedings of the thirty-fifth annual ACM symposium on Theory of computing},
  pages={20--29},
  year={2003}
}

@article{LovaszVempala06Annealing,
  title={{Simulated annealing in convex bodies and an $O^*(n^4)$ volume algorithm}},
  author={Lov{\'a}sz, L{\'a}szl{\'o} and Vempala, Santosh},
  journal={Journal of Computer and System Sciences},
  volume={72},
  number={2},
  pages={392--417},
  year={2006},
  publisher={Elsevier}
}

@article{Letwin26,
  author  = {Brayden Letwin},
  title   = {The {KLS} Constant Is {$O(\log^{1/4} n)$}},
  journal = {arXiv preprint arXiv:2607.24164},
  year    = {2026},
  eprint  = {2607.24164},
  archivePrefix = {arXiv}
}

@article{Zhang26HitRun,
  title={Hit-and-Run Mixes as Fast as the Ball Walk},
  author={Zhang, Ruizhe},
  journal={arXiv preprint arXiv:2608.13487},
  year={2026}
}

@inproceedings{Szegedy04,
  author    = {Mario Szegedy},
  title     = {Quantum Speed-Up of Markov Chain Based Algorithms},
  booktitle = {45th Annual IEEE Symposium on Foundations of Computer Science},
  pages     = {32--41},
  year      = {2004},
  doi       = {10.1109/FOCS.2004.53}
}

@inproceedings{GSLW19,
  author    = {Andr{\'a}s Gily{\'e}n and Yuan Su and Guang Hao Low and Nathan Wiebe},
  title     = {Quantum Singular Value Transformation and Beyond: Exponential Improvements for Quantum Matrix Arithmetics},
  booktitle = {Proceedings of the 51st Annual ACM SIGACT Symposium on Theory of Computing},
  pages     = {193--204},
  year      = {2019},
  doi       = {10.1145/3313276.3316366},
  eprint    = {1806.01838},
  archivePrefix = {arXiv}
}

@article{YoderLowChuang14,
  author  = {Theodore J. Yoder and Guang Hao Low and Isaac L. Chuang},
  title   = {Fixed-Point Quantum Search with an Optimal Number of Queries},
  journal = {Physical Review Letters},
  volume  = {113},
  pages   = {210501},
  year    = {2014},
  doi     = {10.1103/PhysRevLett.113.210501},
  eprint  = {1409.3305},
  archivePrefix = {arXiv}
}

@article{WocjanAbeyesinghe08,
  author  = {Pawel Wocjan and Anura Abeyesinghe},
  title   = {Speedup via Quantum Sampling},
  journal = {Physical Review A},
  volume  = {78},
  pages   = {042336},
  year    = {2008},
  doi     = {10.1103/PhysRevA.78.042336},
  eprint  = {0804.4259},
  archivePrefix = {arXiv}
}

@article{GroverRudolph02,
  author  = {Lov K. Grover and Terry Rudolph},
  title   = {Creating Superpositions that Correspond to Efficiently Integrable Probability Distributions},
  journal = {arXiv preprint quant-ph/0208112},
  year    = {2002},
  eprint  = {quant-ph/0208112},
  archivePrefix = {arXiv}
}

@article{Jerrum06,
  title={On the approximation of one Markov chain by another},
  author={Jerrum, Mark},
  journal={Probability theory and related fields},
  volume={135},
  number={1},
  pages={1--14},
  year={2006},
  publisher={Springer}
}

@article{Montanaro15,
  title={Quantum speedup of Monte Carlo methods},
  author={Montanaro, Ashley},
  journal={Proceedings. Mathematical, Physical, and Engineering Sciences/The Royal Society},
  volume={471},
  number={2181},
  pages={20150301},
  year={2015}
}

@article{HaenerEtAl18,
  title={Optimizing quantum circuits for arithmetic},
  author={H{\"a}ner, Thomas and Roetteler, Martin and Svore, Krysta M},
  journal={arXiv preprint arXiv:1805.12445},
  year={2018}
}

@article{BhaskarEtAl16,
  author       = {Mihir K. Bhaskar and
                  Stuart Hadfield and
                  Anargyros Papageorgiou and
                  Iasonas Petras},
  title        = {Quantum algorithms and circuits for scientific computing},
  journal      = {Quantum Inf. Comput.},
  volume       = {16},
  number       = {3{\&}4},
  pages        = {197--236},
  year         = {2016},
  url          = {https://doi.org/10.26421/QIC16.3-4-2},
  doi          = {10.26421/QIC16.3-4-2}
}

@article{GuoHigham06,
  title={{A schur--newton method for the matrix pth root and its inverse}},
  author={Guo, Chun-Hua and Higham, Nicholas J},
  journal={SIAM Journal on Matrix Analysis and Applications},
  volume={28},
  number={3},
  pages={788--804},
  year={2006},
  publisher={SIAM}
}

@article{DawsonNielsen06,
  author       = {Christopher M. Dawson and
                  Michael A. Nielsen},
  title        = {The Solovay-Kitaev algorithm},
  journal      = {Quantum Inf. Comput.},
  volume       = {6},
  number       = {1},
  pages        = {81--95},
  year         = {2006},
  url          = {https://doi.org/10.26421/QIC6.1-6},
  doi          = {10.26421/QIC6.1-6}
}

@article{BoxMuller58,
  title={A note on the generation of random normal deviates},
  author={Box, George EP and Muller, Mervin E},
  journal={The annals of mathematical statistics},
  volume={29},
  number={2},
  pages={610--611},
  year={1958},
  publisher={Institute of Mathematical Statistics}
}

@article{Bennett73,
  title={Logical reversibility of computation},
  author={Bennett, Charles H},
  journal={IBM journal of Research and Development},
  volume={17},
  number={6},
  pages={525--532},
  year={1973},
  publisher={IBM}
}

@article{AdamczakEtAl10,
  title={Quantitative estimates of the convergence of the empirical covariance matrix in log-concave ensembles},
  author={Adamczak, Rados{\l}aw and Litvak, Alexander and Pajor, Alain and Tomczak-Jaegermann, Nicole},
  journal={Journal of the American Mathematical Society},
  volume={23},
  number={2},
  pages={535--561},
  year={2010}
}

@misc{Tao09Interpolation,
  author = {Terence Tao},
  title = {{245C, Notes 1: Interpolation of $L^p$ spaces}},
  year = {2009},
  howpublished = {Graduate real analysis lecture notes},
  url = {https://terrytao.wordpress.com/2009/03/30/245c-notes-1-interpolation-of-lp-spaces/}
}

@article{kitaev1995quantum,
  title={Quantum measurements and the Abelian stabilizer problem},
  author={Kitaev, A Yu},
  journal={arXiv preprint quant-ph/9511026},
  year={1995}
}

@article{MengWong96,
  title={Simulating ratios of normalizing constants via a simple identity: a theoretical exploration},
  author={Meng, Xiao-Li and Wong, Wing Hung},
  journal={Statistica Sinica},
  pages={831--860},
  year={1996},
  publisher={JSTOR}
}

@article{BaranyFuredi87,
  author       = {Imre B{\'{a}}r{\'{a}}ny and
                  Zolt{\'{a}}n F{\"{u}}redi},
  title        = {Computing the Volume is Difficult},
  journal      = {Discret. Comput. Geom.},
  volume       = {2},
  pages        = {319--326},
  year         = {1987},
  url          = {https://doi.org/10.1007/BF02187886},
  doi          = {10.1007/BF02187886}
}

@article{DyerFrieze88,
  title={On the complexity of computing the volume of a polyhedron},
  author={Dyer, Martin E. and Frieze, Alan M.},
  journal={SIAM Journal on Computing},
  volume={17},
  number={5},
  pages={967--974},
  year={1988},
  publisher={SIAM}
}

@article{Stanley86,
  title={Two poset polytopes},
  author={Stanley, Richard P},
  journal={Discrete \& Computational Geometry},
  volume={1},
  number={1},
  pages={9--23},
  year={1986},
  publisher={Springer}
}

@article{Prekopa73,
  title={On logarithmic concave measures and functions},
  author={Pr{\'e}kopa, Andr{\'a}s},
  journal={Acta Sci. Math.},
  volume={34},
  pages={335},
  year={1973}
}

@article{BrassardEtAl02,
  author  = {Gilles Brassard and Peter H{\o}yer and Michele Mosca and Alain Tapp},
  title   = {Quantum Amplitude Amplification and Estimation},
  journal = {Contemporary Mathematics},
  volume  = {305},
  pages   = {53--74},
  year    = {2002},
  doi     = {10.1090/conm/305/05215},
  eprint  = {quant-ph/0005055},
  archivePrefix = {arXiv}
}
\end{document}